%% file: main.tex
\documentclass[acmsmall,screen,nonacm]{acmart}
\setcopyright{none}
\acmDOI{}
\acmISBN{}

\usepackage{mathpartir}
\usepackage{booktabs}
\usepackage{longtable}
\usepackage{xcolor}
\usepackage{colortbl}
\usepackage{listings}
\usepackage{xspace}
\usepackage{tikz}
\input{figures/comfy-nodes}
\usetikzlibrary{positioning,arrows.meta,patterns}
\usepackage{graphicx}
\usepackage{placeins}
\usepackage{subcaption}
\usepackage{enumitem}
\definecolor{ctsilent}{HTML}{C0392B}
\definecolor{ctsem}{HTML}{D85A30}
\definecolor{dslkeyword}{HTML}{1F5FBF} \definecolor{dsltype}{HTML}{0F7A6B}    \definecolor{dslnode}{HTML}{7A3E9D}    \definecolor{dslfunc}{HTML}{9A5B00}    \definecolor{dslcomment}{HTML}{6B7280} \definecolor{ctType}{HTML}{16665B}     \definecolor{ctNode}{HTML}{6A3688}     \definecolor{ctBrand}{HTML}{9C3B57}

\newcommand{\tool}{\textsc{ComfyTy}\xspace}
\newcommand{\parh}[1]{\smallskip\noindent\textbf{#1}}
\newcommand{\takeaway}[2]{\par\smallskip\noindent
  \begin{tikzpicture}
    \node[draw=ctType!55, line width=0.7pt, rounded corners=3pt,
          fill=ctType!7, inner xsep=9pt, inner ysep=7pt, align=left,
          text width=\dimexpr\columnwidth-24pt\relax, font=\small]
      (tkbox) {\textcolor{ctType}{\textbf{#1.}}\hspace{0.35em}#2};
    \fill[ctType, rounded corners=1pt]
      ([xshift=1.2pt]tkbox.north west) rectangle
      ([xshift=4.2pt]tkbox.south west);
  \end{tikzpicture}\par\smallskip}
\newcommand{\reqdef}[1]{\phantomsection\label{req:R#1}(R#1)}
\newcommand{\req}[1]{\hyperref[req:R#1]{R#1}}
\newcommand{\tycon}[1]{\ensuremath{\textcolor{ctType}{\mathsf{#1}}}}
\newcommand{\brand}[1]{\ensuremath{\textcolor{ctBrand}{\mathsf{#1}}}}
\newcommand{\nodename}[1]{\textcolor{ctNode}{\texttt{#1}}}
\newcommand{\tyImage}[3]{\tycon{Image}(#1,#2,#3)}
\newcommand{\tyLatent}[4]{\tycon{Latent}(#1,#2,#3,#4)}
\newcommand{\tyMask}[2]{\tycon{Mask}(#1,#2)}
\newcommand{\tyCond}[1]{\tycon{Cond}(#1)}
\newcommand{\tyModel}[2]{\tycon{Model}(#1,#2)}
\newcommand{\tyClip}[1]{\tycon{Clip}(#1)}
\newcommand{\tyVae}[1]{\tycon{VAE}(#1)}
\newcommand{\tyLora}[1]{\tycon{Lora}(#1)}
\newcommand{\tyCtrl}[2]{\tycon{Ctrl}(#1,#2)}
\newcommand{\tyInt}[1]{\tycon{Int}(#1)}
\newcommand{\tyDyn}{\mathord{?}}
\newcommand{\chfn}{\mathit{ch}}
\newcommand{\archfn}{\mathit{arch}}
\newcommand{\consistent}{\sim}
\newcommand{\wf}[1]{\vdash #1\ \mathsf{ok}}

\newenvironment{proofsketch}{\begin{proof}[Proof sketch]}{\end{proof}}

\lstdefinelanguage{schemadsl}{
  keywords={node, forall, where, widget, in, out},
  keywordstyle=\bfseries\color{dslkeyword},
  keywords=[2]{Image, Latent, Mask, Cond, Model, Clip, VAE, Lora, Ctrl, Int},
  keywordstyle=[2]\color{dsltype},
  keywords=[3]{sf, ch, arch},
  keywordstyle=[3]\color{dslfunc},
  emph={VAEEncode, WanImageToVideo, ControlNetApplyAdvanced},
  emphstyle=\color{dslnode},
  comment=[l]{\#},
  commentstyle=\color{dslcomment},
  basicstyle=\small\ttfamily,
  columns=fullflexible,
  keepspaces=true,
  sensitive=true,
}

\begin{document}

\title{Well-Typed ComfyUI Cannot Go Wrong: Gradual Refinement Types for Generative-AI Workflow Graphs}

\author{Pingchuan Ma}
\authornote{Both authors contributed equally to this research.}
\orcid{0000-0001-7680-2817}
\affiliation{\institution{Zhejiang University of Technology}
  \city{Hangzhou}
  \country{China}
}
\email{pma@zjut.edu.cn}

\author{Zhantong Xue}
\authornotemark[1]
\orcid{0009-0000-3419-9431}
\affiliation{\institution{Hong Kong University of Science and Technology}
  \city{Hong Kong}
  \country{China}
}
\email{zxueai@cse.ust.hk}

\author{Hui Li}
\authornote{Corresponding author.}
\orcid{0009-0008-3821-9434}
\affiliation{\institution{Institute of Software, Chinese Academy of Sciences}
  \city{Beijing}
  \country{China}
}
\email{lihui2012@otcaix.iscas.ac.cn}

\author{Zhaoyu Wang}
\orcid{0009-0009-6892-1264}
\affiliation{\institution{Hong Kong University of Science and Technology}
  \city{Hong Kong}
  \country{China}
}
\email{zwangjz@cse.ust.hk}

\author{Yuguang Zhou}
\orcid{0000-0001-5069-4563}
\affiliation{\institution{Hong Kong University of Science and Technology}
  \city{Hong Kong}
  \country{China}
}
\email{yzhougv@cse.ust.hk}

\author{Shuai Wang}
\orcid{0000-0002-0866-0308}
\affiliation{\institution{Hong Kong University of Science and Technology}
  \city{Hong Kong}
  \country{China}
}
\email{shuaiw@cse.ust.hk}

\author{Xiaoqin Zhang}
\orcid{0000-0003-0958-7285}
\affiliation{\institution{Zhejiang University of Technology}
  \city{Hangzhou}
  \country{China}
}
\email{xqzhang@zjut.edu.cn}

\renewcommand{\shortauthors}{Ma and Xue et al.}

\begin{abstract}
Generative-AI workflows are programs assembled from independently developed
models and operators. Their correctness depends on relationships among model
families, tensor geometry, control modalities, and quantization levels. In
ComfyUI, these relationships can be violated even when every connection is
accepted by its validator, causing expensive runtime crashes or silent failures.
To measure how often such semantic composition failures occur, we classify all
8,766 issues and 1,368 discussion threads of ComfyUI's GitHub repository with an
LLM, checked for cross-model agreement. They account for 8\% of issue reports
and 21\% of issue-tracker runtime failures, and 32\% of them are silent. Because
execution cannot reliably reveal those failures, we present ComfyTy, a gradual
refinement type system that checks component compatibility before execution as a
compositional property of workflow graphs. Each node declares a reusable schema
that states which model families, sizes, modalities, and quantizations it
accepts and produces. Wiring nodes together combines their schemas into
constraints, which ComfyTy solves to find and report incompatibilities. To
evaluate ComfyTy, we first inject faults into official ComfyUI workflows and
execute every resulting graph. ComfyTy detects 66.4\% of the confirmed
injected faults, all of which ComfyUI's own validator admits. On real workflows
across more model families, it raises no false alarm on 765 known-good workflows
and detects 67\% of 108 failures rebuilt from user reports. Finally, in a user
study with 10 participants, its diagnostics raise fault-localization accuracy on
faulty workflows from 35.0\% to 87.5\% and repair accuracy from 30.0\% to
62.5\%.
\end{abstract}



\maketitle

\input{sections/introduction}
\input{sections/study}
\input{sections/calculus}
\input{sections/gradual}
\input{sections/implementation}
\input{sections/evaluation}
\input{sections/discussion}
\input{sections/related}
\input{sections/conclusion}

\bibliographystyle{ACM-Reference-Format}
\bibliography{references}

\clearpage
\appendix
\input{sections/appendix}
\input{sections/appendix-validity}
\input{sections/appendix-evaluation}

\input{sections/appendix-user-study}

\input{sections/appendix-threats}

\end{document}

%% file: figures/comfy-nodes.tex
\definecolor{tMODEL}{HTML}{8E7CC3}\definecolor{tCLIP}{HTML}{C9A800}\definecolor{tVAE}{HTML}{E05050}\definecolor{tCOND}{HTML}{E08A10}\definecolor{tLATENT}{HTML}{D660D6}\definecolor{tIMAGE}{HTML}{3D95DB}\definecolor{tCNET}{HTML}{00A870}\def\cfth{0.30}\def\cfrh{0.25}\def\cfink{black}\def\cfdot{100}\newcommand{\cfbox}[8]{\expandafter\xdef\csname cfw@#1\endcsname{#4}\coordinate (#1) at (#2,#3);
  \fill[white, rounded corners=1.2pt] (#2,#3) rectangle ++(#4,{-\cfth-#5*\cfrh-0.05});
  \fill[#7, rounded corners=1.2pt] (#2,#3) rectangle ++(#4,-\cfth);
  \draw[#6, rounded corners=1.2pt] (#2,#3) rectangle ++(#4,{-\cfth-#5*\cfrh-0.05});
  \node[anchor=west, font=\cftitle, text=#8, inner sep=0pt] at ({#2+0.07},{#3-0.5*\cfth})}\newcommand{\cfnode}[7]{\cfbox{#1}{#2}{#3}{#4}{#5}{#6, line width=0.6pt}{#6!16}{black}{#7};}\newcommand{\cfghost}[6]{\cfbox{#1}{#2}{#3}{#4}{#5}{black!35, line width=0.5pt,
  dash pattern=on 1.5pt off 1pt}{black!5}{black!50}{#6};}\newcommand{\cfin}[4]{\path (#1) ++(0,{-\cfth-(#2-0.5)*\cfrh}) coordinate (#1-i#2);
  \fill[#4!\cfdot] (#1-i#2) circle (0.05);
  \node[anchor=west, font=\cftext, text=\cfink, inner sep=0pt] at ([xshift=0.08cm]#1-i#2) {#3};}\newcommand{\cfout}[4]{\path (#1) ++(\csname cfw@#1\endcsname,{-\cfth-(#2-0.5)*\cfrh}) coordinate (#1-o#2);
  \fill[#4!\cfdot] (#1-o#2) circle (0.05);
  \node[anchor=east, font=\cftext, text=\cfink, inner sep=0pt] at ([xshift=-0.08cm]#1-o#2) {#3};}\newcommand{\cfwid}[3]{\path (#1) ++({0.5*\csname cfw@#1\endcsname},{-\cfth-(#2-0.5)*\cfrh})
    node[font=\cftext, fill=black!7, rounded corners=0.8pt, inner sep=0.6pt,
         minimum width={(\csname cfw@#1\endcsname-0.16)*1cm}] {#3};}\newcommand{\ghoston}{\def\cfink{black!50}\def\cfdot{45}}\newcommand{\ghostoff}{\def\cfink{black}\def\cfdot{100}}\newcommand{\cftext}{\fontsize{5}{5}\sffamily}\newcommand{\cftitle}{\fontsize{5.5}{5.5}\sffamily\bfseries}\newcommand{\cfpanel}[7]{\fill[rounded corners=3pt, #5!7] (#1,#2) rectangle (#3,#4);
  \node[anchor=north west, inner sep=0pt, font=\fontsize{6}{6}\sffamily\bfseries,
        text=#5] at ({#1+0.12},{#2-0.1}) {#6};
  \node[anchor=north east, inner sep=0pt, font=\fontsize{6}{6}\sffamily,
        text=#5] at ({#3-0.12},{#2-0.1}) {#7};}

%% file: sections/introduction.tex
\section{Introduction}
\label{sec:intro}

\input{figures/fig-silent-failures}

Generative-AI applications are no longer single models but \emph{pipelines}.
They compose diffusion backbones~\cite{ldm2022}, text encoders~\cite{clip2021},
VAE decoders~\cite{vae2014}, LoRA adapters~\cite{lora2022},
ControlNets~\cite{controlnet2023}, and post-processors into multi-stage
dataflows. To author such pipelines, visual workflow systems have become a
widely used medium, and for image and video generation, ComfyUI~\cite{comfyui}
is among the most widely adopted tools. ComfyUI practitioners express pipelines
as node graphs in its editor, draw on about 5{,}700 community node packs
published in the ComfyUI Registry~\cite{comfy-registry}, share workflows on
community sites, and, increasingly, delegate graph construction to LLM agents
that emit ComfyUI workflows
directly~\cite{comfygpt2025,comfybench2024,comfyclaw2026}. In this regime, a
workflow \emph{is} a program~\cite{ko2011}: its author controls how components
are composed, while their implementations and compatibility assumptions come
from many independent developers~\cite{garlan1995}. This makes checking these
assumptions at composition time a software engineering problem.

ComfyUI already addresses a first version of this problem: its editor performs
\emph{nominal} checking on the type labels of node inputs and outputs, so it
refuses to connect an \texttt{IMAGE} output to a \texttt{LATENT} input. However,
the check stops at the label and ignores the semantic compatibility behind it.
\texttt{MODEL} does not say \emph{which} model family, \texttt{LATENT} does not
say how many channels, and \texttt{CONTROL\_NET} does not say which control
modality the network was trained on, yet these are exactly the facts on which
compatibility hinges. Consequently, (i) an SDXL~\cite{sdxl2024} checkpoint wired
to SD1.5 conditioning, (ii) a 4-channel SD latent fed to a 16-channel Flux~\cite{flux}
decoder, and (iii) a CLIP-vision encoder with 1024-wide embeddings attached to
a Wan~\cite{wan2025} video model that projects 1280-wide ones all pass every
editor check, then crash only after the models have been loaded and execution
has begun. Worse, some mismatches never crash at all.

We call such violations of the components' semantic compatibility requirements
\emph{semantic composition failures}, and call those whose execution still
completes without an exception \emph{silent}~\cite{tambon2024}.
Example~\ref{ex:silent} shows how little it takes to produce a silent one.

\begin{example}\label{ex:silent} Figure~\ref{fig:silent-failures} shows four
workflows run with the same prompt and seed. ComfyUI accepts all four, and all
four run to completion. Workflow~(a) produces the intended image. Workflows
(b)--(d) each add one component that disagrees with the rest of the graph on an
attribute no nominal label records. In~(b), an SDXL LoRA is applied to the SD1.5
checkpoint: 486 of its 2{,}364 tensors match coincidentally named layers, so
loading succeeds, but the intended style is lost and the image resembles one
generated without the LoRA. In~(c), an SDXL VAE decodes the SD1.5 latent: both
latent formats have four channels, so all shapes agree, but the colors are
corrupted. In~(d), a depth-trained ControlNet receives Canny edges: neither the
tensor shape nor the nominal label records the control modality, and the subject
degenerates into a nested-box artifact. In each case, the user receives a
plausible but incorrect image and \emph{no diagnostic that points to the cause}.
\end{example}

To establish the practical scale of this problem, we conducted a census of
\emph{all} 10{,}134 issues and discussion threads in the ComfyUI repository
(\S\ref{sec:study}). Unlike prior studies of deep-learning
bugs~\cite{islam2019,humbatova2020}, which examine training code and frameworks,
our census targets failures in composing pre-built components. Among the
reports, we find 753 semantic composition failures, which account for $8.3\%$ of
issue reports and $20.9\%$ of issue-tracker runtime execution failures. Of these
failures, 31.5\% across both venues are silent. These reports show that
accepting a graph's connections and completing its execution leave component
compatibility unchecked.

Preventing these failures requires checking compatibility before execution, when
the workflow is composed. The central challenge is to make compatibility
requirements \emph{compose with the workflow}. In a typical workflow, model
identity originates at loaders, tensor geometry changes through encoders and
resizers, and consumers constrain values arriving along different paths. A
node's requirements must therefore be stated in terms of its inputs, so that
they remain meaningful when the node is reused in a different graph. Moreover,
workflows mix built-in operators with custom nodes whose behavior may be unknown
to the checker, so the checker must combine the facts available across the graph
and track where information is missing.

Existing techniques do not meet this challenge. Shape checkers for deep-learning
code~\cite{ariadne2018,pytea2022,tensorsfitting2021,dex2021} analyze program
source and track tensor dimensions, whereas a workflow is assembled from opaque
components and often fails on attributes with no shape, such as model family and
control modality. Visual dataflow and scientific workflow
systems~\cite{simulink-types,taverna2004,kepler2006} type their ports,
and JSON subschema checking validates ML pipeline
connections~\cite{habib2021}; our focus is relational compatibility in
generative-AI graphs with missing schemas. LLM agents that construct ComfyUI
workflows produce graphs without checking them against component requirements.

To address this challenge, we present \tool, a gradual refinement type
system~\cite{xi1998,liquidtypes2008,lehmann2017} for generative-AI workflow
graphs, implemented for ComfyUI. \tool expresses compatibility requirements as
\emph{refinements of existing workflow interfaces}. It defines a constrained
polymorphic schema for each supported operator, describing what its inputs
require and what its outputs preserve or transform. For example,
\nodename{VAEEncode} relates an image and VAE to a latent whose channels and
spatial compression depend on the VAE architecture. Composing this schema with a
sampler and decoder automatically relates their geometry and architecture
requirements, wherever they occur in the graph. Schemas are independent of any
graph, so the checker derives each workflow's compatibility constraints from its
edges alone, without annotations from its author.

The refinements in these schemas track model-family and architecture brands,
dimension expressions and divisibility, control modalities, and quantization
levels. A workflow type-checks when the constraints from all its edges hold
together, which \tool decides by brand unification and linear integer
arithmetic. Because custom nodes often lack schemas, \tool is
\emph{gradual}~\cite{siek2006}: it treats the refinements on their ports as
unknown ($\tyDyn$) and checks the rest of the graph. We prove a static gradual
guarantee~\cite{siek2015}: removing a schema never causes an accepted workflow
to be rejected. We also prove that, given faithful schemas and conforming values
from nodes without schemas, the checked nodes of an accepted workflow run
without the compatibility errors their types rule out. Each added schema
therefore adds checks, and the boundaries where schemas are missing show where
more are needed.

We evaluate \tool's detection accuracy and cost, its effectiveness on real
workflows, and the usefulness of its diagnostics to users. We inject faults into
132 official ComfyUI workflows to produce graphs with faults. ComfyUI's
validator admits all confirmed faults, whereas \tool detects 66.4\% of them. We also showed that
\tool is fast: checking takes 5.46\,ms at the median. On real workflows, \tool
raises no false alarm on 765 known-good workflows and detects 67\% of 108 failures rebuilt
from user reports. Furthermore, in a controlled user study with 10 participants,
\tool's diagnostics raise fault-localization accuracy on faulty workflows from
35.0\% to 87.5\% and valid-repair selection from 30.0\% to 62.5\%.

This paper contributes (i)~a census of 10{,}134 ComfyUI reports that
characterizes semantic composition failures and the component relationships they
violate (\S\ref{sec:study}); (ii)~a gradual refinement type system that composes
node schemas to check open generative-AI workflow graphs (\S\S\ref{sec:calculus}
and~\ref{sec:gradual}); and (iii)~an open-source checker with 147 node schemas
(114 built-in, 33 custom), each proved faithful, evaluated on real and
fault-injected workflows (\S\S\ref{sec:impl} and~\ref{sec:eval}).

%% file: figures/fig-silent-failures.tex
\begin{figure}[t]
\centering
\colorlet{vA}{black!55}\definecolor{vB}{HTML}{D55E00}\definecolor{vC}{HTML}{0072B2}\definecolor{vD}{HTML}{CC79A7}\begin{subfigure}[t]{0.227\textwidth}
  \tikz\node[draw=vA, line width=1.2pt, inner sep=0.8pt] {\includegraphics[width=\dimexpr\linewidth-2.8pt\relax]{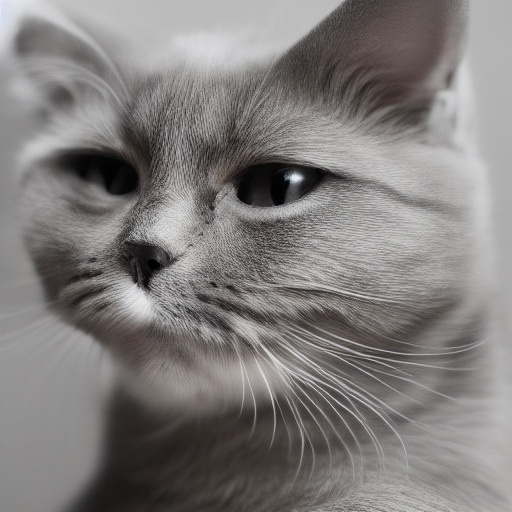}};
  \caption{Intended (SD1.5)}
\end{subfigure}\hfill
\begin{subfigure}[t]{0.227\textwidth}
  \tikz\node[draw=vB, line width=1.2pt, inner sep=0.8pt] {\includegraphics[width=\dimexpr\linewidth-2.8pt\relax]{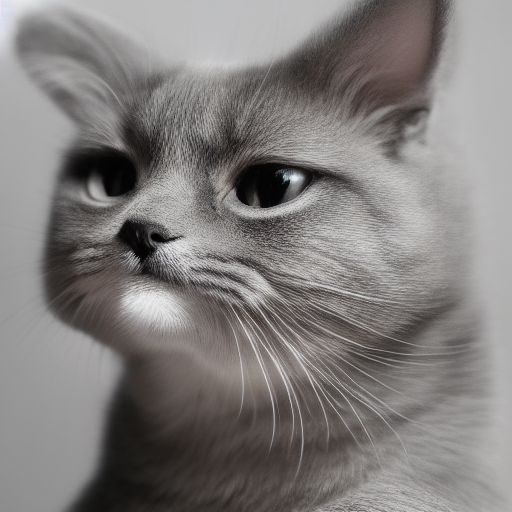}};
  \caption{SDXL LoRA}
\end{subfigure}\hfill
\begin{subfigure}[t]{0.227\textwidth}
  \tikz\node[draw=vC, line width=1.2pt, inner sep=0.8pt] {\includegraphics[width=\dimexpr\linewidth-2.8pt\relax]{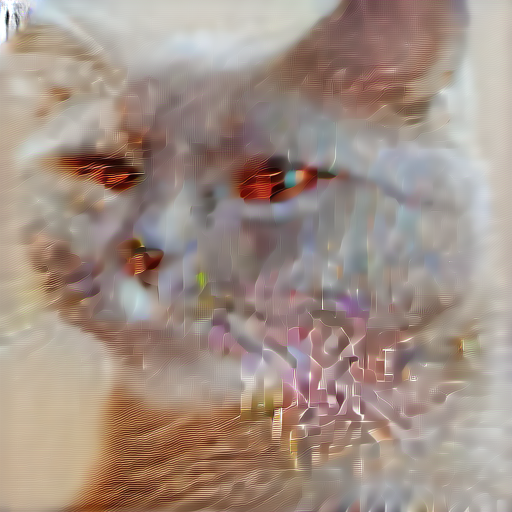}};
  \caption{SDXL VAE}
\end{subfigure}\hfill
\begin{subfigure}[t]{0.227\textwidth}
  \tikz\node[draw=vD, line width=1.2pt, inner sep=0.8pt] {\includegraphics[width=\dimexpr\linewidth-2.8pt\relax]{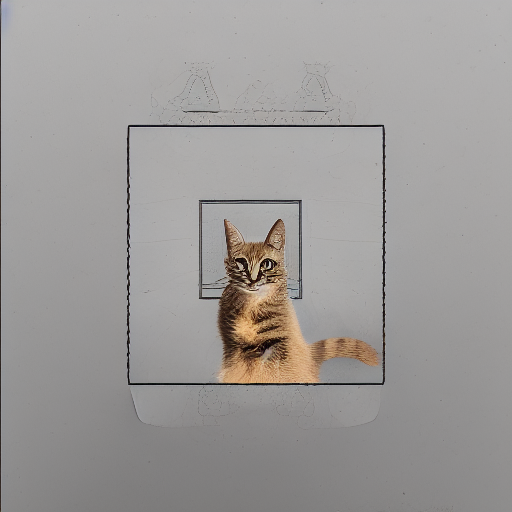}};
  \caption{Depth ControlNet}
\end{subfigure}

\vspace{3pt}
\resizebox{\textwidth}{!}{\begin{tikzpicture}[wire/.style={line width=0.7pt, #1}, wire/.default=black!45]
  \cfpanel{0}{0}{14.3}{-2.15}{vA}{(a) Intended workflow}{}
  \cfpanel{0}{-2.3}{7.15}{-3.95}{vB}{(b) adds Load LoRA}
    {LoRA family: $\mathsf{SDXL}\neq\mathsf{SD15}$}
  \cfpanel{0}{-4.1}{7.15}{-5.75}{vC}{(c) swaps the VAE}
    {VAE family: $\mathsf{SDXL}\neq\mathsf{SD15}$}
  \cfpanel{7.3}{-2.3}{14.3}{-5.75}{vD}{(d) adds a ControlNet}
    {control modality: $\mathsf{Edge}\neq\mathsf{Depth}$}
  \begin{scope}[shift={(1.7,-0.42)}]
    \cfnode{ck}{0}{0}{1.75}{4}{vA}{Load Checkpoint}
    \cfout{ck}{1}{MODEL}{tMODEL}\cfout{ck}{2}{CLIP}{tCLIP}\cfout{ck}{3}{VAE}{tVAE}
    \cfwid{ck}{4}{sd15.safetensors}
    \cfnode{te}{2.3}{-0.55}{2.05}{2}{vA}{CLIP Text Encode}
    \cfin{te}{1}{clip}{tCLIP}\cfout{te}{1}{CONDITIONING}{tCOND}
    \cfwid{te}{2}{a photo of a cat, \dots}
    \cfnode{ks}{5.6}{0}{1.45}{3}{vA}{KSampler}
    \cfin{ks}{1}{model}{tMODEL}\cfout{ks}{1}{LATENT}{tLATENT}
    \cfin{ks}{2}{positive}{tCOND}\cfwid{ks}{3}{seed: 42}
    \cfnode{vd}{7.7}{0}{1.4}{2}{vA}{VAE Decode}
    \cfin{vd}{1}{samples}{tLATENT}\cfout{vd}{1}{IMAGE}{tIMAGE}
    \cfin{vd}{2}{vae}{tVAE}
    \cfnode{sv}{9.7}{0}{1.2}{1}{vA}{Save Image}
    \cfin{sv}{1}{images}{tIMAGE}
    \draw[wire] (ck-o1) -- (ks-i1);
    \draw[wire] (ck-o2) to[out=0,in=180] (te-i1);
    \draw[wire] (te-o1) to[out=0,in=180] (ks-i2);
    \draw[wire] (ck-o3) .. controls +(0.35,0) and +(-0.3,0) .. (2.25,-1.55)
      -- (7.1,-1.55) .. controls +(0.35,0) and +(-0.35,0) .. (vd-i2);
    \draw[wire] (ks-o1) -- (vd-i1);
    \draw[wire] (vd-o1) -- (sv-i1);
  \end{scope}
  \begin{scope}[shift={(1.0,-2.8)}]
    \ghoston\cfghost{b1}{-0.1}{0}{1.65}{1}{Load Checkpoint}\cfout{b1}{1}{MODEL}{tMODEL}
    \cfghost{b3}{4.1}{0}{1.05}{1}{KSampler}\cfin{b3}{1}{model}{tMODEL}\ghostoff
    \cfnode{b2}{1.85}{0}{1.95}{2}{vB}{Load LoRA}
    \cfin{b2}{1}{model}{tMODEL}\cfout{b2}{1}{MODEL}{tMODEL}
    \cfwid{b2}{2}{sdxl\_lora.safetensors}
    \draw[wire=vB] (b1-o1) -- (b2-i1);
    \draw[wire=vB] (b2-o1) -- (b3-i1);
  \end{scope}
  \begin{scope}[shift={(1.95,-4.6)}]
    \ghoston\cfghost{c2}{2.05}{0}{1.2}{1}{VAE Decode}\cfin{c2}{1}{vae}{tVAE}\ghostoff
    \cfnode{c1}{0}{0}{1.7}{2}{vC}{Load Checkpoint}
    \cfout{c1}{1}{VAE}{tVAE}\cfwid{c1}{2}{sdxl.safetensors}
    \draw[wire=vC] (c1-o1) -- (c2-i1);
  \end{scope}
  \begin{scope}[shift={(7.6,-2.78)}]
    \ghoston\cfghost{d1}{0}{0}{1.75}{1}{CLIP Text Encode}\cfout{d1}{1}{CONDITIONING}{tCOND}
    \cfghost{d7}{5.35}{-0.25}{1.05}{1}{KSampler}\cfin{d7}{1}{positive}{tCOND}\ghostoff
    \cfnode{d2}{0}{-0.75}{2.25}{2}{vD}{Load ControlNet Model}
    \cfout{d2}{1}{CONTROL\_NET}{tCNET}\cfwid{d2}{2}{sd15\_depth.safetensors}
    \cfnode{d3}{0}{-1.75}{1.2}{2}{vD}{Load Image}
    \cfout{d3}{1}{IMAGE}{tIMAGE}\cfwid{d3}{2}{ref.png}
    \cfnode{d4}{1.35}{-1.75}{1.4}{1}{vD}{Canny}
    \cfin{d4}{1}{image}{tIMAGE}\cfout{d4}{1}{IMAGE}{tIMAGE}
    \cfnode{d5}{3.3}{-0.25}{1.75}{3}{vD}{Apply ControlNet}
    \cfin{d5}{1}{positive}{tCOND}\cfout{d5}{1}{positive}{tCOND}
    \cfin{d5}{2}{control\_net}{tCNET}\cfin{d5}{3}{image}{tIMAGE}
    \draw[wire=vD] (d1-o1) to[out=0,in=180] (d5-i1);
    \draw[wire=vD] (d2-o1) to[out=0,in=180] (d5-i2);
    \draw[wire=vD] (d3-o1) -- (d4-i1);
    \draw[wire=vD] (d4-o1) to[out=0,in=180] (d5-i3);
    \draw[wire=vD] (d5-o1) -- (d7-i1);
  \end{scope}
\end{tikzpicture}}
\caption{Outputs of four ComfyUI workflows run with the same prompt and seed.
Panel~(a) is the intended workflow; panels (b)--(d) show the change that
produces each variant, with unchanged nodes of~(a) dashed. Workflows are
simplified; Example~\ref{ex:silent} explains each variant.}
\label{fig:silent-failures}
\end{figure}

%% file: sections/study.tex
\section{Background and Failure Study}
\label{sec:study}

This section first introduces the background on ComfyUI that the rest of the
paper relies on: its graph model and how it validates workflows
(\S\ref{sec:study:bg}). It then presents a failure study that answers our first
research question (RQ1): How often do semantic composition failures occur in
ComfyUI, how do they manifest, and which component relationships do they
violate? Section~\ref{sec:study:method} describes the study design, and
\S\ref{sec:study:findings} reports the findings and the requirements they impose
on a checker, with the classifier's full validation in
Appendix~\ref{sec:appendix-validity}.

\subsection{ComfyUI Workflows and Nominal Validation}
\label{sec:study:bg}

A ComfyUI workflow is a directed acyclic graph serialized as JSON. Each node
carries an operator name (its \emph{node class}), a set of \emph{widget} values
(literals chosen in the UI, such as an image width or a sampler name), and input
ports wired to output ports of other nodes. At runtime, the server executes the
nodes in topological order as Python functions over tensors and model objects.

ComfyUI validates graphs at two points. The \emph{editor} refuses to create an
edge whose endpoint port types differ nominally: every port is labeled with a
string such as \texttt{IMAGE}, \texttt{LATENT}, \texttt{MODEL},
\texttt{CONDITIONING}, or \texttt{CONTROL\_NET}, and only equal labels connect.
The \emph{server} re-checks these labels and widget enum membership before
execution. However, both checks stop at the label: \texttt{MODEL} does not say
\emph{which} model family, \texttt{LATENT} does not say how many channels, and
\texttt{CONTROL\_NET} does not say which control modality the network was
trained on. Errors that depend on these facts therefore appear only at run time,
deep inside model code. They typically appear as PyTorch shape errors after the
checkpoint has been loaded onto the GPU, and some never appear at all.

\subsection{Study Design}
\label{sec:study:method}

\emph{Goal.} To answer RQ1, we measure (i)~how common semantic composition
failures are relative to other failure classes, (ii)~how they manifest, in
particular how often they are \emph{silent} and complete execution without an
exception, and (iii)~which compatibility dimensions they involve. Each measure
addresses one aspect of the problem: prevalence establishes its scale, the
silent fraction shows how far execution can be relied on to reveal a failure,
and the dimensions identify the relationships a checker must track.

\emph{Corpus and Census.} We study the ComfyUI core repository
(\texttt{Comfy-Org/ComfyUI}), across its \emph{issue tracker} and its
\emph{Discussions} forum. We enumerated the entire population with {8{,}766
issues} and {1{,}368 discussions}. Then, we classify \emph{every} item with an
LLM, which succeeds for 10{,}121 (99.9\%).

\emph{Classification.} A large language model (\texttt{deepseek-v4-flash})
labels each item by applying one fixed written rubric that encodes the taxonomy
below, following prior use of LLMs to annotate software artifacts and issue
reports~\cite{ahmed2025,colavito2024}. Each label records the item's primary
class, whether it reports a runtime execution failure, its failure mode
(nominal, semantic-crashing, or semantic-silent), whether it passed ComfyUI's
nominal check, and the compatibility dimensions involved. All census statistics
are computed from these labels without manual relabeling. To validate this
classifier, we drew a seeded random subsample of 1{,}583 items and had four
further models relabel it:
\texttt{deepseek-v4-pro}, \texttt{glm-5.2}, \texttt{minimax-m3}, and
\texttt{hy3-preview}. Three of them come from developers other than the
classifier's, which reduces the bias of any single model~\cite{verga2024}, and
agreement between models indicates whether a labeling task suits
LLMs~\cite{ahmed2025}. On the issues, the four models agree with the
classifier's labels at $\kappa=0.71$--$0.81$~\cite{cohen1960}, and their
prevalence estimates lie on both sides of the classifier's. A consensus of the classifier and
\texttt{deepseek-v4-pro}, adjudicated where they disagree, puts the
issue-tracker prevalence within one percentage point of the classifier's.

\emph{Taxonomy.} We classify failures as (a)~\emph{nominal}: caught by ComfyUI's
existing nominal check; (b)~\emph{semantic-crashing}: passes the nominal check,
raises a runtime exception; and (c)~\emph{semantic-silent}: passes the nominal
check, completes execution, and produces wrong or degraded output. Each semantic
case is further tagged with the compatibility dimensions it involves
(\emph{family}, \emph{dims}, \emph{channels}, \emph{divisibility},
\emph{modality}, \emph{quantization}, \emph{routing}).

\subsection{Results}
\label{sec:study:findings}

\input{figures/fig-census}

\parh{Prevalence (Figure~\ref{fig:census}a).} Semantic composition failures make
up \textbf{8.3\%} of all issues (724 of 8{,}753) and \textbf{20.9\%} of issues
that report a runtime failure (724 of 3{,}467). In Discussions, they make up
18.6\% of runtime-failure threads (29 of 156) but only 2.1\% of all threads (29
of 1{,}368), because most threads are usage questions. In total, the census
finds 753 semantic composition failures (724 in issues and 29 in Discussions),
which the rest of this section analyzes.

\parh{Silent Failures (Figure~\ref{fig:census}b).} About 31.5\% semantic
composition failure is \emph{silent} (237 of 753): the run completes without an
exception despite the violation. The share is higher in Discussions (14 of 29)
than in issues (30.8\%), consistent with users reporting silent failures as
``why is my output wrong'' questions, not as bugs. Execution is therefore an
unreliable detector: many such failures raise no exception, and inspecting the
output requires paying the full execution cost first. This motivates checking
compatibility before execution.

\parh{Compatibility Dimensions (Figure~\ref{fig:census}c).} Family mismatches
dominate (421 cases), such as an SDXL LoRA applied to an SD1.5 model or a
CLIP-vision encoder with the wrong embedding width; they motivate the family and
architecture constraints of \S\ref{sec:calculus}. Tensor dimensions,
channels, and divisibility together account for 220 labels, such as video
patchify requirements and frame counts $\equiv 1 \bmod 4$; they motivate
dimension expressions with divisibility constraints. Quantization mismatches
(74 cases) and control-modality mismatches (11) motivate tracking both in types.
Routing labels (157 cases) mark values sent along the wrong path, such as a LoRA
left off a guider's input. Types detect such an error only when it produces an
incompatibility; wiring that is type-compatible but not what the user intended
is out of scope. In all cases, detecting a failure also requires schemas for
the nodes involved and facts about the models.

\parh{Beyond User Error: Bugs in Built-in Nodes.} Types also expose bugs in
built-in nodes. In issue
\href{https://github.com/Comfy-Org/ComfyUI/issues/13858}{\#13858}, a built-in
compositing node hardcoded a 3-channel slice and failed on every model with
16-channel latents; a channel-indexed schema would have exposed the bug when the
schema was written. Node developers are thus a second audience for the checker,
alongside workflow authors.

\parh{From Findings to Design Requirements.} The findings set three requirements
for a workflow type system. \reqdef{1}~Node schemas must express relationships
among model families, tensor geometry, control modality, and quantization.
\reqdef{2}~Compatibility must be checked before execution, because many failures
are silent. \reqdef{3}~Because many nodes lack schemas, checking must extend to
workflows that mix nodes with and without them.

\medskip
\takeaway{RQ1 takeaway}{In a census of \emph{all} 10{,}134 ComfyUI GitHub
reports, semantic composition failures account for $8.3\%$ of issues and
$20.9\%$ of issue-tracker runtime failures, and 31.5\% of them are silent. Their
recurring component relationships motivate making compatibility an explicit,
checkable part of workflow interfaces.}

%% file: figures/fig-census.tex
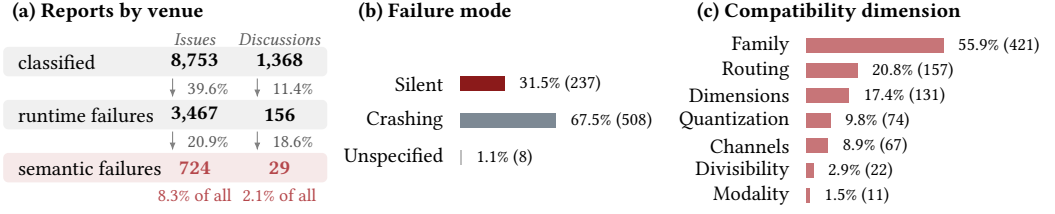
\begin{figure}[t]
\centering
\resizebox{\textwidth}{!}{\input{figures/census_plot}}
\caption{Census results: (a)~reports per venue at each classification stage;
(b)~failure modes and (c)~compatibility dimensions of the semantic composition failures of
both venues. A failure may involve several dimensions.}
\label{fig:census}
\end{figure}

%% file: figures/census_plot.tex
\begin{tikzpicture}[font=\footnotesize]
\definecolor{censussem}{HTML}{B5474B}
\definecolor{censussilent}{HTML}{8B1A1A}
\definecolor{censusneutral}{HTML}{7D8994}
\node[anchor=west] at (0.0000,3.0000) {\textbf{(a) Reports by venue}};
\node[font=\scriptsize\itshape,text=black!65] at (2.6000,2.6300) {Issues};
\node[font=\scriptsize\itshape,text=black!65] at (3.7500,2.6300) {Discussions};
\fill[black!6,rounded corners=2pt] (0,2.530) rectangle (4.4,2.130);
\node[anchor=west] at (0.0800,2.3300) {classified};
\node[text=black] at (2.6000,2.3300) {\textbf{8,753}};
\node[text=black] at (3.7500,2.3300) {\textbf{1,368}};
\fill[black!6,rounded corners=2pt] (0,1.810) rectangle (4.4,1.410);
\node[anchor=west] at (0.0800,1.6100) {runtime failures};
\node[text=black] at (2.6000,1.6100) {\textbf{3,467}};
\draw[-{Latex[length=1.1mm]},black!50] (2.30,2.100) -- (2.30,1.840);
\node[anchor=west,font=\scriptsize,text=black!65] at (2.3800,1.9700) {39.6\%};
\node[text=black] at (3.7500,1.6100) {\textbf{156}};
\draw[-{Latex[length=1.1mm]},black!50] (3.45,2.100) -- (3.45,1.840);
\node[anchor=west,font=\scriptsize,text=black!65] at (3.5300,1.9700) {11.4\%};
\fill[censussem!12,rounded corners=2pt] (0,1.090) rectangle (4.4,0.690);
\node[anchor=west] at (0.0800,0.8900) {semantic failures};
\node[text=censussem] at (2.6000,0.8900) {\textbf{724}};
\draw[-{Latex[length=1.1mm]},black!50] (2.30,1.380) -- (2.30,1.120);
\node[anchor=west,font=\scriptsize,text=black!65] at (2.3800,1.2500) {20.9\%};
\node[text=censussem] at (3.7500,0.8900) {\textbf{29}};
\draw[-{Latex[length=1.1mm]},black!50] (3.45,1.380) -- (3.45,1.120);
\node[anchor=west,font=\scriptsize,text=black!65] at (3.5300,1.2500) {18.6\%};
\node[font=\scriptsize,text=censussem] at (2.6000,0.5200) {8.3\% of all};
\node[font=\scriptsize,text=censussem] at (3.7500,0.5200) {2.1\% of all};
\node[anchor=west] at (4.7000,3.0000) {\textbf{(b) Failure mode}};
\node[anchor=east] at (6.1000,2.0300) {Silent};
\fill[censussilent] (6.2,1.930) rectangle (6.8070,2.130);
\node[anchor=west,font=\scriptsize] at (6.8870,2.0300) {31.5\% (237)};
\node[anchor=east] at (6.1000,1.5300) {Crashing};
\fill[censusneutral] (6.2,1.430) rectangle (7.5011,1.630);
\node[anchor=west,font=\scriptsize] at (7.5811,1.5300) {67.5\% (508)};
\node[anchor=east] at (6.1000,1.0300) {Unspecified};
\fill[black!25] (6.2,0.930) rectangle (6.2205,1.130);
\node[anchor=west,font=\scriptsize] at (6.3005,1.0300) {1.1\% (8)};
\node[anchor=west] at (9.3000,3.0000) {\textbf{(c) Compatibility dimension}};
\node[anchor=east] at (10.8000,2.5500) {Family};
\fill[censussem!75] (10.9,2.450) rectangle (12.7637,2.650);
\node[anchor=west,font=\scriptsize] at (12.8437,2.5500) {55.9\% (421)};
\node[anchor=east] at (10.8000,2.2100) {Routing};
\fill[censussem!75] (10.9,2.110) rectangle (11.5950,2.310);
\node[anchor=west,font=\scriptsize] at (11.6750,2.2100) {20.8\% (157)};
\node[anchor=east] at (10.8000,1.8700) {Dimensions};
\fill[censussem!75] (10.9,1.770) rectangle (11.4799,1.970);
\node[anchor=west,font=\scriptsize] at (11.5599,1.8700) {17.4\% (131)};
\node[anchor=east] at (10.8000,1.5300) {Quantization};
\fill[censussem!75] (10.9,1.430) rectangle (11.2276,1.630);
\node[anchor=west,font=\scriptsize] at (11.3076,1.5300) {9.8\% (74)};
\node[anchor=east] at (10.8000,1.1900) {Channels};
\fill[censussem!75] (10.9,1.090) rectangle (11.1966,1.290);
\node[anchor=west,font=\scriptsize] at (11.2766,1.1900) {8.9\% (67)};
\node[anchor=east] at (10.8000,0.8500) {Divisibility};
\fill[censussem!75] (10.9,0.750) rectangle (10.9974,0.950);
\node[anchor=west,font=\scriptsize] at (11.0774,0.8500) {2.9\% (22)};
\node[anchor=east] at (10.8000,0.5100) {Modality};
\fill[censussem!75] (10.9,0.410) rectangle (10.9487,0.610);
\node[anchor=west,font=\scriptsize] at (11.0287,0.5100) {1.5\% (11)};
\end{tikzpicture}

%% file: sections/calculus.tex
\section{The \tool Calculus}
\label{sec:calculus}

Section~\ref{sec:study} showed that semantic composition failures arise from
relationships between components that ComfyUI's nominal port labels cannot
express. \tool makes these relationships part of each node's interface: each
node contributes a reusable schema whose indices capture the relationships of
\req{1}, and each edge relates the producer's guarantees to the consumer's
requirements, so that compatibility is checked before execution (\req{2}). This
section formalizes workflow graphs (\S\ref{sec:calc:graphs}), the refinement
type and constraint language (\S\ref{sec:calc:types}), node schemas
(\S\ref{sec:calc:schemas}), and the well-typedness judgment
(\S\ref{sec:calc:judgment}). Section~\ref{sec:gradual} then extends checking to
nodes without schemas (\req{3}). Appendix~\ref{app:supp} gives the full
type language, the faithfulness proof of every schema, and the proofs of the
solver's properties.

\subsection{Workflow Graphs}
\label{sec:calc:graphs}

A workflow is a tuple $G = (N, E, \mathit{op}, \mathit{wid})$, where $N$ is a
finite set of nodes, $\mathit{op}(n)$ is the operator of node $n$,
$\mathit{wid}(n)$ maps each widget slot of $n$ to a literal value, and $E
\subseteq (N \times \mathbb{N}) \times (N \times \mathbb{N})$ connects output
ports to input ports. We require $G$ to be acyclic and each input port to have
exactly one incoming edge. As in ComfyUI, each edge must also pass the
\emph{nominal} check: the labels of its two endpoints are equal. \tool thus keeps ComfyUI's validator and adds checks on top of it.

Saved workflows also contain four editor constructs that this core model
omits. \emph{Reroute} nodes forward a value unchanged, \emph{group nodes} and
subgraphs package several nodes behind a single node, \emph{muted} nodes are
excluded from execution, and \emph{bypassed} nodes pass each input through to
the output with the same nominal label. ComfyUI removes these constructs before
it executes a workflow, and \tool removes them in the same way before checking.
It replaces each reroute by a direct edge, inlines groups and subgraphs, deletes
muted nodes together with the nodes that depend on them, and applies ComfyUI's
own bypass rewiring. \tool therefore checks exactly the graph that runs.

\subsection{Types, Indices, and Constraints}
\label{sec:calc:types}

\input{figures/fig-type-syntax}

ComfyUI's port labels state what kind of value flows along an edge, but not the
facts on which compatibility depends. \tool's types add these facts as
\emph{indices}, following indexed and refinement
types~\cite{xi1998,liquidtypes2008}: each type is a nominal ComfyUI label
refined with the indices that the failure study found to matter
(Figure~\ref{fig:type-syntax}). The indices record four kinds of facts.
\emph{Geometry}: images and masks carry their height and width, images also
their channel count, and latents their height, width, and channel count.
\emph{Model family}: every value tied to a model (latents, conditioning, text
encoders, VAEs, LoRAs, ControlNets, and models) carries the \emph{model family}
$f$ it belongs to, which determines, e.g., a latent's channel count and spatial
compression. \emph{Control
modality}: an image carries its \emph{modality} $s$, such as an edge or depth
map, and a ControlNet carries the modality of the hints it was trained on.
ComfyUI records modality nowhere at run time, so it exists only in the types.
\emph{Quantization}: a model carries its \emph{quantization level} $q$, such as
full precision or fp8. Dimensions are linear over $\mathbb{Z}$ with division
only by constants, and an omitted index is a fresh variable.
The figure elides indices that our examples do not need, such as the protocol and control kind on
$\tycon{Cond}$ and the zero state on $\tycon{Latent}$, and index-free labels such as $\tycon{Noise}$.

To relate the values at different ports, for example to require that a model and its LoRA belong to the same family, schemas
state \emph{constraints} $\varphi$ over indices. A constraint is a conjunction
of three kinds of atoms: (i)~equations between families, modalities, or
quantization levels, (ii)~linear equalities between dimensions, and (iii)~divisibility
atoms $k \mid e$ with a constant $k$, such as the requirement that a width be a
multiple of 16. The arithmetic part is a fragment of quantifier-free linear integer arithmetic, which is
decidable~\cite{kroening2016}, and it omits multiplication of variables,
disjunction, and negation, none of which the studied failures required.

The model family needs more structure than a flat set of names. Community
checkpoints such as $\brand{Pony}$ and $\brand{Illustrious}$~\cite{illustrious2024}
are fine-tunes of $\brand{SDXL}$~\cite{sdxl2024} and inherit its tensor geometry
unchanged. A checker that treats them as unrelated families would reject, for
example, an SDXL LoRA applied to a Pony checkpoint, a false alarm that would
affect many popular fine-tune workflows (\S\ref{sec:eval}). We therefore order
families by derivation ($\brand{Pony} \sqsubseteq \brand{SDXL}$), so that they
form a \emph{derivation forest}, and write $\archfn(f)$ for the root of $f$'s
tree, its \emph{architecture}. Tensor-level compatibility (latent geometry, conditioning dimensions, VAE
channels) depends only on the architecture, whereas weight-level compatibility
(which lineage a LoRA was trained on) depends on the family. Accordingly, types
that carry tensors ($\tycon{Latent}$, $\tycon{Cond}$, $\tycon{Clip}$,
$\tycon{VAE}$, $\tycon{Ctrl}$) are indexed by \emph{architectures}, and types
that carry weights ($\tycon{Model}$, $\tycon{Lora}$) by \emph{families}. A
family written in an architecture position abbreviates $\archfn(f)$.

To the checker, families (including architectures), modalities, and
quantization levels are all \emph{brands}: opaque names that unify only with
the same name or with a brand variable. Equations between them are therefore
\emph{brand equations}. A few interpreted functions connect brands to each
other and to dimensions. $\archfn$ maps a family to its architecture. $\chfn$
maps an architecture to its latent channel count ($\chfn(\brand{SD15}) =
\chfn(\brand{SDXL}) = 4$, $\chfn(\brand{SD3}) = \chfn(\brand{Flux}) = 16$,
\ldots), and $\mathit{sf}$ maps it to its VAE's spatial downscaling factor
($\mathit{sf}(\brand{SD15}) = 8$, $\mathit{sf}(\brand{Wan22\_5B}) = 16$).
$\mathit{sz}$ maps a family to its size class (e.g., 14B for Wan~2.1). With $\chfn$
and $\mathit{sf}$, one VAE schema covers every architecture yet rejects channel
and geometry mismatches between architectures. The functions add no case analysis to checking. Each function is evaluated only once
unification has fixed its brand argument, side predicates on brands are
evaluated in the same way, and guards on widget literals are resolved when a
schema is instantiated. Constraints therefore never need disjunction.

\subsection{Node Schemas}
\label{sec:calc:schemas}

Types describe the values that flow along edges, and a \emph{schema} describes
an operator in terms of these types. Each operator has one schema, written once
and independently of any graph, as a constrained polymorphic signature~\cite{odersky1999}
\[
\Sigma(\mathit{op}) \;=\;
\forall \bar\alpha, \bar x.\;
(\tau_1, \ldots, \tau_k) \xrightarrow{\;\varphi\;} (\tau'_1, \ldots, \tau'_m),
\]
where $\bar\alpha$ ranges over brand variables, $\bar x$ over dimension
variables, and the side condition $\varphi$ relates the indices of inputs and
outputs. A node's widget literals instantiate designated variables. Six schemas
illustrate the refinement dimensions:

\begin{align*}
&\nodename{VAEEncode} : \forall a, h, w. (\tyImage{h}{w}{3},\, \tyVae{a})
  \longrightarrow
  \tyLatent{h/\mathit{sf}(a)}{w/\mathit{sf}(a)}{\chfn(a)}{a}
\\[3pt]
&\nodename{KSampler} : \forall f, q, h, w. (\tyModel{f}{q},\, \tyCond{f},\, \tyCond{f},\, \tyLatent{h}{w}{\chfn(f)}{f})\\
&\quad \longrightarrow\; \tyLatent{h}{w}{\chfn(f)}{f}
\\[3pt]
&\nodename{LoraLoader} : \forall f, g, q.
  (\tyModel{f}{q},\, \tyClip{\archfn(f)},\, \tyLora{g})\\
&\quad \xrightarrow{\archfn(f) \,=\, \archfn(g)\,\wedge\, \mathit{sz}(f) \,=\, \mathit{sz}(g)}
  (\tyModel{f}{q},\, \tyClip{\archfn(f)})
\\[3pt]
&\nodename{ControlNetApply} : \forall f, s, m, h, w. (\tyCond{f},\, \tyCtrl{s}{f},\, \tycon{Image}(h,w,3,m))
  \xrightarrow{s \,=\, m} \tyCond{f}
\\[3pt]
&\nodename{EmptyLatentImage} : \forall a. (\tyInt{\nu_w},\, \tyInt{\nu_h})
  \longrightarrow
  \tyLatent{\nu_h/\mathit{sf}(a)}{\nu_w/\mathit{sf}(a)}{\chfn(a)}{a}
\\[3pt]
&\nodename{WanImageToVideo} : \forall a.\ \ldots (\tyCond{a}, \tyCond{a}, \tyVae{a}, \tyInt{\nu_w}, \tyInt{\nu_h})\\
&\quad \xrightarrow{16 \mid \nu_w \,\wedge\, 16 \mid \nu_h}
  (\tyCond{a}, \tyCond{a}, \tyLatent{\nu_h/8}{\nu_w/8}{16}{a})
\end{align*}

These schemas cover the dimensions of the study. \nodename{KSampler} requires
the model, both conditionings, and the latent to share an architecture.
\nodename{LoraLoader} requires the model and the LoRA to share an architecture
and a size class.
It therefore rejects the workflow of Figure~\ref{fig:silent-failures}b but
accepts an SDXL LoRA on a Pony checkpoint. \nodename{VAEEncode} derives the latent's
channel count and size from the VAE's architecture through $\chfn$ and
$\mathit{sf}$. \nodename{ControlNetApply} requires the hint image's modality to
match the ControlNet's. Modalities are compatibility classes: canny, scribble,
lineart, and MLSD edge maps all map to $\brand{Edge}$, so a depth ControlNet fed
an edge map is rejected. \nodename{WanImageToVideo} requires pixel dimensions
divisible by 16, the Wan resolution grid~\cite{wan2025}, which ComfyUI declares
only in widget metadata that API-format workflows and LLM agents bypass.

A schema is a claim about how its node behaves, and the claim can be wrong. We
call a schema \emph{faithful} if the claim holds: whenever the values passed to
the node have the types that the schema states for its inputs, the node's code
runs without a compatibility error and returns values that have the types that
the schema states for its outputs. A schema that is too loose is not faithful,
because it admits inputs on which the node fails, so \tool misses those
failures. A schema that is too strict is still faithful and only rejects some
working workflows. We prove every schema in our library faithful. Each proof follows the node's code path in
ComfyUI v0.27.0 and assumes standard facts about PyTorch, ComfyUI's model
architectures, and the model database.

\subsection{Well-Typedness as Constraint Satisfaction}
\label{sec:calc:judgment}

\input{figures/fig-typing-rules}

A graph is well-typed when the requirements of all its nodes and edges can hold
together. Figure~\ref{fig:typing-rules} states this judgment, and checking a
graph takes three steps. (1)~\emph{Instantiation}: each node's schema is
instantiated with fresh variables, and its widget literals are substituted.
(2)~\emph{Constraint generation}: each edge structurally equates the producer's
output type with the consumer's input type, which yields brand equations
and dimension equations. Each node also adds
its side condition $\varphi$. (3)~\emph{Solving}: brand equations are solved by
first-order unification (union-find~\cite{tarjan1975}). Grounding a family
variable evaluates $\archfn$, and then $\chfn$ and $\mathit{sf}$, which makes
the dimension equations that mention them concrete. $G$ is \emph{well-typed},
written $\wf{G}$, iff the resulting constraint system $\mathcal{C}(G)$ is
satisfiable under the interpretation $\mathcal{H}$ below.

After unification, the remaining constraints are linear equalities and
divisibility atoms over dimensions. \tool decides them under a \emph{hybrid}
interpretation $\mathcal{H}$ that is exact where the studied failures need it
and relaxed elsewhere. Divisibility atoms, which express requirements such as a
width divisible by 16, are decided exactly over $\mathbb{Z}$, and every variable
that occurs in one is an integer. Equalities are solved over the rationals by
Gaussian elimination. To combine the two, elimination pivots on
rational variables first, and eliminating an integer variable adds an
integrality atom. Each atom is checked for gcd feasibility. Atoms on the same
variable are merged by the Chinese Remainder Theorem (CRT), because individually
feasible congruences can conflict (e.g.\ $4 \mid x \wedge 4 \mid x{+}2$),
and components with several variables are decided by integer elimination. A
variable forced to a non-integer \emph{constant} is rejected, because concrete
sizes are integers. The relaxation errs in one direction only. A rejection is never wrong, but an
accepted system need not have an integer solution,
so the runtime guarantee
assumes faithful schemas and an integer solution. Rational division also cannot
model cropping a size down to a multiple, which needs floor division
($k y \leq x < k(y+1)$), outside our fragment.

Well-typedness means that the constraints \emph{can} hold, not that they hold for
every input: free dimension variables, such as the size of an image loaded from
disk, stay symbolic, and \emph{some} assignment must satisfy all constraints. A run whose inputs violate them can still fail, for
example when two \nodename{LoadImage} branches meeting at a blend differ in
size. As in hybrid type checking~\cite{flanagan2006}, the checker therefore
reports such constraints to be checked at run time.
Section~\ref{sec:gradual} extends them to values produced by nodes without
schemas and states which runs the guarantee covers.

Two scope decisions complete the calculus. Operators with genuinely dynamic
behavior, such as an upscaler with a float ratio widget, give their outputs
fresh symbolic dimensions, related to the inputs only when the ratio is a widget
literal. Batch dimensions are out of scope: ComfyUI broadcasts batches at almost
every node, so a batch index would add complexity without matching a failure
class in our study.

\input{figures/fig-worked-example}

\parh{Worked Example.}
Figure~\ref{fig:worked-example} shows the checker on two variants of a LoRA
workflow. In (a), an SD1.5 LoRA is inserted between an SDXL checkpoint and the
sampler. The model database resolves the checkpoint to $\brand{SDXL}$ and the
LoRA to $\brand{SD15}$, so the checkpoint loader emits
$\tyModel{\brand{SDXL}}{\brand{full}}$. \nodename{LoraLoader}'s side condition
then requires $\archfn(\brand{SDXL}) = \archfn(\brand{SD15})$. These are two
distinct roots of the derivation forest, so unification fails and the graph is
ill-typed. In (b), an SDXL LoRA is applied to a Pony checkpoint. Because
$\archfn(\brand{Pony}) = \brand{SDXL}$, the side condition holds and the graph
is well-typed. The model type still carries the precise family $\brand{Pony}$.

%% file: figures/fig-type-syntax.tex
\begin{figure}[t]
\setlength{\arraycolsep}{3pt}
\[
\begin{array}[t]{rcl}
\tau & ::= & \tycon{Image}(e,e,e,s) \mid \tyLatent{e}{e}{e}{f} \mid \tyMask{e}{e} \\
    & \mid & \tyCond{f} \mid \tyModel{f}{q} \mid \tyClip{f} \mid \tyVae{f} \\
    & \mid & \tyLora{f} \mid \tyCtrl{s}{f} \mid \tyInt{\nu} \\
    & \mid & \tycon{Float} \mid \tycon{String} \mid \tyDyn \\
\nu & ::= & k \mid x
\end{array}
\quad
\begin{array}[t]{rcl}
e & ::= & k \mid x \mid e + e \mid k \cdot e \mid e / k \\
f & ::= & \brand{SD15} \mid \brand{SDXL} \mid \brand{Pony} \mid \cdots \mid \alpha \\
s & ::= & \brand{Edge} \mid \brand{Depth} \mid \brand{Pose} \mid \cdots \mid \sigma \\
q & ::= & \brand{full} \mid \brand{fp8} \mid \brand{int8} \mid \cdots \mid \rho \\
b & ::= & f \mid s \mid q \mid \archfn(f) \\
\varphi & ::= & b = b \mid e = e \mid (k \mid e) \mid \varphi \wedge \varphi \mid \top
\end{array}
\]
\caption{Types, linear index expressions, brands, and constraints. $k$ ranges
over integer literals, $x$ over dimension variables, and $\alpha, \sigma, \rho$
over family, modality, and quantization variables. $(k \mid e)$ is a
divisibility atom. $\tyDyn$ is the dynamic type.}
\label{fig:type-syntax}
\end{figure}

%% file: figures/fig-typing-rules.tex
\begin{figure}[t]
\begin{mathpar}
\inferrule*[right=Node]
  {\Sigma(\mathit{op}(n)) = \forall \bar\alpha, \bar x.\,
     (\tau_1, \ldots, \tau_k) \xrightarrow{\varphi} (\tau'_1, \ldots, \tau'_m) \\\\
   \theta_n \text{ fresh renaming of } \bar\alpha, \bar x \\
   \theta_n^{w} \text{ instantiates widget variables by } \mathit{wid}(n)}
  {\Gamma \vdash n :
     (\theta_n^{w}\theta_n\tau_1, \ldots) \xrightarrow{\theta_n^{w}\theta_n\varphi}
     (\theta_n^{w}\theta_n\tau'_1, \ldots)}
\\
\inferrule*[right=Edge]
  {((n, i) \to (n', j)) \in E \\\\
   \mathit{out}_i(n) = \tau \\ \mathit{in}_j(n') = \tau'}
  {\mathcal{C} \mathrel{{+}{=}} \tau \doteq \tau'}
\and
\inferrule*[right=Graph]
  {\mathcal{C}(G) = \bigwedge_{e \in E} \tau_e \doteq \tau'_e
     \;\wedge\; \bigwedge_{n \in N} \varphi_n \\\\
   \mathcal{H} \models \mathcal{C}(G)}
  {\wf{G}}
\end{mathpar}
\caption{Well-typedness rules. $\tau \doteq \tau'$ equates two types
structurally: it requires equal constructors and equates their indices
position by position. $\mathcal{H}$ is the hybrid interpretation
defined in the text. Rules for $\tyDyn$ are in
Figure~\ref{fig:consistency}.}
\label{fig:typing-rules}
\end{figure}

%% file: figures/fig-worked-example.tex
\begin{figure}[t]
\centering
\definecolor{wBad}{HTML}{B5474B}\definecolor{wOk}{HTML}{009E73}\newcommand{\wtype}[4]{\node[anchor=#3, font=\cftext, inner sep=0.5pt,
  fill=white, fill opacity=0.85, text opacity=1] at (#1,#2) {$#4$};}\newcommand{\wdb}[2]{\node[anchor=north west, font=\cftext, text=black!60,
  inner sep=0pt] at ([yshift=-0.07cm]#1-bl) {#2};}\newcommand{\wpanel}[9]{\cfpanel{0}{#1}{12.7}{#1-2.05}{#2}{#3}{}
  \begin{scope}[shift={(0.25,#1-0.42)}]
    \cfnode{ck}{0}{0}{1.75}{4}{#2}{Load Checkpoint}
    \cfout{ck}{1}{MODEL}{tMODEL}\cfout{ck}{2}{CLIP}{tCLIP}\cfout{ck}{3}{VAE}{tVAE}
    \cfwid{ck}{4}{#4}
    \coordinate (ck-bl) at (0,{-\cfth-4*\cfrh-0.05});
    \wdb{ck}{db: $f = \mathsf{#5}$}
    \cfnode{lo}{3.55}{0}{2.1}{4}{#2}{Load LoRA}
    \cfin{lo}{1}{model}{tMODEL}\cfout{lo}{1}{MODEL}{tMODEL}
    \cfin{lo}{2}{clip}{tCLIP}\cfout{lo}{2}{CLIP}{tCLIP}
    \cfwid{lo}{3}{#6}\cfwid{lo}{4}{strength\_model: 1.00}
    \coordinate (lo-bl) at (3.55,{-\cfth-4*\cfrh-0.05});
    \wdb{lo}{db: $g = \mathsf{#7}$}
    \ghoston
    \cfghost{ks}{7.55}{0}{1.1}{1}{KSampler}\cfin{ks}{1}{model}{tMODEL}
    \cfghost{te}{7.55}{-0.75}{1.75}{1}{CLIP Text Encode}\cfin{te}{1}{clip}{tCLIP}
    \ghostoff
    \draw[wire] (ck-o1) -- (lo-i1);
    \draw[wire] (ck-o2) -- (lo-i2);
    \draw[wire] (lo-o1) -- (ks-i1);
    \draw[wire] (lo-o2) to[out=0,in=180] (te-i1);
    \wtype{2.65}{-0.37}{south}{#8}
    \wtype{2.65}{-0.72}{north}{\tyClip{\mathsf{SDXL}}}
    \wtype{6.6}{-0.37}{south}{#8}
  \end{scope}
  \node[anchor=north west, align=left, font=\fontsize{6}{7.5}\sffamily,
        text width=2.6cm, inner sep=0pt] at (9.95,#1-0.55) {#9};}\resizebox{\textwidth}{!}{\begin{tikzpicture}[wire/.style={line width=0.7pt, #1}, wire/.default=black!45]
  \wpanel{0}{wBad}{(a) SD1.5 LoRA on an SDXL checkpoint}
    {sdxl.safetensors}{SDXL}{sd15\_style.safetensors}{SD15}
    {\tyModel{\mathsf{SDXL}}{\mathsf{full}}}
    {Load LoRA requires\\ $\archfn(f) = \archfn(g)$:\\[1pt]
     $\mathsf{SDXL} \doteq \mathsf{SD15}$ fails\\[1pt]
     \textcolor{wBad}{\textbf{ill-typed}}}
  \wpanel{-2.2}{wOk}{(b) SDXL LoRA on a Pony fine-tune}
    {pony\_v6.safetensors}{Pony}{sdxl\_style.safetensors}{SDXL}
    {\tyModel{\mathsf{Pony}}{\mathsf{full}}}
    {Load LoRA requires\\ $\archfn(f) = \archfn(g)$:\\[1pt]
     $\archfn(\mathsf{Pony}) = \mathsf{SDXL}$ holds\\[1pt]
     \textcolor{wOk}{\textbf{well-typed}}}
\end{tikzpicture}}
\caption{Inferred edge types for (a)~an SD1.5 LoRA on an SDXL checkpoint
and (b)~an SDXL LoRA on a Pony fine-tune.}
\label{fig:worked-example}
\end{figure}

%% file: sections/gradual.tex
\section{Gradual Extension and Metatheory}
\label{sec:gradual}

The calculus of \S\ref{sec:calculus} checks only nodes that have schemas.
Real workflows, however, also use custom and hosted-API nodes that have
none. In the real-world workflows of RQ4, such nodes make up 18\% of node
instances (\S\ref{sec:eval:real}). Requirement~\req{3} therefore asks that
checking extend to workflows that mix both kinds of nodes. This section adds a
dynamic type for nodes without schemas (\S\ref{sec:grad:dyn}) and states what
acceptance then guarantees (\S\ref{sec:grad:meta}), with full proofs in
Appendix~\ref{app:supp}.

\subsection{The Dynamic Type and Consistency}
\label{sec:grad:dyn}

Any operator without a schema is assigned the default schema
$(\tyDyn, \ldots, \tyDyn) \to (\tyDyn, \ldots, \tyDyn)$ over its
nominal arity: one $\tyDyn$ per port, but still respecting the
\emph{nominal} port labels, which ComfyUI's own metadata always
provides. $\tyDyn$ thus refines a known nominal label with unknown
indices.

\input{figures/fig-consistency}

An edge between a typed and an untyped port therefore needs its own rule. As in
gradual typing~\cite{siek2006}, an edge requires its two types to be \emph{consistent} ($\consistent$, Figure~\ref{fig:consistency}), not equal. Since
$\tyDyn$ is consistent with every type, an edge with a $\tyDyn$ end adds no
constraint. Moreover, nothing is known about how a node without a schema
transforms its inputs, so its outputs receive fresh $\tyDyn$ types and
constraints do not pass through it. As a result, a missing schema weakens checking only around its node. Constraints among the typed nodes on either side
of such a node are still checked. Only the constraints that would relate the two
sides are lost. One could instead guess how types flow through nodes without schemas,
for example that a custom resize node preserves the family of a latent. We do
not, because a wrong guess would reject working workflows.

\subsection{Metatheory}
\label{sec:grad:meta}

When \tool accepts a workflow, a user wants to know whether its run can still
fail with a compatibility error. We prove that no node with a schema raises one,
whether it would crash or fail silently (Theorem~\ref{thm:soundness}). Silent failures
count: decoding a latent with another architecture's VAE raises no exception but
produces a wrong image, and the proofs treat it as a compatibility error by tracking which
model produced each value. We also prove that
removing a schema never turns an accepted workflow into a rejected one
(Theorem~\ref{thm:gradual}). Both results assume that schemas are faithful to
their nodes (\S\ref{sec:calc:schemas}). They also make one assumption about the
values that \tool cannot see before a run.

Some values are known only when a workflow runs: the sizes of loaded images,
model files the database does not know, outputs whose sizes a schema leaves open
(\S\ref{sec:calc:judgment}), and untyped outputs. We call the
ports where these values enter the checked part of a graph its \emph{frontier},
and the part of it fed by untyped output ports its \emph{boundary}.
\tool treats each such value as an unknown and accepts the graph if \emph{some}
choice of the unknowns satisfies all constraints. Along the way, it derives the
conditions that these values must meet, for example that the sides of a loaded
image are multiples of the VAE's downscaling factor, and reports them to the user.

At runtime, ComfyUI chooses these values, for example through the image a user
loads. We assume that the values supplied during execution jointly extend
to an integer solution $\theta$ of $\mathcal{C}(G)$.
We call such a run \emph{frontier-conformant}.
The checker reports necessary conditions for conformance; these conditions
are not sufficient in general.

\begin{theorem}[Soundness modulo schema faithfulness]
\label{thm:soundness}
Let $G$ be well-typed with all schemas faithful, and let $E$ be a frontier-conformant
execution of $G$ with witness $\theta$. Then every node with a schema that
executes in $E$ raises no shape, channel, architecture, family, modality, or
quantization error, and every value it produces inhabits $\theta\tau'$ for its
solved output type $\tau'$. Hence the first such error in $E$, if any, occurs at
a $\tyDyn$-typed node.
\end{theorem}

\begin{proofsketch}
By induction over the nodes in topological order. Each input of a node with a
schema comes either from another such node or from the frontier. In the first
case, the producer's output inhabits its type by the induction hypothesis, and
the type carries across the edge because $\theta$ satisfies the edge's
equations. In the second case, conformance gives inhabitation directly. Since
$\theta$ also satisfies the node's side condition, faithfulness yields normal
execution and outputs that inhabit their types.
\end{proofsketch}

Theorem~\ref{thm:soundness} leaves two questions for mixed workflows. Can
removing a schema cause an accepted workflow to be rejected? And when a node
with a schema still fails, where does the fault lie? Theorem~\ref{thm:gradual}
answers both.

\begin{theorem}[Gradual guarantee]
\label{thm:gradual}
(i)~If $\wf{G}$ and $G'$ replaces the schema of any node of $G$ with the
$\tyDyn$ default, then $\wf{G'}$.
(ii)~Let $G$ be well-typed with all schemas faithful, and let the graph inputs of an
execution $E$ extend to an integer solution of $\mathcal{C}(G)$. If $E$ raises
a shape, channel, architecture, family, modality, or quantization error at a
node with a schema, then the graph inputs and the values entering typed ports
from untyped output ports do not jointly extend to such a solution.
\end{theorem}

\begin{proofsketch}
(i)~Removing a schema only removes constraints. The node's side condition and
widget bindings disappear, and its edges become consistency checks that add
nothing. Removing constraints preserves satisfiability. (ii)~follows from
Theorem~\ref{thm:soundness}. The error shows that $E$ is not
frontier-conformant. Since the graph inputs extend by assumption, the conflict
must involve a boundary value.
\end{proofsketch}

Part~(i) is the static gradual guarantee~\cite{siek2015} for the precision
order that erases whole schemas. It lets the schema library grow one node at a
time, because a missing schema never causes a false alarm. Part~(ii) shows that
a remaining error involves a value from the boundary. It is weaker than
blame in gradual typing~\cite{wadler2009}, which names a single culprit through
runtime casts. \tool inserts no casts, and the conflict may involve several
boundary values. Together, the theorems ensure that a workflow
with missing schemas is checked as far as its schemas reach, and that any
remaining compatibility error starts where they end.

%% file: figures/fig-consistency.tex
\begin{figure}[t]
\[
\inferrule{ }{\tyDyn \consistent \tau}
\quad
\inferrule{ }{\tau \consistent \tyDyn}
\quad
\inferrule{ }{\mathsf{C}(\bar\iota) \consistent \mathsf{C}(\bar\iota')}
\qquad
\inferrule
  {((n,i) \to (n',j)) \in E \quad
   \mathit{out}_i(n) = \tau \quad \mathit{in}_j(n') = \tau' \quad
   \tau \consistent \tau'}
  {\mathcal{C} \mathrel{{+}{=}} \tau \doteq \tau'}
\]
\caption{Consistency (left three rules) and \textsc{Edge-Grad} (right), which
replaces \textsc{Edge} of Figure~\ref{fig:typing-rules}. $\mathsf{C}$ is a
type constructor and $\bar\iota$ its indices. When $\tau$ or $\tau'$ is
$\tyDyn$, $\tau \doteq \tau'$ is read as $\top$, so the edge adds no
constraint.}
\label{fig:consistency}
\end{figure}

%% file: sections/implementation.tex
\section{Implementation}
\label{sec:impl}

\tool checks ComfyUI workflows directly from their JSON representation, in two
parts. Node schemas and a model database supply the facts that ComfyUI's nominal
labels omit (\S\ref{sec:impl:dsl}). The checker then turns a workflow into
constraints over these facts and solves them (\S\ref{sec:impl:solve}). The
implementation totals $\sim$3.7k lines of Python, including the diagnostics
evaluated in RQ5 (\S\ref{sec:eval:users}).

\subsection{Schema DSL and Database}
\label{sec:impl:dsl}

The schemas of \S\ref{sec:calc:schemas} are written in a small declarative DSL:

\noindent
\begin{minipage}[t]{0.6\linewidth}
\begin{lstlisting}[language=schemadsl,basicstyle=\footnotesize\ttfamily,frame=single,framerule=0.4pt,framesep=3pt,rulecolor=\color{black!60},xleftmargin=4pt,xrightmargin=4pt,aboveskip=2pt,belowskip=2pt,showlines=true]
node VAEEncode
  forall a, h, w
  in  pixels : Image(h, w, 3)
  in  vae    : VAE(a)
  out latent : Latent(h/sf(a), w/sf(a), ch(a), a)

node WanImageToVideo   # excerpt
  widget width  : Int(w)
  widget height : Int(h)
  where 16 | w, 16 | h
\end{lstlisting}
\end{minipage}\hfill
\begin{minipage}[t]{0.39\linewidth}
\begin{lstlisting}[language=schemadsl,basicstyle=\footnotesize\ttfamily,frame=single,framerule=0.4pt,framesep=3pt,rulecolor=\color{black!60},xleftmargin=4pt,xrightmargin=4pt,aboveskip=2pt,belowskip=2pt,showlines=true]
node ControlNetApplyAdvanced
  forall f, s, m, h, w
  in  positive : Cond(f)
  in  negative : Cond(f)
  in  control_net : Ctrl(s, f)
  in  image : Image(h, w, 3, m)
  out positive : Cond(f)
  out negative : Cond(f)
  where s = m

\end{lstlisting}
\end{minipage}

Each schema names its type variables (\texttt{forall}), gives a type to every
input and output port, binds widgets to variables, and states its side condition
(\texttt{where}). We wrote schemas for 147 nodes: 114 ComfyUI built-in nodes and
33 custom nodes from popular community packs.

A schema relates the brands of its ports, for example requiring a LoRA and the
model it patches to share an architecture. Which brands a particular file has,
however, depends on the file the user selects. Therefore we provide a
\emph{model database} supplying these brands. It records the derivation forest
of families (\S\ref{sec:calc:types}) and maps checkpoint, LoRA,
ControlNet, VAE, and CLIP files to family, modality, and quantization brands, by
filename and, more reliably, by safetensors header metadata such as tensor key
patterns and shapes. When a file is
unknown, its brand variables remain free. The graph is still checked for
internal consistency, and the file's constraints still propagate, but its actual
brands go unchecked and the checker reports it as unresolved.

\subsection{Constraint Generation and Solving}
\label{sec:impl:solve}

Two obstacles separate the calculus from a usable checker. Saved workflows
contain editor constructs that the calculus does not model, and a checker
meant to run in the editor must answer within milliseconds.

\emph{From saved workflows to core graphs.} The checker lowers these constructs
as described in \S\ref{sec:calc:graphs} and then generates constraints exactly
as the calculus prescribes.

\emph{Fast solving.} Solving is fast because brands are solved first. Brand
equations are solved by union-find, and grounding a family evaluates $\archfn$,
$\chfn$, and $\mathit{sf}$, which makes the dimension equations that mention it
concrete. The resulting systems are also small, with a median of 100 constraints per
graph in our evaluation (\S\ref{sec:eval:cost}). A dedicated engine
without native dependencies therefore suffices to decide them exactly with the
procedure of \S\ref{sec:calc:judgment}
when side predicates have ground family arguments.
As a result, median checking time is 5.46\,ms.

%% file: sections/evaluation.tex
\section{Evaluation}
\label{sec:eval}

Section~\ref{sec:study} answered RQ1: semantic composition failures are common,
and 31.5\% of them are silent. The evaluation asks whether \tool
addresses them, through four further questions:
\begin{itemize}
  \item[\textbf{RQ2}] How reliably does \tool detect compatibility faults
    injected into real workflows, and how does detection change when schemas
    are unavailable?
  \item[\textbf{RQ3}] How much time, memory, and solver work does
    checking these workflows require?
  \item[\textbf{RQ4}] What do known-good workflows and reported real
    failures reveal about practical coverage and failure mechanisms?
  \item[\textbf{RQ5}] How do \tool's diagnostics help
    users to locate compatibility faults and select valid repairs in
    ComfyUI workflows?
\end{itemize}

\parh{Evaluation data.} RQ2 and RQ3 use faults injected into 132
official ComfyUI workflows. RQ4 tests whether the results carry over to
known-good and reported real workflows, and RQ5 involves 10 participants
working on 12 workflows, with full details in Appendices~\ref{app:eval}
and~\ref{app:user-study}.

\subsection{RQ2: Effectiveness on Injected Faults}
\label{sec:eval:synthetic}

Measuring detection requires faulty workflows across every compatibility
dimension of the census, which real failure reports rarely provide. Only 89 of
the 753 failures in our census include the faulty workflow graph, and only one
of these involves control modality. Rebuilding graphs from the reports' text
adds failures, but still yields only 108 that execution reproduces (RQ4), 82\%
of them model-family mismatches. We therefore inject faults into official
workflows that are known to run, with operators that cover every dimension of
the census. RQ4 then establishes realism on the failures that users report.

ComfyUI ships 765 official workflows, so we build our corpus from them. The
operators swap model files, rewire values within a workflow, graft one
workflow's latent or image into another, and break dimension, quantization, and
modality requirements. Of the 765 workflows, 156 use only built-in nodes and
public models within our GPU budget and contain a component that some injection
operator can change, such as a model file to swap. Of these, 24 fail to run
unmodified, mostly because they are reusable subgraph components, which leaves
132. Each injection operator may yield both candidate faults and benign graphs.
We use model documentation and ComfyUI's source code, never \tool, to determine
whether each resulting graph is a candidate fault or benign. This gives 1,837
graphs, 619 of them candidate faults. We execute every graph to confirm the
candidate faults. Of the 619, 12 count in no measurement because their
baseline does not run correctly. Of the other 607, 417 raise an error, and 83
are confirmed as silent faults, 42 by an image oracle and 41 by manual review.
This yields 500 confirmed faults. The remaining 107 stay unconfirmed and do not
count either. We compare \tool with ComfyUI's own validator.

\input{figures/fig-synthetic-detection}

\parh{Results.} ComfyUI's own validator admits all 500 confirmed faults (i.e.,
ComfyUI rejects none). Of these, \tool detects 332
(66.4\%) (Figure~\ref{fig:synthetic-detection}a), including all
116 latent grafts, 174 of 275 file swaps, 22 of 34 rewirings, and
16 of 52 dimension faults.

\parh{Measurement on Missing Schemas.} Because many nodes lack schemas, we also
measure how missing schemas affect detection. We use the 169
confirmed faulty graphs that \tool detects and in which every node has a schema,
so that only the schemas we remove are missing. We then remove the schemas of
node classes in 100 random orders (Figure~\ref{fig:synthetic-detection}b). With
50\% of the node classes removed, 48.0\% of nodes keep their
schemas but only 17.7\% of these 169 detections
remain, because 29 of the 123 classes carry them
all. No run raises a false alarm on schema removals, as
Theorem~\ref{thm:gradual}(i) guarantees.

\medskip
\takeaway{RQ2 takeaway}{\tool detects injected faults that ComfyUI's own
validator admits. Missing schemas
can hide faults but never cause false alarms.}

\subsection{RQ3: Cost and Resource Use}
\label{sec:eval:cost}

\input{study/data/rq3_real_macros}
A checker that runs before every execution must cost far less than the execution
it guards. Figure~\ref{fig:synthetic-cost} reports checking time, including
diagnostics, solver work, and peak memory on the graphs of RQ2.

\input{figures/fig-synthetic-cost}

\parh{Time (Figure~\ref{fig:synthetic-cost}a).} Checking takes
\CostCheckMedianMs\,ms at the median and \CostCheckPctMs\,ms at the 95th
percentile. Execution takes \CostExecRatio{} times as long at the median over
graphs. A rejection therefore costs the user milliseconds instead of a GPU run.

\parh{Solver work and memory (Figure~\ref{fig:synthetic-cost}b--d).} The
constraint systems are small, with a median of 100 constraints over 96 variables
per graph. Accepting a graph takes one solver call. Rejecting one takes 110
calls at the median, because diagnosis minimizes the unsatisfiable
core~\cite{marco2016} and computes repair hints. Peak memory, with the Python
interpreter, is \CostMemMedianMiB\,MiB at the median and at most
\CostMemMaxMiB\,MiB.

\medskip
\takeaway{RQ3 takeaway}{\tool is cheap enough to check every workflow before it
runs, and it spends extra work only on explaining errors.}

\subsection{RQ4: Effectiveness Study on Real-World Workflows}
\label{sec:eval:real}

RQ2 injects faults into official workflows that use only built-in nodes. Real
workflows span more model families and contain custom nodes without schemas, and
the mistakes that users make may differ from the injected faults. RQ4 therefore
examines whether \tool's results carry over to real workflows, from two sides.
(1) Known-good workflows show whether the checker accepts real graphs and how
much of them it types. We use all 765 official workflows of RQ2. Because ComfyUI
maintains them as working workflows and components, any rejection would be a
false alarm. (2) Reported failures show whether the checker detects the failures
that users actually meet. Such failures are hard to obtain, because few reports
include a workflow that reproduces the failure. We therefore draw them from the
census of 753 reports in \S\ref{sec:study} and rebuild their graphs from the
reports.

\input{figures/fig-coverage}

\parh{Acceptance and coverage (Figure~\ref{fig:coverage}).} \tool accepts all
765 known-good workflows, so it raises no false alarm on real graphs across the
families they span. It also types most nodes. Across the 765 workflows, 82\% of
all node instances are typed. Per workflow, the median coverage is 85.7\%, and
39.1\% of workflows are fully typed. The untyped nodes come from custom packs
and hosted APIs, which receive $\tyDyn$. At the left end of the distribution,
13.6\% of workflows have at most 50\% of their nodes typed and 3.1\% have none,
because they consist mostly of such nodes. Low coverage limits detection in
these workflows but never causes a false alarm.

\begin{table}[t]
\caption{Representative reported failures that \tool detects. Numbers refer to ComfyUI GitHub threads.}
\label{tab:decidable}
\resizebox{0.85\columnwidth}{!}{\begin{tabular}{@{}lll cc@{}}
\toprule
ID & Mode & Fault & ComfyUI & \tool \\
\midrule
Discussion~\#2115 & crash  & LoRA file loaded as a checkpoint                          & $\times$ & \checkmark \\
Issue~\#3690      & crash  & SD3 checkpoint without text encoders used for prompts     & $\times$ & \checkmark \\
Issue~\#5098      & crash  & SD1.5 ControlNet on an SDXL checkpoint                    & $\times$ & \checkmark \\
Issue~\#7335      & crash  & Control-LoRA on a 9-channel SDXL inpaint UNet             & $\times$ & \checkmark \\
Issue~\#8362      & crash  & Wan VACE width not divisible by 16                        & $\times$ & \checkmark \\
Issue~\#9479      & crash  & Wan text encoder loaded with the wrong loader type        & $\times$ & \checkmark \\
\midrule
Issue~\#6813      & silent & Scaled-fp8 HunyuanVideo weights without their scales      & $\times$ & \checkmark \\
Issue~\#7407      & silent & Flux depth model combined with a depth LoRA               & $\times$ & \checkmark \\
Issue~\#4611      & silent & Flux text encoders given two T5 files and no CLIP-L       & $\times$ & \checkmark \\
\bottomrule
\end{tabular}}
\end{table}

\parh{Reported failures.} We obtain these failures by filtering the census of
RQ1 in two steps. First, we keep the semantic failures whose reports include the
faulty workflow graph, or whose text and logs describe it precisely enough to
rebuild it, and whose components are all built-in nodes. This leaves 129 of
the 753 reports. Second, we try to reconstruct
each graph in two versions (faulty and fixed) based on the report and the
available logs, execute them, and keep the graph only if it fails while its
fixed version runs. This results in 108 faulty workflow graphs. We deliberately
keep graphs whose faults lie outside what \tool's types can express, as well as
graphs with nodes without schemas. The evaluation therefore
covers the broad range of failures that users meet, not only those that the
checker is designed to catch.

We compare \tool with ComfyUI on those graphs. ComfyUI's own validator admits
all 108 graphs, but \tool rejects 72 (67\%) of them from the graphs alone,
without the reporters' models or environments.
It detects 64 of 88 crashing failures and 8 of 20 silent failures.
Table~\ref{tab:decidable} lists
representative failures.

\medskip
\takeaway{RQ4 takeaway}{On real workflows, \tool raises no false alarm and types
most nodes. Its mechanisms carry over from injected faults to reported failures,
across newer model families and including silent ones.}

\subsection{\texorpdfstring{RQ5: User Study}{RQ5: User Study}}
\label{sec:eval:users}

\parh{Setup.}
To assess whether \tool's diagnostics help users locate faults and choose
repairs, we conduct a within-participant study with 10 participants. Each
participant works both with and without diagnostics, allowing us to compare
their performance while accounting for differences in prior experience. The
study includes 12 workflows: 6 with silent compatibility faults, 2 that crash,
and 4 valid ones, which test whether participants can also recognize when no
repair is necessary. Each participant sees every workflow once to avoid recall
from an earlier attempt, with 6 tasks without diagnostics (baseline) and 6 with
diagnostics (treatment). To balance task composition, each set includes 3 silent
faults, 1 crash, and 2 valid workflows, and each workflow appears in each
version for 5 participants. Both versions show the workflow graph, intended
behavior, model-file information, and output image or error trace. Only the
treatment adds \tool's report, isolating diagnostic availability as the
difference in task information.

\parh{Procedure.}
We recruit participants with at least 3 months of experience with generative AI
workflows, so that the study focuses on users familiar with these tools.
Participants work in a read-only web interface, so the study measures how
participants interpret diagnostics and choose repairs, not how they edit or
execute workflows. A common tutorial, unscored practice tasks, and a checker
briefing introduce the interface and the report format. For each workflow,
participants judge whether it is faulty and, if so, identify the relevant region
and node, then select a repair category and a specific repair. They answer
through multiple-choice questions, which make scoring consistent. The interface
records all responses and interaction events.

\input{figures/fig-user-study}

\parh{Scoring.} We score locating a fault separately from repairing it. A
faulty workflow counts as diagnosed when the participant identifies the fault's
region and node within the time limit and, for a silent fault, also judges the
workflow faulty. A valid workflow counts as diagnosed when the participant
reports no fault. A repair additionally requires the correct repair category
and a valid specific repair. Timeouts receive no credit.

\parh{Results.} Diagnostics help participants both find faults and fix them.
With diagnostics, diagnosis accuracy rises from 55.0\% to 91.7\%
(Figure~\mbox{\ref{fig:user-study}}a). On faulty workflows, localization rises
from 35.0\% to 87.5\%, and repair from 30.0\% to 62.5\%. The improvement is
consistent across participants. No participant is less accurate with diagnostics
on any outcome (Figure~\mbox{\ref{fig:user-study}}b), and exploratory two-sided
exact sign tests on paired accuracies are significant for all three outcomes
(Holm-adjusted~\cite{holm1979} \mbox{$p=0.0117$} for diagnosis and localization
and \mbox{$p=0.0156$} for repair). Diagnosis accuracy rises from 36.7\% to
83.3\% on silent faults and from 30.0\% to 100\% on crashes
(Figure~\mbox{\ref{fig:user-study}}c). On valid workflows, it stays near perfect
(19 of 20 trials without diagnostics, 20 of 20 with), so the diagnostics do not
lead participants to see faults that are not there.

\medskip
\takeaway{RQ5 takeaway}{\tool's diagnostics help users locate compatibility
faults and choose valid repairs, especially for silent faults.}
\FloatBarrier

%% file: figures/fig-synthetic-detection.tex
\begin{figure}[t]
\centering
\includegraphics{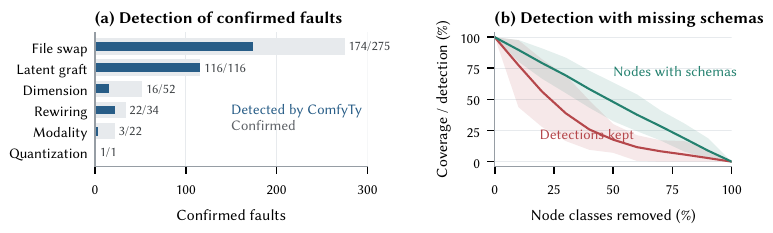}
\caption{\textbf{Detection and missing schemas.} (a)~Confirmed faults per operator family and
their detection. (b)~Detections kept and node coverage as schemas are removed.}
\Description{Horizontal bars per operator family compare confirmed faults with the
detections of ComfyTy. ComfyTy detects all latent grafts and most rewirings. Two lines with shaded bands show
that, as the schemas of node classes are removed, the detections kept fall faster than the
fraction of nodes with schemas.}
\label{fig:synthetic-detection}
\end{figure}

%% file: study/data/rq3_real_macros.tex
\newcommand{\CostGraphs}{1,837}
\newcommand{\CostCheckMedianMs}{5.46}
\newcommand{\CostCheckPctMs}{83.72}
\newcommand{\CostExecRatio}{3,426}
\newcommand{\CostConstraintsLo}{0}
\newcommand{\CostConstraintsHi}{756}
\newcommand{\CostVariablesLo}{0}
\newcommand{\CostVariablesHi}{603}
\newcommand{\CostAcceptCallsLo}{1}
\newcommand{\CostAcceptCallsHi}{1}
\newcommand{\CostRejectCallsLo}{1}
\newcommand{\CostRejectCallsHi}{627}
\newcommand{\CostZThreeCalls}{0}
\newcommand{\CostMemMedianMiB}{25.62}
\newcommand{\CostMemMaxMiB}{26.68}

%% file: figures/fig-synthetic-cost.tex
\begin{figure}[t]
\centering
\includegraphics{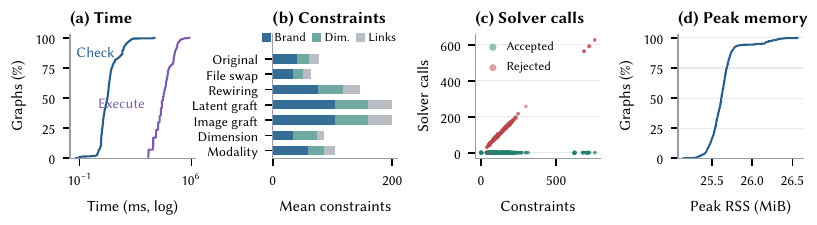}
\caption{\textbf{Checking cost and resource use.} (a)~Checking and execution time. (b)~Mean constraints per operator family, by kind.
(c)~Solver calls per graph, including diagnosis. (d)~Peak process memory.}
\Description{Four panels compare checking and execution time, show
constraint composition for each operator family, plot solver calls against
constraint size for accepted and rejected graphs, and show the distribution of
peak memory.}
\label{fig:synthetic-cost}
\end{figure}

%% file: figures/fig-coverage.tex
\begin{figure}[t]
\centering
\resizebox{0.8\columnwidth}{!}{\input{figures/coverage_plot}}
\caption{Distribution of per-workflow typing coverage across the 765
known-good workflows.}
\label{fig:coverage}
\end{figure}
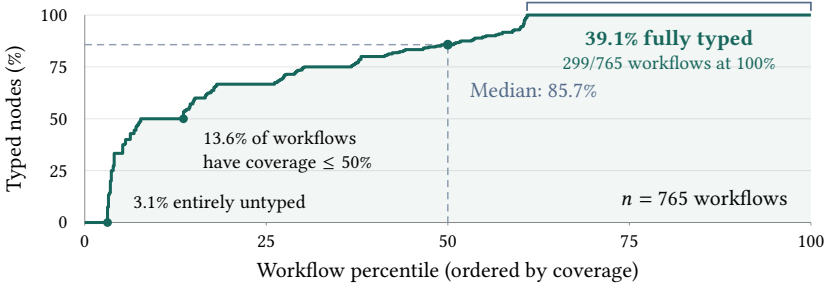

%% file: figures/coverage_plot.tex
\begin{tikzpicture}[x=0.105cm,y=0.03cm,font=\small]
\definecolor{covline}{HTML}{16665B}
\definecolor{covnote}{HTML}{546B89}
\foreach \y in {0,25,50,75,100} {
  \draw[black!10,line width=0.3pt] (0,\y) -- (100,\y);
  \node[anchor=east,font=\footnotesize] at (-1.2,\y) {\y};}
\foreach \x in {0,25,50,75,100} {
  \draw[black!45,line width=0.4pt] (\x,0) -- (\x,-1.4);
  \node[anchor=north,font=\footnotesize] at (\x,-2) {\x};}
\draw[black!45,line width=0.4pt] (0,100) -- (0,0) -- (100,0);
\node[rotate=90] at (-9.5,50) {Typed nodes (\%)};
\node at (50,-24) {Workflow percentile (ordered by coverage)};
\node[anchor=south east] at (98,4) {$n=765$ workflows};
\draw[covnote!65,densely dashed,line width=0.6pt] (0,85.714286) -- (50,85.714286) -- (50,0);
\path[fill=covline,fill opacity=0.055] (0.000000,0.000000) --
  (3.137255,0.000000) --
  (3.137255,7.407407) --
  (3.267974,7.407407) --
  (3.267974,13.333333) --
  (3.398693,13.333333) --
  (3.398693,14.285714) --
  (3.529412,14.285714) --
  (3.529412,20.000000) --
  (3.660131,20.000000) --
  (3.660131,25.000000) --
  (3.921569,25.000000) --
  (3.921569,27.272727) --
  (4.052288,27.272727) --
  (4.052288,33.333333) --
  (5.228758,33.333333) --
  (5.228758,37.500000) --
  (5.490196,37.500000) --
  (5.490196,37.931034) --
  (5.620915,37.931034) --
  (5.620915,40.000000) --
  (6.274510,40.000000) --
  (6.274510,42.857143) --
  (6.797386,42.857143) --
  (6.797386,44.444444) --
  (7.058824,44.444444) --
  (7.058824,46.153846) --
  (7.189542,46.153846) --
  (7.189542,46.666667) --
  (7.320261,46.666667) --
  (7.320261,47.368421) --
  (7.450980,47.368421) --
  (7.450980,48.148148) --
  (7.581699,48.148148) --
  (7.581699,48.780488) --
  (7.712418,48.780488) --
  (7.712418,50.000000) --
  (13.594771,50.000000) --
  (13.594771,53.333333) --
  (13.725490,53.333333) --
  (13.725490,53.846154) --
  (13.986928,53.846154) --
  (13.986928,54.545455) --
  (14.379085,54.545455) --
  (14.379085,55.555556) --
  (14.509804,55.555556) --
  (14.509804,57.142857) --
  (14.901961,57.142857) --
  (14.901961,58.333333) --
  (15.032680,58.333333) --
  (15.032680,59.322034) --
  (15.163399,59.322034) --
  (15.163399,60.000000) --
  (16.601307,60.000000) --
  (16.601307,61.111111) --
  (16.732026,61.111111) --
  (16.732026,61.764706) --
  (16.862745,61.764706) --
  (16.862745,62.068966) --
  (16.993464,62.068966) --
  (16.993464,62.500000) --
  (17.385621,62.500000) --
  (17.385621,63.333333) --
  (17.516340,63.333333) --
  (17.516340,63.636364) --
  (17.777778,63.636364) --
  (17.777778,65.116279) --
  (17.908497,65.116279) --
  (17.908497,65.384615) --
  (18.039216,65.384615) --
  (18.039216,65.789474) --
  (18.169935,65.789474) --
  (18.169935,66.666667) --
  (26.143791,66.666667) --
  (26.143791,67.391304) --
  (26.274510,67.391304) --
  (26.274510,67.500000) --
  (26.405229,67.500000) --
  (26.405229,67.647059) --
  (26.535948,67.647059) --
  (26.535948,67.875648) --
  (26.666667,67.875648) --
  (26.666667,68.421053) --
  (26.797386,68.421053) --
  (26.797386,68.571429) --
  (26.928105,68.571429) --
  (26.928105,68.888889) --
  (27.058824,68.888889) --
  (27.058824,69.072165) --
  (27.189542,69.072165) --
  (27.189542,69.444444) --
  (27.320261,69.444444) --
  (27.320261,70.000000) --
  (27.450980,70.000000) --
  (27.450980,70.967742) --
  (27.581699,70.967742) --
  (27.581699,71.111111) --
  (27.712418,71.111111) --
  (27.712418,71.428571) --
  (28.888889,71.428571) --
  (28.888889,71.875000) --
  (29.019608,71.875000) --
  (29.019608,72.549020) --
  (29.150327,72.549020) --
  (29.150327,72.727273) --
  (29.281046,72.727273) --
  (29.281046,73.333333) --
  (29.542484,73.333333) --
  (29.542484,73.469388) --
  (29.934641,73.469388) --
  (29.934641,74.418605) --
  (30.065359,74.418605) --
  (30.065359,74.509804) --
  (30.196078,74.509804) --
  (30.196078,75.000000) --
  (36.732026,75.000000) --
  (36.732026,75.531915) --
  (36.862745,75.531915) --
  (36.862745,75.757576) --
  (36.993464,75.757576) --
  (36.993464,76.086957) --
  (37.124183,76.086957) --
  (37.124183,76.315789) --
  (37.254902,76.315789) --
  (37.254902,76.811594) --
  (37.385621,76.811594) --
  (37.385621,76.923077) --
  (37.647059,76.923077) --
  (37.647059,77.777778) --
  (37.777778,77.777778) --
  (37.777778,78.571429) --
  (38.039216,78.571429) --
  (38.039216,80.000000) --
  (41.307190,80.000000) --
  (41.307190,80.487805) --
  (41.437908,80.487805) --
  (41.437908,80.555556) --
  (41.568627,80.555556) --
  (41.568627,80.645161) --
  (41.699346,80.645161) --
  (41.699346,81.034483) --
  (41.830065,81.034483) --
  (41.830065,81.250000) --
  (42.352941,81.250000) --
  (42.352941,81.818182) --
  (43.137255,81.818182) --
  (43.137255,82.051282) --
  (43.398693,82.051282) --
  (43.398693,82.352941) --
  (43.790850,82.352941) --
  (43.790850,82.608696) --
  (43.921569,82.608696) --
  (43.921569,82.857143) --
  (44.052288,82.857143) --
  (44.052288,83.050847) --
  (44.183007,83.050847) --
  (44.183007,83.333333) --
  (46.535948,83.333333) --
  (46.535948,83.636364) --
  (46.666667,83.636364) --
  (46.666667,84.090909) --
  (46.928105,84.090909) --
  (46.928105,84.210526) --
  (47.320261,84.210526) --
  (47.320261,84.444444) --
  (47.450980,84.444444) --
  (47.450980,84.615385) --
  (48.235294,84.615385) --
  (48.235294,84.722222) --
  (48.366013,84.722222) --
  (48.366013,85.000000) --
  (49.019608,85.000000) --
  (49.019608,85.365854) --
  (49.150327,85.365854) --
  (49.150327,85.714286) --
  (50.849673,85.714286) --
  (50.849673,86.363636) --
  (51.111111,86.363636) --
  (51.111111,86.666667) --
  (51.503268,86.666667) --
  (51.503268,86.956522) --
  (51.764706,86.956522) --
  (51.764706,87.500000) --
  (53.202614,87.500000) --
  (53.202614,87.564767) --
  (53.333333,87.564767) --
  (53.333333,88.000000) --
  (53.594771,88.000000) --
  (53.594771,88.461538) --
  (53.725490,88.461538) --
  (53.725490,88.888889) --
  (54.771242,88.888889) --
  (54.771242,89.473684) --
  (55.163399,89.473684) --
  (55.163399,89.743590) --
  (55.294118,89.743590) --
  (55.294118,90.000000) --
  (56.601307,90.000000) --
  (56.601307,90.476190) --
  (56.993464,90.476190) --
  (56.993464,90.909091) --
  (57.516340,90.909091) --
  (57.516340,91.666667) --
  (58.823529,91.666667) --
  (58.823529,91.803279) --
  (58.954248,91.803279) --
  (58.954248,92.592593) --
  (59.084967,92.592593) --
  (59.084967,92.753623) --
  (59.215686,92.753623) --
  (59.215686,92.857143) --
  (59.869281,92.857143) --
  (59.869281,93.333333) --
  (60.000000,93.333333) --
  (60.000000,94.117647) --
  (60.130719,94.117647) --
  (60.130719,94.444444) --
  (60.392157,94.444444) --
  (60.392157,95.652174) --
  (60.522876,95.652174) --
  (60.522876,97.468354) --
  (60.653595,97.468354) --
  (60.653595,98.305085) --
  (60.784314,98.305085) --
  (60.784314,98.809524) --
  (60.915033,98.809524) --
  (60.915033,100.000000) --
  (100.000000,100.000000) -- (100,0) -- cycle;
\draw[covline,line width=1.25pt,line join=round] (0.000000,0.000000) --
  (3.137255,0.000000) --
  (3.137255,7.407407) --
  (3.267974,7.407407) --
  (3.267974,13.333333) --
  (3.398693,13.333333) --
  (3.398693,14.285714) --
  (3.529412,14.285714) --
  (3.529412,20.000000) --
  (3.660131,20.000000) --
  (3.660131,25.000000) --
  (3.921569,25.000000) --
  (3.921569,27.272727) --
  (4.052288,27.272727) --
  (4.052288,33.333333) --
  (5.228758,33.333333) --
  (5.228758,37.500000) --
  (5.490196,37.500000) --
  (5.490196,37.931034) --
  (5.620915,37.931034) --
  (5.620915,40.000000) --
  (6.274510,40.000000) --
  (6.274510,42.857143) --
  (6.797386,42.857143) --
  (6.797386,44.444444) --
  (7.058824,44.444444) --
  (7.058824,46.153846) --
  (7.189542,46.153846) --
  (7.189542,46.666667) --
  (7.320261,46.666667) --
  (7.320261,47.368421) --
  (7.450980,47.368421) --
  (7.450980,48.148148) --
  (7.581699,48.148148) --
  (7.581699,48.780488) --
  (7.712418,48.780488) --
  (7.712418,50.000000) --
  (13.594771,50.000000) --
  (13.594771,53.333333) --
  (13.725490,53.333333) --
  (13.725490,53.846154) --
  (13.986928,53.846154) --
  (13.986928,54.545455) --
  (14.379085,54.545455) --
  (14.379085,55.555556) --
  (14.509804,55.555556) --
  (14.509804,57.142857) --
  (14.901961,57.142857) --
  (14.901961,58.333333) --
  (15.032680,58.333333) --
  (15.032680,59.322034) --
  (15.163399,59.322034) --
  (15.163399,60.000000) --
  (16.601307,60.000000) --
  (16.601307,61.111111) --
  (16.732026,61.111111) --
  (16.732026,61.764706) --
  (16.862745,61.764706) --
  (16.862745,62.068966) --
  (16.993464,62.068966) --
  (16.993464,62.500000) --
  (17.385621,62.500000) --
  (17.385621,63.333333) --
  (17.516340,63.333333) --
  (17.516340,63.636364) --
  (17.777778,63.636364) --
  (17.777778,65.116279) --
  (17.908497,65.116279) --
  (17.908497,65.384615) --
  (18.039216,65.384615) --
  (18.039216,65.789474) --
  (18.169935,65.789474) --
  (18.169935,66.666667) --
  (26.143791,66.666667) --
  (26.143791,67.391304) --
  (26.274510,67.391304) --
  (26.274510,67.500000) --
  (26.405229,67.500000) --
  (26.405229,67.647059) --
  (26.535948,67.647059) --
  (26.535948,67.875648) --
  (26.666667,67.875648) --
  (26.666667,68.421053) --
  (26.797386,68.421053) --
  (26.797386,68.571429) --
  (26.928105,68.571429) --
  (26.928105,68.888889) --
  (27.058824,68.888889) --
  (27.058824,69.072165) --
  (27.189542,69.072165) --
  (27.189542,69.444444) --
  (27.320261,69.444444) --
  (27.320261,70.000000) --
  (27.450980,70.000000) --
  (27.450980,70.967742) --
  (27.581699,70.967742) --
  (27.581699,71.111111) --
  (27.712418,71.111111) --
  (27.712418,71.428571) --
  (28.888889,71.428571) --
  (28.888889,71.875000) --
  (29.019608,71.875000) --
  (29.019608,72.549020) --
  (29.150327,72.549020) --
  (29.150327,72.727273) --
  (29.281046,72.727273) --
  (29.281046,73.333333) --
  (29.542484,73.333333) --
  (29.542484,73.469388) --
  (29.934641,73.469388) --
  (29.934641,74.418605) --
  (30.065359,74.418605) --
  (30.065359,74.509804) --
  (30.196078,74.509804) --
  (30.196078,75.000000) --
  (36.732026,75.000000) --
  (36.732026,75.531915) --
  (36.862745,75.531915) --
  (36.862745,75.757576) --
  (36.993464,75.757576) --
  (36.993464,76.086957) --
  (37.124183,76.086957) --
  (37.124183,76.315789) --
  (37.254902,76.315789) --
  (37.254902,76.811594) --
  (37.385621,76.811594) --
  (37.385621,76.923077) --
  (37.647059,76.923077) --
  (37.647059,77.777778) --
  (37.777778,77.777778) --
  (37.777778,78.571429) --
  (38.039216,78.571429) --
  (38.039216,80.000000) --
  (41.307190,80.000000) --
  (41.307190,80.487805) --
  (41.437908,80.487805) --
  (41.437908,80.555556) --
  (41.568627,80.555556) --
  (41.568627,80.645161) --
  (41.699346,80.645161) --
  (41.699346,81.034483) --
  (41.830065,81.034483) --
  (41.830065,81.250000) --
  (42.352941,81.250000) --
  (42.352941,81.818182) --
  (43.137255,81.818182) --
  (43.137255,82.051282) --
  (43.398693,82.051282) --
  (43.398693,82.352941) --
  (43.790850,82.352941) --
  (43.790850,82.608696) --
  (43.921569,82.608696) --
  (43.921569,82.857143) --
  (44.052288,82.857143) --
  (44.052288,83.050847) --
  (44.183007,83.050847) --
  (44.183007,83.333333) --
  (46.535948,83.333333) --
  (46.535948,83.636364) --
  (46.666667,83.636364) --
  (46.666667,84.090909) --
  (46.928105,84.090909) --
  (46.928105,84.210526) --
  (47.320261,84.210526) --
  (47.320261,84.444444) --
  (47.450980,84.444444) --
  (47.450980,84.615385) --
  (48.235294,84.615385) --
  (48.235294,84.722222) --
  (48.366013,84.722222) --
  (48.366013,85.000000) --
  (49.019608,85.000000) --
  (49.019608,85.365854) --
  (49.150327,85.365854) --
  (49.150327,85.714286) --
  (50.849673,85.714286) --
  (50.849673,86.363636) --
  (51.111111,86.363636) --
  (51.111111,86.666667) --
  (51.503268,86.666667) --
  (51.503268,86.956522) --
  (51.764706,86.956522) --
  (51.764706,87.500000) --
  (53.202614,87.500000) --
  (53.202614,87.564767) --
  (53.333333,87.564767) --
  (53.333333,88.000000) --
  (53.594771,88.000000) --
  (53.594771,88.461538) --
  (53.725490,88.461538) --
  (53.725490,88.888889) --
  (54.771242,88.888889) --
  (54.771242,89.473684) --
  (55.163399,89.473684) --
  (55.163399,89.743590) --
  (55.294118,89.743590) --
  (55.294118,90.000000) --
  (56.601307,90.000000) --
  (56.601307,90.476190) --
  (56.993464,90.476190) --
  (56.993464,90.909091) --
  (57.516340,90.909091) --
  (57.516340,91.666667) --
  (58.823529,91.666667) --
  (58.823529,91.803279) --
  (58.954248,91.803279) --
  (58.954248,92.592593) --
  (59.084967,92.592593) --
  (59.084967,92.753623) --
  (59.215686,92.753623) --
  (59.215686,92.857143) --
  (59.869281,92.857143) --
  (59.869281,93.333333) --
  (60.000000,93.333333) --
  (60.000000,94.117647) --
  (60.130719,94.117647) --
  (60.130719,94.444444) --
  (60.392157,94.444444) --
  (60.392157,95.652174) --
  (60.522876,95.652174) --
  (60.522876,97.468354) --
  (60.653595,97.468354) --
  (60.653595,98.305085) --
  (60.784314,98.305085) --
  (60.784314,98.809524) --
  (60.915033,98.809524) --
  (60.915033,100.000000) --
  (100.000000,100.000000);
\fill[covline] (50,85.714286) circle[radius=2pt];
\node[anchor=north west,inner sep=3pt,text=covnote] at (52,71.714286) {Median: 85.7\%};
\draw[covnote,line width=0.7pt] (60.915033,102) -- (60.915033,106) -- (100,106) -- (100,102);
\node[anchor=north,align=center,inner sep=3pt,text=covline] at (80.457516,97) {\textbf{39.1\% fully typed}\\[-1pt]\footnotesize 299/765 workflows at 100\%};
\fill[covline] (13.594771,50) circle[radius=1.8pt];
\node[anchor=north west,align=left,inner sep=2pt,font=\footnotesize] at (15.594771,47) {13.6\% of workflows\\have coverage $\leq 50\%$};
\fill[covline] (3.137255,0) circle[radius=1.8pt];
\node[anchor=west,inner sep=2pt,font=\footnotesize] at (6,9) {3.1\% entirely untyped};
\end{tikzpicture}

%% file: figures/fig-user-study.tex
\begin{figure}[t]
\centering
\includegraphics{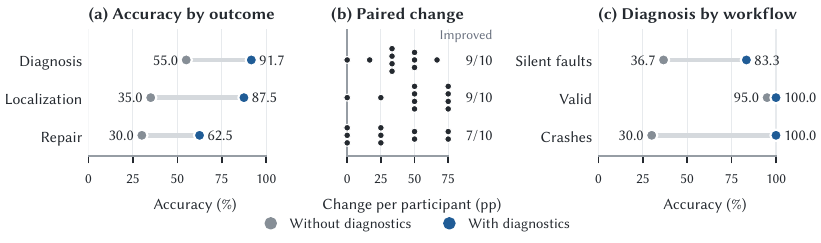}
\caption{\textbf{User-study results.} (a)~Accuracy without and with diagnostics
(60 trials per condition, 40 on faulty workflows). (b)~Change in each
participant's accuracy. (c)~Diagnosis accuracy by workflow type.}
\Description{Three panels. In panel (a), accuracy with diagnostics rises
from 55.0 to 91.7 percent for diagnosis, from 35.0 to 87.5 percent for
localization, and from 30.0 to 62.5 percent for repair. Panel (b) plots
each participant's change: every change is zero or positive; 9 of 10
participants improve in diagnosis and localization and 7 of 10 in repair.
In panel (c), diagnosis accuracy rises from 36.7 to 83.3 percent for silent
faults, from 95.0 to 100 percent for valid workflows, and from 30.0 to 100
percent for crashes.}
\label{fig:user-study}
\end{figure}

%% file: sections/discussion.tex
\section{Discussion}
\label{sec:discussion}

\subsection{Diagnostics in the Editor}
\label{sec:discussion:diagnostics}

\tool's diagnostics already localize a fault. They reduce the conflicting
constraints to a minimal unsatisfiable core, name the nodes that produced it,
and, where possible, point to the file or widget whose change restores
consistency, which lets users locate faults far more accurately (RQ5). An
editor could also highlight the implicated nodes during editing, and the model
database could suggest files from the user's library whose family matches the
rest of the graph.

\subsection{Growing the Schema Library}
\label{sec:discussion:schemas}

Detection depends on the schemas along a failure path (RQ2), so each new schema
extends what \tool can check. The gradual design lets the library grow one node
at a time. Inferring schemas from node code and execution traces could speed
this growth, because these sources reveal how nodes transform tensor dimensions
and which model architectures they require. We see growing the schema library
largely an engineering effort, which could strengthen \tool's capability without
affecting its calculus, guarantees, or results.

\subsection{Threats to Validity}
\label{sec:discussion:validity}

Appendix~\ref{sec:appendix-threats} discusses each threat in full.

\emph{Construct validity}. The failure labels come from an LLM classifier. Four
models agree with its labels at $\kappa=0.71$--$0.81$. This shows that the
labels are reliable, not that they match human judgment. Silent failures may
also go unreported, so the census may undercount them.

\emph{Internal validity.} As is standard, our faithfulness proofs assume that
PyTorch operators and model architectures behave as ComfyUI v0.27.0 implements
them, and that the model database identifies files correctly. A later release
may require re-checking affected schemas.

\emph{External validity.} RQ2 injects faults into official workflows. Injected
faults may differ from the mistakes users make, and official workflows use fewer
custom nodes than community ones. RQ4 complements RQ2 with 108 reported
failures, a small sample because few reports allow rebuilding. The census covers
ComfyUI's core GitHub venues, whose mix of failures may differ from that of
community forums. In RQ5, participants select a repair from a fixed menu without
applying it, so the study measures whether users recognize a valid repair, not
the effort of carrying it out.

%% file: sections/related.tex
\section{Related Work}
\label{sec:related}

\parh{Typing Tensor Programs.}
A rich line of work assigns types to tensor dimensions in programming languages
and ML frameworks: Dependent ML's indexed types~\cite{xi1998}, Dex's typed index
sets~\cite{dex2021}, named-axis proposals~\cite{tensorsfitting2021}, and static
shape checkers for Python ML code such as Ariadne~\cite{ariadne2018} and
PyTea~\cite{pytea2022}. \tool builds on the same foundations: refinement types
in the tradition of liquid types~\cite{liquidtypes2008}, here restricted to a
decidable equality-plus-divisibility fragment, and gradual
typing~\cite{siek2006,siek2015} for nodes without schemas, with
Theorem~\ref{thm:gradual}(ii) a coarse-grained relative of
blame~\cite{wadler2009}. These systems type \emph{programs}; \tool checks the
wiring of \emph{opaque, independently authored components} in end-user graphs.
Beyond shape, it checks family, modality, and quantization: semantic brands that
dimension-only disciplines cannot express.

\parh{Gradual Typing.} Gradual refinement types~\cite{lehmann2017} and gradual
tensor shape checking~\cite{hattori2023} insert runtime checks where static
information is imprecise, and they prove both the static and the dynamic gradual
guarantee~\cite{siek2015}. \tool adopts the static half.

\parh{Typed Dataflow and Generative-AI Workflows.}
Dataflow and scientific workflow systems type operator
connections~\cite{simulink-types,taverna2004,kepler2006}.
JSON subschema checking detects incompatible output/input schemas in ML
pipelines before training~\cite{habib2021}; refinement types also check
threading effects across composed stream operators~\cite{stein2018}.
\tool develops relational, polymorphic schemas for generative-AI components,
covering model families, modalities, and geometry, with a static gradual
guarantee when node schemas are missing.
A newer class of tools produces such graphs but does not type them:
agentic systems generate ComfyUI workflows from natural language
(ComfyGPT~\cite{comfygpt2025}, ComfyAgent / ComfyBench~\cite{comfybench2024}),
and \textsc{ComfyClaw} frames workflow editing over a typed graph
abstraction~\cite{comfyclaw2026}. These are complementary to \tool, which could
check the graphs they construct against component requirements and guide their
edits.

\parh{Empirical Studies of ML System Failures.}
Prior taxonomies of deep-learning bugs~\cite{islam2019,humbatova2020}
study defects in training code and frameworks. Our study is, to our
knowledge, the first to characterize \emph{composition} failures in
end-user generative-AI pipelines, and the first to quantify the silent
fraction.

%% file: sections/conclusion.tex
\section{Conclusion}
\label{sec:conclusion}

Semantic composition failures cause 20.9\% of issue-tracker runtime failures in
ComfyUI, and 31.5\% of these failures are silent. \tool checks component
compatibility before execution using reusable node schemas and gradual
refinement types, while tolerating missing schemas. \tool detects injected
faults that ComfyUI's own validator admits, raises no false alarm on known-good
workflows, and rejects 67\% of the user-reported failures that we rebuild and
confirm by execution. Its diagnostics help users locate faults and choose valid
repairs. The approach could extend to other workflow systems with component
schemas.

%% file: sections/appendix.tex
\section{Full Definitions and Proofs}
\label{app:supp}

This appendix proves Theorems~\ref{thm:soundness} and~\ref{thm:gradual}
and the correctness of the constraint solver of \S\ref{sec:impl:solve}.
The argument has two layers. The first layer
(\S\ref{app:values}--\S\ref{app:proof-gradual}) proves the theorems under
a single assumption: every hand-written schema is faithful to the node
it describes (Definition~\ref{def:fidelity}). The second layer
(\S\ref{app:faithful}) proves this assumption for all 147 schemas.

\parh{Roadmap.} The appendix proceeds in six steps.
\begin{itemize}
\item \S\ref{app:values} defines when a runtime value inhabits a type,
  which failures count as compatibility errors, and when a schema is faithful.
\item \S\ref{app:frontier} defines executions and the frontier, and
  states the conformance hypothesis of Theorem~\ref{thm:soundness}.
\item \S\ref{app:lemmas} proves three lemmas that the two theorems use,
  and a fourth lemma showing that the solver decides well-typedness.
\item \S\ref{app:proof-soundness} proves Theorem~\ref{thm:soundness}.
\item \S\ref{app:proof-gradual} proves Theorem~\ref{thm:gradual}.
\item \S\ref{app:faithful} proves that the 147 schemas are faithful.
\end{itemize}

\parh{Notation.} Throughout, $G = (N, E, \mathit{op}, \mathit{wid})$ is
a normalized workflow (\S\ref{sec:calc:graphs}), $\mathcal{C}(G)$ is
its constraint system, and $\Sigma$ is the schema database. The schema
files carry some indices and side predicates ($\mathit{proto}$,
$\mathit{ctrl\_ok}$, $\mathit{iw}$, $\mathit{sz}$, $\mathit{zero\_ok}$)
that Figure~\ref{fig:type-syntax} elides, and the checker completes
each schema per node before using it. Convention~(iv) of
\S\ref{app:faithful} lists these indices, predicates, and completions.
The definitions below need only the fact that such a completion takes
place.

\subsection{Instrumented Values and Inhabitation}
\label{app:values}

Theorem~\ref{thm:soundness} claims that values inhabit their types. To
make this claim precise, we must say what a value is and when it
inhabits a type. Two of our indices, control-signal modality and, in
part, model family, have no native runtime representation. For
example, a depth map and a photograph are both plain RGB tensors.
Hence ``the value inhabits its type'' cannot be read off a raw tensor.
We therefore state the metatheory over an \emph{instrumented}
(ghost-annotated) semantics. This is a standard device:
instrumentation adds bookkeeping and does not change what the program
computes.

The four definitions of this section build on each other.
Definition~\ref{def:values} attaches ghosts to runtime values.
Definition~\ref{def:inhabitation} uses the ghosts to say when a value
inhabits a type. Definition~\ref{def:errors} fixes which failures
count as compatibility errors. Definition~\ref{def:fidelity} combines the three
to say when a schema is faithful. This is the assumption on which the
rest of the appendix builds.

\begin{definition}[Instrumented values]
\label{def:values}
An instrumented value is a pair $\langle v, g \rangle$ of a raw
runtime value $v$ (tensor, model object, conditioning list, \ldots)
and a ghost record $g$. The ghost records the following, depending on
the kind of value.
\begin{itemize}
\item \emph{Image.} A modality tag $g.\mathit{mod} \in \mathcal{S}
  \cup \{\ast\}$, where $\ast$ means ``unclassified'', e.g.\ a
  photograph loaded from disk.
\item \emph{Model, LoRA, VAE, CLIP, control network, or latent.} The
  family (resp.\ architecture) of the weights it was derived from. A
  model also records a quantization tag, and a control network a kind.
  A latent created without weights (an all-zero tensor) has the
  unclassified ghost $\ast$.
\item \emph{Conditioning.} The architectures of its producing
  encoders, of every attached control network (with its kind), of the
  VAE that such a network stores for encoding its hint, and of every
  concat or reference latent. It also records the width of every
  attached image embedding.
\item \emph{Guider.} The family of its model and the ghosts of its
  conditionings.
\end{itemize}
Ghosts are created at source nodes. A loader's ghost is the true
provenance of the file it loads, and \nodename{LoadImage} yields
$\mathit{mod} = \ast$. Every operator propagates ghosts according to
its actual dataflow. For example, a preprocessor stamps its output's
$\mathit{mod}$, a sampler's output latent inherits the model's
architecture, and a LoRA that reshapes the model's input projection
gives the patched model its own family.
\end{definition}

\begin{definition}[Inhabitation]
\label{def:inhabitation}
Let $\tau$ be a \emph{ground} type, i.e.\ all its dimension indices
are in $\mathbb{Z}$ and all its brand indices are ground. Inhabitation
$\langle v, g\rangle \models \tau$ holds when the raw geometry matches
the dimension indices and the ghost matches the brand indices. A
family ghost $f_0$ matches an architecture index $a$ iff $\archfn(f_0)
= a$, and a family index $f$ iff $f_0 = f$. The latent ghost $\ast$
matches every architecture index. For each type, this means the
following.
\begin{itemize}
\item \emph{Images.} An image of height $h$, width $w$, and $c$
  channels whose ghost $\mathit{mod}$ is $\ast$ or $m$ inhabits
  $\tycon{Image}(h,w,c,m)$. An unclassified image thus imposes no
  modality constraint, matching \S\ref{sec:calc:schemas}.
\item \emph{Latents.} A latent's height and width are the last two
  axes of its tensor. Its type carries a zero state that
  Figure~\ref{fig:type-syntax} elides.
  \begin{itemize}
  \item A latent inhabits $\tycon{Latent}(h,w,c,a,\brand{Dense})$ when
    its tensor has $c$ channels and size $h\times w$, and it carries
    no downscale-ratio tag other than architecture $a$'s own.
  \item An all-zero latent inhabits
    $\tycon{Latent}(h,w,c,a,\brand{Empty})$ when normalizing it for a
    model of architecture $a$ (\texttt{fix\_empty\_latent\_channels},
    Lemma~\ref{lem:f:empty}) yields $c$ channels and size $h\times w$.
  \end{itemize}
\item \emph{Conditionings.} A conditioning inhabits
  $\tycon{Cond}(a,p,k)$ when every architecture in its ghost is $a$,
  every attached image embedding has a width that models of
  architecture $a$ consume, and every attached control network has
  kind $k$. The protocol $p$ imposes nothing.
\item \emph{Guiders.} A guider inhabits $\tycon{Guider}(a)$ when its
  model's family $f$ has $\archfn(f) = a$, and each of its
  conditionings inhabits $\tycon{Cond}(a,p,k)$ for some $p$ and some
  $k$ with $\mathit{ctrl\_ok}(f,k)$.
\item \emph{Integers.} A value inhabits $\tycon{Int}(\nu)$ when the
  carried integer equals $\nu$.
\item \emph{Index-free types.} A value inhabits an index-free type
  such as $\tycon{Noise}$ or $\tycon{Sampler}$ when it carries the
  label and meets the interface that ComfyUI's built-in nodes assume of
  it. For example, a noise generator returns noise of its latent's
  shape, and an upscale model maps RGB images to RGB images.
\item \emph{The dynamic type.} Every value inhabits $\tyDyn$.
\end{itemize}
Finally, the absent value (e.g.\ the text encoder of a checkpoint that
has none) inhabits every type.
\end{definition}

\begin{definition}[Compatibility errors]
\label{def:errors}
A node's execution raises a \emph{compatibility error} if one of the following
four events occurs.
\begin{itemize}
\item \emph{Shape/channel error.} The implementation aborts because
  the height, width, or channel count of a tensor argument violates
  the operator's precondition. Batch extents are out of scope
  (\S\ref{sec:calc:judgment}).
\item \emph{Architecture/family/quantization error.} The
  implementation aborts because the geometry or the ghost family of a
  weight argument is incompatible. A weight patch (a LoRA or a merge)
  also raises this error if it skips keys, either because it follows
  another architecture's key map or because a key's tensor shape
  disagrees with the weight it addresses.
\item \emph{Modality error} (in the instrumented semantics only). A
  control network is applied to a hint image whose ghost modality is
  neither $\ast$ nor in the network's ghost modality class.
\item {\sloppy
  \emph{Latent-architecture error} (in the instrumented semantics
  only). \nodename{VAEDecode} or \nodename{VAEDecodeTiled} decodes, or
  a core sampler (Proposition~\ref{prop:f:sampler}) samples, a latent
  whose ghost architecture is neither $\ast$ nor the architecture of the
  VAE or of the model.\par}
\end{itemize}
Other aborts, such as a missing optional input or an absent component,
are not compatibility errors.
\end{definition}

\parh{Two remarks on Definition~\ref{def:errors}.} First, the
instrumented semantics raises a skipped key or a shape mismatch of a
weight patch at the node that adds the patch. ComfyUI itself detects
such a mismatch only later, when it applies the patch during sampling
(\texttt{weight\_adapter/lora.py:284--285}). The instrumented semantics
checks it eagerly, so that the error is attributed to the node that
adds the patch. Second, modality and latent-architecture errors are
silent in the raw semantics, which is the problem this paper
addresses. For example, an SDXL VAE decodes an SD1.5 latent without an
exception but corrupts the colors (Figure~\ref{fig:silent-failures}c).
In the instrumented semantics these mismatches are \emph{defined}
events, and this is what makes the claim of
Theorem~\ref{thm:soundness} about them meaningful.

\parh{Towards faithfulness.} Informally, a schema is faithful if
inputs that inhabit its input types make the node run without a type
error and produce outputs that inhabit its output types. Making this
precise needs two ingredients. The first is the schema in the form the
checker uses for a particular node, together with the notions of
designation and freshness (Definition~\ref{def:designation}). The
second is a ground instantiation of that schema
(Definition~\ref{def:fidelity}).

\begin{definition}[Instantiated schema, designation, freshness]
\label{def:designation}
Let $\Sigma(\mathit{op}) = \forall \bar\alpha, \bar x.\ (\tau_1,
\ldots, \tau_k) \xrightarrow{\varphi} (\tau'_1, \ldots, \tau'_m)$ be
the schema of a node $n$ \emph{as the checker instantiates it for
$n$}. This is the schema file's entry after the widget-dependent
completions of convention~(iv) of \S\ref{app:faithful}, namely guarded
bindings and constraints, and the converter rule for LoRA files. The
quantifier list $\bar\alpha$ includes every brand variable that a
completion introduces. Widgets of $n$ designate variables as follows.
\begin{itemize}
\item A literal widget \emph{designates} its variable when the literal
  meets the widget's guard (\S\ref{app:faithful}).
\item A model-file widget designates its brand variables.
\item A widget wired to a producer designates nothing.
\end{itemize}
A variable is \emph{fresh} at $n$ if it occurs in no input type and no
widget of $n$ designates it.
\end{definition}

Fresh variables are the indices that a schema leaves open. For
example, when a node resizes an image to a size given by a wired
widget, the output size is fresh.

\begin{definition}[Ground instantiation; faithful schema]
\label{def:fidelity}
Let $\Sigma(\mathit{op})$ be instantiated for $n$ as in
Definition~\ref{def:designation}. A \emph{ground instantiation} is a
map $\sigma$ from $\bar\alpha, \bar x$ to ground brands and integers
such that $\sigma\varphi$ holds and $\sigma$ evaluates every
interpreted application ($\archfn$, $\chfn$, $\mathit{sf}$, and the
side predicates) according to the derivation forest. It \emph{agrees
with the widget literals} of $n$ if it maps every designated variable
to its literal, where the literal of a model-file widget is the ghost
of the file it names.

The schema is \emph{faithful} if the following holds for every ground
instantiation $\sigma$ that agrees with the widget literals of $n$ and
for all instrumented values $w_i \models \sigma\tau_i$ on the declared
inputs, with undeclared inputs unconnected. The real implementation of
$\mathit{op}$, applied to the raw parts of the $w_i$, satisfies two
conditions.
\begin{enumerate}[label=(\alph*)]
\item It raises no compatibility error (Definition~\ref{def:errors}).
\item If it returns outputs $w'_1, \ldots, w'_m$, then some $\sigma'$
  that agrees with $\sigma$ except on fresh variables satisfies $w'_j
  \models \sigma'\tau'_j$ for every $j$.
\end{enumerate}
\end{definition}

\subsection{Executions, the Frontier, and Conformance}
\label{app:frontier}

Theorem~\ref{thm:soundness} holds only for executions in which the
values that the checker cannot see conform to their types. This
section makes that hypothesis precise. It defines an execution, the
frontier of positions whose values the checker cannot see, and
conformance at the frontier.

\begin{definition}[Execution]
\label{def:execution}
An execution $E$ of $G$ evaluates nodes in a topological order. It
applies each node's implementation to the instrumented values on the
node's in-edges and to its widget literals. It stops when all nodes
have executed or when some node raises an error, which ends the
execution. We write $w_e$ for the value carried by edge $e$ in $E$.
\end{definition}

\begin{definition}[Frontier]
\label{def:frontier}
The \emph{frontier} $F(G)$ is the set of positions whose values the
static system cannot relate to anything upstream. It consists of two
groups. The \emph{graph inputs} are the following positions.
\begin{enumerate}[label=(\roman*)]
\item Output ports of runtime-determined source nodes (e.g.\
  \nodename{LoadImage}).
\item Output ports of nodes with schemas whose instantiated type
  mentions a variable that is fresh at that node
  (Definition~\ref{def:designation}). This includes a variable that the
  checker introduces when instantiating the node.
\item Model-file widgets whose brand variables the model database does
  not all resolve to ground brands. The observed value at such a
  widget is the ghost of the loaded file.
\end{enumerate}
The \emph{boundary} consists of the remaining positions.
\begin{enumerate}[label=(\roman*),resume]
\item Input ports of nodes with schemas fed by a $\tyDyn$-typed port,
  i.e.\ by a port of a $\tyDyn$-typed node or by an output that its
  node's schema does not declare.
\end{enumerate}
\end{definition}

\begin{definition}[Frontier conformance]
\label{def:conformance}
An execution $E$ is \emph{frontier-conformant} with witness $\theta$
if $\theta$ is an \emph{integral} satisfying assignment of
$\mathcal{C}(G)$ that matches every observed frontier value. Integral
means that all dimension variables are mapped to $\mathbb{Z}$, all
brand variables to ground brands, and all interpreted applications are
evaluated. Matching means two things.
\begin{itemize}
\item For every frontier position $p \in F(G)$ with instantiated
  static type $\tau_p$, the observed value satisfies $w_p \models
  \theta\tau_p$.
\item At a model-file widget, the file's ghost equals $\theta$ of the
  widget's brand variables.
\end{itemize}
Input-conformance (\S\ref{sec:calc:judgment}) is the special case in
which the boundary positions are absent or ignored.
\end{definition}

\parh{Why an integral witness.} Definition~\ref{def:conformance}
requires an integral witness, although the checker accepts under the
rational relaxation of equalities (\S\ref{sec:calc:judgment}). The two
differ when a division such as $h/\mathit{sf}(a)$ is inexact. The real
node then floors the result, and the rational equality does not encode
the floor. The runtime guarantee therefore covers only executions
whose observed sizes admit an integer witness, i.e.\ executions in
which every such division is exact. Executions with an inexact
division fall outside the theorem's hypotheses. Covering them would
require schemas that model the cropping faithfully.

\subsection{Lemmas and Solver Correctness}
\label{app:lemmas}

The proofs of \S\ref{app:proof-soundness} and
\S\ref{app:proof-gradual} rely on three lemmas about the constraint
system.
\begin{itemize}
\item Lemma~\ref{lem:inst} turns a solution of $\mathcal{C}(G)$ into a
  ground instantiation of each node's schema, as
  Definition~\ref{def:fidelity} requires.
\item Lemma~\ref{lem:edge} carries inhabitation across an edge.
\item Lemma~\ref{lem:mono} shows that erasing a schema only removes
  constraints.
\end{itemize}
The first two serve Theorem~\ref{thm:soundness}, and the third serves
Theorem~\ref{thm:gradual}(i). A fourth lemma, Lemma~\ref{lem:decision},
connects the semantic notion of well-typedness in the theorems to what
the solver actually accepts.

\begin{lemma}[Instantiation]
\label{lem:inst}
Let $\theta \models \mathcal{C}(G)$ and let $n$ be a node with a schema
with fresh renaming $\theta_n$ and widget substitution $\theta_n^w$
(\textsc{Node} rule, Figure~\ref{fig:typing-rules}). Let $\sigma_n :=
\theta \circ \theta_n^w \theta_n$, restricted to the variables of
$\Sigma(\mathit{op}(n))$ as instantiated for $n$
(Definition~\ref{def:designation}). Then $\sigma_n$ is a ground
instantiation of $\Sigma(\mathit{op}(n))$. It agrees with $n$'s
designating literal widgets and with the ghosts of the model files that
the database resolves to ground brands.
\end{lemma}

\begin{proof}
For $n$, $\mathcal{C}(G)$ contains four kinds of constraints: the
instantiated side condition $\theta_n^w\theta_n\varphi$, the
widget-binding equations $x = \nu$ for each literal that meets its
guard, the database equations of its model-file widgets, and the
interpreted-function links, i.e.\ the equations for $\chfn$ and
$\mathit{sf}$ generated at instantiation. Since $\theta$ satisfies all
of them, four properties follow.
\begin{itemize}
\item $\sigma_n\varphi$ holds.
\item The designated variables of $\sigma_n$ equal the literals.
\item A variable that the database binds to a ground brand equals that
  brand. By the database premise (convention~(iii) of
  \S\ref{app:faithful}), this brand is the file's ghost.
\item Every interpreted application evaluates consistently with the
  forest. The solver releases $\chfn$ and $\mathit{sf}$ equations only
  at their forest values, and it enforces $\archfn$ facts during
  unification.
\end{itemize}
Finally, $\sigma_n$ is ground: by integrality of $\theta$, brand
variables receive ground brands and dimension variables receive
integers.
\end{proof}

\begin{lemma}[Edge transfer]
\label{lem:edge}
Let $\theta \models \mathcal{C}(G)$ and let $e = ((n,i),(n',j))$ be a
fully typed edge with producer output type $\tau$ and consumer input
type $\tau'$. Then $\theta\tau = \theta\tau'$. Hence
$w \models \theta\tau$ implies $w \models \theta\tau'$.
\end{lemma}

\begin{proof}
Neither end of $e$ is $\tyDyn$. \textsc{Edge-Grad}
(Figure~\ref{fig:consistency}) therefore records the structural
equation $\tau \doteq \tau'$. The constructors are equal, since
otherwise $\tau \not\consistent \tau'$, a clash was recorded, and $G$
is not well-typed. Moreover, $\mathcal{C}(G)$ contains an index-wise equation
for every position, namely brand equations and dimension equations as
in the signature of Figure~\ref{fig:type-syntax}. $\theta$ satisfies
each of them, so the ground types coincide index-wise, i.e.\
$\theta\tau = \theta\tau'$.
\end{proof}

\begin{lemma}[Erasure monotonicity]
\label{lem:mono}
Let $G'$ be $G$ with one operator's schema replaced by the $\tyDyn$
default. Then, up to the shared fresh-variable namespace,
$\mathcal{C}(G') \subseteq \mathcal{C}(G)$.
\end{lemma}

\begin{proof}
Fresh renamings are per-node, so the instantiation of every other node
is unchanged. Their side conditions, widget bindings, and
interpreted-function links therefore appear in both systems. The erased
node $n_0$ contributes to $\mathcal{C}(G)$ its own side condition,
widget bindings, links, and, for each incident edge, structural
equations. In $\mathcal{C}(G')$ all of these vanish. The ports of
$n_0$ are $\tyDyn$, so by \textsc{Edge-Grad}
(Figure~\ref{fig:consistency}) each incident edge contributes no
constraint, and the $\tyDyn$ default has no side condition. Nothing
else changes, and no new constraint is created.
\end{proof}

\parh{Solver correctness.} Well-typedness $\wf{G}$ is defined
semantically, as satisfiability of $\mathcal{C}(G)$ under the hybrid
interpretation $\mathcal{H}$ (Figure~\ref{fig:typing-rules}).
Precisely, $\mathcal{H} \models \mathcal{C}(G)$ holds if the brand
equations have a unifier and some assignment of rationals to the
dimension variables satisfies every equality and divisibility atom,
where two sets of variables must receive integers: the
\emph{integer variables}, which occur in some divisibility atom, and
the \emph{fixed variables}, whose value the equalities alone
determine. Every integral satisfying assignment of $\mathcal{C}(G)$ is
such a model, so $\mathcal{H}$ relaxes integer satisfiability. The
converse fails: $2x = 2y + 1$ is satisfiable under $\mathcal{H}$ but
has no integer solution. Deleting constraints can only shrink both
sets, which is why $\mathcal{H}$-satisfiability survives deletion
(\S\ref{app:proof-gradual}). The checker, however, runs a concrete
procedure. The next lemma relates the two.

\begin{lemma}[Decision procedure]
\label{lem:decision}
The procedure of \S\ref{sec:impl:solve} never rejects a system that
is satisfiable under the hybrid interpretation of
\S\ref{sec:calc:judgment}. Conversely, acceptance implies
satisfiability whenever the side predicates have ground family
arguments.
\end{lemma}

\begin{proof}
The procedure handles brands, equalities, and congruences in turn. We
show for each part that it never rejects a satisfiable system, and
state when its acceptance is exact.

\emph{Brands.} Syntactic unification with union-find is sound and
complete for the equational part. The interpreted layer only
\emph{adds} ground facts. $\archfn$ of a grounded family is a total
forest lookup. Pending $\archfn$ facts are carried across
variable-variable merges: they migrate to the representative and are
checked on grounding. A clash between two ground brands is a genuine
inconsistency. Hence rejection at this stage implies unsatisfiability,
and acceptance yields a most general unifier of the syntactic brand
equations. An $\archfn$ fact on a variable that never becomes ground
stays pending and is accepted, which can only fail to reject. The side
predicates
($\mathit{proto}$, $\mathit{ctrl\_ok}$, $\mathit{iw}$, $\mathit{sz}$,
$\mathit{zero\_ok}$) are evaluated once their family arguments are
ground. Until then they are accepted, which can only fail to reject.

\emph{Equalities.} Gaussian elimination over $\mathbb{Q}$ is sound
and complete for the equalities. It pivots on a rational variable
whenever the equation has one. When it must eliminate an integer
variable $x$ as $x = e$, it adds the \emph{integrality atom}
$1 \mid e$, which the congruence stage decides with the other atoms.
Without these atoms, elimination would let an integer variable take a
fractional value, and the verdict would depend on the pivot order. For
example, pivoting on $x$ in $2x = y \wedge 2 \mid 2x{+}1$ would leave
the satisfiable atom $2 \mid y{+}1$. With the atoms, the remaining
constraints are congruences over free variables that are integers
exactly when the theory requires it. A variable forced to a
non-integer constant is a fixed variable, so rejecting it is part of
the theory.

\emph{Congruences.} Each residual atom $k \mid e$ with rational
coefficients is normalized by clearing denominators. Writing $e =
E/L$ with $E$ integer-coefficiented, the atom becomes $E \equiv 0 \
(\mathrm{mod}\ kL)$. The two conditions are equivalent for
integer-valued variables. Per-atom feasibility is the standard
linear-congruence criterion $\gcd(\text{coefficients}, kL) \mid
\text{constant}$. For the joint check, we partition the atoms into
components by shared variables. Distinct components constrain disjoint
variables, so their satisfying assignments merge. Two kinds of
component remain.
\begin{itemize}
\item \emph{Univariate components.} Every atom of such a component is
  in the same variable $x$. Each atom $ax + c \equiv 0\
  (\mathrm{mod}\,M)$ either is a tautology, is infeasible, or
  determines a residue class $x \equiv r \ (\mathrm{mod}\ m)$ with $m
  = M/\gcd(a,M)$ and $r = -(c/g)\,(a/g)^{-1} \bmod m$. Two classes
  $(r_1,m_1)$ and $(r_2,m_2)$ intersect iff $\gcd(m_1,m_2) \mid r_2 -
  r_1$. In that case the intersection is again a residue class, modulo
  $\mathrm{lcm}(m_1,m_2)$ (generalized CRT). Folding the component's
  atoms therefore decides it exactly.
\item \emph{Multivariate components.} Each congruence becomes an
  equation with a fresh integer multiplier. Unimodular column
  operations (extended Euclid on each row) bring the system to echelon
  form while preserving integer solvability, and forward substitution
  then decides it exactly.
\end{itemize}
We additionally validate the implementation differentially against a
Z3~\cite{demoura2008} encoding of the same theory. On 18{,}000 random systems, the two
agree on every verdict, and the verdicts depend neither on the order of
the constraints nor, for accepted systems, on deleting any one of
them.
\end{proof}

\parh{Consequence.} Lemma~\ref{lem:decision} has two consequences for
the checker. First, the solver never rejects a well-typed graph.
Second, the solver \emph{decides} well-typedness whenever the side
predicates' family arguments are ground. It needs no SMT back end for
this.

\subsection{Proof of Theorem~\ref{thm:soundness}}
\label{app:proof-soundness}

We restate Theorem~\ref{thm:soundness}. Let $G$ be well-typed with all
schemas faithful, and let $E$ be a frontier-conformant execution of $G$
with witness $\theta$. The theorem makes three claims.
\begin{enumerate}[label=(\alph*)]
\item Every node with a schema that executes in $E$ raises no compatibility error
  (Definition~\ref{def:errors}).
\item Every value that such a node produces inhabits $\theta\tau'$ for
  its solved output type $\tau'$.
\item Hence the first compatibility error in $E$, if any, occurs at a
  $\tyDyn$-typed node.
\end{enumerate}
The proof also shows that every value on a fully typed edge inhabits
the edge's solved type under $\theta$.

\begin{proof}
\emph{Idea.} We induct over the nodes in the order in which $E$
evaluates them. Every input of a node with a schema inhabits its type,
either by the induction hypothesis and Lemma~\ref{lem:edge}, or by
conformance at the frontier. The node's widgets and side condition
agree with $\theta$ by Lemma~\ref{lem:inst}. Faithfulness of the
node's schema then yields error-free execution and well-typed outputs.

\emph{Invariant.} Let $n_1, \ldots, n_N$ be the topological order in
which $E$ evaluates $G$. We prove by strong induction on $t$ the
invariant $P(t)$: \emph{if $n_t$ has a schema and $E$ reaches $n_t$,
then $n_t$ raises no compatibility error and, if it returns, the value it
produces on every output port $i$ with instantiated type $\tau'_i$
satisfies $w \models \theta\tau'_i$.} Assume $P(s)$ for all $s < t$ and
suppose that $E$ reaches the node with a schema $n_t$.

\emph{Inputs.} Consider any input port $j$ of $n_t$ with instantiated
type $\tau_j$ and an incoming edge from producer $p$. Two cases arise.
\begin{itemize}
\item \emph{Case 1: $p$ has a schema that declares the output.} Since
  $E$ reached $n_t$, $p$ executed earlier without error. By the
  induction hypothesis, its output value $w$ satisfies $w \models
  \theta\tau$ for the producer's instantiated output type $\tau$. The
  edge is fully typed, so Lemma~\ref{lem:edge} gives $\theta\tau =
  \theta\tau_j$ and hence $w \models \theta\tau_j$.
\item \emph{Case 2: $p$ is $\tyDyn$-typed, or the port is fed by an
  output that $p$'s schema does not declare.} Then the position is in
  the boundary of $F(G)$ (clause~(iv) of Definition~\ref{def:frontier}),
  and Definition~\ref{def:conformance} gives $w \models \theta\tau_j$
  directly.
\end{itemize}
Undeclared inputs of $n_t$ are unconnected, because the checker types
a node with a connected undeclared input as $\tyDyn$.

\emph{Widgets and side condition.} By Lemma~\ref{lem:inst}, the ground
instantiation $\sigma_{n_t}$ agrees with the designating literal
widgets and with the model files that the database resolves. Frontier
conformance at positions~(iii) extends this agreement to the remaining
model files. A wired widget designates nothing, so it needs no
agreement. Thus all declared inputs of $n_t$ inhabit
$\sigma_{n_t}\tau_j = \theta\tau_j$. Moreover, $\sigma_{n_t}\varphi$
holds (Lemma~\ref{lem:inst}, using integrality of $\theta$ for the
divisibility atoms).

\emph{Applying faithfulness.} All hypotheses of
Definition~\ref{def:fidelity} now hold for $\sigma_{n_t}$. Faithfulness
of $\Sigma(\mathit{op}(n_t))$ therefore yields execution without a type
error. If $n_t$ returns, it also yields some $\sigma'$ that agrees with
$\sigma_{n_t}$ off the fresh variables and under which every output
inhabits $\sigma'\tau'_i$.

\emph{Outputs.} It remains to replace $\sigma'$ by $\theta$. For an
output port whose type has no fresh variable, $\sigma'\tau'_i =
\theta\tau'_i$. For an output port whose type has a fresh variable, the
port is in $F(G)$ by clause~(ii) of Definition~\ref{def:frontier}, and
Definition~\ref{def:conformance} gives $w \models \theta\tau'_i$
directly. This establishes $P(t)$.

\emph{Conclusion.} The invariant establishes claims~(a) and~(b) for
every node with a schema that executes. Together with
Lemma~\ref{lem:edge}, it also shows that every value on a fully typed
edge inhabits the edge's solved type under $\theta$. For claim~(c),
suppose that $E$ raises a compatibility error. The invariant rules out a
node with a schema as its source, so the node is $\tyDyn$-typed.
\end{proof}

\parh{Label-only schemas.} For a node without a hand-written schema,
the checker derives a schema from ComfyUI's port labels
(\texttt{/object\_info}) in which every index is a distinct fresh
variable. Such a schema claims no faithfulness, and the theorems treat
its node as $\tyDyn$. This is safe for acceptance. The schema
contributes only constructor equations on its edges and equations with
its own fresh variables. Replacing it by $\tyDyn$ therefore removes
constraints as in Lemma~\ref{lem:mono}, and acceptance carries over.

\parh{Scope of the guarantee.} The proof places no restriction on how
$\tyDyn$-typed nodes compute. Their outputs, however, must jointly conform at
the frontier for the guarantee to hold downstream. When they do not,
Theorem~\ref{thm:gradual}(ii), proved in \S\ref{app:proof-gradual}, locates the
incompatibility at the boundary. The CLI reports solved input equalities and
residual congruences, but does not enforce them at runtime. Because the
equalities use the rational relaxation, the printed conditions are necessary but
not sufficient for frontier conformance, which requires an integral satisfying
assignment.

\subsection{Proof of Theorem~\ref{thm:gradual}}
\label{app:proof-gradual}

We restate Theorem~\ref{thm:gradual}. It has two parts.
\begin{enumerate}[label=(\roman*)]
\item If $\wf{G}$ and $G'$ replaces the schema of any node of $G$ with
  the $\tyDyn$ default, then $\wf{G'}$.
\item Let $G$ be well-typed with all schemas faithful, and let the
  graph inputs (Definition~\ref{def:frontier}) of an execution $E$
  extend to an integral satisfying assignment of $\mathcal{C}(G)$ under
  which they inhabit their frontier types. If $E$ raises a compatibility error
  at a node with a schema, then the graph inputs and the boundary values
  do not jointly admit such an assignment.
\end{enumerate}
Part~(i) follows from Lemma~\ref{lem:mono}. Part~(ii) follows from
Theorem~\ref{thm:soundness} by contradiction.

\begin{proof}[Proof of (i)]
By Lemma~\ref{lem:mono}, $\mathcal{C}(G') \subseteq \mathcal{C}(G)$ up
to the shared fresh-variable namespace. Since $\wf{G}$, an
$\mathcal{H}$-model $\theta$ of $\mathcal{C}(G)$ exists. It satisfies
every constraint of $\mathcal{C}(G')$. It also meets the integrality
requirements of $\mathcal{C}(G')$. An integer variable of
$\mathcal{C}(G')$ occurs in one of its divisibility atoms, which is
also an atom of $\mathcal{C}(G)$. A fixed variable of $\mathcal{C}(G')$
is fixed to the same value by the larger system $\mathcal{C}(G)$.
Hence $\theta$ is an $\mathcal{H}$-model of $\mathcal{C}(G')$, i.e.\
$\wf{G'}$. Iterating the argument erases any set of schemas.
\end{proof}

\begin{proof}[Proof of (ii)]
\emph{Idea.} If the boundary values conformed as well, $E$ would be
frontier-conformant, and Theorem~\ref{thm:soundness} would rule out the
error.

Suppose, for contradiction, that the graph inputs and all boundary
values jointly admit an integral satisfying assignment $\theta$ of
$\mathcal{C}(G)$ under which they inhabit their frontier types. Then
$E$ is frontier-conformant with witness $\theta$, and
Theorem~\ref{thm:soundness} rules out a compatibility error at a node
with a schema. This contradicts the hypothesis, so no such joint assignment
exists. By hypothesis, the graph inputs alone admit a conforming
assignment. The inconsistency therefore involves at least one
observation at a boundary position. Shared downstream constraints may
relate several boundary values. The argument therefore establishes
their joint incompatibility without singling out one producer as the
fault.
\end{proof}

\subsection{Schema Faithfulness}
\label{app:faithful}
\sloppy

This section proves that all 147 hand-written schemas are faithful in
the sense of Definition~\ref{def:fidelity}. Of these, 114 are for
ComfyUI v0.27.0 built-in nodes and 33 are for community-pack nodes.
The proofs are relative to nine trusted lemmas about PyTorch operators
and about the model architectures as ComfyUI implements them
(Lemmas~\ref{lem:f:resize}--\ref{lem:f:clone}). They are also relative
to a premise about the model database. Each proof follows the node's
code path in ComfyUI v0.27.0 (commit \texttt{bb131be9}) or in the
pack's source.

Faithfulness is a \emph{soundness} property. A schema stricter than
its node is still faithful and only costs precision. The proofs make
no completeness claim. Schemas derived from port labels alone claim no
faithfulness (\S\ref{app:proof-soundness}).

\parh{Roadmap.} The section has three parts.
\begin{enumerate}
\item A proof pattern and six conventions fix what each proof must show
  and what it may assume.
\item Nine trusted lemmas state the facts about PyTorch and ComfyUI's
  model code that the proofs use.
\item Twenty-two propositions cover the schemas, grouped by the kind of
  node:
  \begin{itemize}
  \item index-free nodes (Proposition~\ref{prop:f:indexfree}),
  \item loaders and model patchers
    (Propositions~\ref{prop:f:ckpt}--\ref{prop:f:patch}),
  \item text encoders and conditioning
    (Propositions~\ref{prop:f:text}--\ref{prop:f:cn}),
  \item latents and VAEs
    (Propositions~\ref{prop:f:empty}--\ref{prop:f:latentops}),
  \item guiders, samplers, and video models
    (Propositions~\ref{prop:f:guider}--\ref{prop:f:wan}),
  \item image and mask nodes
    (Propositions~\ref{prop:f:imgsrc}--\ref{prop:f:preproc}),
  \item sinks and pipes
    (Propositions~\ref{prop:f:sinks} and~\ref{prop:f:pipe}).
  \end{itemize}
\end{enumerate}
A closing remark lists the nodes that have no schema. Each proof opens
with a short \emph{Idea}, so a reader can take the intuition and skip
the details.

\subsubsection*{Proof pattern and conventions}
Each proof instantiates
Definition~\ref{def:fidelity} in four steps.
\begin{enumerate}
\item It fixes a ground instantiation $\sigma$ that satisfies the side
  condition and agrees with the designating widget literals. It also
  fixes inputs $w_i \models \sigma\tau_i$. We call this step ``fix
  $\sigma$ and inputs''.
\item It follows the node's code path.
\item It shows that no error of Definition~\ref{def:errors} occurs.
\item It exhibits an extension $\sigma'$ of $\sigma$ on the fresh
  variables with $w'_j \models \sigma'\tau'_j$. Here $\sigma' =
  \sigma$ unless stated.
\end{enumerate}
The proofs mark these steps with the run-in labels \emph{Setting},
\emph{Code path}, \emph{No error}, and \emph{Outputs}, after an
opening \emph{Idea}. Short proofs merge adjacent labels. When the
nodes of a proposition take different code paths, the proof treats
them in groups, each in the same order. \emph{Other nodes} reduces the
remaining nodes of a proposition to the first one.

We write $A = \sigma(a)$ and $f_0 = \sigma(f)$. Files are named by
basename in the ComfyUI tree and prefixed by the pack's name
otherwise.

\parh{Conventions.} The proofs rely on the following six conventions.
\begin{enumerate}[label=(\roman*)]
\item \emph{Exact division.} A ground instantiation makes every index
  expression an integer (Definition~\ref{def:inhabitation}). Each
  division $e/k$ is therefore exact under $\sigma$. ComfyUI's floor
  division $\lfloor e/k\rfloor$, written $e/\!/k$ below, then equals
  $e/k$.

\item \emph{Validated workflows and untyped occurrences.} Workflows
  pass ComfyUI's validator, which checks port labels and widget ranges
  (\texttt{execution.py:1008}, \S\ref{sec:calc:graphs}). The checker
  types a node as $\tyDyn$ when a connected input is not declared by
  its schema. An output the schema does not declare is $\tyDyn$. Each
  schema is therefore proved with its undeclared optional inputs
  absent. The same rule types a node as $\tyDyn$ when a widget takes
  it outside its schema. This applies to \nodename{FaceDetailer}'s
  LoRA-loading \texttt{wildcard} (Proposition~\ref{prop:f:pipe}).

\item \emph{Database premise.} A brand that the model database
  resolves for a file or a CLIP type is the ghost of the loaded
  weights. An architecture that it resolves alone is the ghost's
  architecture. This covers a name without its size variant, such as
  \texttt{wan}. Loaded files are well-formed models of the kind their
  loader expects. Only LoRAs that the database marks as converters
  address the input projection. We rely on this premise without
  further argument. Variants outside the forest stay unresolved. These
  are the SDXL refiner, Flux Lite, Flux.2, HunyuanVideo~1.5, and
  LTX-2.

\item \emph{Instantiated schemas.} The schema files carry more than
  Figure~\ref{fig:type-syntax} shows, and the checker completes them
  per node.
  \begin{itemize}
  \item \emph{Elided indices.} The files carry indices that
    Figure~\ref{fig:type-syntax} elides. These are a protocol and a
    control kind on $\tycon{Cond}$, and a zero state on
    $\tycon{Latent}$. They also include an architecture on
    $\tycon{Guider}$, on hook groups, and on the Wan wrapper's VAE,
    and an embedding width on CLIP-vision outputs.
  \item \emph{Side predicates.} The files use the side predicates
    $\mathit{proto}$, $\mathit{ctrl\_ok}$, $\mathit{iw}$,
    $\mathit{sz}$ (size class), and $\mathit{zero\_ok}$. Side
    predicates only restrict $\sigma$.
  \item \emph{Guarded constraints and bindings.} A constraint may be
    guarded by a widget literal. A guarded constraint is dropped only
    when the widget holds a literal other than the guard's. A guarded
    widget binding $\tycon{Int}(x) \ge k$ designates $x$ only for
    literals of at least $k$.
  \item \emph{Omitted indices.} An index written $\_$ and an index
    that a schema omits are fresh variables.
  \item \emph{LoRA completion.} The LoRA loaders' schemas are
    completed from the model database. Here $g$ is the family of the
    LoRA file, which the file widget designates. For a file that the
    database marks as a converter, the checker distinguishes three
    cases by the strength literals.
    \begin{itemize}
    \item If some strength literal is nonzero, it adds the side
      condition $\mathit{iw}(g)>\mathit{iw}(f)$. It retypes the model
      output as $\tycon{Model}(g,q)$.
    \item If every strength literal is $0$, it keeps
      $\tycon{Model}(f,q)$ and adds no side condition.
    \item Otherwise a strength is wired and no literal is nonzero. The
      checker adds the same side condition. It retypes the output as
      $\tycon{Model}(g',q)$ for a fresh brand variable $g'$.
    \end{itemize}
    Proposition~\ref{prop:f:lora} proves the completed schemas.
  \end{itemize}

\item \emph{Ghosts, attachments, and patch errors.} Ghosts follow
  Definition~\ref{def:values}. An output derived from one weighted
  input inherits its ghost. A resized or cropped image keeps its
  modality. A preprocessor stamps its output's modality. A
  latent-space control network stores the VAE given at its application
  (\texttt{controlnet.py:103--110}). A conditioning's ghost records
  that VAE's architecture, so $\tycon{Cond}(a,p,k)$ requires it to be
  $a$. Patch errors are raised by the node that adds the patch
  (Definition~\ref{def:errors}).

\item \emph{Fresh indices.} The proofs instantiate fresh modality,
  protocol, kind, and zero-state indices without comment. They name
  every fresh dimension and weight brand.
\end{enumerate}

\subsubsection*{Trusted lemmas}

The following lemmas state facts about PyTorch and about ComfyUI's
model code that we take as given. Each is followed by the code that
justifies it. Lemmas~\ref{lem:f:resize} and~\ref{lem:f:pointwise}
concern tensor shapes. Lemmas~\ref{lem:f:vae}--\ref{lem:f:text}
concern the geometry of VAEs, models, and text encoders.
Lemmas~\ref{lem:f:attach}--\ref{lem:f:clone} concern attachments,
weight patches, and clones.

\begin{lemma}[Resizing]
\label{lem:f:resize}
Called with a target $(w,h)$, the helper \texttt{common\_upscale}
returns a tensor whose last two dimensions are $(h,w)$. Its other
dimensions are those of its input.
\end{lemma}

\noindent\emph{Justification.} The routine optionally center-crops. It
then resizes with \texttt{interpolate} and \texttt{size}${}=(h,w)$,
\texttt{bislerp}, or \texttt{lanczos}. It restores the leading
dimensions of inputs with more than four axes
(\texttt{utils.py:1069--1101}).

\begin{lemma}[Size-preserving operations]
\label{lem:f:pointwise}
Pointwise arithmetic, \texttt{zeros\_like}, \texttt{clone}, channel
selection and expansion, batch slicing, repetition, and concatenation
along the batch or channel axis keep height and width. So do
symmetrically padded filters.
\end{lemma}

\noindent\emph{Justification.} PyTorch and SciPy semantics.

\begin{lemma}[VAE geometry]
\label{lem:f:vae}
A VAE of architecture $a$ has $\chfn(a)$ latent channels and spatial
ratio $\mathit{sf}(a)$.
\begin{itemize}
\item Its \texttt{encode} crops each side of an RGB image to a multiple
  of $\mathit{sf}(a)$. It returns a latent with $\chfn(a)$ channels and
  size $\lfloor h/\mathit{sf}(a)\rfloor \times \lfloor
  w/\mathit{sf}(a)\rfloor$.
\item Its \texttt{decode} maps a latent with $\chfn(a)$ channels and
  size $h \times w$ to an RGB image of size $\mathit{sf}(a)h \times
  \mathit{sf}(a)w$.
\item Its \texttt{decode} raises a channel error on a latent with
  another channel count.
\end{itemize}
\end{lemma}

\noindent\emph{Justification.} ComfyUI sets the latent channels and the
ratio per architecture, with one default (\texttt{sd.py:483--485}) and
overrides for LTXV (\texttt{:664--670}), Wan~2.2 (\texttt{:731--737}),
and Wan~2.1 (\texttt{:745--750}). Before encoding, the VAE crops the
image to a multiple of its ratio (\texttt{:956--966}).

\begin{lemma}[Model geometry]
\label{lem:f:model}
A model of family $f$ has a latent format with $\chfn(\archfn(f))$
channels and spatial ratio $\mathit{sf}(\archfn(f))$.
\begin{itemize}
\item $8$ divides every $\mathit{sf}(a)$.
\item The model accepts a latent with $\chfn(\archfn(f))$ channels at
  any height and width. So do its reference-latent inputs.
\item Sampling returns a tensor of the latent's shape.
\item Absent inpainting inputs receive defaults.
\end{itemize}
\end{lemma}

\noindent\emph{Justification.} Latent formats are declared in
\texttt{latent\_formats.py} and \texttt{supported\_models.py}. The UNet
upsamples to the skip shape (\texttt{openaimodel.py:92}). Transformers
pad to the patch size (\texttt{common\_dit.py:5}). This includes
reference latents (\texttt{model\_base.py:1020--1035},
\texttt{flux/model.py:360--365}). The sampling entry points keep the
latent's shape (\texttt{sample.py:9--11, 71--82}). The inpainting
defaults are set in the model base (\texttt{model\_base.py:298--302}).

\begin{lemma}[Empty-latent normalization]
\label{lem:f:empty}
Every core sampler passes its latent through
\texttt{fix\_empty\_latent\_channels} before sampling. Let $a$ be the
model's architecture.
\begin{itemize}
\item On an all-zero latent, the call repeats or truncates the channel
  axis to the model's channel count. If the latent is tagged with a
  spatial ratio $r \neq \mathit{sf}(a)$, the call also rescales height
  and width by $r/\mathit{sf}(a)$ with rounding.
\item Other latents it leaves unchanged, apart from adding a time axis
  for video models.
\end{itemize}
\end{lemma}

\noindent\emph{Justification.} The normalization repeats or truncates
the channels and rescales tagged latents (\texttt{sample.py:40--61},
\texttt{utils.py:859--864}). All three core samplers apply it before
sampling (\texttt{nodes.py:1557},
\texttt{nodes\_custom\_sampler.py:755, 1037}).

\begin{lemma}[Text encoders]
\label{lem:f:text}
A text encoder of architecture $a$ emits embeddings of the width that
models of architecture $a$ consume. Its class is fixed by the
checkpoint or by the \texttt{type} widget of its loader. The helper
\texttt{conditioning\_set\_values} edits only the metadata of a
conditioning.
\end{lemma}

\noindent\emph{Justification.} Loading a CLIP selects its encoder class
from the file or the \texttt{type} widget (\texttt{sd.py:1275, 1315}).
\texttt{conditioning\_set\_values} copies each entry and sets only
metadata keys (\texttt{node\_helpers.py:9}).

\begin{lemma}[Attachments]
\label{lem:f:attach}
Attachments are adapted to the sampled latent as follows.
\begin{itemize}
\item Concatenated latents and masks attached to a conditioning are
  resized to the sampled latent.
\item A pixel-space control network resizes its hint to the latent
  size.
\item A latent-space control network encodes its hint with the VAE
  given at its application. Without that VAE it raises an error that
  names the missing VAE.
\item A control network of architecture $a$ is built for the models of
  $a$. If it works in latent space, it is also built for hints encoded
  by a VAE of $a$.
\end{itemize}
\end{lemma}

\noindent\emph{Justification.} The model base resizes conditioning
attachments (\texttt{model\_base.py:257--288}). The control network
code handles hints (\texttt{controlnet.py:103--110, 269--283,
731--886}).

\begin{lemma}[Weight patches]
\label{lem:f:weights}
LoRA and merge patches address weights by key, and key maps are
determined by the architecture.
\begin{itemize}
\item A key outside the target's key map is skipped with a warning.
\item A key whose tensor shape disagrees with the weight is logged and
  skipped.
\item Families of equal architecture and size class have weights of
  equal shape, except the input projection, whose width is
  $\mathit{iw}$.
\item A converter LoRA pads the input projection to its own width
  before adding its update.
\end{itemize}
\end{lemma}

\noindent\emph{Justification.} Key maps are derived from the model's
architecture (\texttt{lora.py:97, 187}). A LoRA key that the map does
not contain is skipped with a warning (\texttt{:93}). A patch whose
shape differs from its weight is skipped with an error
(\texttt{weight\_adapter/lora.py:270, 284--285}). A converter reshapes
the input projection before adding its update
(\texttt{weight\_adapter/lora.py:243--246}). Weight shapes follow each
architecture's model configuration (\texttt{supported\_models.py}).

\begin{lemma}[Clones and object patches]
\label{lem:f:clone}
\texttt{ModelPatcher.clone} and \texttt{CLIP.clone} keep the weights
and the latent format. So do object patches of the schedule, the CFG
function, the denoise mask, the attention, or hooks. The one exception
is the sampling mode \texttt{edm\_playground\_v2.5}, which installs a
four-channel latent format.
\end{lemma}

\noindent\emph{Justification.} Cloning a model or a CLIP shares its
weights and copies its patch lists (\texttt{model\_patcher.py:382},
\texttt{sd.py:290}). The exception installs a four-channel latent
format in the \texttt{edm\_playground\_v2.5} mode
(\texttt{nodes\_model\_advanced.py:224--236}).

\subsubsection*{Index-free nodes}
Every schema is named in exactly one of
Propositions~\ref{prop:f:indexfree}--\ref{prop:f:pipe}.

\begin{proposition}[Nodes without index claims]
\label{prop:f:indexfree}
The schemas of
\nodename{CLIPVisionLoader}, \nodename{StyleModelLoader},
\nodename{UpscaleModelLoader}, \nodename{SAMLoader},
\nodename{UltralyticsDetectorProvider}, \nodename{WanVideoModelLoader},
\nodename{LoadAudio}, \nodename{LoadVideo},
\nodename{KSamplerSelect}, \nodename{SamplerEulerAncestral},
\nodename{SamplerDPMPP\_2M\_SDE}, \nodename{KarrasScheduler},
\nodename{RandomNoise}, \nodename{DisableNoise},
\nodename{PrimitiveFloat}, \nodename{PrimitiveString},
\nodename{PrimitiveStringMultiline}, \nodename{PrimitiveBoolean},
\nodename{FloatConstant}, \nodename{StringConstant},
\nodename{PreviewAny}, \nodename{Display Any (rgthree)},
\nodename{SaveAudio}, \nodename{SaveAudioMP3},
\nodename{SaveAudioOpus}, \nodename{SaveVideo}, \nodename{SaveGLB},
and \nodename{Preview3D} are faithful.
\end{proposition}

\begin{proof}
\emph{Idea.} None of these nodes looks at a tensor's shape or patches
weights, so no index claim can fail.

\emph{Setting.} Every port of these schemas has an index-free type or
$\tyDyn$. The schemas therefore have no variables, and $\sigma$ is
empty.

\emph{Code path.} None of the nodes reads the height, width, or
channel count of an image or latent, and none patches weights.
\begin{itemize}
\item The loaders read files.
\item The sampler, scheduler, and noise nodes build objects from their
  widget values (\texttt{nodes\_custom\_sampler.py:63, 390, 437, 508,
  990, 1007}).
\item The primitives return their widget values.
\item The sinks serialize or display their inputs.
\end{itemize}

\emph{No error.} Hence no error of Definition~\ref{def:errors} occurs.

\emph{Outputs.} Each output carries its label. It meets the interface
that ComfyUI's built-in nodes assume of that label. The built-in samplers,
schedulers, and noise sources are the implementations those nodes are
written against. Loaded files meet the interface by convention~(iii).

\emph{Other nodes.} \nodename{WanVideoModelLoader} needs two further
remarks. It reads its $\tyDyn$-typed block-swap argument only for
block counts (WanVideoWrapper \texttt{nodes\_model\_loading.py:846--851,
1835}). Its weight-patching inputs \texttt{lora} and
\texttt{multitalk\_model} are undeclared and therefore absent by
convention~(ii).
\end{proof}

\subsubsection*{Loaders and model patchers}
These nodes load models, CLIPs, VAEs, and control networks, or modify
a loaded model or CLIP.

\begin{proposition}[Checkpoint and diffusion-model loaders]
\label{prop:f:ckpt}
The schemas of \nodename{CheckpointLoaderSimple},
\nodename{UNETLoader}, and \nodename{UnetLoaderGGUF} are faithful.
\end{proposition}

\begin{proof}
\emph{Idea.} A loader returns objects built from the named file, so
they carry that file's ghost.

\emph{Setting.} Consider \nodename{CheckpointLoaderSimple}, whose
schema is $\forall f,q.\ \tycon{Ckpt}(f) \to (\tycon{Model}(f,q),
\tycon{Clip}(\archfn(f)), \tycon{VAE}(\archfn(f)))$. The file widget
designates $f$, so $f_0 = \sigma(f)$ is the ghost of the named file
(Definition~\ref{def:fidelity}). The variable $q$ is fresh.

\emph{Code path.} The node calls
\texttt{load\_checkpoint\_guess\_config}, which detects the model
configuration from the weights (\texttt{nodes.py:627--630}).

\emph{No error.} The file is a well-formed checkpoint by
convention~(iii). Loading reads no tensor argument, so no error
occurs.

\emph{Outputs.} The model is built from the file, so its ghost family
is $f_0$. With $\sigma'(q)$ its quantization ghost, it inhabits
$\tycon{Model}(f_0, \sigma'(q))$. The CLIP and the VAE are built from
the same file. Their ghost architecture is therefore $\archfn(f_0) =
\sigma(\archfn(f))$. A component that the checkpoint lacks is returned
as \texttt{None}, the absent value, which inhabits every type.

\emph{Other nodes.} \nodename{UNETLoader} and \nodename{UnetLoaderGGUF}
return only the model through the same detection
(\texttt{nodes.py:977--989}, GGUF \texttt{nodes.py:176}). The argument
for the model output applies.
\end{proof}

\begin{proposition}[Component loaders]
\label{prop:f:comp}
The schemas of \nodename{VAELoader}, \nodename{CLIPLoader},
\nodename{DualCLIPLoader}, \nodename{DualCLIPLoaderGGUF},
\nodename{TripleCLIPLoader}, \nodename{QuadrupleCLIPLoader},
\nodename{ControlNetLoader}, and \nodename{WanVideoVAELoader} are
faithful.
\end{proposition}

\begin{proof}
\emph{Idea.} The file, encoder-type widget, or loader's fixed
architecture determines every brand of the output.

\emph{Setting.} In the single- and dual-file loaders, one widget
designates every brand variable of the output type. The triple- and
quadruple-file loaders instead fix the output architecture to SD3
and HiDream, respectively.

\emph{Code path, no error, and outputs.} The argument depends on how
the loader fixes the output brands.
\begin{itemize}
\item \emph{File loaders.} For \nodename{VAELoader},
  \nodename{ControlNetLoader}, and \nodename{WanVideoVAELoader} the
  widget is the file, so $\sigma$ assigns the brands of the file's
  ghost. The files are well-formed by convention~(iii), so loading
  raises no error. The loaded object carries the file's ghost, which
  inhabits the output type. The Wan wrapper builds a 16- or 48-channel
  VAE according to the file (WanVideoWrapper
  \texttt{nodes\_model\_loading.py:1900--1928}).
\item \emph{CLIP loaders.} For the three CLIP loaders the
  \texttt{type} widget designates $a$ through the database premise. It
  also fixes the class of the text encoder (Lemma~\ref{lem:f:text},
  \texttt{nodes.py:1007--1016, 1035--1046}, GGUF \texttt{nodes.py:272}).
  The returned CLIP therefore has the ghost architecture $A$.
\item \emph{Fixed-architecture CLIP loaders.}
  \nodename{TripleCLIPLoader} and \nodename{QuadrupleCLIPLoader}
  pass three and four files to \texttt{comfy.sd.load\_clip}, respectively
  (\texttt{nodes\_sd3.py}, \texttt{nodes\_hidream.py}). The corresponding
  branches build SD3 and HiDream encoders (\texttt{sd.py:1734--1739}).
  Under the well-formed-file premise of convention~(iii), loading raises
  no compatibility error, and the outputs inhabit
  $\tycon{Clip}(\brand{SD3})$ and $\tycon{Clip}(\brand{HiDream})$.
\end{itemize}
\end{proof}

\begin{proposition}[LoRA loaders]
\label{prop:f:lora}
The schemas of \nodename{LoraLoader} and \nodename{LoraLoaderModelOnly}
are faithful.
\end{proposition}

\begin{proof}
\emph{Idea.} A LoRA built for the model's architecture and size class
addresses the model's keys with matching shapes, so no patch is
skipped. Only a converter LoRA changes the model's family, and the
checker retypes the output to match.

\emph{Setting.} The schemas quantify over $f,q,g$. They take a model of
type $\tycon{Model}(f,q)$ and, at \nodename{LoraLoader}, a CLIP of type
$\tycon{Clip}(\archfn(f))$. They require $\archfn(f)=\archfn(g)$ and
$\mathit{sz}(f)=\mathit{sz}(g)$. Fix $\sigma$ and inputs. The file
widget designates $g$, so $g_0 = \sigma(g)$ is the ghost of the LoRA.
Hence $\archfn(f_0) = \archfn(g_0)$ and $\mathit{sz}(f_0) =
\mathit{sz}(g_0)$.

\emph{Code path.} The node branches on the strengths.
\begin{itemize}
\item If both strengths are $0$, \texttt{load\_lora} returns its inputs
  (\texttt{nodes.py:719--721}). \nodename{LoraLoaderModelOnly} passes a
  CLIP strength of $0$ (\texttt{:751--752}).
\item Otherwise the node loads the file and calls
  \texttt{load\_lora\_for\_models} (\texttt{sd.py:90--122}). This
  function builds one key map from the model's key map and, if a CLIP
  is given, the CLIP's key map. It converts the LoRA's key naming. It
  adds the resulting patches to clones of the model and the CLIP.
\end{itemize}

\emph{No error.} The only errors of Definition~\ref{def:errors} that a
weight patch can raise are skipped keys and shape mismatches. We rule
out both.
\begin{itemize}
\item \emph{Skipped keys.} Key maps are determined by the architecture
  (Lemma~\ref{lem:f:weights}). The LoRA's keys follow the key map of
  $\archfn(g_0) = \archfn(f_0)$. This is the key map of the model. The
  CLIP has ghost architecture $\archfn(f_0)$, so it is the key map of
  the CLIP too. No key is therefore skipped for following another
  architecture's key map. Text-encoder keys that
  \nodename{LoraLoaderModelOnly} skips are skipped because no CLIP is
  given. Definition~\ref{def:errors} does not count them.
\item \emph{Shape mismatches.} Families of equal architecture and size
  class, such as $f_0$ and $g_0$, have weights of equal shape except
  the input projection (Lemma~\ref{lem:f:weights}). An ordinary LoRA
  does not address the input projection (convention~(iii)). Every patch
  then has the shape of its weight, and none is skipped.
\item \emph{Converter LoRAs.} A converter LoRA pads the input
  projection to its own width before adding its update
  (Lemma~\ref{lem:f:weights}). For a converter file, the checker adds
  the strict side condition $\mathit{iw}(g) > \mathit{iw}(f)$ whenever some
  strength literal is nonzero or some strength is wired. So whenever
  the node may add the patch, the padding widens the projection. The
  update then matches the padded shape. When every strength is the
  literal $0$, the node adds no patch and the checker adds no side condition.
\end{itemize}

\emph{Outputs.} The model output depends on the kind of LoRA.
\begin{itemize}
\item \emph{Ordinary LoRA.} The patched clones keep the weights'
  geometry, the latent format, the quantization, and the ghost family
  $f_0$ (Lemma~\ref{lem:f:clone}). The outputs therefore inhabit
  $\tycon{Model}(f_0,\sigma(q))$ and $\tycon{Clip}(\archfn(f_0))$.
\item \emph{Converter file.} A converter file changes the model
  output's family according to the strength literals. The node adds the
  patch whenever some strength is nonzero (\texttt{nodes.py:720}). The
  padding of the input projection precedes the strength-scaled update
  (\texttt{weight\_adapter/lora.py:243--246, 270}). It therefore
  happens at any strength once the patch is added. In particular, it
  happens for \texttt{strength\_model}~$=0$ with a nonzero
  \texttt{strength\_clip}. Three cases arise.
  \begin{enumerate}
  \item \emph{Some strength literal is nonzero.} The checker retypes
    the output as $\tycon{Model}(g,q)$. The patched model has the input
    projection of $g_0$. By Definition~\ref{def:values} it has the
    ghost family $g_0$, so it inhabits $\tycon{Model}(g_0,\sigma(q))$.
  \item \emph{Every strength is the literal $0$.} The checker keeps
    $\tycon{Model}(f,q)$. The node returns its input model, which
    inhabits $\tycon{Model}(f_0,\sigma(q))$.
  \item \emph{Otherwise.} Some strength is wired and no literal is
    nonzero. Whether the node patches is then decided at run time. The
    checker gives the output the fresh family variable $g'$. The
    output position's frontier observation fixes it
    (Definition~\ref{def:frontier}).
  \end{enumerate}
\end{itemize}
\end{proof}

\begin{proposition}[Model and CLIP patchers]
\label{prop:f:patch}
The schemas of \nodename{ModelSamplingSD3},
\nodename{ModelSamplingAuraFlow}, \nodename{ModelSamplingFlux},
\nodename{ModelSamplingDiscrete},
\nodename{ModelSamplingContinuousEDM}, \nodename{RescaleCFG},
\nodename{CFGNorm}, \nodename{DifferentialDiffusion},
\nodename{PathchSageAttentionKJ}, \nodename{CLIPSetLastLayer},
\nodename{SetClipHooks}, and \nodename{CLIPMergeSimple} are faithful.
\end{proposition}

\begin{proof}
\emph{Idea.} These nodes clone their input and change how it samples
or encodes, not the geometry of its weights. The type therefore passes
through unchanged.

\emph{Setting.} The model schemas map $\tycon{Model}(f,q)$ to
$\tycon{Model}(f,q)$. The CLIP schemas map CLIPs of type
$\tycon{Clip}(a)$ to $\tycon{Clip}(a)$.

\emph{Model patchers.} Each node clones its model and installs an
object patch of the sampling schedule, the CFG function, the
denoise-mask function, or the attention kernel
(\texttt{nodes\_model\_advanced.py:65--304}, \texttt{nodes\_cfg.py:78},
\texttt{nodes\_differential\_diffusion.py:34}, KJNodes
\texttt{model\_optimization\_nodes.py:127}). No tensor or weight
argument is processed, so no error occurs. By Lemma~\ref{lem:f:clone}
the clone keeps the weights, the quantization, and the latent format.
Its ghost family is therefore $f_0$, and it inhabits
$\tycon{Model}(f_0,\sigma(q))$.

\emph{The EDM exception.} The one patch that changes the latent format
is \nodename{ModelSamplingContinuousEDM} in mode
\texttt{edm\_playground\_v2.5}. It installs a four-channel format
(\texttt{nodes\_model\_advanced.py:224--236}). Its schema carries the
constraint $\chfn(\archfn(f)) = 4$ guarded by that literal. The
constraint is kept when the widget is wired. Under $\sigma$ the
installed format therefore has the $\chfn(\archfn(f_0))$ channels that
Lemma~\ref{lem:f:model} associates with $f_0$.

\emph{CLIP patchers.} Each node returns its CLIP or a clone of it, and
no key is skipped.
\begin{itemize}
\item \nodename{CLIPSetLastLayer} clones its CLIP and sets the output
  layer (\texttt{nodes.py:688--691}). This keeps the architecture
  (Lemma~\ref{lem:f:clone}).
\item \nodename{SetClipHooks} returns its CLIP unchanged when no hooks
  are given (\texttt{nodes\_hooks.py:249--260}). Otherwise it registers
  the hooks' LoRA patches on a clone through the CLIP's key map
  (\texttt{hooks.py:154--171}). The hooks inhabit $\tycon{Hooks}(A)$,
  so their LoRAs follow the key map of $A$. That is the key map of the
  CLIP, so no key is skipped (Lemma~\ref{lem:f:weights}).
\item \nodename{CLIPMergeSimple} adds the key patches of one CLIP to a
  clone of the other. Both have architecture $A$ and thus the same keys
  and shapes. No key is skipped.
\end{itemize}
In all three cases the output is a clone of a CLIP of architecture $A$
and inhabits $\tycon{Clip}(A)$.
\end{proof}

\subsubsection*{Text encoders and conditioning}
These nodes build and combine conditionings. Their proofs track the
conditioning's architecture, protocol, and control kind.

\begin{proposition}[Text encoders]
\label{prop:f:text}
The schemas of \nodename{CLIPTextEncode}, \nodename{CLIPTextEncodeSDXL},
and \nodename{CLIPTextEncodeFlux} are faithful.
\end{proposition}

\begin{proof}
\emph{Idea.} A text encoder of architecture $A$ produces embeddings
that models of $A$ consume, and it attaches nothing else.

\emph{Setting.} Consider \nodename{CLIPTextEncode}, with schema
$\forall a.\ \tycon{Clip}(a) \to \tycon{Cond}(a)$. Fix $\sigma$ and a
CLIP of ghost architecture $A$.

\emph{Code path.} The node tokenizes the text and calls
\texttt{encode\_from\_tokens\_scheduled} (\texttt{nodes.py:73--77}).

\emph{No error.} An absent CLIP makes the node raise an error about the
missing component, which is not a compatibility error. Otherwise no tensor
argument is read.

\emph{Outputs.} By Lemma~\ref{lem:f:text} the embeddings have the width
that models of $A$ consume. The output's ghost thus has the single
architecture $A$, no image embedding, and no control network. It
therefore inhabits $\tycon{Cond}(A,p,k)$ for every $p$ and $k$.

\emph{Other nodes.} \nodename{CLIPTextEncodeSDXL} reads the token
streams \texttt{g} and \texttt{l} (\texttt{nodes\_clip\_sdxl.py:52--59}).
\nodename{CLIPTextEncodeFlux} reads \texttt{l} and \texttt{t5xxl}
(\texttt{nodes\_flux.py:30--31}). The SDXL and Flux text encoders
provide these streams, and so does the SD3 encoder
(\texttt{sd3\_clip.py:47--52}). The two schemas fix the architecture to
$\brand{SDXL}$ and $\brand{Flux}$. This is stricter than the nodes, and
the same argument applies.
\end{proof}

\begin{proposition}[Conditioning operators]
\label{prop:f:condops}
The schemas of \nodename{ConditioningCombine},
\nodename{ConditioningConcat}, \nodename{ConditioningSetArea},
\nodename{ConditioningSetAreaPercentage},
\nodename{ConditioningSetTimestepRange},
\nodename{ConditioningZeroOut}, and \nodename{FluxGuidance} are
faithful.
\end{proposition}

\begin{proof}
\emph{Idea.} These operators rearrange or relabel conditionings of one
architecture. The architecture, protocol, and kind are unchanged.

\emph{Setting.} All inputs and outputs have architecture $a$.

\emph{Code path and no error.} Each node either merges lists, joins
embeddings of equal width, or edits metadata (\texttt{nodes.py:80--311},
\texttt{nodes\_flux.py}).
\begin{itemize}
\item \nodename{ConditioningCombine} concatenates two lists. Its output
  carries the embeddings and the control networks of both inputs. Both
  inputs inhabit $\tycon{Cond}(A,\cdot,k)$ for the same kind $k$, so
  the combined list does too.
\item \nodename{ConditioningConcat} concatenates token embeddings along
  the token axis. Both inputs have architecture $A$. They thus have
  embeddings of the one width that models of $A$ consume
  (Lemma~\ref{lem:f:text}). The concatenation therefore raises no
  error and keeps the width.
\item The area, range, and guidance operators edit metadata only
  (\texttt{conditioning\_set\_values}, Lemma~\ref{lem:f:text}).
\item \nodename{ConditioningZeroOut} replaces tensors by
  \texttt{zeros\_like} tensors (Lemma~\ref{lem:f:pointwise}).
\end{itemize}

\emph{Outputs.} None of the nodes removes or adds an attachment. Each
output therefore has the input's ghost. It inhabits $\tycon{Cond}(A)$
with the input's protocol and kind.
\end{proof}

\begin{proposition}[Reference encoders]
\label{prop:f:ref}
The schemas of \nodename{TextEncodeQwenImageEdit},
\nodename{TextEncodeQwenImageEditPlus}, and \nodename{ReferenceLatent}
are faithful.
\end{proposition}

\begin{proof}
\emph{Idea.} The reference latents come from a VAE of the same
architecture as the text encoder, so the conditioning keeps a single
architecture.

\emph{Setting.} The Qwen encoders take $\tycon{Clip}(a)$,
$\tycon{VAE}(a)$, and images of independent sizes.

\emph{Code path.} They pass the prompt and the images to the text
encoder of architecture $A$. They attach, as reference latents, the
images encoded by the VAE (\texttt{nodes\_qwen.py:43--49, 101--105}).

\emph{No error.} The VAE receives RGB images, so encoding raises no
error (Lemma~\ref{lem:f:vae}).

\emph{Outputs.} The reference latents have the VAE's architecture $A$.
The embeddings have the width that models of $A$ consume
(Lemma~\ref{lem:f:text}). The output therefore inhabits
$\tycon{Cond}(A,\brand{QwenEdit},k)$, respectively
$\tycon{Cond}(A,\brand{QwenEditPlus},k)$, for the fresh $k$. This holds
because the protocol index imposes nothing on inhabitation.

\emph{Other nodes.} \nodename{ReferenceLatent} appends its latent, if
given, to the conditioning as a reference latent
(\texttt{nodes\_edit\_model.py:24--27}). The latent inhabits
$\tycon{Latent}(h,w,\chfn(A),A,\brand{Dense})$. Its architecture is
therefore $A$, and models of $A$ accept it at any size
(Lemma~\ref{lem:f:model}). The output inhabits $\tycon{Cond}(A)$.
\end{proof}

\begin{proposition}[ControlNet application]
\label{prop:f:cn}
The schemas of \nodename{ControlNetApply},
\nodename{ControlNetApplyAdvanced}, and \nodename{ControlNetApplySD3}
are faithful.
\end{proposition}

\begin{proof}
\emph{Idea.} Applying a control network only attaches it. The schema
makes the network, its VAE, and the conditioning share one architecture
and one kind. It also makes the hint's modality match the network's.

\emph{Setting.} Consider \nodename{ControlNetApplyAdvanced}. Its schema
takes the following inputs.
\begin{itemize}
\item Conditionings of types $\tycon{Cond}(a,p,k)$ and
  $\tycon{Cond}(a,p_2,k)$.
\item A control network of type $\tycon{Ctrl}(s,a,k)$.
\item A hint of type $\tycon{Image}(h,w,3,m)$.
\item An optional VAE of type $\tycon{VAE}(a)$.
\end{itemize}
It requires $s = m$. Fix $\sigma$ and inputs.

\emph{Code path.} With strength $0$ the node returns its conditionings
(\texttt{nodes.py:937--939}). Otherwise it copies each conditioning
entry. It attaches a copy of the network with the hint and the VAE
through \texttt{set\_cond\_hint}. It chains the copy to the entry's
previous network (\texttt{nodes.py:941--963}). The call
\texttt{set\_cond\_hint} stores the hint and, for a latent-space
network only, the VAE (\texttt{controlnet.py:103--110}). No tensor is
processed.

\emph{No error.} The only error this node can raise is a modality
error. The hint's ghost modality is $\ast$ or $\sigma(m) = \sigma(s)$,
the network's modality class, so none occurs. Hint size and network
geometry matter only at sampling time (Lemma~\ref{lem:f:attach}). A
latent-space network applied without a VAE only logs a warning here
(\texttt{controlnet.py:109}). The sampler later raises an error naming
the missing VAE, which is not a compatibility error (\texttt{:278}).

\emph{Outputs.} Each output entry keeps its embeddings and its previous
networks. These have architecture $A$ and kind $\sigma(k)$, because
the input inhabits $\tycon{Cond}(A,\cdot,\sigma(k))$. The new network
has architecture $A$ and kind $\sigma(k)$. The VAE input is handled by
the kind of network.
\begin{itemize}
\item \emph{Latent-space network.} The network stores the VAE. The VAE
  has architecture $A$, as convention~(v) requires.
\item \emph{Pixel-space network.} The VAE matters only for a
  latent-space network. For a pixel-space one the schema is therefore
  stricter than the node.
\end{itemize}
The outputs therefore inhabit $\tycon{Cond}(A,\sigma(p),\sigma(k))$ and
$\tycon{Cond}(A,\sigma(p_2),\sigma(k))$.

\emph{Other nodes.} \nodename{ControlNetApply} follows the same path
with one conditioning and no VAE (\texttt{nodes.py:898--913}).
\nodename{ControlNetApplySD3} follows the advanced node's path with
an explicit VAE and an empty \texttt{extra\_concat} list
(\texttt{nodes\_sd3.py:112--166}). Its schema has the same constraints,
including the strength-zero guards, so the same argument applies.
\end{proof}

\subsubsection*{Latents and VAEs}
These nodes create, encode, decode, and transform latents. Their
proofs track channel counts, spatial ratios, and the zero state.

\begin{proposition}[Empty latents]
\label{prop:f:empty}
The schemas of \nodename{EmptyLatentImage},
\nodename{EmptySD3LatentImage}, and \nodename{EmptyFlux2LatentImage}
are faithful.
\end{proposition}

\begin{proof}
\emph{Idea.} The node returns a zero tensor at ratio $8$. Samplers
later normalize it to the model's channels and ratio, and the state
$\brand{Empty}$ accounts for that. The state $\brand{Dense}$ is
allowed only when normalization changes nothing.

\emph{Setting.} Consider \nodename{EmptyLatentImage}, whose schema is
$\forall a, z, h, w.\ (\tycon{Int}(w), \tycon{Int}(h)) \to
\tycon{Latent}(h/\mathit{sf}(a), w/\mathit{sf}(a), \chfn(a), a, z)$
with side condition $\mathit{zero\_ok}(z, a, 4, 8)$. The size widgets
designate $h$ and $w$ for literals $\nu_h$ and $\nu_w$. The variables
$a$ and $z$ are fresh. Fix $\sigma$.

\emph{Code path and no error.} The node returns the all-zero tensor
\texttt{zeros([b,4,$\nu_h$//8,$\nu_w$//8])} tagged with spatial ratio
$8$ (\texttt{nodes.py:1247--1249}). It processes no tensor argument, so
no error occurs. Its ghost is $\ast$, which inhabits every
architecture index.

\emph{Divisibility.} By convention~(i), $\nu_h/\mathit{sf}(A)$ is an
integer. Also $8 \mid \mathit{sf}(A)$ (Lemma~\ref{lem:f:model}). Hence
$8 \mid \nu_h$ and $\nu_h/\!/8 = \nu_h/8$. The same holds for $\nu_w$.

\emph{Outputs.} We show that the output inhabits $\sigma\tau'$, taking
$\sigma' = \sigma$. The two zero states are treated separately.
\begin{itemize}
\item \emph{State $\brand{Empty}$.} Normalizing the tensor for a model
  of architecture $A$ repeats or truncates its four channels to
  $\chfn(A)$ (Lemma~\ref{lem:f:empty}). Its tag $8$ may differ from
  $\mathit{sf}(A)$. Normalization then rescales each side to $(\nu/8)
  \cdot 8/\mathit{sf}(A) = \nu/\mathit{sf}(A)$. The rounding is exact
  because the quotient is an integer. By
  Definition~\ref{def:inhabitation} the output inhabits
  $\tycon{Latent}(\nu_h/\mathit{sf}(A), \nu_w/\mathit{sf}(A),
  \chfn(A), A, \brand{Empty})$.
\item \emph{State $\brand{Dense}$.} Here $\mathit{zero\_ok}$ holds only
  when $\chfn(A) = 4$ and $\mathit{sf}(A) = 8$. The raw tensor then has
  $\chfn(A)$ channels and size $\nu/8 = \nu/\mathit{sf}(A)$. Its tag
  $8$ is $A$'s own ratio. It therefore inhabits the $\brand{Dense}$
  type.
\end{itemize}
If a size widget is wired, its variable is fresh as well. Then
$\sigma'$ assigns it $\mathit{sf}(A)$ times the side length obtained
above, which is an integer.

\emph{Other nodes.} \nodename{EmptySD3LatentImage} returns sixteen
channels with the same size and tag (\texttt{nodes\_sd3.py:58--59}).
Its schema admits only $\brand{Empty}$, so the first case applies.
\nodename{EmptyFlux2LatentImage} returns
\texttt{zeros([b,128,$\nu_h$//16,$\nu_w$//16])} without a tag
(\texttt{nodes\_flux.py:56--57}). Normalization adjusts only its
channels and keeps its size $\nu/16$, which its schema states.
\end{proof}

\begin{proposition}[VAE encoders]
\label{prop:f:encode}
The schemas of \nodename{VAEEncode}, \nodename{VAEEncodeForInpaint},
\nodename{InpaintModelConditioning}, and
\nodename{InstructPixToPixConditioning} are faithful.
\end{proposition}

\begin{proof}
\emph{Idea.} Encoding crops to the VAE's ratio and yields a latent of
the VAE's architecture. Exact division makes the crop invisible to the
types.

\emph{Setting.} Consider \nodename{VAEEncode}, whose schema is
$\forall a,h,w.\ (\tycon{Image}(h,w,3), \tycon{VAE}(a)) \to
\tycon{Latent}(h/\mathit{sf}(a), w/\mathit{sf}(a), \chfn(a), a,
\brand{Dense})$. Fix $\sigma$, an RGB image of size $\sigma(h) \times
\sigma(w)$, and a VAE of architecture $A$.

\emph{Code path.} The node calls \texttt{vae.encode}
(\texttt{nodes.py:383--385}).

\emph{No error.} Encoding an RGB image raises no error
(Lemma~\ref{lem:f:vae}).

\emph{Outputs.} By the same lemma, encoding crops each side to a
multiple of $\mathit{sf}(A)$. It returns a latent with $\chfn(A)$
channels and size $\lfloor \sigma(h)/\mathit{sf}(A)\rfloor$. This
equals $\sigma(h)/\mathit{sf}(A)$ by convention~(i). The crop lets
other sizes run as well, so the schema needs no divisibility premise.
The latent carries no ratio tag, and its ghost is the VAE's
architecture $A$. It therefore inhabits the $\brand{Dense}$ type.

\emph{Other nodes.} The other encoders add masks or attachments around
the same encoding.
\begin{itemize}
\item \nodename{VAEEncodeForInpaint} also interpolates a mask of any
  size to the pixels and crops both (\texttt{nodes.py:414--444}). It
  grows the mask with a padded convolution
  (Lemmas~\ref{lem:f:resize} and~\ref{lem:f:pointwise}). It returns the
  encoding with a noise mask, which is not a ratio tag.
\item \nodename{InpaintModelConditioning} returns the encoding of the
  uncropped pixels (\texttt{nodes.py:463--495}). It attaches to both
  conditionings a mask and a concat latent from the same VAE. The
  conditionings therefore keep the single architecture $A$. The
  sampler resizes the attachments (Lemma~\ref{lem:f:attach}).
\item \nodename{InstructPixToPixConditioning} crops to a multiple of
  $8$ (\texttt{nodes\_ip2p.py:27--50}). This is no crop, since $8 \mid
  \mathit{sf}(A) \mid \sigma(h)$. The node attaches the encoding as a
  concat latent. It returns \texttt{zeros\_like} of the encoding
  without a tag, which inhabits the same $\brand{Dense}$ type.
\end{itemize}
\end{proof}

\begin{proposition}[VAE decoders]
\label{prop:f:decode}
The schemas of \nodename{VAEDecode} and \nodename{VAEDecodeTiled} are
faithful.
\end{proposition}

\begin{proof}
\emph{Idea.} A $\brand{Dense}$ latent has exactly the channels the VAE
expects. Decoding therefore succeeds and scales each side by the VAE's
ratio.

\emph{Setting.} The schemas take
$\tycon{Latent}(h,w,\chfn(a),a,\brand{Dense})$ and $\tycon{VAE}(a)$.
They return $\tycon{Image}(\mathit{sf}(a)h, \mathit{sf}(a)w, 3)$. Fix
$\sigma$ and inputs. A $\brand{Dense}$ latent of architecture $A$ has
exactly $\chfn(A)$ channels and size $\sigma(h) \times \sigma(w)$
(Definition~\ref{def:inhabitation}). An $\brand{Empty}$ latent need not
have them.

\emph{Code path and no error.} \nodename{VAEDecode} calls
\texttt{vae.decode} (\texttt{nodes.py:330--338}). For a VAE of
architecture $A$ this raises no channel error (Lemma~\ref{lem:f:vae}).
The latent inhabits a type of architecture $A$, so its ghost
architecture is $\ast$ or $A$, and no latent-architecture error
occurs.

\emph{Outputs.} By the same lemma, the call returns an RGB image of
size $\mathit{sf}(A)\sigma(h) \times \mathit{sf}(A)\sigma(w)$. Video
outputs are flattened into a batch of frames, which keeps height,
width, and channels.

\emph{Other nodes.} \nodename{VAEDecodeTiled} decodes the same latent
in tiles (\texttt{nodes.py:354--371}). This changes memory use, not the
output shape.
\end{proof}

\begin{proposition}[Latent operators]
\label{prop:f:latentops}
The schemas of \nodename{RepeatLatentBatch}, \nodename{LatentFromBatch},
\nodename{SetLatentNoiseMask}, \nodename{LatentBlend},
\nodename{LatentUpscale}, and \nodename{LatentUpscaleBy} are faithful.
\end{proposition}

\begin{proof}
\emph{Idea.} Most of these operators keep the latent's shape and
zero state. The upscalers produce a size that the schema either reads
from a widget or leaves fresh.

\emph{Setting.} Each node copies the input latent's dictionary and
replaces the tensor or adds a key (\texttt{nodes.py:1266--1372,
1484--1499, 1550--1553}). Ratio tags therefore pass through.

\emph{Shape-preserving operators.} These nodes keep channels and size.
\begin{itemize}
\item \nodename{RepeatLatentBatch} and \nodename{LatentFromBatch}
  repeat and slice the batch axis. \nodename{SetLatentNoiseMask} adds a
  noise mask of any size. These operations keep channels and size
  (Lemma~\ref{lem:f:pointwise}). They map zero tensors to zero tensors.
  Normalization depends only on channels, size, tag, and whether the
  tensor is zero. The output therefore inhabits the input's type in
  either zero state.
\item \nodename{LatentBlend} takes two $\brand{Dense}$ latents of equal
  type. It resizes the second only when the shapes differ, which keeps
  height and width (Lemma~\ref{lem:f:resize}). It blends them pointwise
  under the first input's tags, which keeps the $\brand{Dense}$ type.
\end{itemize}

\emph{\nodename{LatentUpscale}.} The node takes a $\brand{Dense}$
latent. It has guarded bindings $\tycon{Int}(tw) \ge 64$ and
$\tycon{Int}(th) \ge 64$.
\begin{itemize}
\item \emph{Literals at or above the guard.} Suppose both widgets hold
  literals $\nu_w,\nu_h \ge 64$. The clamps \texttt{max(64,$\cdot$)}
  leave them unchanged. Lemma~\ref{lem:f:resize} gives a tensor of size
  $\nu_h/\!/8 \times \nu_w/\!/8$ with the input's channels. The
  quotient equals $\sigma(th)/8$ by convention~(i). The tags are
  unchanged. The output therefore inhabits
  $\tycon{Latent}(\sigma(th)/8, \sigma(tw)/8, \sigma(c), A,
  \brand{Dense})$.
\item \emph{Literals below the guard, or wired widgets.} A literal below
  $64$ is clamped or replaced by a size derived from the aspect ratio.
  Both literals $0$ return the input. In these cases, and for a wired
  widget, the guard leaves the variable fresh. Then $\sigma'$ assigns
  it $8$ times the resulting side.
\end{itemize}
The input must be $\brand{Dense}$ because a sampler would rescale an
empty latent a second time.

\emph{\nodename{LatentUpscaleBy}.} The node has fresh output sizes
$oh, ow$.
\begin{itemize}
\item For state $\brand{Dense}$, $\sigma'$ assigns them the new size.
\item For state $\brand{Empty}$, the resized tensor is still zero with
  the same channels and tag. Then $\sigma'$ assigns them the size that
  normalization yields.
\end{itemize}
\end{proof}

\subsubsection*{Guiders, samplers, and video models}
These nodes bring a model, conditionings, and a latent together. Their
proofs combine the lemmas on model geometry, attachments, and
normalization.

\begin{proposition}[Guiders and scheduler]
\label{prop:f:guider}
The schemas of \nodename{CFGGuider}, \nodename{BasicGuider},
\nodename{DualCFGGuider}, and \nodename{BasicScheduler} are faithful.
\end{proposition}

\begin{proof}
\emph{Idea.} A guider only stores its model and conditionings. The
schema's side condition is exactly what inhabiting $\tycon{Guider}$
requires.

\emph{Setting.} Consider \nodename{CFGGuider}. It takes
$\tycon{Model}(f,q)$, $\tycon{Cond}(a,\_,c)$, and
$\tycon{Cond}(a,\_,c_2)$ with side condition $\archfn(f) = a \wedge
\mathit{ctrl\_ok}(f,c) \wedge \mathit{ctrl\_ok}(f,c_2)$. It returns
$\tycon{Guider}(a)$.

\emph{Code path and no error.} The node builds a guider object and
stores the model, the conditionings, and the CFG scale
(\texttt{nodes\_custom\_sampler.py:833--839}). No tensor is processed,
so no error occurs.

\emph{Outputs.} The guider's model has family $f_0$ with $\archfn(f_0)
= A$. Each conditioning inhabits $\tycon{Cond}(A,p,\sigma(c))$ for its
own protocol $p$ and a kind with $\mathit{ctrl\_ok}(f_0,\sigma(c))$.
This is the definition of $\tycon{Guider}(A)$
(Definition~\ref{def:inhabitation}).

\emph{Other nodes.} \nodename{BasicGuider} and \nodename{DualCFGGuider}
store one and three conditionings under the same side condition
(\texttt{:809--813, 891--897}). \nodename{BasicScheduler} reads the
model's sampling object and returns a sigma schedule
(\texttt{:33--40}). The schedule is an index-free value that meets its
label's interface.
\end{proof}

\begin{proposition}[Samplers]
\label{prop:f:sampler}
The schemas of \nodename{KSampler}, \nodename{KSamplerAdvanced},
\nodename{SamplerCustom}, and \nodename{SamplerCustomAdvanced} are
faithful.
\end{proposition}

\begin{proof}
\emph{Idea.} After normalization the latent has the model's channel
count and the schema's size. Every attachment either has the model's
architecture or is resized by the sampler. Sampling then preserves the
shape.

\emph{Setting.} Consider \nodename{KSampler}. Its schema takes the
following inputs.
\begin{itemize}
\item A model of type $\tycon{Model}(f,q)$.
\item Conditionings of types $\tycon{Cond}(a,p,c)$ and
  $\tycon{Cond}(a,p_2,c_2)$.
\item A latent of type $\tycon{Latent}(h,w,\chfn(a),a)$ with a free
  zero state.
\end{itemize}
It returns $\tycon{Latent}(h,w,\chfn(a),a,\brand{Dense})$. Its side
condition is $\archfn(f) = a$, $\mathit{proto}(f) = p = p_2$, and
$\mathit{ctrl\_ok}(f,c) \wedge \mathit{ctrl\_ok}(f,c_2)$. Fix $\sigma$
and inputs, so that $\archfn(f_0) = A$.

\emph{Code path.} The node calls \texttt{common\_ksampler}
(\texttt{nodes.py:1555--1578}). The first step is
\texttt{fix\_empty\_latent\_channels} with the model and the latent's
ratio tags (Lemma~\ref{lem:f:empty}).
\begin{itemize}
\item If the latent is $\brand{Dense}$, it has $\chfn(A)$ channels,
  size $\sigma(h) \times \sigma(w)$, and no ratio tag other than
  $\mathit{sf}(A)$. Normalization changes neither channels nor size.
\item If it is $\brand{Empty}$, normalization yields $\chfn(A)$
  channels and size $\sigma(h) \times \sigma(w)$
  (Definition~\ref{def:inhabitation}).
\end{itemize}
In both cases the sampled latent has $\chfn(A)$ channels and size
$\sigma(h) \times \sigma(w)$. The noise is generated in its shape.

\emph{No error.} We check each value that the model reads.
\begin{itemize}
\item \emph{Latent.} The model of family $f_0$ has a latent format
  with $\chfn(\archfn(f_0)) = \chfn(A)$ channels. It accepts the latent
  at any size (Lemma~\ref{lem:f:model}). The latent inhabits a type of
  architecture $A = \archfn(f_0)$, so its ghost architecture is $\ast$
  or $A$, and no latent-architecture error occurs.
\item \emph{Embeddings.} Every architecture in the conditionings'
  ghosts is $A$. Their embeddings have the width that models of $A$
  consume (Lemma~\ref{lem:f:text}). So do their image embeddings
  (Definition~\ref{def:inhabitation}).
\item \emph{Attached latents and masks.} The concat latents and masks,
  and the latent's noise mask, are resized to the sampled latent
  (Lemma~\ref{lem:f:attach}). The reference latents are latents of
  $A$, which the model accepts (Lemma~\ref{lem:f:model}).
\item \emph{Control networks.} Each attached control network has
  architecture $A$ and is built for the models of $A$. It resizes its
  hint or encodes it with a VAE of architecture $A$ (convention~(v)).
  This yields the hint it expects (Lemma~\ref{lem:f:attach}). If no
  VAE was stored, it raises an error that is not a compatibility error
  (\texttt{controlnet.py:278}). Its kind satisfies
  $\mathit{ctrl\_ok}(f_0,\cdot)$. This excludes a control-LoRA on a
  model with a nonstandard nine-channel input, since a control-LoRA
  copies the model's input convolution.
\item \emph{Weight patches.} Skipped keys and shape mismatches of the
  model's weight patches belong to the nodes that added them
  (convention~(v)).
\end{itemize}
Hence no error of Definition~\ref{def:errors} occurs. The protocol
constraint only restricts $\sigma$, since a protocol mismatch degrades
the output without an error.

\emph{Outputs.} Sampling returns a tensor of the sampled latent's shape
(Lemma~\ref{lem:f:model}). The node copies the latent's dictionary,
removes both ratio tags, and stores the samples. The result has
$\chfn(A)$ channels, size $\sigma(h) \times \sigma(w)$, and no ratio
tag. Its ghost is the model's architecture $A$
(Definition~\ref{def:values}). It therefore inhabits
$\tycon{Latent}(\sigma(h),\sigma(w),\chfn(A),A,\brand{Dense})$ with
$\sigma' = \sigma$.

\emph{Other nodes.} The other samplers reuse this argument.
\begin{itemize}
\item \nodename{KSamplerAdvanced} calls \texttt{common\_ksampler} with
  other step and noise settings (\texttt{nodes.py:1634--1641}). The
  argument is unchanged.
\item \nodename{SamplerCustom} performs the same normalization,
  sampling, and tag removal with a sampler and sigmas of index-free
  types (\texttt{nodes\_custom\_sampler.py:751--786}). Its schema has no
  protocol constraint, which only weakens the restriction on $\sigma$.
\item \nodename{SamplerCustomAdvanced} takes the model and the
  conditionings from a guider that inhabits $\tycon{Guider}(A)$
  (\texttt{:1033--1060}). The guider supplies exactly the premises used
  above. These are a model of a family $f_0$ with $\archfn(f_0) = A$,
  and conditionings of architecture $A$ whose kinds satisfy
  $\mathit{ctrl\_ok}(f_0,\cdot)$. Its noise source returns noise of the
  latent's shape by the interface of $\tycon{Noise}$. Its second output
  keeps the input's ratio tags and is undeclared.
\end{itemize}
\end{proof}

\begin{proposition}[Wan video nodes]
\label{prop:f:wan}
The schemas of \nodename{WanImageToVideo}, \nodename{WanVaceToVideo},
and \nodename{WanVideoDecode} are faithful.
\end{proposition}

\begin{proof}
\emph{Idea.} The node resizes all pixel inputs to the requested size
and encodes them with the one VAE, so every attachment has architecture
$A$. The divisibility side condition makes the empty latent's size
exact.

\emph{Setting.} Consider \nodename{WanImageToVideo}. Its schema takes
conditionings and a VAE of architecture $a$, a start image of
independent size, and a CLIP-vision output of width $1280$. The size
widgets designate $w$ and $h$. The side condition is $16 \mid w \wedge
16 \mid h$. The outputs are the two conditionings and
$\tycon{Latent}(h/8,w/8,16,a,\brand{Dense})$. Fix $\sigma$ and inputs.

\emph{Code path and no error.} The node resizes the start image to
$\nu_w \times \nu_h$ (Lemma~\ref{lem:f:resize},
\texttt{nodes\_wan.py:41--61}). It encodes the image's RGB channels
with the VAE, which raises no error (Lemma~\ref{lem:f:vae}). It
attaches the result and a mask as concat latent and mask to both
conditionings.

\emph{Conditioning outputs.} The concat latent has the VAE's
architecture $A$, so the conditionings keep the single architecture
$A$. The CLIP-vision output is attached as metadata. Of the
architectures in the forest, only Wan models read it
(\texttt{model\_base.py:1634}). HunyuanVideo~1.5 also reads it but lies
outside the forest (\texttt{:2399}, convention~(iii)). Wan models read
it through an image projection of width $1280$
(\texttt{wan/model.py:514}). Wan models without that projection ignore
it (\texttt{:585--586}). Its width is therefore one that models of $A$
consume, and both conditionings inhabit $\tycon{Cond}(A)$.

\emph{Latent output.} The latent is
\texttt{zeros([b,16,t,$\nu_h$//8,$\nu_w$//8])} without a tag. Since
$16 \mid \nu$, we have $\nu/\!/8 = \nu/8$. Its ghost $\ast$ inhabits
every architecture. It therefore inhabits
$\tycon{Latent}(\nu_h/8,\nu_w/8,16,A,\brand{Dense})$.

\emph{Other nodes.} The other two nodes follow the same pattern.
\begin{itemize}
\item \nodename{WanVaceToVideo} resizes the control video, the masks,
  and the reference image to $\nu_w \times \nu_h$
  (\texttt{nodes\_wan.py:325--366}). It encodes the RGB channels with
  the VAE and attaches the latents to both conditionings. The view of
  the mask as blocks of $8 \times 8$ pixels is exact because $8 \mid
  \nu$. Its trim count is fresh, and $\sigma'$ assigns it the returned
  integer.
\item \nodename{WanVideoDecode} takes a VAE of type
  $\tycon{WanVae}(a)$ and a wrapper latent with $\chfn(a)$ channels.
  The wrapper decodes it with the loaded VAE, whose channel count
  $\chfn(A)$ matches (WanVideoWrapper \texttt{nodes.py:2081--2149}).
  The result is RGB frames whose fresh size $\sigma'$ assigns.
\end{itemize}
\end{proof}

\subsubsection*{Image and mask nodes}
These nodes create, measure, and transform images and masks. Their
proofs track height, width, channels, and modality.

\begin{proposition}[Sources, primitives, and queries]
\label{prop:f:imgsrc}
The schemas of \nodename{PrimitiveInt}, \nodename{INTConstant},
\nodename{EmptyImage}, \nodename{SolidMask},
\nodename{GetImageSize}, \nodename{LoadImage},
\nodename{GetVideoComponents}, \nodename{VHS\_LoadVideo}, and
\nodename{MathExpression|pysssss} are faithful.
\end{proposition}

\begin{proof}
\emph{Idea.} These nodes create values from widgets or files, or read
sizes off an image. Each output index is designated by a widget, read
from a typed input, or fresh.

\emph{Code path and outputs.} We go through the nodes in groups.
\begin{itemize}
\item \emph{Integer primitives.} They have the schema $\tycon{Int}(v)
  \to \tycon{Int}(v)$ and return their widget value
  (\texttt{nodes\_primitive.py:57}, KJNodes \texttt{nodes.py:55}). A
  literal $\nu$ designates $v$, so the output carries $\sigma(v) =
  \nu$. For a wired widget $v$ is fresh, and $\sigma'(v)$ is the
  received integer.
\item \emph{Filled tensors.} \nodename{EmptyImage} and
  \nodename{SolidMask} have unguarded bindings $\tycon{Int}(w)$ and
  $\tycon{Int}(h)$. They return tensors filled to size $\nu_h \times
  \nu_w$, with three channels for the image (\texttt{nodes.py:1975--1981},
  \texttt{nodes\_mask.py:196--197}).
\item \emph{Size query.} \nodename{GetImageSize} returns
  \texttt{shape[2]} and \texttt{shape[1]} of an image of type
  $\tycon{Image}(h,w,3)$ (\texttt{nodes\_images.py:591--593}). These
  are $\sigma(w)$ and $\sigma(h)$. Its batch output is fresh.
\item \emph{File loaders.} \nodename{LoadImage} converts every frame to
  RGB and returns one minus the alpha channel as a mask. For an image
  without alpha it returns a $64 \times 64$ zero mask
  (\texttt{nodes.py:1745, 1772}). Its image therefore has three
  channels, and its mask size is independent.
  \nodename{GetVideoComponents} returns RGB frames of the video's size.
  \nodename{VHS\_LoadVideo} returns four channels for a video with
  alpha, so its channel index is fresh (VHS
  \texttt{load\_video\_nodes.py:379}).
\item \emph{Expressions.} \nodename{MathExpression|pysssss} evaluates
  an arbitrary expression, and its integer is fresh.
\end{itemize}
The remaining outputs are fresh, and $\sigma'$ assigns them the
observed values.

\emph{No error.} None of these nodes has a tensor precondition that its
inputs could violate.
\end{proof}

\begin{proposition}[Size-preserving image and mask operators]
\label{prop:f:imgpres}
The schemas of \nodename{ImageInvert}, \nodename{ImageFromBatch},
\nodename{ImageBatch}, \nodename{SaveImage}, \nodename{PreviewImage},
\nodename{PreviewBridge}, \nodename{MaskToImage},
\nodename{InvertMask}, \nodename{GrowMask}, \nodename{FeatherMask},
and \nodename{Canny} are faithful.
\end{proposition}

\begin{proof}
\emph{Idea.} Each output keeps the height and width of the first input
and has the channel count its schema states.

\emph{Code path and outputs.} We go through the nodes in groups.
\begin{itemize}
\item \nodename{ImageInvert}, \nodename{ImageFromBatch},
  \nodename{SaveImage}, and \nodename{PreviewImage} apply a pointwise
  map, a batch slice, or return their input.
\item \nodename{MaskToImage}, \nodename{InvertMask},
  \nodename{GrowMask}, and \nodename{FeatherMask} apply a channel
  expansion to three, a pointwise map, a padded filter, or a clamped
  feathering (\texttt{nodes\_mask.py:124--384}).
\item \nodename{ImageBatch} takes two RGB images, so it skips its
  channel padding. It resizes the second to the first
  (Lemma~\ref{lem:f:resize}) and concatenates the batches
  (\texttt{nodes.py:1948--1957}).
\item \nodename{Canny} returns its edge map expanded to three channels
  and stamps the modality \brand{Edge} (\texttt{nodes\_canny.py:32--35}).
\item \nodename{PreviewBridge} passes its image and modality through.
  Its mask is the drawn mask or a $64 \times 64$ default, whose fresh
  size $\sigma'$ assigns (Impact \texttt{bridge\_nodes.py:71, 74,
  174}).
\end{itemize}
The first two groups keep height and width by
Lemma~\ref{lem:f:pointwise}.

\emph{No error.} None of the nodes has a precondition on height,
width, or channels that its typed inputs violate.
\end{proof}

\begin{proposition}[Resizing and cropping]
\label{prop:f:resize}
The schemas of \nodename{ImageScale}, \nodename{ImageScaleBy},
\nodename{ImageScaleToTotalPixels}, \nodename{ImageUpscaleWithModel},
\nodename{ImageResize+}, \nodename{ImageCrop},
\nodename{ImagePadForOutpaint}, \nodename{ImageStitch}, and
\nodename{CropMask} are faithful.
\end{proposition}

\begin{proof}
\emph{Idea.} The output size of \nodename{ImageScale} is the widget
literal when the guard holds and fresh otherwise. The other nodes leave
output sizes fresh, so only channels, modality, and shared extents
need checking.

\emph{Setting.} Consider \nodename{ImageScale}, whose schema maps
$\tycon{Image}(h,w,3,m)$ to $\tycon{Image}(th,tw,3,m)$. It has guarded
bindings $\tycon{Int}(tw) \ge 1$ and $\tycon{Int}(th) \ge 1$.

\emph{Code path, no error, and outputs.} The argument splits on the
guard.
\begin{itemize}
\item \emph{Literals at or above the guard.} Suppose both widgets hold
  literals $\nu_w,\nu_h \ge 1$. The node calls
  \texttt{common\_upscale} with $(\nu_w,\nu_h)$
  (\texttt{nodes.py:1882--1895}). This returns size $\nu_h \times
  \nu_w$ with three channels (Lemma~\ref{lem:f:resize}), and it keeps
  the modality. The output inhabits
  $\tycon{Image}(\sigma(th),\sigma(tw),3,\sigma(m))$.
\item \emph{Literal $0$, or wired widgets.} A literal $0$ makes the
  node derive that side from the aspect ratio. When both are $0$ it
  returns the input. In these cases, and for a wired widget, the guard
  leaves the variable fresh. Then $\sigma'$ assigns it the resulting
  side.
\end{itemize}

\emph{Other nodes.} The other schemas give all output sizes fresh
variables, which $\sigma'$ sets to the observed sizes. Only channels,
modality, and shared extents remain to be checked.
\begin{itemize}
\item \nodename{ImageScaleBy} and \nodename{ImageScaleToTotalPixels}
  resize RGB images and keep the modality.
\item \nodename{ImageUpscaleWithModel} applies an upscale model. The
  model maps RGB images to RGB images by the interface of its label.
\item \nodename{ImageCrop} and \nodename{CropMask} slice.
\item \nodename{ImageStitch} resizes or pads its inputs to a common
  extent before concatenating them (\texttt{nodes\_images.py:318--373}).
\item \nodename{ImagePadForOutpaint} pads the image and builds the mask
  with one extent (\texttt{nodes.py:2004--2045}). This matches the
  shared variables $oh, ow$ of its two outputs.
\item \nodename{ImageResize+} returns its output's width and height as
  integers (essentials \texttt{image.py:458}). This matches the shared
  variables $h_2, w_2$.
\end{itemize}
\end{proof}

\begin{proposition}[Preprocessors]
\label{prop:f:preproc}
The schemas of \nodename{CannyEdgePreprocessor},
\nodename{DepthAnythingPreprocessor},
\nodename{DepthAnythingV2Preprocessor},
\nodename{MiDaS-DepthMapPreprocessor},
\nodename{Zoe-DepthMapPreprocessor}, \nodename{OpenposePreprocessor},
\nodename{DWPreprocessor}, \nodename{LineArtPreprocessor},
\nodename{AnimeLineArtPreprocessor}, \nodename{ScribblePreprocessor},
\nodename{SemSegPreprocessor}, and \nodename{M-LSDPreprocessor} are
faithful.
\end{proposition}

\begin{proof}
\emph{Idea.} A preprocessor accepts any RGB image and stamps its
modality. The output size is left fresh.

\emph{Setting.} Each schema maps $\tycon{Image}(h,w,3)$ to
$\tycon{Image}(oh,ow,3,s)$ for a ground modality $s$ and fresh $oh,
ow$.

\emph{Code path.} The pack resizes the RGB input to the detection
resolution and runs the detector. It returns a three-channel result at
that resolution (controlnet\_aux \texttt{utils.py:64--92},
\texttt{util.py:155--171}).

\emph{No error.} The detectors accept any RGB input size, so no error
occurs.

\emph{Outputs.} The output's ghost modality is the stamped $s$. Then
$\sigma'$ assigns $oh$ and $ow$ the returned size.
\end{proof}

\subsubsection*{Sinks and pipes}
These nodes consume values or bundle them into a pipe.

\begin{proposition}[Sinks and the vision encoder]
\label{prop:f:sinks}
The schemas of \nodename{CLIPVisionEncode}, \nodename{CreateVideo},
\nodename{VHS\_VideoCombine}, \nodename{MaskPreview},
\nodename{ImageCompare}, \nodename{Image Comparer (rgthree)},
\nodename{CheckpointSave}, and \nodename{ModelSave} are faithful.
\end{proposition}

\begin{proof}
\emph{Idea.} These nodes consume their inputs without shape
preconditions. The one indexed output gets a fresh width.

\emph{No error.} None of these nodes has a precondition beyond its
input labels.
\begin{itemize}
\item The CLIP-vision encoder resizes its RGB input to the model's
  resolution and emits an embedding (\texttt{clip\_vision.py:63--80}).
\item The video writers accept frames of any size, and VHS pads odd
  sizes (VHS \texttt{nodes.py:445}).
\item The previews and comparers display their inputs.
\item The savers serialize whatever they receive. The independent
  indices of \nodename{CheckpointSave} are therefore sound.
\end{itemize}

\emph{Outputs.} The embedding's width is fresh, and $\sigma'$ sets it
to the observed width. The other outputs are index-free.
\end{proof}

\begin{proposition}[Impact pipe nodes]
\label{prop:f:pipe}
The schemas of \nodename{ToBasicPipe}, \nodename{FromBasicPipe}, and
\nodename{FaceDetailer} are faithful.
\end{proposition}

\begin{proof}
\emph{Idea.} A pipe is a tuple, so packing and unpacking keep the
component types. \nodename{FaceDetailer} runs the encode, sample, and
decode steps whose faithfulness is already proved.

\emph{Setting.} A pipe is the tuple (model, CLIP, VAE, positive,
negative) (Impact \texttt{pipe.py:148, 164}). We read
$\tycon{BasicPipe}(f,q,a)$ as inhabited by a tuple whose first two
components inhabit $\tycon{Model}(f,q)$ and $\tycon{Clip}(a)$.

\emph{Pipe nodes.} The two pipe nodes pack and unpack the tuple.
\begin{itemize}
\item \nodename{ToBasicPipe} builds the tuple from its inputs, which
  inhabit these types.
\item \nodename{FromBasicPipe} returns the components. Its model and
  CLIP outputs inhabit $\tycon{Model}(f_0,\sigma(q))$ and
  $\tycon{Clip}(A)$. Its VAE and conditioning outputs have fresh
  indices, which $\sigma'$ sets to their ghosts.
\end{itemize}

\emph{\nodename{FaceDetailer}: setting and code path.} The node takes a
model of family $f$ and a CLIP, a VAE, and conditionings of
architecture $\archfn(f)$. It requires $\mathit{ctrl\_ok}$ for both
conditionings. For each detected crop it encodes the crop with its
VAE, samples it, and decodes it (Impact \texttt{core.py:352--397}).

\emph{\nodename{FaceDetailer}: no error and outputs.} Encoding raises no
error (Lemma~\ref{lem:f:vae}). The argument of
Proposition~\ref{prop:f:sampler} applies to the sampling step with
architecture $\archfn(f_0)$ for every weight. The decoded latent is
$\brand{Dense}$ with $\chfn(\archfn(f_0))$ channels, so decoding raises
no error (Lemma~\ref{lem:f:vae}). The output image keeps the input's
modality. Its size and the mask's size are fresh.

\emph{\nodename{FaceDetailer}: untyped occurrences.} The undeclared
optional inputs \texttt{detailer\_hook} and
\texttt{scheduler\_func\_opt} are absent by convention~(ii). The
\texttt{wildcard} widget may name LoRAs with \texttt{<lora:...>} tags.
The node loads them through \nodename{LoraLoader} (Impact
\texttt{core.py:267--268}, \texttt{wildcards.py:919--960}). A
wildcard-file reference \texttt{\_\_name\_\_} or an option group
\texttt{\{a|b\}} can produce such a tag from text the literal does not
contain (\texttt{wildcards.py:654, 728}). The checker types an
occurrence as $\tyDyn$ when its wildcard is wired or contains a tag, a
wildcard-file reference, or an option group. The schema is therefore
proved for the remaining wildcards. These contain no LoRA tag, and for
them the node patches no weights.
\end{proof}

\parh{Remark (nodes left untyped).} Some preconditions lie outside the
index language, so the following nodes have no schema.
\begin{itemize}
\item \nodename{StyleModelApply} and \nodename{unCLIPConditioning}
  attach image embeddings of unindexed width (\texttt{nodes.py:1159,
  1180}).
\item \nodename{ImageToMask} may select a missing alpha channel.
\item \nodename{MaskComposite} and \nodename{ImageBlur} need an
  inequality between a widget and an image side.
\item \nodename{DiffControlNetLoader} adds model weights whose input
  convolution may be wider.
\item \nodename{KSampler (Efficient)} runs scripts.
\item Audio latents have one axis and lie outside $\tycon{Latent}$.
\end{itemize}
These nodes are typed $\tyDyn$, and Theorem~\ref{thm:soundness} makes
no claim about them.
\fussy

%% file: sections/appendix-validity.tex
\section{Classifier Validity Study}
\label{sec:appendix-validity}

The failure study (\S\ref{sec:study}) answers RQ1 by counting failures in
ComfyUI's GitHub reports. For each of the 10{,}134 issues and discussions, it
asks whether the report describes a \emph{semantic composition failure}, i.e.\
a failure that passes ComfyUI's nominal label check because it depends on facts
the labels omit, such as a model's family. If so, it also asks
whether the failure is \emph{silent}, i.e.\ whether the workflow completes
without an exception but produces wrong output. Labeling this many reports by
hand is impractical. The study therefore lets a large language model label
every report, following one fixed written rubric that encodes the taxonomy of
\S\ref{sec:study:method}. We call this model, \texttt{deepseek-v4-flash}, the
\emph{classifier}.

Every percentage in \S\ref{sec:study} is computed from the classifier's labels.
If the classifier mislabels reports, these percentages are wrong. This appendix
therefore asks whether the classifier's labels can be trusted.

\parh{Quantities checked.} The validity check concerns two quantities.
\begin{itemize}
\item The \emph{prevalence} is the fraction of reports labeled as semantic
  composition failures. The census, for example, finds 8.3\% of issues. We call
  the per-report yes/no decision behind it the \emph{semantic label}.
\item The \emph{silent fraction} is the fraction of these failures labeled
  silent. The census finds 31.5\%.
\end{itemize}

\parh{Checks.} Without human labels, we test the classifier by comparing it
with other language models. We run two checks.
\begin{enumerate}
\item \emph{Agreement across model developers.} Do models from other developers,
  applying the same rubric to the same reports, reach the same labels and the
  same percentages?
\item \emph{Two-rater consensus.} Does combining the classifier with a second
  rater, and resolving their disagreements, change the prevalence?
\end{enumerate}
All these checks support the classifier. In Check~1, four further models agree
with it at $\kappa=0.71$--$0.81$, and their own prevalences lie on both sides of
its estimate. In Check~2, the consensus prevalence lies within one percentage
point of the classifier's. Because these checks compare models with models, they
establish reliability but not accuracy against human judgment. We return to this
limitation at the end.

\parh{Validation subsample.} Re-labeling all 10{,}134 reports with further
models is unnecessary, because a uniform random subsample suffices to estimate
how well raters agree. We drew a seeded subsample of 1{,}583 reports (985 issues
and 598 discussions), and all results below refer to them.

On this subsample, the classifier labels 9.8\% of the issues as semantic
composition failures. This is slightly above its census estimate of 8.3\%, a
difference within sampling variation (95\% confidence interval 8.1\%--11.9\%).
The difference does not affect the checks, because they compare raters on the
same items.

\parh{Check 1: agreement across model developers.} Suppose the classifier's labels
reflected a bias of its model and not the content of the reports. Then a
model from another developer would disagree with them and would give
different percentages. To test this, four further models independently re-labeled
every subsample item under the identical rubric: \texttt{deepseek-v4-pro} from
the same developer, and \texttt{glm-5.2} (Zhipu), \texttt{minimax-m3} (MiniMax),
and \texttt{hy3-preview} (Tencent) from three other developers. We report results
on the 985 issues. The 598 discussions contain only 11 of the classifier's 108
semantic labels in the subsample, too few for a stable estimate of agreement.

Table~\ref{tab:validity} reports two kinds of evidence for each rater.
\begin{itemize}
\item \emph{Agreement with the classifier} on the semantic label shows whether
  the labels match item by item. We report raw agreement and Cohen's $\kappa$,
  which corrects raw agreement for agreement by chance. The correction matters
  here because semantic composition failures are rare. A rater that never assigns the
  semantic label would still agree with the classifier on 90\% of the issues.
\item \emph{The rater's own estimates}, i.e.\ the prevalence and the silent
  fraction computed from its labels alone, show whether the numbers of
  \S\ref{sec:study} still hold under a different model. This is the more important
  evidence.
\end{itemize}

\begin{table}[h]
\caption{Agreement of four further raters with the classifier
(\texttt{deepseek-v4-flash}) on the 985 issues of the validation subsample, and
each rater's own semantic prevalence and silent fraction. $\kappa$ and
agreement refer to the binary semantic label.}
\label{tab:validity}
\begin{tabular}{llrrrr}
\toprule
Rater & Developer & $\kappa$ & Agreement & Prevalence & Silent \\
\midrule
\texttt{deepseek-v4-flash} & DeepSeek & --- & --- & 9.8\% & 31\% \\
\midrule
\texttt{deepseek-v4-pro} & DeepSeek & 0.80 & 96.8\% & 8.2\% & 30\% \\
\texttt{glm-5.2} & Zhipu & 0.74 & 96.0\% & 6.7\% & 26\% \\
\texttt{minimax-m3} & MiniMax & 0.71 & 95.1\% & 8.4\% & 31\% \\
\texttt{hy3-preview} & Tencent & 0.81 & 96.7\% & 9.9\% & 29\% \\
\bottomrule
\end{tabular}
\end{table}

We identify two kinds of evidence supporting the classifier. First, all four
raters agree with it substantially, at $\kappa=0.71$--$0.81$ and raw agreement
of $95$--$97\%$. Second, their own prevalences ($6.7\%$--$9.9\%$) lie on both
sides of the classifier's $9.8\%$, and their silent fractions ($26$--$31\%$)
match its $31\%$. This agreement is meaningful because three of the raters are
trained by other organizations. Correlated errors are less plausible across such
models than between two models from one lab. We therefore conclude that the
percentages of \S\ref{sec:study} reflect the reports, not the developer of the
classifier's model.

\parh{Check 2: two-rater consensus.} Check~1 compares each rater with the
classifier one at a time. A complementary question is whether a consensus of
two raters, which can correct some individual errors, would arrive at a
different prevalence. We therefore ran a dual-rater protocol on the subsample.
\texttt{deepseek-v4-flash} and \texttt{deepseek-v4-pro} label each item
independently. Where they disagree, an adjudication pass shows a model the
report and asks it for the final label.

Over all 1{,}583 items, the two raters agree substantially on every part of
the label.
\begin{itemize}
\item $\kappa=0.79$ on the semantic label.
\item $\kappa=0.81$ on the seven-way primary class of the taxonomy.
\item $\kappa=0.88$ on whether the report describes a failure at run time. This
  flag defines the denominator of the runtime percentages in \S\ref{sec:study},
  e.g.\ the share of semantic composition failures among issues that report a runtime
  failure (20.9\%).
\item $\kappa=0.97$ on silent versus crashing, among the items that both raters
  label semantic.
\end{itemize}

The consensus also leaves the issue prevalence almost unchanged. Adjudication
gives 9.0\%. Two simpler consensus rules bound the effect of adjudication. The
strict rule counts an issue only when both raters label it semantic and gives
7.4\%. The lenient rule counts it when either rater does and gives 10.7\%. The
classifier's 9.8\% lies inside this range and within one percentage point of
the adjudicated 9.0\%. The classifier is thus marginally more liberal than the
consensus, but close to it.

%% file: sections/appendix-evaluation.tex
\section{Evaluation Details}
\label{app:eval}

This appendix gives the details needed to interpret and reproduce results for
RQ2--RQ4. Appendix~\ref{app:user-study} gives the details of the user study (RQ5).

\subsection{Fault-Injection Benchmark}
\label{app:eval:injection}

Measuring detection requires faulty workflows whose ground truth is known. Real
failure reports rarely provide them, because most lack an executable graph or a
confirmed diagnosis. We therefore inject faults into workflows that are known to
run, so that each injected change is the only difference between a faulty graph
and a working one. Below we describe the operators and their selection rules.

\parh{Seed workflow.} The seed workflows come from ComfyUI's 765 official
workflows: the 673 official templates and 92 workflows from the official
\texttt{ComfyUI\_examples} gallery. They span SD1.5, SDXL, SD3.5, Flux, Wan,
Hunyuan-video, LTXV, and other families. A seed workflow must be one that we can
execute: every node is a built-in node, every model file has a public download
URL, the model files total at most 48\,GiB, the workflow contains a sampler, and
it has no audio or 3D node. These rules select 157 workflows. One of them yields
no graph to test: all its model nodes are switched off (i.e., bypassed) in the
saved file. Of the other 156, 24 fail to run unmodified, which leaves 132. Most
of these 24 files are subgraph blueprints, i.e.\ reusable components without an
output node, which ComfyUI does not run on their own. We convert every workflow
with ComfyUI's own frontend, the code path behind its \emph{Queue} button, so
that each graph is executed exactly as a user would submit it.

\parh{Operators.} The operators cover the dimensions of the census
(\S\ref{sec:study}) and create both candidate faults and benign controls
(Table~\ref{tab:injection-design}).
\begin{itemize}
\item \emph{File swap.} Replace a loaded model file (checkpoint, diffusion
  model, VAE, text encoder, LoRA, or ControlNet) by another public file of the
  same role.
\item \emph{Rewiring.} Reconnect an input to another output of the same nominal
  type in the same workflow, as a user does when dragging a wire.
\item \emph{Grafts.} Combine two workflows into one graph. A latent graft feeds
  the sampled latent of a generator to the image-conditioned sampler of another
  workflow, and an image graft feeds the decoded image instead. Image grafts
  are benign by construction and serve as the baselines of latent grafts.
\item \emph{Dimension, quantization, and modality.} Change a widget or a file so
  that a size, divisibility, quantization, or hint-modality requirement fails or
  holds. Each operator applies only where ComfyUI's source at version 0.27.0
  raises, or documents a wrong result, when the requirement fails.
\end{itemize}

\begin{table}[t]
\caption{Operator families of the fault-injection benchmark and the graphs they
contribute. A benign change keeps the documented property that its candidate
fault breaks.}
\label{tab:injection-design}
\centering\small
\begin{tabular}{@{}p{0.17\linewidth}p{0.53\linewidth}rr@{}}
\toprule
Family & Injected change & Candidate faults & Benign \\
\midrule
File swap & A model file of another architecture, latent format, or modality & 308 & 95 \\
Rewiring & An output of the same nominal type but another provenance & 48 & 446 \\
Latent graft & Another workflow's latent fed to a sampler & 175 & 150 \\
Image graft & Another workflow's decoded image fed to the image input & -- & 325 \\
Dimension & A size or frame count that breaks a divisibility or shape rule & 54 & 46 \\
Quantization & An fp8 model without its required scale file & 1 & 1 \\
Modality & A ControlNet or hint of another modality & 33 & 23 \\
\midrule
Total & With the 132 originals, 1,837 graphs & 619 & 1,086 \\
\bottomrule
\end{tabular}
\end{table}

\parh{Labels.} The expected verdict of each graph must not come from \tool: if
it did, the checker would agree with it by construction. We label file
properties from public documentation only. A separate labeling procedure, run by
an LLM agent in a sandbox that exposes the file list, the public model
repositories, and model cards, but never the checker, its schemas, its model
database, or any verdict, records each file's role, architecture, latent format,
and, for ControlNets, modality. A swap, rewiring, or graft is a candidate fault
if the documented property of the new value differs from that of the old one,
and benign if they agree. The dimension, quantization, and modality operators
take their labels from the cited code condition in ComfyUI's source.

\parh{Ground truth.} Labels come from documentation, which can be wrong, so we
execute all 1,837 graphs. Each graph runs once on ComfyUI~v0.27.0 with a 900\,s
timeout, on a server with one 48\,GB NVIDIA GPU. Of the 619 candidate faults, 12
count in no measure because their baseline does not run correctly. Of the other
607, 417 raise an error, 188 complete, and 2 time out. A candidate fault that
completes is a silent fault if its output is wrong. We therefore check the 188
completed candidate faults in two steps. We use an image oracle to compare each
output with its baseline's output through CLIP ViT-L/14 embeddings, followed by
a visual review to confirm image oracle's flagged candidate faults. It confirms
42 candidate faults whose output is blank or distorted. We then judge the other
146 by labor, confirming 41 additional silent faults. The remaining 107
candidate faults stay unconfirmed and do not count in the benchmark.

\subsection{Schema Availability in RQ2}
\label{app:eval:detection}

In practice, many nodes lack schemas. This experiment asks how detection changes
as schemas are removed. So that only the schemas we remove are missing. Each
workflow of the fault-injection benchmark yields many graphs: its unmodified
original and the derived graphs we inject changes into, 1,837 graphs from 132
workflows in total. We take the workflows whose original workflow has a schema
for every node, and from their derived graphs we keep those in which every node
still has a schema, and \tool already detects those injected graphs as faulty
graphs. This yields 169 faulty graphs.

The experiment removes schemas in random orders, so that no particular choice of
missing classes drives the result. We draw 100 orders. For order $s$, we hash
the string ``$s$/\textit{class}'' with SHA-256 and map its first 32 bits to
$[0,1)$. At removal level $t \in \{0, 0.1, \ldots, 1\}$, we remove the schema of
every class whose value is below $t$. Each class is thus removed with
probability $t$. Because the hash value of a class is fixed within an order, the
classes removed at a lower level are also removed at every higher level. Within
an order, differences between levels are therefore due only to the additional
removals. At every level of every order, we check all graphs and count a fault
as detected only if \tool rejects it and accepts its baseline at that level.
Execution outcomes do not depend on schemas, so only the static verdicts are
recomputed, and with no class removed they equal the verdicts of
RQ2. Figure~\ref{fig:synthetic-detection}b and
Table~\ref{tab:synthetic-erasure} report the detections kept and the mean
fraction of nodes that keep their schemas. No run rejects any of the 82 baselines of these faults.

Random removal shows how detection degrades on average, but not which classes it
depends on. We therefore also remove the schema of each of the node classes
alone (Table~\ref{tab:synthetic-single-class}). Only 29 classes carry
detections. They introduce a model's family or latent format, such as the model
and VAE loaders, or require two components to agree on it, such as the sampler,
the VAE decoder, and the reference-latent and video-conditioning nodes.

\input{study/data/rq2_erasure_real}

\subsection{Time, Constraints, Solver Work, and Memory}
\label{app:eval:resources}

RQ3 asks whether checking is cheap enough to run before every execution. We
answer it with four measures, each defined below.

\parh{Time.} We time the public checking pipeline from parsing and normalization
through constraint solving and diagnostic generation. After one warm-up check,
we time 25 checks per graph and take the graph's median. Execution time starts
at submission and includes model loading, computation, and client polling until
completion or a runtime error. We compute the ratio of execution to checking
time for each graph before taking the median over graphs. Execution times come
from the GPU run of \S\ref{app:eval:injection}, in which the server caches no
node outputs, so each run includes loading its models.
Table~\ref{tab:synthetic-cost} breaks the times down by operator family.

\parh{Constraint size.} Constraint size indicates how much work each graph gives
the solver. We count the constraints generated before solving: brand equations,
dimension equations, divisibility atoms, and links to model properties such as
channel counts and VAE scale factors. Variables are the distinct brand and
dimension indices that these constraints reference. An interpreted link is
counted once, when it is generated. The equations it later releases are not
counted again. Figure~\ref{fig:synthetic-cost}b groups the counts by operator
family and reports the mean over the graphs in each group.

\parh{Solver work.} Constraint size alone does not show how often the solver
runs, because diagnosis calls it repeatedly. A separate instrumented check
therefore records every call to the solver, including the calls made by core
minimization and repair-hint generation. The instrumentation observes the public
pipeline without changing its inputs or solver decisions. As a safeguard, the
instrumented verdict must match the independently measured verdict for every
graph. In Figure~\ref{fig:synthetic-cost}c, each point is one graph.

\parh{Peak memory.} We measure the memory that a complete checking process
needs, since this is what a user would run.
\begin{itemize}
\item \emph{Process.} For each graph, we launch three fresh Python processes in
  sequence. Each process loads the schema registry, checks the graph, and
  generates its diagnostics. It then reads \texttt{VmHWM} from Linux
  \texttt{/proc/self/status}, the resident-memory high-water mark of its
  address space. This peak includes Python, the imported modules, and the
  schema registry, in addition to checking itself.
\item \emph{Reported values.} We convert KiB to MiB and plot the median of the
  three peaks per graph. The reported maximum is the largest peak over all
  processes. The memory workers run without tracing, so instrumentation does
  not inflate the peaks.
\item \emph{Choice of metric.} We use \texttt{VmHWM} because
  \texttt{ru\_maxrss} reports the same startup high-water mark for every worker
  and thus cannot distinguish graphs.
\item \emph{Environment.} The cost run uses the 1,837 graphs and the checker of
  RQ2. It uses Python~3.12 on an exclusively allocated node with an AMD
  EPYC~9754 CPU. The released report records the checker version, CPU model, and
  Python version.
\end{itemize}

\input{study/data/rq3_real_cost}

\subsection{Effectiveness Study on Real-World Workflows}
\label{app:eval:real}

RQ4 examines whether the results carry over to real workflows and shows
representative case studies. In this appendix, we give the full list.
Table~\ref{tab:rq4-cases} lists all 108 reported failures, each linked to its
upstream ComfyUI thread, with its failure mode and whether \tool detects it.
\tool detects 72 of them, 64 of the 88 crashing and 8 of the 20 silent ones.

\input{study/data/rq4_cases_table}

%% file: study/data/rq2_erasure_real.tex
\begin{table}[t]
\caption{Detection as node-class schemas are removed in 100 random orders from the 169 fully typed faulty graphs that the full library detects. Typed nodes refers to these graphs, and detections kept is relative to them. Cells give the mean over orders and, in parentheses, the 10th and 90th percentiles. False alarms is the maximum over orders, out of the 82 baselines of these faults.}
\label{tab:synthetic-erasure}
\centering\small
\begin{tabular}{@{}rrrr@{}}
\toprule
Classes removed & Typed nodes & Detections kept & False alarms \\
\midrule
0\% & 100.0\% (100.0--100.0) & 100.0\% (100.0--100.0) & 0/82 \\
10\% & 89.8\% (80.3--97.8) & 77.6\% (43.8--97.6) & 0/82 \\
20\% & 79.2\% (66.6--91.0) & 56.3\% (27.8--82.2) & 0/82 \\
30\% & 69.3\% (55.1--83.6) & 38.9\% (16.6--69.2) & 0/82 \\
40\% & 58.2\% (43.3--73.0) & 25.8\% (9.5--49.1) & 0/82 \\
50\% & 48.0\% (32.5--63.6) & 17.7\% (7.1--30.2) & 0/82 \\
60\% & 37.8\% (24.5--54.1) & 11.7\% (0.6--21.3) & 0/82 \\
70\% & 28.5\% (17.0--41.2) & 8.3\% (0.0--17.2) & 0/82 \\
80\% & 18.7\% (8.1--30.7) & 5.7\% (0.0--16.6) & 0/82 \\
90\% & 8.9\% (2.5--18.9) & 2.9\% (0.0--9.5) & 0/82 \\
100\% & 0.0\% (0.0--0.0) & 0.0\% (0.0--0.0) & 0/82 \\
\bottomrule
\end{tabular}
\end{table}
\begin{table}[t]
\caption{Detections lost, out of 169, when a single node class loses its schema. Of the 123 classes in the subset, the 29 listed carry all detections. No removal causes a false alarm.}
\label{tab:synthetic-single-class}
\centering\small
\begin{tabular}{@{}lr@{\qquad}lr@{}}
\toprule
Class & Lost & Class & Lost \\
\midrule
\texttt{KSampler} & 93 & \texttt{ModelSamplingAuraFlow} & 7 \\
\texttt{UNETLoader} & 53 & \texttt{BasicGuider} & 6 \\
\texttt{VAELoader} & 46 & \texttt{LoraLoaderModelOnly} & 6 \\
\texttt{VAEDecode} & 34 & \texttt{DualCFGGuider} & 4 \\
\texttt{CLIPTextEncode} & 28 & \texttt{CFGGuider} & 3 \\
\texttt{CLIPLoader} & 22 & \texttt{Canny} & 3 \\
\texttt{DualCLIPLoader} & 20 & \texttt{SamplerCustom} & 3 \\
\texttt{ReferenceLatent} & 19 & \texttt{DifferentialDiffusion} & 2 \\
\texttt{CheckpointLoaderSimple} & 17 & \texttt{EmptySD3LatentImage} & 2 \\
\texttt{WanVaceToVideo} & 16 & \texttt{TextEncodeQwenImageEditPlus} & 2 \\
\texttt{QuadrupleCLIPLoader} & 14 & \texttt{CFGNorm} & 1 \\
\texttt{ModelSamplingSD3} & 13 & \texttt{EmptyLatentImage} & 1 \\
\texttt{SamplerCustomAdvanced} & 8 & \texttt{FluxGuidance} & 1 \\
\texttt{ControlNetApplyAdvanced} & 7 & \texttt{KSamplerAdvanced} & 1 \\
\texttt{ControlNetLoader} & 7 &  &  \\
\bottomrule
\end{tabular}
\end{table}

%% file: study/data/rq3_real_cost.tex
\begin{table}[t]
\caption{Checking and execution time on the graphs of RQ2, by operator family. Checking statistics summarize each graph's median over 25 checks, including diagnostics. Execution time covers submission through completion or a runtime error, with each workflow's own settings and model loading.}
\label{tab:synthetic-cost}
\centering\small
\begin{tabular}{@{}lrrrrr@{}}
\toprule
Graphs & $n$ & Check med.\ (ms) & p95 (ms) & Max (ms) & Exec.\ med.\ (s) \\
\midrule
Original & 132 & 2.51 & 5.89 & 47.33 & 16.0 \\
File swap & 403 & 3.03 & 22.76 & 58.42 & 8.0 \\
Rewiring & 494 & 3.15 & 47.93 & 128.37 & 14.0 \\
Latent graft & 325 & 32.23 & 191.92 & 4892.11 & 26.0 \\
Image graft & 325 & 7.44 & 56.82 & 67.45 & 38.0 \\
Dimension & 100 & 5.10 & 62.86 & 156.96 & 6.0 \\
Modality & 56 & 3.20 & 9.04 & 15.45 & 15.1 \\
Quantization & 2 & 0.90 & 1.46 & 1.46 & 35.3 \\
\midrule
All & 1,837 & 5.46 & 83.72 & 4892.11 & 16.0 \\
\bottomrule
\end{tabular}
\end{table}

%% file: study/data/rq4_cases_table.tex
\begingroup\footnotesize
\begin{longtable}{@{}llp{0.6\linewidth}c@{}}
\caption{The 108 reported failures of RQ4. ComfyUI's validator admits every faulty graph. A check mark means that \tool rejects the faulty graph (and accepts its fixed version).}\label{tab:rq4-cases}\\
\toprule
Report & Mode & Fault & \tool \\
\midrule
\endfirsthead
\toprule
Report & Mode & Fault & \tool \\
\midrule
\endhead
\bottomrule
\endlastfoot
\href{https://github.com/Comfy-Org/ComfyUI/discussions/1098}{Discussion~\#1098} & silent & SDXL latent sampled by an SD1.5 model & \checkmark \\
\href{https://github.com/Comfy-Org/ComfyUI/discussions/2115}{Discussion~\#2115} & crash & LoRA file loaded as a checkpoint & \checkmark \\
\href{https://github.com/Comfy-Org/ComfyUI/discussions/11477}{Discussion~\#11477} & silent & Flux.2 text encoder in a Flux.1 encoder slot & \checkmark \\
\href{https://github.com/Comfy-Org/ComfyUI/issues/65}{Issue~\#65} & crash & ViT-H CLIP vision fed to a ViT-L style adapter & $\times$ \\
\href{https://github.com/Comfy-Org/ComfyUI/issues/424}{Issue~\#424} & crash & Plain checkpoint loaded with the unCLIP loader & $\times$ \\
\href{https://github.com/Comfy-Org/ComfyUI/issues/554}{Issue~\#554} & crash & Plain checkpoint loaded with the unCLIP loader & $\times$ \\
\href{https://github.com/Comfy-Org/ComfyUI/issues/817}{Issue~\#817} & crash & SDXL base text encoder used with the refiner & \checkmark \\
\href{https://github.com/Comfy-Org/ComfyUI/issues/880}{Issue~\#880} & silent & SD1.5 LoRA on an SDXL checkpoint & $\times$ \\
\href{https://github.com/Comfy-Org/ComfyUI/issues/932}{Issue~\#932} & crash & SD1.5 ControlNet on an SDXL checkpoint & \checkmark \\
\href{https://github.com/Comfy-Org/ComfyUI/issues/1064}{Issue~\#1064} & crash & Plain checkpoint loaded with the unCLIP loader & $\times$ \\
\href{https://github.com/Comfy-Org/ComfyUI/issues/1316}{Issue~\#1316} & crash & SDXL base ControlNet with the SDXL refiner & \checkmark \\
\href{https://github.com/Comfy-Org/ComfyUI/issues/1453}{Issue~\#1453} & crash & SD1.5 ControlNet on an SDXL checkpoint & \checkmark \\
\href{https://github.com/Comfy-Org/ComfyUI/issues/1461}{Issue~\#1461} & silent & SD1.5 LoRA on an SDXL checkpoint & \checkmark \\
\href{https://github.com/Comfy-Org/ComfyUI/issues/1732}{Issue~\#1732} & crash & Style adapter loaded as a ControlNet & \checkmark \\
\href{https://github.com/Comfy-Org/ComfyUI/issues/1835}{Issue~\#1835} & crash & SDXL ControlNet on an SD1.5 checkpoint & \checkmark \\
\href{https://github.com/Comfy-Org/ComfyUI/issues/2084}{Issue~\#2084} & crash & SVD checkpoint without text encoder used for prompts & \checkmark \\
\href{https://github.com/Comfy-Org/ComfyUI/issues/2119}{Issue~\#2119} & crash & SVD checkpoint without text encoder used for prompts & \checkmark \\
\href{https://github.com/Comfy-Org/ComfyUI/issues/2224}{Issue~\#2224} & crash & Full checkpoint loaded as a ControlNet & \checkmark \\
\href{https://github.com/Comfy-Org/ComfyUI/issues/2440}{Issue~\#2440} & crash & IP-Adapter file loaded as a ControlNet & \checkmark \\
\href{https://github.com/Comfy-Org/ComfyUI/issues/2535}{Issue~\#2535} & silent & Inpaint gray fill never denoised by a regular model & $\times$ \\
\href{https://github.com/Comfy-Org/ComfyUI/issues/2594}{Issue~\#2594} & crash & Image batches of different lengths blended & $\times$ \\
\href{https://github.com/Comfy-Org/ComfyUI/issues/2639}{Issue~\#2639} & crash & PhotoMaker (SDXL only) with an SD1.5 text encoder & $\times$ \\
\href{https://github.com/Comfy-Org/ComfyUI/issues/2670}{Issue~\#2670} & crash & T2I adapter loaded as a style model & $\times$ \\
\href{https://github.com/Comfy-Org/ComfyUI/issues/2897}{Issue~\#2897} & crash & SD1.5 text encoder in the SDXL text-encode node & \checkmark \\
\href{https://github.com/Comfy-Org/ComfyUI/issues/2922}{Issue~\#2922} & crash & Control-LoRA on a 9-channel SDXL inpaint UNet & \checkmark \\
\href{https://github.com/Comfy-Org/ComfyUI/issues/3022}{Issue~\#3022} & silent & ControlNet with a deep-shrink patched model & $\times$ \\
\href{https://github.com/Comfy-Org/ComfyUI/issues/3210}{Issue~\#3210} & crash & SD1.5 ControlNet on an SDXL checkpoint & \checkmark \\
\href{https://github.com/Comfy-Org/ComfyUI/issues/3372}{Issue~\#3372} & silent & ControlNet with a deep-shrink patched model & $\times$ \\
\href{https://github.com/Comfy-Org/ComfyUI/issues/3690}{Issue~\#3690} & crash & SD3 checkpoint without text encoders used for prompts & \checkmark \\
\href{https://github.com/Comfy-Org/ComfyUI/issues/3714}{Issue~\#3714} & crash & SD1.5 ControlNet on an SDXL checkpoint & \checkmark \\
\href{https://github.com/Comfy-Org/ComfyUI/issues/3726}{Issue~\#3726} & crash & SD3 checkpoint without text encoders used for prompts & \checkmark \\
\href{https://github.com/Comfy-Org/ComfyUI/issues/3950}{Issue~\#3950} & crash & SD3 checkpoint without text encoders used for prompts & \checkmark \\
\href{https://github.com/Comfy-Org/ComfyUI/issues/4049}{Issue~\#4049} & crash & SD1.5 ControlNet on an SDXL checkpoint & \checkmark \\
\href{https://github.com/Comfy-Org/ComfyUI/issues/4250}{Issue~\#4250} & crash & 16-channel and 4-channel latents interpolated & $\times$ \\
\href{https://github.com/Comfy-Org/ComfyUI/issues/4329}{Issue~\#4329} & crash & Flux diffusion-only file loaded as a checkpoint & \checkmark \\
\href{https://github.com/Comfy-Org/ComfyUI/issues/4428}{Issue~\#4428} & crash & LoRA file loaded as a checkpoint & \checkmark \\
\href{https://github.com/Comfy-Org/ComfyUI/issues/4530}{Issue~\#4530} & crash & Flux diffusion-only file loaded as a checkpoint & \checkmark \\
\href{https://github.com/Comfy-Org/ComfyUI/issues/4611}{Issue~\#4611} & silent & Flux text encoders given two T5 files and no CLIP-L & \checkmark \\
\href{https://github.com/Comfy-Org/ComfyUI/issues/4624}{Issue~\#4624} & crash & Flux diffusion-only file loaded as a checkpoint & \checkmark \\
\href{https://github.com/Comfy-Org/ComfyUI/issues/4669}{Issue~\#4669} & crash & InstantX Flux ControlNet applied without a VAE & $\times$ \\
\href{https://github.com/Comfy-Org/ComfyUI/issues/4699}{Issue~\#4699} & crash & Flux text encoders loaded with type \texttt{sdxl} & \checkmark \\
\href{https://github.com/Comfy-Org/ComfyUI/issues/4740}{Issue~\#4740} & crash & 4-channel latent sampled by a 16-channel SD3 model & $\times$ \\
\href{https://github.com/Comfy-Org/ComfyUI/issues/4866}{Issue~\#4866} & crash & Flux diffusion-only file loaded as a checkpoint & \checkmark \\
\href{https://github.com/Comfy-Org/ComfyUI/issues/4875}{Issue~\#4875} & crash & Full checkpoint loaded as a VAE & \checkmark \\
\href{https://github.com/Comfy-Org/ComfyUI/issues/4884}{Issue~\#4884} & crash & Flux VAE loaded as the CLIP-L text encoder & \checkmark \\
\href{https://github.com/Comfy-Org/ComfyUI/issues/4892}{Issue~\#4892} & crash & Flux text-encoder conditioning with an SD1.5 model & \checkmark \\
\href{https://github.com/Comfy-Org/ComfyUI/issues/4900}{Issue~\#4900} & crash & Flux latent decoded with the SDXL VAE & \checkmark \\
\href{https://github.com/Comfy-Org/ComfyUI/issues/4948}{Issue~\#4948} & crash & Flux text encoders loaded with type \texttt{sdxl} & \checkmark \\
\href{https://github.com/Comfy-Org/ComfyUI/issues/5098}{Issue~\#5098} & crash & SD1.5 ControlNet on an SDXL checkpoint & \checkmark \\
\href{https://github.com/Comfy-Org/ComfyUI/issues/5149}{Issue~\#5149} & crash & SD1.5 ControlNet on an SDXL checkpoint & \checkmark \\
\href{https://github.com/Comfy-Org/ComfyUI/issues/5239}{Issue~\#5239} & crash & LoRA file loaded as a checkpoint & \checkmark \\
\href{https://github.com/Comfy-Org/ComfyUI/issues/5359}{Issue~\#5359} & silent & Flux LoRA on an SD1.5 checkpoint & \checkmark \\
\href{https://github.com/Comfy-Org/ComfyUI/issues/5388}{Issue~\#5388} & crash & SD3.5 checkpoint without text encoders used for prompts & \checkmark \\
\href{https://github.com/Comfy-Org/ComfyUI/issues/5492}{Issue~\#5492} & silent & XLabs-format Flux LoRA whose keys are ignored & $\times$ \\
\href{https://github.com/Comfy-Org/ComfyUI/issues/5651}{Issue~\#5651} & crash & SDXL ControlNet on an SD1.5 checkpoint & $\times$ \\
\href{https://github.com/Comfy-Org/ComfyUI/issues/5655}{Issue~\#5655} & silent & SDXL LoRA on a Flux model & \checkmark \\
\href{https://github.com/Comfy-Org/ComfyUI/issues/5783}{Issue~\#5783} & crash & Flux depth control LoRA on the Flux Fill model & \checkmark \\
\href{https://github.com/Comfy-Org/ComfyUI/issues/5820}{Issue~\#5820} & crash & Flux diffusion-only file loaded as a checkpoint & \checkmark \\
\href{https://github.com/Comfy-Org/ComfyUI/issues/5824}{Issue~\#5824} & crash & Flux diffusion-only file loaded as a checkpoint & \checkmark \\
\href{https://github.com/Comfy-Org/ComfyUI/issues/5880}{Issue~\#5880} & crash & PhotoMaker V2 file with the V1 loader & $\times$ \\
\href{https://github.com/Comfy-Org/ComfyUI/issues/6194}{Issue~\#6194} & crash & ControlNet-LLLite file loaded as a ControlNet & \checkmark \\
\href{https://github.com/Comfy-Org/ComfyUI/issues/6241}{Issue~\#6241} & crash & CLIP-L text encoder used for Mochi & \checkmark \\
\href{https://github.com/Comfy-Org/ComfyUI/issues/6380}{Issue~\#6380} & crash & Flux ControlNet on an SD1.5 checkpoint & \checkmark \\
\href{https://github.com/Comfy-Org/ComfyUI/issues/6392}{Issue~\#6392} & crash & SD1.5 checkpoint loaded for SVD image-to-video & $\times$ \\
\href{https://github.com/Comfy-Org/ComfyUI/issues/6419}{Issue~\#6419} & crash & Flux latent decoded with the SD1.5 TAESD VAE & \checkmark \\
\href{https://github.com/Comfy-Org/ComfyUI/issues/6468}{Issue~\#6468} & crash & Flux text encoders loaded with type \texttt{sdxl} & \checkmark \\
\href{https://github.com/Comfy-Org/ComfyUI/issues/6575}{Issue~\#6575} & crash & SDXL ControlNet on an SD1.5 checkpoint & \checkmark \\
\href{https://github.com/Comfy-Org/ComfyUI/issues/6588}{Issue~\#6588} & crash & Flux diffusion-only file loaded as a checkpoint & \checkmark \\
\href{https://github.com/Comfy-Org/ComfyUI/issues/6590}{Issue~\#6590} & crash & SDXL ControlNet on an SD1.5 checkpoint & \checkmark \\
\href{https://github.com/Comfy-Org/ComfyUI/issues/6609}{Issue~\#6609} & crash & NF4-quantized Flux checkpoint with the core loader & $\times$ \\
\href{https://github.com/Comfy-Org/ComfyUI/issues/6774}{Issue~\#6774} & silent & CFG-distilled Flux sampled with classic CFG & $\times$ \\
\href{https://github.com/Comfy-Org/ComfyUI/issues/6813}{Issue~\#6813} & silent & Scaled-fp8 HunyuanVideo weights without their scales & \checkmark \\
\href{https://github.com/Comfy-Org/ComfyUI/issues/6880}{Issue~\#6880} & crash & SD3.5 checkpoint without text encoders used for prompts & \checkmark \\
\href{https://github.com/Comfy-Org/ComfyUI/issues/6893}{Issue~\#6893} & crash & SD3.5 Large ControlNet on SD3.5 Medium & $\times$ \\
\href{https://github.com/Comfy-Org/ComfyUI/issues/6940}{Issue~\#6940} & crash & Flux Redux style model loaded as a diffusion model & \checkmark \\
\href{https://github.com/Comfy-Org/ComfyUI/issues/7022}{Issue~\#7022} & crash & ViT-L CLIP vision used for Wan image-to-video & $\times$ \\
\href{https://github.com/Comfy-Org/ComfyUI/issues/7053}{Issue~\#7053} & crash & Wrapper-format Wan VAE with the core loader & $\times$ \\
\href{https://github.com/Comfy-Org/ComfyUI/issues/7058}{Issue~\#7058} & crash & HunyuanVideo text encoders loaded with type \texttt{sdxl} & \checkmark \\
\href{https://github.com/Comfy-Org/ComfyUI/issues/7080}{Issue~\#7080} & crash & SD3.5 checkpoint without text encoders used for prompts & \checkmark \\
\href{https://github.com/Comfy-Org/ComfyUI/issues/7095}{Issue~\#7095} & crash & Conditioning area outside the latent & $\times$ \\
\href{https://github.com/Comfy-Org/ComfyUI/issues/7114}{Issue~\#7114} & crash & Wrapper-format Wan text encoder with the core loader & $\times$ \\
\href{https://github.com/Comfy-Org/ComfyUI/issues/7240}{Issue~\#7240} & crash & Wan latent with fewer frames than the conditioning & $\times$ \\
\href{https://github.com/Comfy-Org/ComfyUI/issues/7335}{Issue~\#7335} & crash & Control-LoRA on a 9-channel SDXL inpaint UNet & \checkmark \\
\href{https://github.com/Comfy-Org/ComfyUI/issues/7407}{Issue~\#7407} & silent & Flux depth model combined with a depth LoRA & \checkmark \\
\href{https://github.com/Comfy-Org/ComfyUI/issues/7649}{Issue~\#7649} & silent & HiDream with regular instead of long CLIP encoders & $\times$ \\
\href{https://github.com/Comfy-Org/ComfyUI/issues/7715}{Issue~\#7715} & silent & HiDream with the Flux CLIP-L text encoder & $\times$ \\
\href{https://github.com/Comfy-Org/ComfyUI/issues/8362}{Issue~\#8362} & crash & Wan VACE width not divisible by 16 & \checkmark \\
\href{https://github.com/Comfy-Org/ComfyUI/issues/8463}{Issue~\#8463} & crash & SD1.5 ControlNet on an SDXL checkpoint & \checkmark \\
\href{https://github.com/Comfy-Org/ComfyUI/issues/8697}{Issue~\#8697} & crash & Flux Kontext text encoders loaded with type \texttt{sdxl} & \checkmark \\
\href{https://github.com/Comfy-Org/ComfyUI/issues/9091}{Issue~\#9091} & silent & Wan 2.2 14B sampled with the 5B model's frame count & $\times$ \\
\href{https://github.com/Comfy-Org/ComfyUI/issues/9095}{Issue~\#9095} & crash & Flux diffusion-only file loaded as a checkpoint & \checkmark \\
\href{https://github.com/Comfy-Org/ComfyUI/issues/9190}{Issue~\#9190} & silent & Qwen-Image with the \texttt{fp8\_e4m3fn\_fast} dtype & $\times$ \\
\href{https://github.com/Comfy-Org/ComfyUI/issues/9204}{Issue~\#9204} & crash & Qwen-Image with the 3B instead of the 7B text encoder & \checkmark \\
\href{https://github.com/Comfy-Org/ComfyUI/issues/9248}{Issue~\#9248} & crash & Wan text encoder loaded with the wrong loader type & \checkmark \\
\href{https://github.com/Comfy-Org/ComfyUI/issues/9318}{Issue~\#9318} & crash & Wan 2.2 VAE latent sampled by the 14B model & \checkmark \\
\href{https://github.com/Comfy-Org/ComfyUI/issues/9479}{Issue~\#9479} & crash & Wan text encoder loaded with the wrong loader type & \checkmark \\
\href{https://github.com/Comfy-Org/ComfyUI/issues/9596}{Issue~\#9596} & crash & Wan 2.2 14B latent decoded with the 5B model's VAE & \checkmark \\
\href{https://github.com/Comfy-Org/ComfyUI/issues/9877}{Issue~\#9877} & crash & Chroma diffusion-only file loaded as a checkpoint & \checkmark \\
\href{https://github.com/Comfy-Org/ComfyUI/issues/11255}{Issue~\#11255} & silent & Qwen Image Edit with the \texttt{fp8\_e4m3fn\_fast} dtype & $\times$ \\
\href{https://github.com/Comfy-Org/ComfyUI/issues/11262}{Issue~\#11262} & crash & Qwen text encoder with a Wan-based model & \checkmark \\
\href{https://github.com/Comfy-Org/ComfyUI/issues/11690}{Issue~\#11690} & crash & LTX-2 audio VAE file without its config metadata & $\times$ \\
\href{https://github.com/Comfy-Org/ComfyUI/issues/12132}{Issue~\#12132} & crash & Flux.2 Klein with a Z-Image text encoder & \checkmark \\
\href{https://github.com/Comfy-Org/ComfyUI/issues/12161}{Issue~\#12161} & crash & Video and audio latents of different batch sizes & $\times$ \\
\href{https://github.com/Comfy-Org/ComfyUI/issues/12177}{Issue~\#12177} & silent & UNet attention guidance on the Chroma DiT & $\times$ \\
\href{https://github.com/Comfy-Org/ComfyUI/issues/12278}{Issue~\#12278} & crash & ACE-Step LM encoder loaded without its partner & $\times$ \\
\href{https://github.com/Comfy-Org/ComfyUI/issues/13833}{Issue~\#13833} & crash & SD1.5 ControlNet on a Lumina2 checkpoint & \checkmark \\
\href{https://github.com/Comfy-Org/ComfyUI/issues/14154}{Issue~\#14154} & crash & Pixel-space Flux.2 output decoded with the Flux.2 VAE & \checkmark \\
\href{https://github.com/Comfy-Org/ComfyUI/issues/14339}{Issue~\#14339} & crash & 4-channel RGBA images saved as RGB video & $\times$ \\
\end{longtable}
\endgroup

%% file: sections/appendix-user-study.tex
\section{\texorpdfstring{User Study Details}{User Study Details}}
\label{app:user-study}

This appendix provides the full setup and procedure for the RQ5 study
in \S\mbox{\ref{sec:eval:users}}.

\subsection{Study Interface}
\label{app:user-study:interface}

\input{figures/fig-user-study-interface}

Figure~\ref{fig:user-study-interface} shows the web interface during a task. The
left pane shows the workflow graph in the ComfyUI editor, which participants can
inspect but not edit. A tag above each node names its pipeline stage. The
sidebar on the right presents the task information in a fixed order: the
intended behavior of the workflow, facts about each model file, and the outcome
of the run, which is either the output image or the error trace. In the
treatment version, the \tool report follows. The answer form comes last. It
first asks whether the workflow is faulty. A \emph{yes} answer opens the
localization questions, which ask for a region and then a node, followed by the
repair questions (Appendix~\ref{app:user-study:tasks}). Choosing a region
highlights its nodes in the graph, and each node reference (\#$N$) in the report
or the choices centers the graph on that node. The top bar shows the time left
in the task's 8-minute budget, a quick reference to the tutorial, a pause button
that stops the clock, and the participant's progress.

\subsection{Reading the Report}
\label{app:user-study:report}

The report uses a few fixed terms. Its first line gives the verdict. \texttt{OK}
means that the workflow is well-typed. \texttt{ILL-TYPED} means that it is
rejected, and the tag names the kind of conflict: \texttt{brand} for components
that disagree on a model family, control modality, or quantization level,
\texttt{dim} for conflicting tensor dimensions, and \texttt{clash} for
incompatible types or a model-database fact. The bracket gives the share of
nodes that have a schema. The indented lines explain the conflict in words. The
list \emph{Involved (minimal core)} names the nodes and edges whose constraints
conflict, and changing any one of them can resolve the conflict. A parenthesis
after a loader gives the fact that the model database derived for its file and
its source, which is the filename or the file's header (\texttt{headerfacts}).
\emph{Possible fixes}, when present, name a change that restores consistency.
The reports in Appendix~\ref{app:user-study:tasks} reproduce what participants
saw. In this prototype, the two rejections that depend on model-database facts
(T03 and T04) print a generic conflict message and omit the explanation and fix
that the fact carries.

\subsection{Participants}
\label{app:user-study:participants}

We recruited 10 students from a university. All of them use generative-AI tools
regularly. Seven of them report hands-on experience with image or video
generation, including Stable Diffusion. The other three work mainly with
language models, for coding purposes. Before the scored tasks, every participant
completed the same tutorial and guided practice (\S\mbox{\ref{sec:eval:users}}),
which introduce ComfyUI workflows and a debugging procedure. Participants
received a gift card for their time.

\subsection{\texorpdfstring{Task Definitions}{Task Definitions}}
\label{app:user-study:tasks}

\input{study/data/user_study_tasks}

\subsection{\texorpdfstring{Detailed Study Results}{Detailed Study Results}}
\label{app:user-study:results}

Table~\ref{tab:user-study} gives the scored outcome of every trial
summarized in Figure~\ref{fig:user-study}.

\input{study/data/user_study_results}

\input{study/data/user_study_details}

%% file: figures/fig-user-study-interface.tex
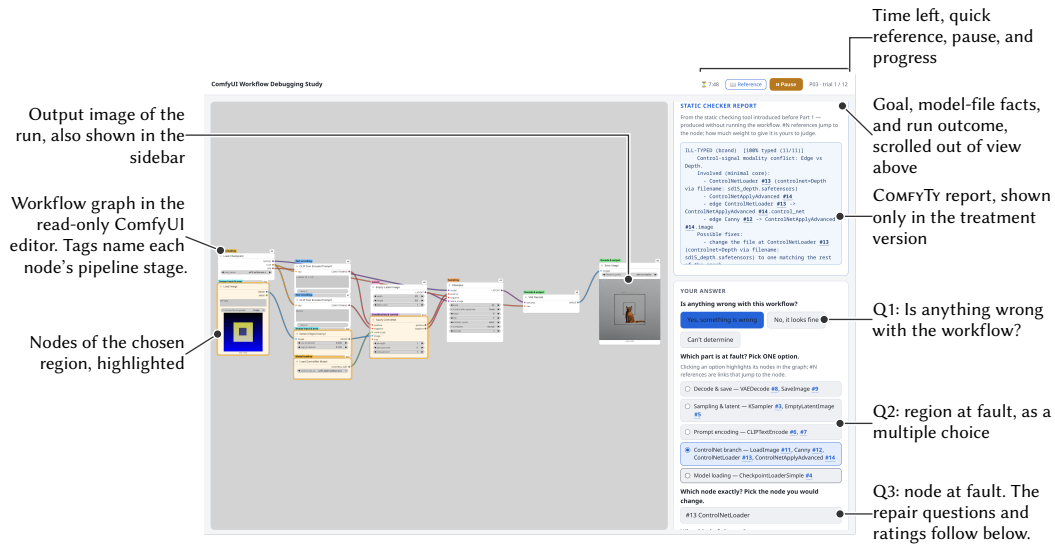
\begin{figure}[!htb]
\centering
\begin{tikzpicture}[x=0.005375cm, y=-0.005375cm,
    lbl/.style={font=\sffamily\scriptsize, text width=2.4cm,
      inner sep=0pt, align=flush left},
    lead/.style={line width=0.5pt, black!80, preaction={draw=white,
      line width=1.6pt}},
    dot/.style={circle, fill=black!80, draw=white, line width=0.5pt,
      inner sep=0pt, minimum size=3.2pt}]
  \node[anchor=north west, inner sep=0pt, draw=black!35, line width=0.4pt]
    at (0,0) {\includegraphics[width=8.6cm]{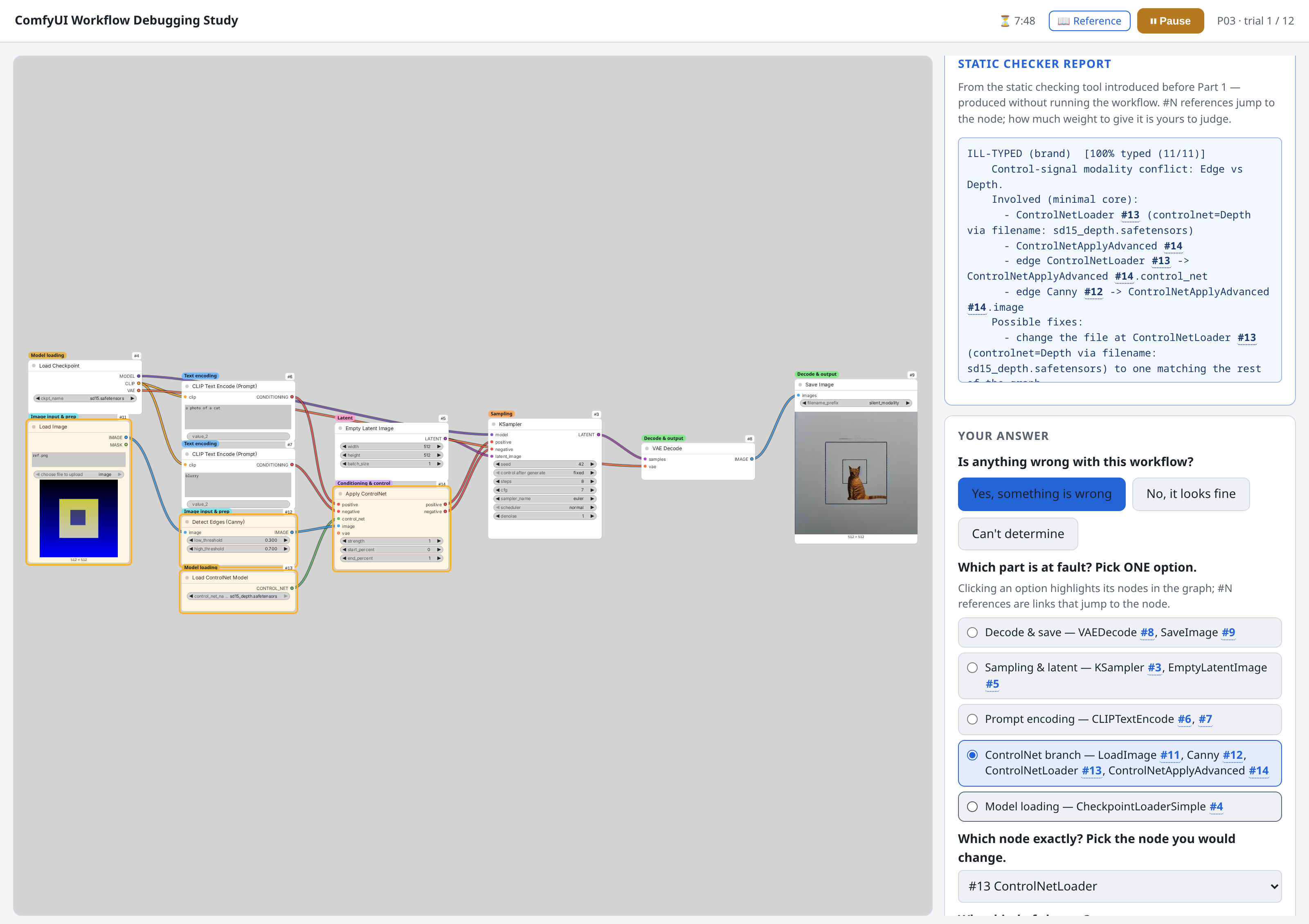}};
  \node[lbl, anchor=east, align=flush right] (la) at (-46,150)
    {Output image of the run, also shown in the sidebar};
  \draw[lead] (la.east) -- (1045,150) -- (1045,505);
  \node[dot] at (1045,505) {};
  \node[lbl, anchor=east, align=flush right] (lb) at (-46,400)
    {Workflow graph in the read-only ComfyUI editor. Tags name each node's
     pipeline stage.};
  \draw[lead] (lb.east) -- (40,434);
  \node[dot] at (40,434) {};
  \node[lbl, anchor=east, align=flush right] (lc) at (-46,690)
    {Nodes of the chosen region, highlighted};
  \draw[lead] (lc.east) -- (33,640);
  \node[dot] at (33,640) {};
  \draw[black!80, line width=0.5pt] (1222,-6) -- (1222,-14) -- (1590,-14)
    -- (1590,-6);
  \node[lbl, anchor=west] (r1) at (1646,-90)
    {Time left, quick reference, pause, and progress};
  \draw[lead] (r1.west) -- (1590,-90) -- (1590,-14);
  \node[lbl, anchor=west] (r2) at (1646,150)
    {Goal, model-file facts, and run outcome, scrolled out of view above};
  \draw[lead] (r2.west) -- (1572,78);
  \node[dot] at (1572,78) {};
  \node[lbl, anchor=west] (r3) at (1646,350)
    {\tool report, shown only in the treatment version};
  \draw[lead] (r3.west) -- (1566,350);
  \node[dot] at (1566,350) {};
  \node[lbl, anchor=west] (r4) at (1646,604)
    {Q1: Is anything wrong with the workflow?};
  \draw[lead] (r4.west) -- (1527,604);
  \node[dot] at (1527,604) {};
  \node[lbl, anchor=west] (r5) at (1646,860)
    {Q2: region at fault, as a multiple choice};
  \draw[lead] (r5.west) -- (1566,860);
  \node[dot] at (1566,860) {};
  \node[lbl, anchor=west] (r6) at (1646,1083)
    {Q3: node at fault. The repair questions and ratings follow below.};
  \draw[lead] (r6.west) -- (1566,1083);
  \node[dot] at (1566,1083) {};
\end{tikzpicture}
\caption{The study interface during a treatment task. The participant has
judged the workflow faulty and chosen a region and a node.}
\Description{A screenshot of the study web page. The left pane shows a
ComfyUI workflow graph with four ControlNet nodes highlighted. The right
sidebar shows the ComfyTy report and the answer form with the first three
questions answered. Labels on both sides explain each part.}
\label{fig:user-study-interface}
\end{figure}

%% file: study/data/user_study_tasks.tex
\begingroup
\input{sections/user-study-format}
Task numbers T01--T12 match the results table. Each task entry provides its source or construction, workflow, questions, answer key, and diagnostic report. The notation \#N identifies node N in the workflow graph. Region options use labels R1, R2, and so on. Repair categories use A--E, and labels such as A1 identify individual options within a category. These labels support cross-referencing within each entry; the study interface shuffles option order.\par
Baseline tasks show the workflow, intended behavior, model-file facts, and observed output or error trace; treatment tasks additionally show ComfyTy's report. Bold Accepted labels identify reference answers for readers. Correct diagnosis requires the accepted fault decision, region, and node within the recorded task time limit; for a valid workflow, the correct answer is \emph{No, it looks fine}. Repair success additionally requires an accepted repair category and specific repair. A timeout or give-up response counts as unsuccessful.\par

\ustask{T01\quad\texttt{s-cn-modality}}
\usfield{Task class}{Silent compatibility fault}
\usfield{Source and construction}{Controlled modality-mismatch example, validated by GPU execution. An SD1.5 text-to-image workflow pairs canny-edge conditioning with a depth ControlNet.}
\usfield{Upstream basis}{\href{https://github.com/comfyanonymous/ComfyUI_examples/tree/master/controlnet}{The official ComfyUI ControlNet examples} specify that the conditioning image format must match the ControlNet model, including depth and canny maps. This task instantiates that requirement with canny maps and a depth model.}
\usfield{Upstream basis}{\href{https://github.com/Comfy-Org/ComfyUI/issues/5070\#issuecomment-2378769137}{A discussion in ComfyUI issue \#5070} describes a related guidance-mode mismatch: a workflow produces an image while its ControlNet uses a different mode from the intended depth conditioning.}
\usfield{Task description}{Generate a cat photo whose composition follows the edge map extracted from a reference image (ControlNet).}
\usfield{Reference diagnosis}{A *depth* ControlNet (sd15\_depth) is fed *canny edge* maps; guidance silently collapses.}

\ussection{Answer key}
\usfield{Q1. Decision}{Yes, something is wrong}
\usfield{Q2. Region}{R2: ControlNet branch --- LoadImage \#11, Canny \#12, ControlNetLoader \#13, ControlNetApplyAdvanced \#14}
\usfield{Q3. Node}{\#12 Canny; \#13 ControlNetLoader}
\usfield{Q4. Category}{A: Swap the model/file it loads}
\usfield{Q5. Specific repair}{A1: change \#13 sd15\_depth.safetensors to an SD1.5 *edge (canny)* ControlNet}

\ussection{Workflow and observed behavior}
\begin{center}
\includegraphics[width=\linewidth,height=0.38\textheight,keepaspectratio]{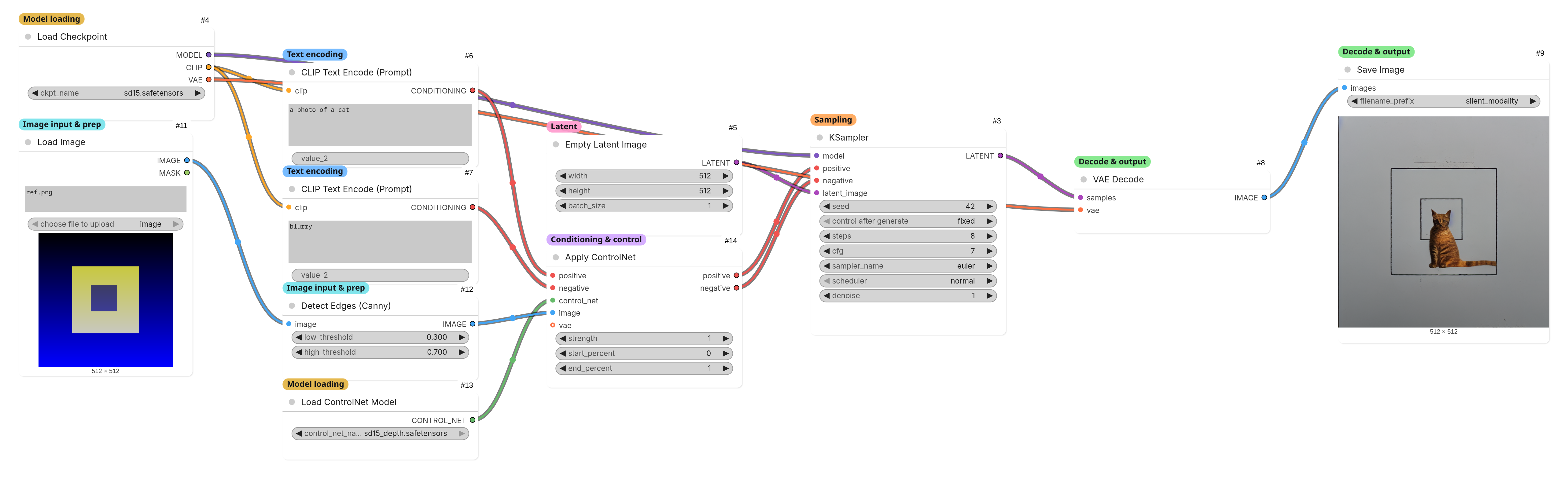}
\end{center}
\usfield{Graph}{11 nodes; 14 links. Node numbers connect the graph to the choices and diagnostic report.}
\usfield{Observed output}{The run completes and produces the following image. Participants inspect this output.}
\begin{center}
\includegraphics[width=0.25\linewidth,height=0.13\textheight,keepaspectratio]{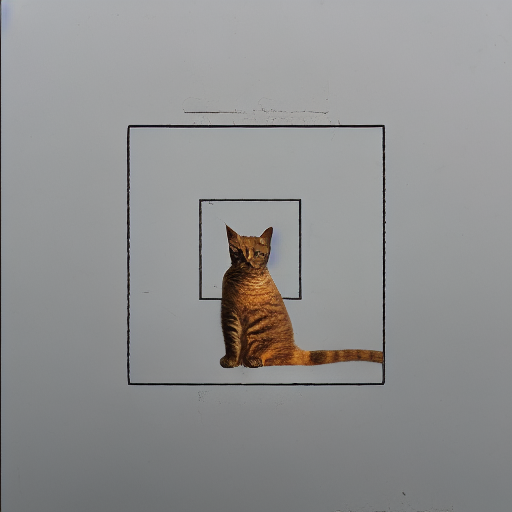}
\end{center}

\ussection{Model-file facts (shown in both versions)}
\usfield{Node \#4}{\textbf{sd15.safetensors}\newline Stable Diffusion 1.5. The classic text-to-image checkpoint; one file holding the UNet, the CLIP-L text encoder and the VAE. Trained at 512x512, and the base that most community 1.5 finetunes and LoRAs build on.}
\usfield{Node \#13}{\textbf{sd15\_depth.safetensors}\newline ControlNet 1.1, depth conditioning. Trained against Stable Diffusion 1.5 on roughly 3M depth-map and caption pairs. The control image it expects is a depth map, such as the output of the MiDaS or Zoe estimators.}

\ussection{Q1. Fault decision}
\usfield{Question}{Is anything wrong with this workflow?}
\begin{itemize}
\item[$\bullet$] Yes, something is wrong\quad\textbf{(Accepted)}
\item[$\bullet$] No, it looks fine
\item[$\bullet$] Can't determine
\end{itemize}
\noindent A Yes answer opens Q2--Q5. The other answers lead to the participant ratings.\par\smallskip

\ussection{Q2--Q3. Localization choices}
\usfield{Q2. Region}{Which part is at fault? Pick ONE option.}
\usfield{Q3. Node}{Which node exactly? Pick the node you would change.}
\noindent Each region below lists its Q2 option followed by the complete Q3 node menu. Select one node within the selected region.\par\smallskip
\begin{enumerate}
\item[\texttt{R1}] Model loading --- CheckpointLoaderSimple \#4
\begin{itemize}
\item \#4 CheckpointLoaderSimple
\item I don't know
\end{itemize}
\item[\texttt{R2}] ControlNet branch --- LoadImage \#11, Canny \#12, ControlNetLoader \#13, ControlNetApplyAdvanced \#14\quad\textbf{(Accepted)}
\begin{itemize}
\item \#11 LoadImage
\item \#12 Canny
\item \#13 ControlNetLoader
\item \#14 ControlNetApplyAdvanced
\item I don't know
\end{itemize}
\item[\texttt{R3}] Prompt encoding --- CLIPTextEncode \#6, \#7
\begin{itemize}
\item \#6 CLIPTextEncode
\item \#7 CLIPTextEncode
\item I don't know
\end{itemize}
\item[\texttt{R4}] Sampling \& latent --- KSampler \#3, EmptyLatentImage \#5
\begin{itemize}
\item \#3 KSampler
\item \#5 EmptyLatentImage
\item I don't know
\end{itemize}
\item[\texttt{R5}] Decode \& save --- VAEDecode \#8, SaveImage \#9
\begin{itemize}
\item \#8 VAEDecode
\item \#9 SaveImage
\item I don't know
\end{itemize}
\end{enumerate}

\ussection{Q4--Q5. Repair choices}
\usfield{Q4. Question}{What kind of change?}
\begin{enumerate}
\item[\texttt{A}] Swap the model/file it loads\quad\textbf{(Accepted)}
\item[\texttt{B}] Change a setting/widget value
\item[\texttt{C}] Rewire connections
\item[\texttt{D}] Replace with a different node
\item[\texttt{E}] Other
\item[$\bullet$] I don't know
\end{enumerate}
\noindent Q5 displays the menu for the selected category A--D. Each menu also offers I don't know, which counts as an unsuccessful repair selection. Category E (Other) requires a free-text explanation. Selecting I don't know for Q4 proceeds directly to the participant ratings.\par\smallskip

\ussection{A. Swap the model/file it loads}
\usfield{Question}{Change the file to what?}
\begin{enumerate}
\item[\texttt{A1}] change \#13 sd15\_depth.safetensors to an SD1.5 *edge (canny)* ControlNet\quad\textbf{(Accepted)}
\item[\texttt{A2}] change \#13 sd15\_depth.safetensors to an SDXL depth ControlNet
\item[\texttt{A3}] change \#13 sd15\_depth.safetensors to an SD1.5 *pose* ControlNet
\item[\texttt{A4}] change \#4 sd15.safetensors to an SDXL checkpoint
\item[\texttt{A5}] change \#4 sd15.safetensors to a depth-specialized finetune
\item[$\bullet$] I don't know
\end{enumerate}

\ussection{B. Change a setting/widget value}
\usfield{Question}{Which setting --- and to what?}
\begin{enumerate}
\item[\texttt{B1}] raise `strength` on ControlNetApplyAdvanced \#14 (now 1.0)
\item[\texttt{B2}] extend `end\_percent` on \#14 (now 1.0)
\item[\texttt{B3}] tune Canny \#12 thresholds (0.3 / 0.7)
\item[\texttt{B4}] increase sampler `steps`
\item[$\bullet$] I don't know
\end{enumerate}

\ussection{C. Rewire connections}
\usfield{Question}{Which connection would you change?}
\begin{enumerate}
\item[\texttt{C1}] connect LoadImage \#11 directly into \#14 `image` (skip Canny \#12)
\item[\texttt{C2}] swap the positive/negative conditioning inputs of \#14
\item[\texttt{C3}] apply the ControlNet after sampling instead of before
\item[$\bullet$] I don't know
\end{enumerate}

\ussection{D. Replace with a different node}
\usfield{Question}{Replace it with what?}
\begin{enumerate}
\item[\texttt{D1}] replace Canny \#12 with a depth-map preprocessor (match the depth ControlNet)
\item[\texttt{D2}] replace \#14 with the basic ControlNetApply node
\item[\texttt{D3}] replace ControlNetLoader \#13 with a T2I-Adapter loader
\item[$\bullet$] I don't know
\end{enumerate}

\ussection{Q6--Q8. Participant ratings}
\usfield{Q6. Confidence}{"How confident are you?" Options: 1, 2, 3, 4, 5 (1: not at all; 5: very).}
\usfield{Q7. Report helpfulness}{"How helpful was the checker report on this task?" Options: 1, 2, 3, 4, 5 (1: not helpful; 5: very helpful). Appears in the treatment version.}
\usfield{Q8. Mental effort}{"How much mental effort did this task take?" Options: 1, 2, 3, 4, 5, 6, 7, 8, 9 (1: very, very low; 9: very, very high).}
\noindent These ratings record participants' subjective judgments and remain separate from accuracy scoring. The rating fields appear for every Q1 answer.\par\smallskip
\usfield{Optional comment}{"Anything else about your answer? (optional)" Participants may add a free-text comment.}

\ussection{ComfyTy diagnosis (treatment only)}
\begin{lstlisting}[style=usreport]
ILL-TYPED (brand)  [100% typed (11/11)]
    Control-signal modality conflict: Edge vs Depth.
    Involved (minimal core):
      - ControlNetLoader #13 (controlnet=Depth via filename: sd15_depth.safetensors)
      - ControlNetApplyAdvanced #14
      - edge ControlNetLoader #13 -> ControlNetApplyAdvanced #14.control_net
      - edge Canny #12 -> ControlNetApplyAdvanced #14.image
    Possible fixes:
      - change the file at ControlNetLoader #13 (controlnet=Depth via filename: sd15_depth.safetensors) to one matching the rest of the graph
\end{lstlisting}

\ustask{T02\quad\texttt{s-latent-blend}}
\usfield{Task class}{Silent compatibility fault}
\usfield{Source and construction}{Controlled resolution-mismatch example, validated by GPU execution. Two EmptyLatentImage nodes use resolution settings of 512x512 and 768x768; their sampled outputs feed a LatentBlend node.}
\usfield{Upstream basis}{\href{https://github.com/Comfy-Org/ComfyUI/pull/1058}{ComfyUI pull request \#1058} introduces LatentBlend and explicitly discusses different-sized inputs; its implementation resizes one input before blending. The 512x512/768x768 mutation and its observed output are specific to this study.}
\usfield{Task description}{Generate a detailed photo of a golden retriever dog in a park (text-to-image).}
\usfield{Reference diagnosis}{LatentBlend \#16 blends a 512 and a 768 latent (EmptyLatentImage \#11 vs \#12); ComfyUI silently rescales one branch and the blend ghosts.}

\ussection{Answer key}
\usfield{Q1. Decision}{Yes, something is wrong}
\usfield{Q2. Region}{R2: Latent sources \& blend --- EmptyLatentImage \#11, \#12, LatentBlend \#16}
\usfield{Q3. Node}{\#11 EmptyLatentImage; \#12 EmptyLatentImage; \#16 LatentBlend}
\usfield{Q4. Category}{B: Change a setting/widget value}
\usfield{Q5. Specific repair}{B1: set EmptyLatentImage \#11 and \#12 to the SAME width/height (now 512 vs 768)}

\ussection{Workflow and observed behavior}
\begin{center}
\includegraphics[width=\linewidth,height=0.38\textheight,keepaspectratio]{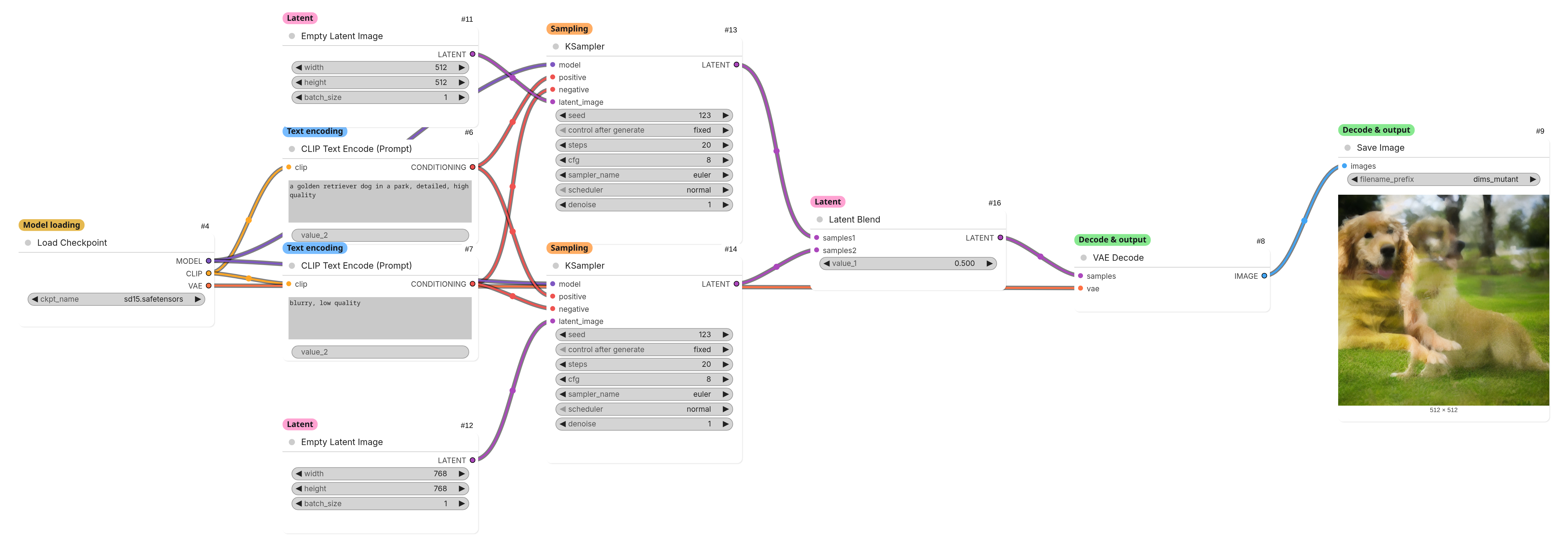}
\end{center}
\usfield{Graph}{10 nodes; 15 links. Node numbers connect the graph to the choices and diagnostic report.}
\usfield{Observed output}{The run completes and produces the following image. Participants inspect this output.}
\begin{center}
\includegraphics[width=0.25\linewidth,height=0.13\textheight,keepaspectratio]{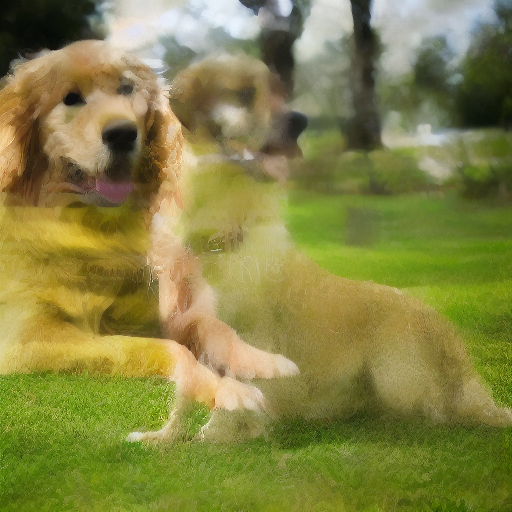}
\end{center}

\ussection{Model-file facts (shown in both versions)}
\usfield{Node \#4}{\textbf{sd15.safetensors}\newline Stable Diffusion 1.5. The classic text-to-image checkpoint; one file holding the UNet, the CLIP-L text encoder and the VAE. Trained at 512x512, and the base that most community 1.5 finetunes and LoRAs build on.}

\ussection{Q1. Fault decision}
\usfield{Question}{Is anything wrong with this workflow?}
\begin{itemize}
\item[$\bullet$] Yes, something is wrong\quad\textbf{(Accepted)}
\item[$\bullet$] No, it looks fine
\item[$\bullet$] Can't determine
\end{itemize}
\noindent A Yes answer opens Q2--Q5. The other answers lead to the participant ratings.\par\smallskip

\ussection{Q2--Q3. Localization choices}
\usfield{Q2. Region}{Which part is at fault? Pick ONE option.}
\usfield{Q3. Node}{Which node exactly? Pick the node you would change.}
\noindent Each region below lists its Q2 option followed by the complete Q3 node menu. Select one node within the selected region.\par\smallskip
\begin{enumerate}
\item[\texttt{R1}] Model loading --- CheckpointLoaderSimple \#4
\begin{itemize}
\item \#4 CheckpointLoaderSimple
\item I don't know
\end{itemize}
\item[\texttt{R2}] Latent sources \& blend --- EmptyLatentImage \#11, \#12, LatentBlend \#16\quad\textbf{(Accepted)}
\begin{itemize}
\item \#11 EmptyLatentImage
\item \#12 EmptyLatentImage
\item \#16 LatentBlend
\item I don't know
\end{itemize}
\item[\texttt{R3}] The two samplers --- KSampler \#13, \#14
\begin{itemize}
\item \#13 KSampler
\item \#14 KSampler
\item I don't know
\end{itemize}
\item[\texttt{R4}] Prompt encoding --- CLIPTextEncode \#6, \#7
\begin{itemize}
\item \#6 CLIPTextEncode
\item \#7 CLIPTextEncode
\item I don't know
\end{itemize}
\item[\texttt{R5}] Decode \& save --- VAEDecode \#8, SaveImage \#9
\begin{itemize}
\item \#8 VAEDecode
\item \#9 SaveImage
\item I don't know
\end{itemize}
\end{enumerate}

\ussection{Q4--Q5. Repair choices}
\usfield{Q4. Question}{What kind of change?}
\begin{enumerate}
\item[\texttt{A}] Swap the model/file it loads
\item[\texttt{B}] Change a setting/widget value\quad\textbf{(Accepted)}
\item[\texttt{C}] Rewire connections
\item[\texttt{D}] Replace with a different node
\item[\texttt{E}] Other
\item[$\bullet$] I don't know
\end{enumerate}
\noindent Q5 displays the menu for the selected category A--D. Each menu also offers I don't know, which counts as an unsuccessful repair selection. Category E (Other) requires a free-text explanation. Selecting I don't know for Q4 proceeds directly to the participant ratings.\par\smallskip

\ussection{A. Swap the model/file it loads}
\usfield{Question}{Change the file to what?}
\begin{enumerate}
\item[\texttt{A1}] change \#4 sd15.safetensors to an SDXL checkpoint (1024-native)
\item[\texttt{A2}] change \#4 sd15.safetensors to a high-res SD1.5 finetune
\item[\texttt{A3}] change \#4 sd15.safetensors to an SD2.x (768) checkpoint
\item[$\bullet$] I don't know
\end{enumerate}

\ussection{B. Change a setting/widget value}
\usfield{Question}{Which setting --- and to what?}
\begin{enumerate}
\item[\texttt{B1}] set EmptyLatentImage \#11 and \#12 to the SAME width/height (now 512 vs 768)\quad\textbf{(Accepted)}
\item[\texttt{B2}] change the blend factor on LatentBlend \#16 (now 0.5)
\item[\texttt{B3}] increase `steps` on KSampler \#13/\#14 (now 20)
\item[\texttt{B4}] lower `cfg` on both samplers (now 7.5)
\item[\texttt{B5}] use different seeds on the two samplers (now both 123)
\item[$\bullet$] I don't know
\end{enumerate}

\ussection{C. Rewire connections}
\usfield{Question}{Which connection would you change?}
\begin{enumerate}
\item[\texttt{C1}] swap `samples1`/`samples2` into LatentBlend \#16
\item[\texttt{C2}] feed \#16 from only one sampler (drop the other branch)
\item[\texttt{C3}] blend the empty latents BEFORE sampling instead of after
\item[$\bullet$] I don't know
\end{enumerate}

\ussection{D. Replace with a different node}
\usfield{Question}{Replace it with what?}
\begin{enumerate}
\item[\texttt{D1}] replace LatentBlend \#16 with LatentComposite
\item[\texttt{D2}] replace the EmptyLatentImages with a VAEEncode of a reference image
\item[\texttt{D3}] decode both branches and average the pixel images instead
\item[$\bullet$] I don't know
\end{enumerate}

\ussection{Q6--Q8. Participant ratings}
\usfield{Q6. Confidence}{"How confident are you?" Options: 1, 2, 3, 4, 5 (1: not at all; 5: very).}
\usfield{Q7. Report helpfulness}{"How helpful was the checker report on this task?" Options: 1, 2, 3, 4, 5 (1: not helpful; 5: very helpful). Appears in the treatment version.}
\usfield{Q8. Mental effort}{"How much mental effort did this task take?" Options: 1, 2, 3, 4, 5, 6, 7, 8, 9 (1: very, very low; 9: very, very high).}
\noindent These ratings record participants' subjective judgments and remain separate from accuracy scoring. The rating fields appear for every Q1 answer.\par\smallskip
\usfield{Optional comment}{"Anything else about your answer? (optional)" Participants may add a free-text comment.}

\ussection{ComfyTy diagnosis (treatment only)}
\begin{lstlisting}[style=usreport]
ILL-TYPED (dim)  [100% typed (10/10)]
    Dimension/channel conflict: constraints are unsatisfiable.
    Involved (minimal core):
      - widget height of EmptyLatentImage #11
      - widget height of EmptyLatentImage #12
      - edge EmptyLatentImage #11 -> KSampler #13.latent_image
      - edge EmptyLatentImage #12 -> KSampler #14.latent_image
      - edge KSampler #13 -> LatentBlend #16.samples1
      - edge KSampler #14 -> LatentBlend #16.samples2
\end{lstlisting}

\ustask{T03\quad\texttt{s-6813-fp8noise}}
\usfield{Task class}{Silent compatibility fault}
\usfield{Source}{Adapted from \href{https://github.com/Comfy-Org/ComfyUI/issues/6813}{ComfyUI issue \#6813}.}
\usfield{Task description}{Text-to-video with HunyuanVideo (official fp8 release). The produced frames look like the image shown.}
\usfield{Reference diagnosis}{The scaled-fp8 checkpoint is loaded without its scaled-fp8 weight\_dtype (UNETLoader \#12); sampling silently produces pure noise.}

\ussection{Answer key}
\usfield{Q1. Decision}{Yes, something is wrong}
\usfield{Q2. Region}{R1: Diffusion model loading --- UNETLoader \#12, ModelSamplingSD3 \#67}
\usfield{Q3. Node}{\#12 UNETLoader}
\usfield{Q4. Category}{B: Change a setting/widget value}
\usfield{Q5. Specific repair}{B1: set `weight\_dtype` on UNETLoader \#12 to the scaled-fp8 mode (now 'default')}

\ussection{Workflow and observed behavior}
\begin{center}
\includegraphics[width=\linewidth,height=0.38\textheight,keepaspectratio]{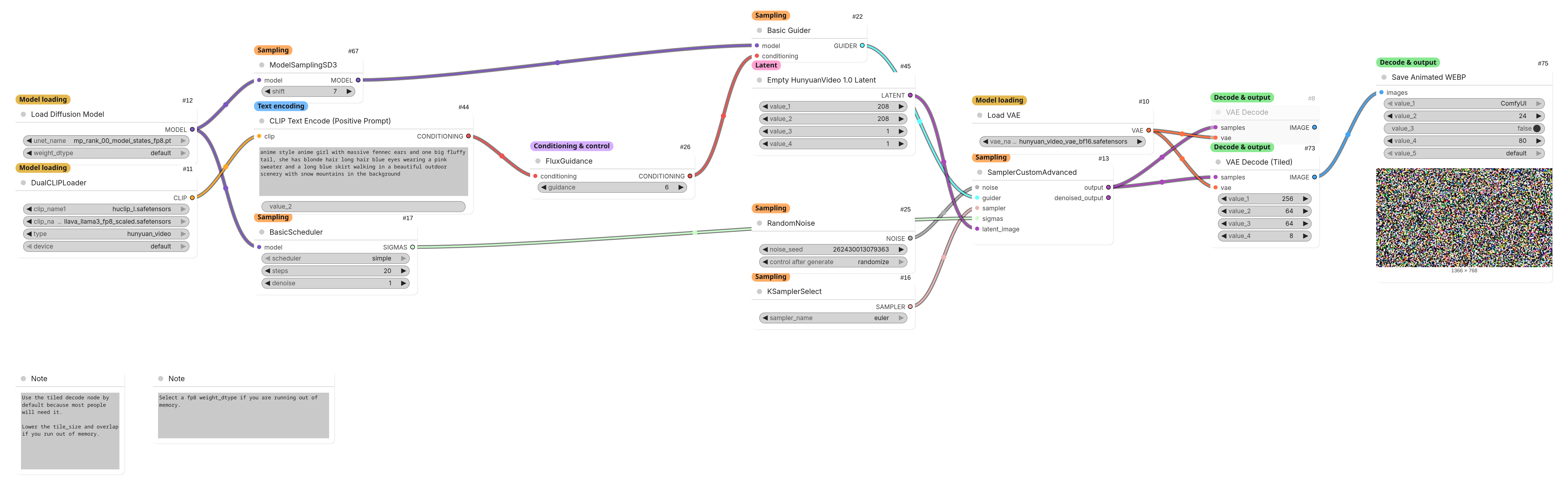}
\end{center}
\usfield{Graph}{17 nodes; 16 links. Node numbers connect the graph to the choices and diagnostic report.}
\usfield{Observed output}{The run completes and produces the following image. Participants inspect this output.}
\begin{center}
\includegraphics[width=0.25\linewidth,height=0.13\textheight,keepaspectratio]{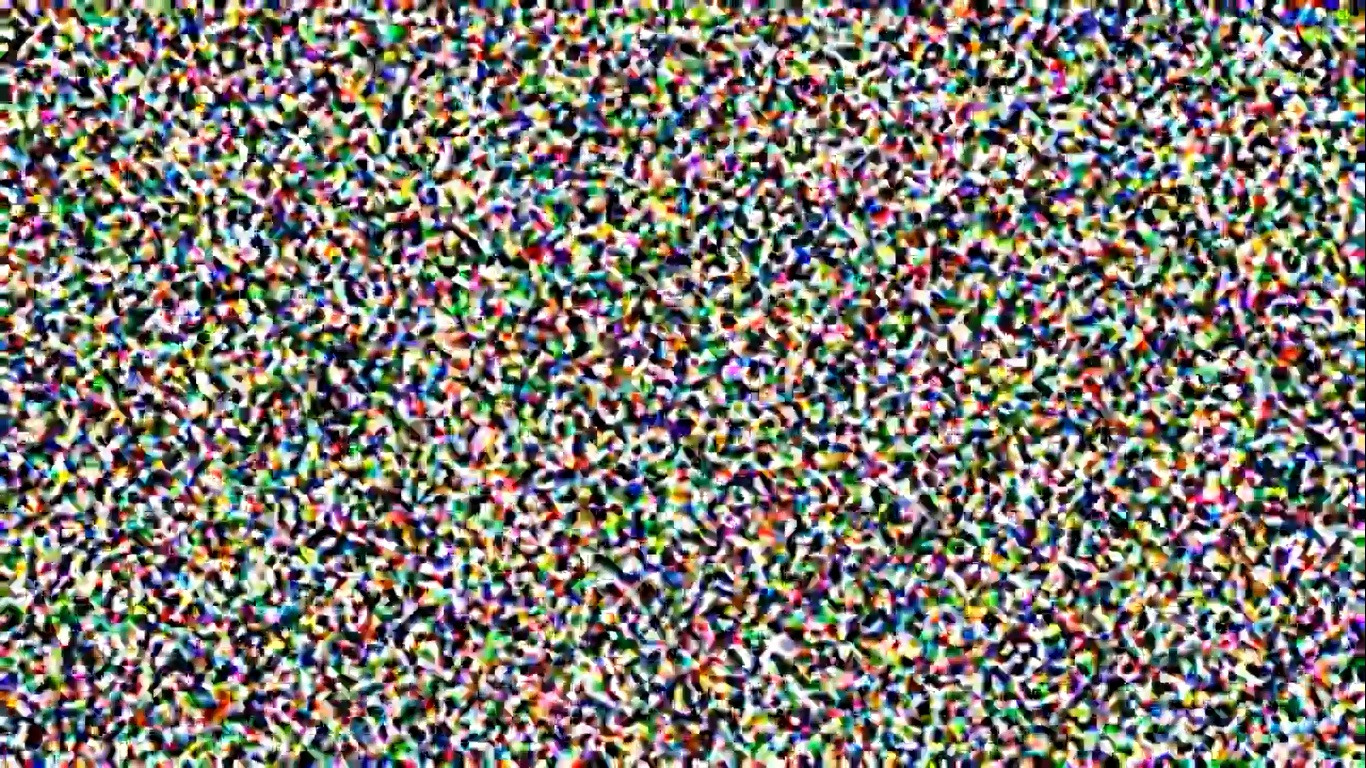}
\end{center}

\ussection{Model-file facts (shown in both versions)}
\usfield{Node \#12}{\textbf{mp\_rank\_00\_model\_states\_fp8.pt}\newline HunyuanVideo transformer weights, official fp8 release. The tensors are stored in fp8 together with per-tensor scale factors, a layout usually called scaled fp8, which halves VRAM against the bf16 release. 13B parameters, up to 720p at 129 frames.}
\usfield{Node \#11}{\textbf{huclip\_l.safetensors}\newline CLIP-L/14 text encoder as packaged for HunyuanVideo. One of the two text towers this model expects; the other is a LLaVA-Llama-3 encoder. Same architecture as the standard CLIP-L release.}
\usfield{Node \#11}{\textbf{llava\_llama3\_fp8\_scaled.safetensors}\newline LLaVA-Llama-3 8B text encoder, repackaged for HunyuanVideo. Stored in fp8 with per-tensor scale factors so it fits alongside the transformer on consumer cards.}
\usfield{Node \#10}{\textbf{hunyuan\_video\_vae\_bf16.safetensors}\newline HunyuanVideo autoencoder, bf16 build. Compresses video four times temporally and eight times spatially into its latent space. Published alongside the transformer weights and loaded with the standard VAE loader.}

\ussection{Q1. Fault decision}
\usfield{Question}{Is anything wrong with this workflow?}
\begin{itemize}
\item[$\bullet$] Yes, something is wrong\quad\textbf{(Accepted)}
\item[$\bullet$] No, it looks fine
\item[$\bullet$] Can't determine
\end{itemize}
\noindent A Yes answer opens Q2--Q5. The other answers lead to the participant ratings.\par\smallskip

\ussection{Q2--Q3. Localization choices}
\usfield{Q2. Region}{Which part is at fault? Pick ONE option.}
\usfield{Q3. Node}{Which node exactly? Pick the node you would change.}
\noindent Each region below lists its Q2 option followed by the complete Q3 node menu. Select one node within the selected region.\par\smallskip
\begin{enumerate}
\item[\texttt{R1}] Diffusion model loading --- UNETLoader \#12, ModelSamplingSD3 \#67\quad\textbf{(Accepted)}
\begin{itemize}
\item \#12 UNETLoader
\item \#67 ModelSamplingSD3
\item I don't know
\end{itemize}
\item[\texttt{R2}] Text encoders \& guidance --- DualCLIPLoader \#11, CLIPTextEncode \#44, FluxGuidance \#26
\begin{itemize}
\item \#11 DualCLIPLoader
\item \#44 CLIPTextEncode
\item \#26 FluxGuidance
\item I don't know
\end{itemize}
\item[\texttt{R3}] Sampler stack --- SamplerCustomAdvanced \#13, KSamplerSelect \#16, BasicScheduler \#17, BasicGuider \#22, RandomNoise \#25
\begin{itemize}
\item \#13 SamplerCustomAdvanced
\item \#16 KSamplerSelect
\item \#17 BasicScheduler
\item \#22 BasicGuider
\item \#25 RandomNoise
\item I don't know
\end{itemize}
\item[\texttt{R4}] Latent source --- EmptyHunyuanLatentVideo \#45
\begin{itemize}
\item \#45 EmptyHunyuanLatentVideo
\item I don't know
\end{itemize}
\item[\texttt{R5}] VAE \& output --- VAELoader \#10, VAEDecode \#8, VAEDecodeTiled \#73, SaveAnimatedWEBP \#75
\begin{itemize}
\item \#10 VAELoader
\item \#8 VAEDecode
\item \#73 VAEDecodeTiled
\item \#75 SaveAnimatedWEBP
\item I don't know
\end{itemize}
\end{enumerate}

\ussection{Q4--Q5. Repair choices}
\usfield{Q4. Question}{What kind of change?}
\begin{enumerate}
\item[\texttt{A}] Swap the model/file it loads
\item[\texttt{B}] Change a setting/widget value\quad\textbf{(Accepted)}
\item[\texttt{C}] Rewire connections
\item[\texttt{D}] Replace with a different node
\item[\texttt{E}] Other
\item[$\bullet$] I don't know
\end{enumerate}
\noindent Q5 displays the menu for the selected category A--D. Each menu also offers I don't know, which counts as an unsuccessful repair selection. Category E (Other) requires a free-text explanation. Selecting I don't know for Q4 proceeds directly to the participant ratings.\par\smallskip

\ussection{A. Swap the model/file it loads}
\usfield{Question}{Change the file to what?}
\begin{enumerate}
\item[\texttt{A1}] change \#10 hunyuan\_video\_vae\_bf16 to an fp16 VAE build
\item[\texttt{A2}] change \#11's llava\_llama3\_fp8\_scaled to a non-fp8 build
\item[\texttt{A3}] change \#12 mp\_rank\_00\_model\_states\_fp8.pt to a GGUF build
\item[\texttt{A4}] change \#12 to a Stable-Diffusion-family video model
\item[$\bullet$] I don't know
\end{enumerate}

\ussection{B. Change a setting/widget value}
\usfield{Question}{Which setting --- and to what?}
\begin{enumerate}
\item[\texttt{B1}] set `weight\_dtype` on UNETLoader \#12 to the scaled-fp8 mode (now 'default')\quad\textbf{(Accepted)}
\item[\texttt{B2}] change `shift` on ModelSamplingSD3 \#67 (now 7)
\item[\texttt{B3}] change the `type` on DualCLIPLoader \#11 (now hunyuan\_video)
\item[\texttt{B4}] increase sampling steps
\item[\texttt{B5}] lower the guidance/cfg
\item[$\bullet$] I don't know
\end{enumerate}

\ussection{C. Rewire connections}
\usfield{Question}{Which connection would you change?}
\begin{enumerate}
\item[\texttt{C1}] bypass ModelSamplingSD3 \#67 (wire \#12 straight to the sampler)
\item[\texttt{C2}] decode with VAEDecode \#8 instead of VAEDecodeTiled \#73
\item[\texttt{C3}] feed only one text encoder into the guider
\item[$\bullet$] I don't know
\end{enumerate}

\ussection{D. Replace with a different node}
\usfield{Question}{Replace it with what?}
\begin{enumerate}
\item[\texttt{D1}] replace UNETLoader \#12 with UnetLoaderGGUF
\item[\texttt{D2}] replace UNETLoader \#12 with CheckpointLoaderSimple
\item[\texttt{D3}] replace VAEDecodeTiled \#73 with plain VAEDecode
\item[$\bullet$] I don't know
\end{enumerate}

\ussection{Q6--Q8. Participant ratings}
\usfield{Q6. Confidence}{"How confident are you?" Options: 1, 2, 3, 4, 5 (1: not at all; 5: very).}
\usfield{Q7. Report helpfulness}{"How helpful was the checker report on this task?" Options: 1, 2, 3, 4, 5 (1: not helpful; 5: very helpful). Appears in the treatment version.}
\usfield{Q8. Mental effort}{"How much mental effort did this task take?" Options: 1, 2, 3, 4, 5, 6, 7, 8, 9 (1: very, very low; 9: very, very high).}
\noindent These ratings record participants' subjective judgments and remain separate from accuracy scoring. The rating fields appear for every Q1 answer.\par\smallskip
\usfield{Optional comment}{"Anything else about your answer? (optional)" Participants may add a free-text comment.}

\ussection{ComfyTy diagnosis (treatment only)}
\begin{lstlisting}[style=usreport]
ILL-TYPED (clash)  [100% typed (14/14)]
    Type mismatch: incompatible constructors across an edge.
    Involved (minimal core):
      - UNETLoader #12 (fp8_scales_external via headerfacts: mp_rank_00_model_states_fp8.pt)
\end{lstlisting}

\ustask{T04\quad\texttt{s-13271-fluxblack}}
\usfield{Task class}{Silent compatibility fault}
\usfield{Source}{Adapted from \href{https://github.com/Comfy-Org/ComfyUI/discussions/13271}{ComfyUI discussion \#13271}.}
\usfield{Task description}{Text-to-image with Flux schnell. The run completes without any error; the saved image is shown.}
\usfield{Reference diagnosis}{flux1-schnell (a bare diffusion model) is loaded via CheckpointLoaderSimple \#4; the pipeline silently produces a black image.}

\ussection{Answer key}
\usfield{Q1. Decision}{Yes, something is wrong}
\usfield{Q2. Region}{R1: Checkpoint loading --- CheckpointLoaderSimple \#4}
\usfield{Q3. Node}{\#4 CheckpointLoaderSimple}
\usfield{Q4. Category}{A: Swap the model/file it loads; D: Replace with a different node}
\usfield{Q5. Specific repair}{A1: change \#4 flux1-schnell.safetensors to an all-in-one Flux *checkpoint* release (bundles CLIP+VAE); D1: replace CheckpointLoaderSimple \#4 with UNETLoader (keep the separate CLIP/VAE loaders)}

\ussection{Workflow and observed behavior}
\begin{center}
\includegraphics[width=\linewidth,height=0.38\textheight,keepaspectratio]{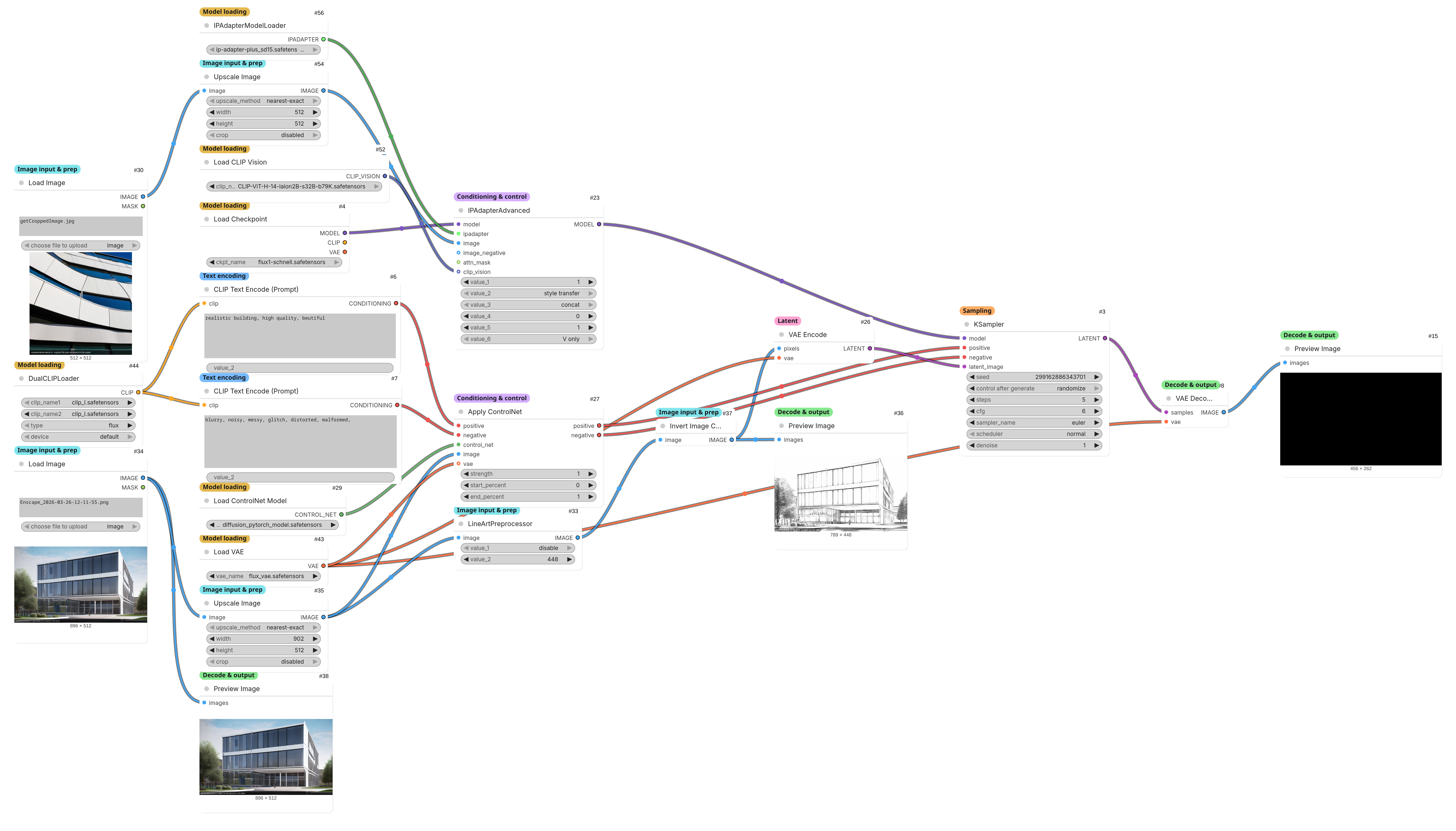}
\end{center}
\usfield{Graph}{22 nodes; 26 links. Node numbers connect the graph to the choices and diagnostic report.}
\usfield{Observed output}{The run completes and produces the following image. Participants inspect this output.}
\begin{center}
\includegraphics[width=0.25\linewidth,height=0.13\textheight,keepaspectratio]{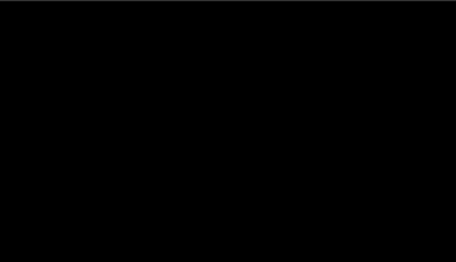}
\end{center}

\ussection{Model-file facts (shown in both versions)}
\usfield{Node \#4}{\textbf{flux1-schnell.safetensors}\newline FLUX.1 [schnell]. The diffusion transformer weights on their own; the text encoders (CLIP-L and T5-XXL) and the autoencoder are published as separate files in the same repository. 12B parameters, timestep-distilled for 1-4 step sampling.}
\usfield{Node \#43}{\textbf{flux\_vae.safetensors}\newline Autoencoder from the FLUX.1 release. Encodes images to sixteen-channel latents, four times the channel count of the Stable Diffusion 1.5 era autoencoders. Required by both the dev and the schnell transformers.}
\usfield{Node \#52}{\textbf{CLIP-ViT-H-14-laion2B-s32B-b79K.safetensors}\newline OpenCLIP ViT-H/14, trained on LAION-2B. A vision tower producing 1280-dimensional image embeddings, larger than the ViT-L/14 encoder. Commonly paired with the image adapters of the SD1.5 era.}
\usfield{Node \#44}{\textbf{clip\_l.safetensors}\newline CLIP-L/14 text encoder. The standard text tower of Stable Diffusion 1.5, and one of the two or three towers used by SDXL, SD3.5 and FLUX. 77-token context, 768-dimensional embeddings.}
\usfield{Node \#44}{\textbf{clip\_l.safetensors}\newline CLIP-L/14 text encoder. The standard text tower of Stable Diffusion 1.5, and one of the two or three towers used by SDXL, SD3.5 and FLUX. 77-token context, 768-dimensional embeddings.}
\usfield{Node \#29}{\textbf{diffusion\_pytorch\_model.safetensors}\newline The default file name diffusers writes when a component is saved. Many unrelated models ship under exactly this name, so the file name alone does not identify the architecture or the base model.}
\usfield{Node \#56}{\textbf{ip-adapter-plus\_sd15.safetensors}\newline IP-Adapter Plus. An image-prompt adapter trained for Stable Diffusion 1.5, using a ViT-H/14 CLIP vision tower to turn a reference image into conditioning. The Plus variant carries more image tokens than the base IP-Adapter.}

\ussection{Q1. Fault decision}
\usfield{Question}{Is anything wrong with this workflow?}
\begin{itemize}
\item[$\bullet$] Yes, something is wrong\quad\textbf{(Accepted)}
\item[$\bullet$] No, it looks fine
\item[$\bullet$] Can't determine
\end{itemize}
\noindent A Yes answer opens Q2--Q5. The other answers lead to the participant ratings.\par\smallskip

\ussection{Q2--Q3. Localization choices}
\usfield{Q2. Region}{Which part is at fault? Pick ONE option.}
\usfield{Q3. Node}{Which node exactly? Pick the node you would change.}
\noindent Each region below lists its Q2 option followed by the complete Q3 node menu. Select one node within the selected region.\par\smallskip
\begin{enumerate}
\item[\texttt{R1}] Checkpoint loading --- CheckpointLoaderSimple \#4\quad\textbf{(Accepted)}
\begin{itemize}
\item \#4 CheckpointLoaderSimple
\item I don't know
\end{itemize}
\item[\texttt{R2}] Text encoding --- DualCLIPLoader \#44, CLIPTextEncode \#6, \#7
\begin{itemize}
\item \#44 DualCLIPLoader
\item \#6 CLIPTextEncode
\item \#7 CLIPTextEncode
\item I don't know
\end{itemize}
\item[\texttt{R3}] IPAdapter branch --- IPAdapterModelLoader \#56, CLIPVisionLoader \#52, IPAdapterAdvanced \#23, LoadImage \#34, ImageScale \#54
\begin{itemize}
\item \#56 IPAdapterModelLoader
\item \#52 CLIPVisionLoader
\item \#23 IPAdapterAdvanced
\item \#34 LoadImage
\item \#54 ImageScale
\item I don't know
\end{itemize}
\item[\texttt{R4}] ControlNet branch --- ControlNetLoader \#29, ControlNetApplyAdvanced \#27, LineArtPreprocessor \#33, ImageInvert \#37, LoadImage \#30, ImageScale \#35, PreviewImage \#36, \#38
\begin{itemize}
\item \#29 ControlNetLoader
\item \#27 ControlNetApplyAdvanced
\item \#33 LineArtPreprocessor
\item \#37 ImageInvert
\item \#30 LoadImage
\item \#35 ImageScale
\item \#36 PreviewImage
\item \#38 PreviewImage
\item I don't know
\end{itemize}
\item[\texttt{R5}] Sampling \& img2img latent --- KSampler \#3, VAEEncode \#26
\begin{itemize}
\item \#3 KSampler
\item \#26 VAEEncode
\item I don't know
\end{itemize}
\item[\texttt{R6}] VAE \& output --- VAELoader \#43, VAEDecode \#8, PreviewImage \#15
\begin{itemize}
\item \#43 VAELoader
\item \#8 VAEDecode
\item \#15 PreviewImage
\item I don't know
\end{itemize}
\end{enumerate}

\ussection{Q4--Q5. Repair choices}
\usfield{Q4. Question}{What kind of change?}
\begin{enumerate}
\item[\texttt{A}] Swap the model/file it loads\quad\textbf{(Accepted)}
\item[\texttt{B}] Change a setting/widget value
\item[\texttt{C}] Rewire connections
\item[\texttt{D}] Replace with a different node\quad\textbf{(Accepted)}
\item[\texttt{E}] Other
\item[$\bullet$] I don't know
\end{enumerate}
\noindent Q5 displays the menu for the selected category A--D. Each menu also offers I don't know, which counts as an unsuccessful repair selection. Category E (Other) requires a free-text explanation. Selecting I don't know for Q4 proceeds directly to the participant ratings.\par\smallskip

\ussection{A. Swap the model/file it loads}
\usfield{Question}{Change the file to what?}
\begin{enumerate}
\item[\texttt{A1}] change \#4 flux1-schnell.safetensors to an all-in-one Flux *checkpoint* release (bundles CLIP+VAE)\quad\textbf{(Accepted)}
\item[\texttt{A2}] change \#4 flux1-schnell.safetensors to flux1-dev.safetensors
\item[\texttt{A3}] change \#29 diffusion\_pytorch\_model.safetensors to a Flux ControlNet
\item[\texttt{A4}] change \#44's second clip\_l.safetensors to t5xxl
\item[\texttt{A5}] change \#56 ip-adapter-plus\_sd15 to a Flux IPAdapter
\item[$\bullet$] I don't know
\end{enumerate}

\ussection{B. Change a setting/widget value}
\usfield{Question}{Which setting --- and to what?}
\begin{enumerate}
\item[\texttt{B1}] set the sampler `cfg` to 1 (Flux guidance style)
\item[\texttt{B2}] change the `type` on DualCLIPLoader \#44 (now flux)
\item[\texttt{B3}] lower `denoise` on KSampler \#3
\item[\texttt{B4}] lower the ControlNet strength on \#27
\item[$\bullet$] I don't know
\end{enumerate}

\ussection{C. Rewire connections}
\usfield{Question}{Which connection would you change?}
\begin{enumerate}
\item[\texttt{C1}] bypass the ControlNet branch (\#27-\#38)
\item[\texttt{C2}] bypass the IPAdapter branch (\#23/\#56)
\item[\texttt{C3}] take VAEDecode \#8's `vae` from checkpoint \#4 instead of VAELoader \#43
\item[$\bullet$] I don't know
\end{enumerate}

\ussection{D. Replace with a different node}
\usfield{Question}{Replace it with what?}
\begin{enumerate}
\item[\texttt{D1}] replace CheckpointLoaderSimple \#4 with UNETLoader (keep the separate CLIP/VAE loaders)\quad\textbf{(Accepted)}
\item[\texttt{D2}] replace ControlNetLoader \#29 with an SD1.5-specific loader
\item[\texttt{D3}] replace DualCLIPLoader \#44 with a single CLIPLoader
\item[\texttt{D4}] replace the IPAdapter stack with Flux-native style nodes
\item[$\bullet$] I don't know
\end{enumerate}

\ussection{Q6--Q8. Participant ratings}
\usfield{Q6. Confidence}{"How confident are you?" Options: 1, 2, 3, 4, 5 (1: not at all; 5: very).}
\usfield{Q7. Report helpfulness}{"How helpful was the checker report on this task?" Options: 1, 2, 3, 4, 5 (1: not helpful; 5: very helpful). Appears in the treatment version.}
\usfield{Q8. Mental effort}{"How much mental effort did this task take?" Options: 1, 2, 3, 4, 5, 6, 7, 8, 9 (1: very, very low; 9: very, very high).}
\noindent These ratings record participants' subjective judgments and remain separate from accuracy scoring. The rating fields appear for every Q1 answer.\par\smallskip
\usfield{Optional comment}{"Anything else about your answer? (optional)" Participants may add a free-text comment.}

\ussection{ComfyTy diagnosis (treatment only)}
\begin{lstlisting}[style=usreport]
ILL-TYPED (clash)  [91% typed (20/22)]
    Type mismatch: incompatible constructors across an edge.
    Involved (minimal core):
      - CheckpointLoaderSimple #4 (file-kind=diffusion_model via headerfacts: flux1-schnell.safetensors)
\end{lstlisting}

\ustask{T05\quad\texttt{s-lora-fam}}
\usfield{Task class}{Silent compatibility fault}
\usfield{Source and construction}{Controlled example of silent LoRA loading failures. This task instantiates a model-family mismatch by applying an SDXL style LoRA to an SD1.5 text-to-image workflow; GPU execution validates the observed behavior.}
\usfield{Upstream basis}{\href{https://github.com/Comfy-Org/ComfyUI/issues/3194\#issuecomment-2041488007}{A comment on ComfyUI issue \#3194} describes an SDXL LoRA applied to an SD1.5 model: tensor-size warnings appear, while the workflow continues and the LoRA affects the result. This task instantiates the same model-family pairing.}
\usfield{Task description}{Generate a detailed cat photo with a style LoRA applied.}
\usfield{Reference diagnosis}{SDXL LoRA on an SD1.5 checkpoint: most LoRA keys do not bind; the style is silently absent/degraded.}

\ussection{Answer key}
\usfield{Q1. Decision}{Yes, something is wrong}
\usfield{Q2. Region}{R1: Checkpoint loading --- CheckpointLoaderSimple \#4; R2: LoRA loading --- LoraLoader \#10}
\usfield{Q3. Node}{\#4 CheckpointLoaderSimple; \#10 LoraLoader}
\usfield{Q4. Category}{A: Swap the model/file it loads}
\usfield{Q5. Specific repair}{A1: change \#10 sdxl\_lora.safetensors to an SD1.5 build of the LoRA; A2: change \#4 sd15.safetensors to an SDXL checkpoint}

\ussection{Workflow and observed behavior}
\begin{center}
\includegraphics[width=\linewidth,height=0.38\textheight,keepaspectratio]{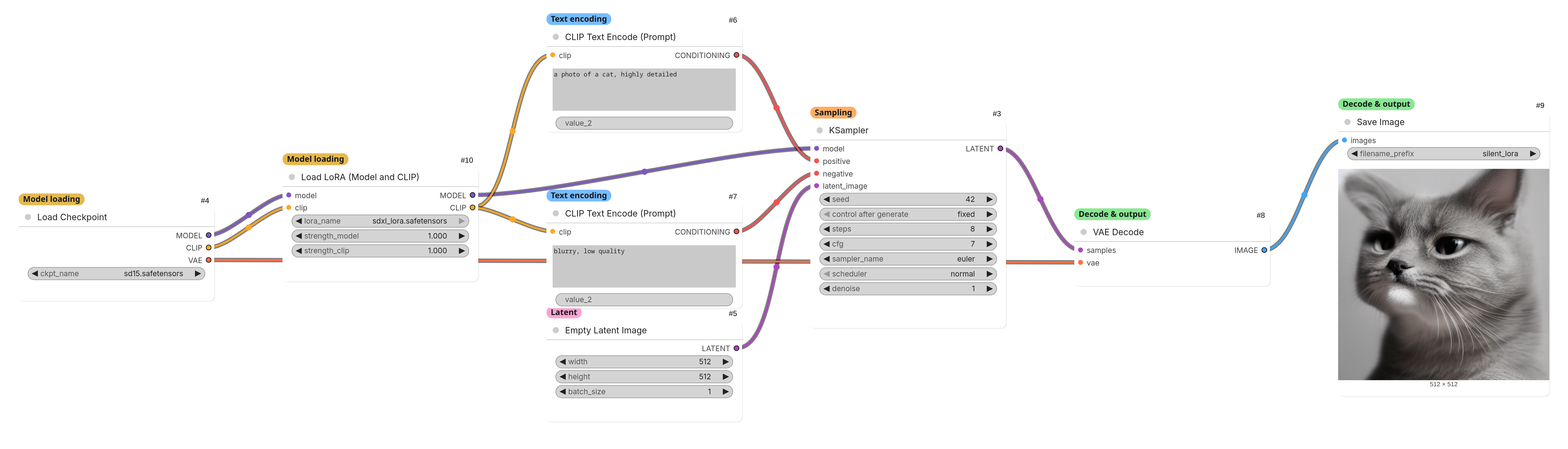}
\end{center}
\usfield{Graph}{8 nodes; 11 links. Node numbers connect the graph to the choices and diagnostic report.}
\usfield{Observed output}{The run completes and produces the following image. Participants inspect this output.}
\begin{center}
\includegraphics[width=0.25\linewidth,height=0.13\textheight,keepaspectratio]{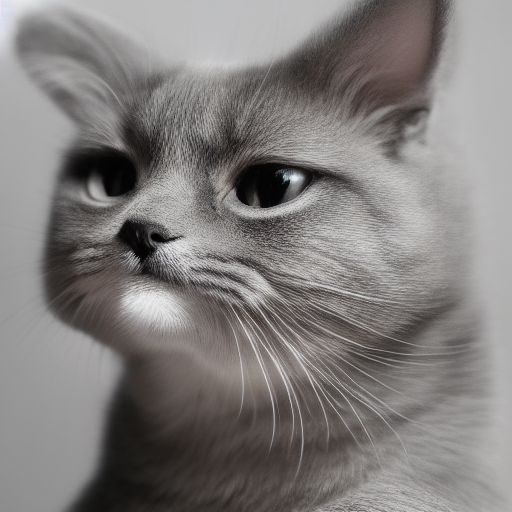}
\end{center}

\ussection{Model-file facts (shown in both versions)}
\usfield{Node \#4}{\textbf{sd15.safetensors}\newline Stable Diffusion 1.5. The classic text-to-image checkpoint; one file holding the UNet, the CLIP-L text encoder and the VAE. Trained at 512x512, and the base that most community 1.5 finetunes and LoRAs build on.}
\usfield{Node \#10}{\textbf{sdxl\_lora.safetensors}\newline Style LoRA. Trained on SDXL Base 1.0 at 1024x1024 with rank 32. Recommended strength 0.6-1.0; no trigger word is needed, the look is carried by the UNet blocks.}

\ussection{Q1. Fault decision}
\usfield{Question}{Is anything wrong with this workflow?}
\begin{itemize}
\item[$\bullet$] Yes, something is wrong\quad\textbf{(Accepted)}
\item[$\bullet$] No, it looks fine
\item[$\bullet$] Can't determine
\end{itemize}
\noindent A Yes answer opens Q2--Q5. The other answers lead to the participant ratings.\par\smallskip

\ussection{Q2--Q3. Localization choices}
\usfield{Q2. Region}{Which part is at fault? Pick ONE option.}
\usfield{Q3. Node}{Which node exactly? Pick the node you would change.}
\noindent Each region below lists its Q2 option followed by the complete Q3 node menu. Select one node within the selected region.\par\smallskip
\begin{enumerate}
\item[\texttt{R1}] Checkpoint loading --- CheckpointLoaderSimple \#4\quad\textbf{(Accepted)}
\begin{itemize}
\item \#4 CheckpointLoaderSimple
\item I don't know
\end{itemize}
\item[\texttt{R2}] LoRA loading --- LoraLoader \#10\quad\textbf{(Accepted)}
\begin{itemize}
\item \#10 LoraLoader
\item I don't know
\end{itemize}
\item[\texttt{R3}] Prompt encoding --- CLIPTextEncode \#6, \#7
\begin{itemize}
\item \#6 CLIPTextEncode
\item \#7 CLIPTextEncode
\item I don't know
\end{itemize}
\item[\texttt{R4}] Sampling \& latent --- KSampler \#3, EmptyLatentImage \#5
\begin{itemize}
\item \#3 KSampler
\item \#5 EmptyLatentImage
\item I don't know
\end{itemize}
\item[\texttt{R5}] Decode \& save --- VAEDecode \#8, SaveImage \#9
\begin{itemize}
\item \#8 VAEDecode
\item \#9 SaveImage
\item I don't know
\end{itemize}
\end{enumerate}

\ussection{Q4--Q5. Repair choices}
\usfield{Q4. Question}{What kind of change?}
\begin{enumerate}
\item[\texttt{A}] Swap the model/file it loads\quad\textbf{(Accepted)}
\item[\texttt{B}] Change a setting/widget value
\item[\texttt{C}] Rewire connections
\item[\texttt{D}] Replace with a different node
\item[\texttt{E}] Other
\item[$\bullet$] I don't know
\end{enumerate}
\noindent Q5 displays the menu for the selected category A--D. Each menu also offers I don't know, which counts as an unsuccessful repair selection. Category E (Other) requires a free-text explanation. Selecting I don't know for Q4 proceeds directly to the participant ratings.\par\smallskip

\ussection{A. Swap the model/file it loads}
\usfield{Question}{Change the file to what?}
\begin{enumerate}
\item[\texttt{A1}] change \#10 sdxl\_lora.safetensors to an SD1.5 build of the LoRA\quad\textbf{(Accepted)}
\item[\texttt{A2}] change \#4 sd15.safetensors to an SDXL checkpoint\quad\textbf{(Accepted)}
\item[\texttt{A3}] change \#10 sdxl\_lora.safetensors to a Flux build of the LoRA
\item[\texttt{A4}] change \#4 sd15.safetensors to an SD2.x checkpoint
\item[\texttt{A5}] change \#10 to a different style LoRA of the same family
\item[$\bullet$] I don't know
\end{enumerate}

\ussection{B. Change a setting/widget value}
\usfield{Question}{Which setting --- and to what?}
\begin{enumerate}
\item[\texttt{B1}] raise the LoRA strengths on \#10 (now 1.0 / 1.0)
\item[\texttt{B2}] raise the sampler `cfg`
\item[\texttt{B3}] increase `steps`
\item[$\bullet$] I don't know
\end{enumerate}

\ussection{C. Rewire connections}
\usfield{Question}{Which connection would you change?}
\begin{enumerate}
\item[\texttt{C1}] take the CLIP for the text encoders from \#4 directly (bypass \#10's CLIP)
\item[\texttt{C2}] route the MODEL around LoraLoader \#10
\item[\texttt{C3}] apply the LoRA after sampling
\item[$\bullet$] I don't know
\end{enumerate}

\ussection{D. Replace with a different node}
\usfield{Question}{Replace it with what?}
\begin{enumerate}
\item[\texttt{D1}] replace LoraLoader \#10 with LoraLoaderModelOnly
\item[\texttt{D2}] replace the LoRA with an IPAdapter style reference
\item[\texttt{D3}] replace \#4 with a checkpoint that has the style merged in
\item[$\bullet$] I don't know
\end{enumerate}

\ussection{Q6--Q8. Participant ratings}
\usfield{Q6. Confidence}{"How confident are you?" Options: 1, 2, 3, 4, 5 (1: not at all; 5: very).}
\usfield{Q7. Report helpfulness}{"How helpful was the checker report on this task?" Options: 1, 2, 3, 4, 5 (1: not helpful; 5: very helpful). Appears in the treatment version.}
\usfield{Q8. Mental effort}{"How much mental effort did this task take?" Options: 1, 2, 3, 4, 5, 6, 7, 8, 9 (1: very, very low; 9: very, very high).}
\noindent These ratings record participants' subjective judgments and remain separate from accuracy scoring. The rating fields appear for every Q1 answer.\par\smallskip
\usfield{Optional comment}{"Anything else about your answer? (optional)" Participants may add a free-text comment.}

\ussection{ComfyTy diagnosis (treatment only)}
\begin{lstlisting}[style=usreport]
ILL-TYPED (brand)  [100% typed (8/8)]
    Model family/architecture conflict: SD15 vs SDXL.
    Involved (minimal core):
      - CheckpointLoaderSimple #4 (ckpt=SD15 via headerfacts: sd15.safetensors)
      - LoraLoader #10 (lora=SDXL via headerfacts: sdxl_lora.safetensors)
      - LoraLoader #10
      - edge CheckpointLoaderSimple #4 -> LoraLoader #10.model
    Possible fixes:
      - change the file at CheckpointLoaderSimple #4 (ckpt=SD15 via headerfacts: sd15.safetensors) to one matching the rest of the graph
      - change the file at LoraLoader #10 (lora=SDXL via headerfacts: sdxl_lora.safetensors) to one matching the rest of the graph
\end{lstlisting}

\ustask{T06\quad\texttt{s-wrong-vae}}
\usfield{Task class}{Silent compatibility fault}
\usfield{Source and construction}{Controlled example of cross-family VAE decoding. This task routes an SD1.5 sampled latent through the VAE supplied by an SDXL checkpoint; GPU execution validates the observed behavior.}
\usfield{Upstream basis}{\href{https://github.com/Comfy-Org/ComfyUI/issues/1602\#issuecomment-1741763137}{A comment on ComfyUI issue \#1602} describes noisy output when an SD1.4 latent is decoded with an SD2.1/XL VAE. This task instantiates the cross-family decoding pattern with SD1.5 and SDXL.}
\usfield{Task description}{Generate a detailed photo of a cat (text-to-image).}
\usfield{Reference diagnosis}{The SD1.5 latent is decoded with the SDXL checkpoint's VAE (from \#20); output is silently corrupted.}

\ussection{Answer key}
\usfield{Q1. Decision}{Yes, something is wrong}
\usfield{Q2. Region}{R2: Second checkpoint (supplies the VAE) --- CheckpointLoaderSimple \#20; R5: Decode \& save --- VAEDecode \#8, SaveImage \#9}
\usfield{Q3. Node}{\#20 CheckpointLoaderSimple; \#8 VAEDecode}
\usfield{Q4. Category}{A: Swap the model/file it loads; C: Rewire connections}
\usfield{Q5. Specific repair}{A1: change \#20 sdxl.safetensors to an SD1.5 checkpoint (its VAE then matches); C1: take VAEDecode \#8's `vae` input from checkpoint \#4 instead of checkpoint \#20}

\ussection{Workflow and observed behavior}
\begin{center}
\includegraphics[width=\linewidth,height=0.38\textheight,keepaspectratio]{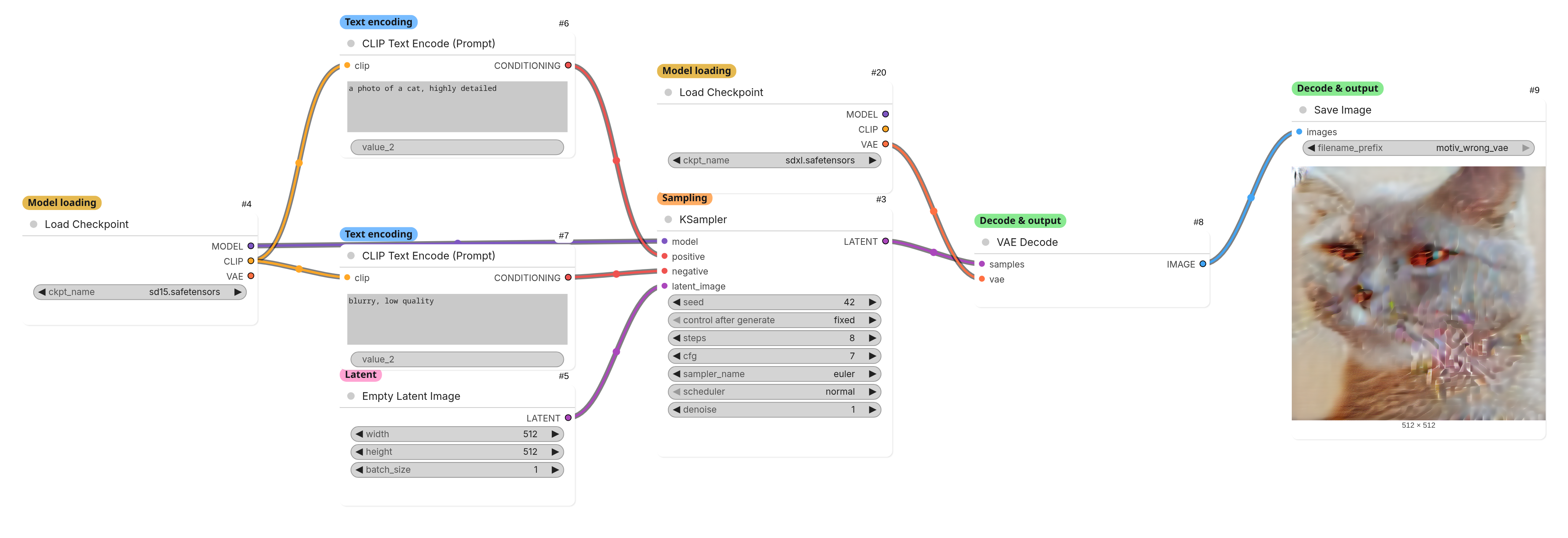}
\end{center}
\usfield{Graph}{8 nodes; 9 links. Node numbers connect the graph to the choices and diagnostic report.}
\usfield{Observed output}{The run completes and produces the following image. Participants inspect this output.}
\begin{center}
\includegraphics[width=0.25\linewidth,height=0.13\textheight,keepaspectratio]{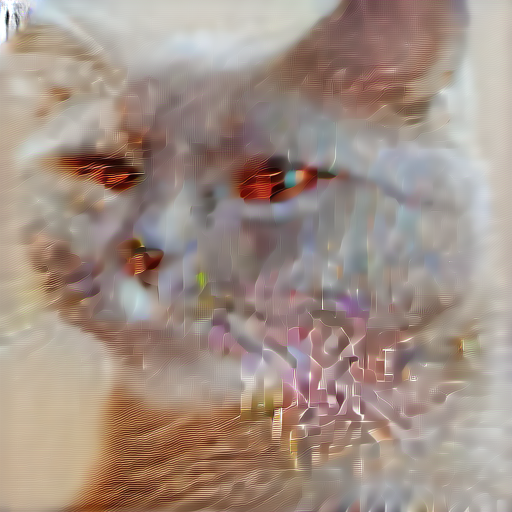}
\end{center}

\ussection{Model-file facts (shown in both versions)}
\usfield{Node \#4}{\textbf{sd15.safetensors}\newline Stable Diffusion 1.5. The classic text-to-image checkpoint; one file holding the UNet, the CLIP-L text encoder and the VAE. Trained at 512x512, and the base that most community 1.5 finetunes and LoRAs build on.}
\usfield{Node \#20}{\textbf{sdxl.safetensors}\newline SDXL Base 1.0. All-in-one checkpoint holding the UNet, both text encoders (CLIP-L and OpenCLIP bigG) and the VAE. Trained at 1024x1024; a separate refiner model is published alongside it.}

\ussection{Q1. Fault decision}
\usfield{Question}{Is anything wrong with this workflow?}
\begin{itemize}
\item[$\bullet$] Yes, something is wrong\quad\textbf{(Accepted)}
\item[$\bullet$] No, it looks fine
\item[$\bullet$] Can't determine
\end{itemize}
\noindent A Yes answer opens Q2--Q5. The other answers lead to the participant ratings.\par\smallskip

\ussection{Q2--Q3. Localization choices}
\usfield{Q2. Region}{Which part is at fault? Pick ONE option.}
\usfield{Q3. Node}{Which node exactly? Pick the node you would change.}
\noindent Each region below lists its Q2 option followed by the complete Q3 node menu. Select one node within the selected region.\par\smallskip
\begin{enumerate}
\item[\texttt{R1}] Main checkpoint --- CheckpointLoaderSimple \#4
\begin{itemize}
\item \#4 CheckpointLoaderSimple
\item I don't know
\end{itemize}
\item[\texttt{R2}] Second checkpoint (supplies the VAE) --- CheckpointLoaderSimple \#20\quad\textbf{(Accepted)}
\begin{itemize}
\item \#20 CheckpointLoaderSimple
\item I don't know
\end{itemize}
\item[\texttt{R3}] Prompt encoding --- CLIPTextEncode \#6, \#7
\begin{itemize}
\item \#6 CLIPTextEncode
\item \#7 CLIPTextEncode
\item I don't know
\end{itemize}
\item[\texttt{R4}] Sampling \& latent --- KSampler \#3, EmptyLatentImage \#5
\begin{itemize}
\item \#3 KSampler
\item \#5 EmptyLatentImage
\item I don't know
\end{itemize}
\item[\texttt{R5}] Decode \& save --- VAEDecode \#8, SaveImage \#9\quad\textbf{(Accepted)}
\begin{itemize}
\item \#8 VAEDecode
\item \#9 SaveImage
\item I don't know
\end{itemize}
\end{enumerate}

\ussection{Q4--Q5. Repair choices}
\usfield{Q4. Question}{What kind of change?}
\begin{enumerate}
\item[\texttt{A}] Swap the model/file it loads\quad\textbf{(Accepted)}
\item[\texttt{B}] Change a setting/widget value
\item[\texttt{C}] Rewire connections\quad\textbf{(Accepted)}
\item[\texttt{D}] Replace with a different node
\item[\texttt{E}] Other
\item[$\bullet$] I don't know
\end{enumerate}
\noindent Q5 displays the menu for the selected category A--D. Each menu also offers I don't know, which counts as an unsuccessful repair selection. Category E (Other) requires a free-text explanation. Selecting I don't know for Q4 proceeds directly to the participant ratings.\par\smallskip

\ussection{A. Swap the model/file it loads}
\usfield{Question}{Change the file to what?}
\begin{enumerate}
\item[\texttt{A1}] change \#20 sdxl.safetensors to an SD1.5 checkpoint (its VAE then matches)\quad\textbf{(Accepted)}
\item[\texttt{A2}] change \#4 sd15.safetensors to a high-res SD1.5 finetune
\item[\texttt{A3}] change \#20 sdxl.safetensors to the SDXL inpaint variant
\item[\texttt{A4}] change \#4 sd15.safetensors to an SD2.x checkpoint
\item[$\bullet$] I don't know
\end{enumerate}

\ussection{B. Change a setting/widget value}
\usfield{Question}{Which setting --- and to what?}
\begin{enumerate}
\item[\texttt{B1}] change `denoise` on the sampler
\item[\texttt{B2}] change the sampler `cfg`
\item[\texttt{B3}] change the latent batch size
\item[$\bullet$] I don't know
\end{enumerate}

\ussection{C. Rewire connections}
\usfield{Question}{Which connection would you change?}
\begin{enumerate}
\item[\texttt{C1}] take VAEDecode \#8's `vae` input from checkpoint \#4 instead of checkpoint \#20\quad\textbf{(Accepted)}
\item[\texttt{C2}] take the text encoders' CLIP from \#20 instead of \#4
\item[\texttt{C3}] take the sampler's MODEL from \#20 instead of \#4
\item[$\bullet$] I don't know
\end{enumerate}

\ussection{D. Replace with a different node}
\usfield{Question}{Replace it with what?}
\begin{enumerate}
\item[\texttt{D1}] insert a dedicated VAELoader (loading an SDXL VAE) to feed VAEDecode \#8
\item[\texttt{D2}] replace VAEDecode \#8 with VAEDecodeTiled
\item[\texttt{D3}] delete checkpoint \#20 and its outputs
\item[$\bullet$] I don't know
\end{enumerate}

\ussection{Q6--Q8. Participant ratings}
\usfield{Q6. Confidence}{"How confident are you?" Options: 1, 2, 3, 4, 5 (1: not at all; 5: very).}
\usfield{Q7. Report helpfulness}{"How helpful was the checker report on this task?" Options: 1, 2, 3, 4, 5 (1: not helpful; 5: very helpful). Appears in the treatment version.}
\usfield{Q8. Mental effort}{"How much mental effort did this task take?" Options: 1, 2, 3, 4, 5, 6, 7, 8, 9 (1: very, very low; 9: very, very high).}
\noindent These ratings record participants' subjective judgments and remain separate from accuracy scoring. The rating fields appear for every Q1 answer.\par\smallskip
\usfield{Optional comment}{"Anything else about your answer? (optional)" Participants may add a free-text comment.}

\ussection{ComfyTy diagnosis (treatment only)}
\begin{lstlisting}[style=usreport]
ILL-TYPED (brand)  [100% typed (8/8)]
    Model family/architecture conflict: SDXL vs SD15.
    Involved (minimal core):
      - CheckpointLoaderSimple #4 (ckpt=SD15 via headerfacts: sd15.safetensors)
      - CheckpointLoaderSimple #20 (ckpt=SDXL via headerfacts: sdxl.safetensors)
      - edge CheckpointLoaderSimple #4 -> CLIPTextEncode #7.clip
      - edge CLIPTextEncode #7 -> KSampler #3.negative
      - edge KSampler #3 -> VAEDecode #8.samples
      - edge CheckpointLoaderSimple #20 -> VAEDecode #8.vae
    Possible fixes:
      - change the file at CheckpointLoaderSimple #4 (ckpt=SD15 via headerfacts: sd15.safetensors) to one matching the rest of the graph
      - change the file at CheckpointLoaderSimple #20 (ckpt=SDXL via headerfacts: sdxl.safetensors) to one matching the rest of the graph
\end{lstlisting}

\ustask{T07\quad\texttt{g-fox-ok}}
\usfield{Task class}{Valid workflow}
\usfield{Source and construction}{Valid text-to-image control, validated by GPU execution. A standard SD1.5 pipeline generates a fox image at 512x512 resolution.}
\usfield{Upstream basis}{\href{https://github.com/Comfy-Org/ComfyUI/blob/11ad6060fcc79e2a0088fc5f1eb906ac445755a1/web/scripts/defaultGraph.js}{The official ComfyUI default workflow} provides the same seven-node text-to-image topology, node identifiers, and connections. This study supplies the fox prompt and generation settings.}
\usfield{Task description}{Generate a photo of a red fox in an autumn forest (text-to-image).}
\usfield{Reference diagnosis}{The workflow is valid and its components are compatible.}

\ussection{Answer key}
\usfield{Q1. Decision}{No, it looks fine}
\usfield{Q2. Region}{The correct Q1 answer completes diagnosis.}
\usfield{Q3. Node}{The workflow's nodes are valid.}
\usfield{Q4--Q5. Repair}{Keep the valid workflow unchanged. Repair scoring applies to faulty workflows.}

\ussection{Workflow and observed behavior}
\begin{center}
\includegraphics[width=\linewidth,height=0.38\textheight,keepaspectratio]{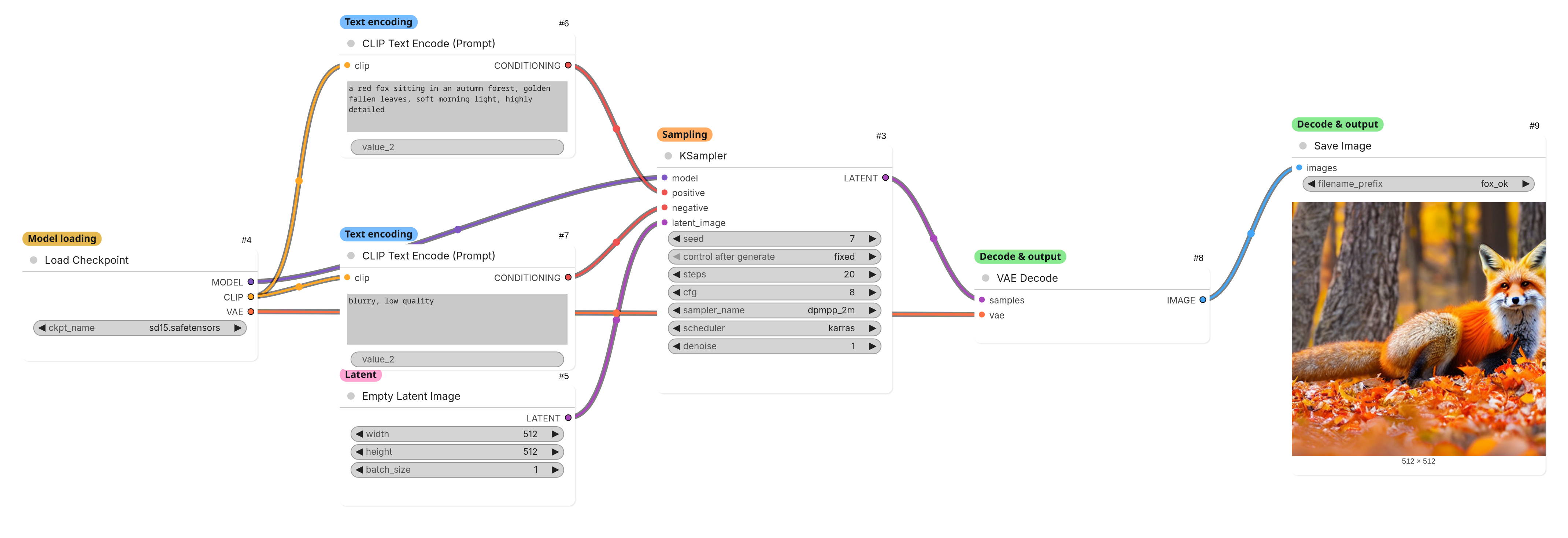}
\end{center}
\usfield{Graph}{7 nodes; 9 links. Node numbers connect the graph to the choices and diagnostic report.}
\usfield{Observed output}{The run completes and produces the following image. Participants inspect this output.}
\begin{center}
\includegraphics[width=0.25\linewidth,height=0.13\textheight,keepaspectratio]{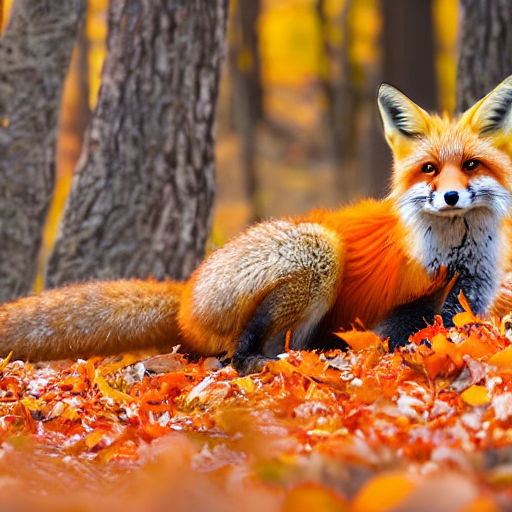}
\end{center}

\ussection{Model-file facts (shown in both versions)}
\usfield{Node \#4}{\textbf{sd15.safetensors}\newline Stable Diffusion 1.5. The classic text-to-image checkpoint; one file holding the UNet, the CLIP-L text encoder and the VAE. Trained at 512x512, and the base that most community 1.5 finetunes and LoRAs build on.}

\ussection{Q1. Fault decision}
\usfield{Question}{Is anything wrong with this workflow?}
\begin{itemize}
\item[$\bullet$] Yes, something is wrong
\item[$\bullet$] No, it looks fine\quad\textbf{(Accepted)}
\item[$\bullet$] Can't determine
\end{itemize}
\noindent A Yes answer opens Q2--Q5. The other answers lead to the participant ratings.\par\smallskip

\ussection{Q2--Q3. Localization choices}
\usfield{Q2. Region}{Which part is at fault? Pick ONE option.}
\usfield{Q3. Node}{Which node exactly? Pick the node you would change.}
\noindent Each region below lists its Q2 option followed by the complete Q3 node menu. Select one node within the selected region.\par\smallskip
\begin{enumerate}
\item[\texttt{R1}] Model loading --- CheckpointLoaderSimple \#4
\begin{itemize}
\item \#4 CheckpointLoaderSimple
\item I don't know
\end{itemize}
\item[\texttt{R2}] Prompt encoding --- CLIPTextEncode \#6, \#7
\begin{itemize}
\item \#6 CLIPTextEncode
\item \#7 CLIPTextEncode
\item I don't know
\end{itemize}
\item[\texttt{R3}] Sampling \& latent --- KSampler \#3, EmptyLatentImage \#5
\begin{itemize}
\item \#3 KSampler
\item \#5 EmptyLatentImage
\item I don't know
\end{itemize}
\item[\texttt{R4}] Decode \& save --- VAEDecode \#8, SaveImage \#9
\begin{itemize}
\item \#8 VAEDecode
\item \#9 SaveImage
\item I don't know
\end{itemize}
\end{enumerate}
\noindent A Yes answer opens these region and node choices and counts as a false alarm on this valid workflow.\par\smallskip

\ussection{Q4--Q5. Repair choices}
\usfield{Q4. Question}{What kind of change?}
\begin{enumerate}
\item[\texttt{A}] Swap the model/file it loads
\item[\texttt{B}] Change a setting/widget value
\item[\texttt{C}] Rewire connections
\item[\texttt{D}] Replace with a different node
\item[\texttt{E}] Other
\item[$\bullet$] I don't know
\end{enumerate}
\usfield{Q5. Branch behavior}{A Yes answer exposes Q4. Categories A--D and I don't know lead directly to the ratings; category E (Other) requires a free-text explanation. Q5 is omitted for this valid workflow. The correct response preserves the workflow unchanged.}

\ussection{Q6--Q8. Participant ratings}
\usfield{Q6. Confidence}{"How confident are you?" Options: 1, 2, 3, 4, 5 (1: not at all; 5: very).}
\usfield{Q7. Report helpfulness}{"How helpful was the checker report on this task?" Options: 1, 2, 3, 4, 5 (1: not helpful; 5: very helpful). Appears in the treatment version.}
\usfield{Q8. Mental effort}{"How much mental effort did this task take?" Options: 1, 2, 3, 4, 5, 6, 7, 8, 9 (1: very, very low; 9: very, very high).}
\noindent These ratings record participants' subjective judgments and remain separate from accuracy scoring. The rating fields appear for every Q1 answer.\par\smallskip
\usfield{Optional comment}{"Anything else about your answer? (optional)" Participants may add a free-text comment.}

\ussection{ComfyTy diagnosis (treatment only)}
\begin{lstlisting}[style=usreport]
OK  [100% typed (7/7)]
\end{lstlisting}

\ustask{T08\quad\texttt{g-lora-ok}}
\usfield{Task class}{Valid workflow}
\usfield{Source and construction}{Valid control for the LoRA compatibility pattern in T05, validated by GPU execution. An SDXL text-to-image pipeline uses a compatible SDXL pixel-art LoRA at 1024x1024 resolution.}
\usfield{Upstream basis}{\href{https://github.com/comfyanonymous/ComfyUI_examples/tree/master/lora}{The official ComfyUI LoRA examples} show a LoraLoader applied to the checkpoint's MODEL and CLIP outputs, as in this task.}
\usfield{Upstream basis}{\href{https://github.com/Comfy-Org/ComfyUI/issues/1461\#issuecomment-1712399049}{A discussion in ComfyUI issue \#1461} explains that an SDXL model requires an SDXL LoRA. This valid control uses that compatible pairing.}
\usfield{Task description}{Generate a corgi photo in pixel-art style via a pixel-art LoRA.}
\usfield{Reference diagnosis}{The workflow is valid and its components are compatible.}

\ussection{Answer key}
\usfield{Q1. Decision}{No, it looks fine}
\usfield{Q2. Region}{The correct Q1 answer completes diagnosis.}
\usfield{Q3. Node}{The workflow's nodes are valid.}
\usfield{Q4--Q5. Repair}{Keep the valid workflow unchanged. Repair scoring applies to faulty workflows.}

\ussection{Workflow and observed behavior}
\begin{center}
\includegraphics[width=\linewidth,height=0.38\textheight,keepaspectratio]{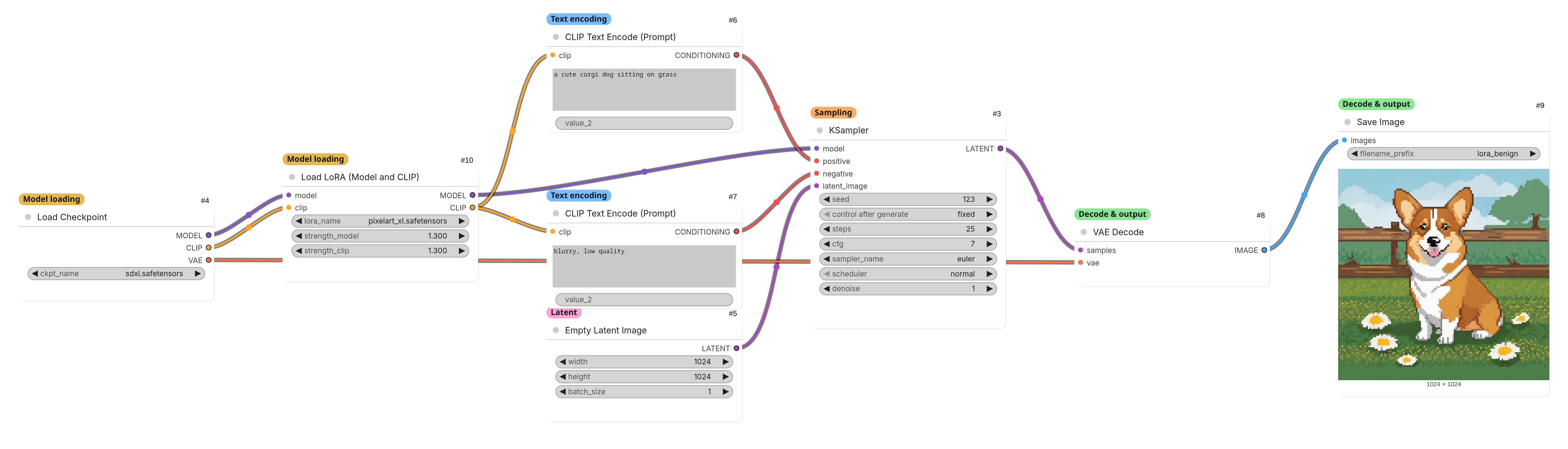}
\end{center}
\usfield{Graph}{8 nodes; 11 links. Node numbers connect the graph to the choices and diagnostic report.}
\usfield{Observed output}{The run completes and produces the following image. Participants inspect this output.}
\begin{center}
\includegraphics[width=0.25\linewidth,height=0.13\textheight,keepaspectratio]{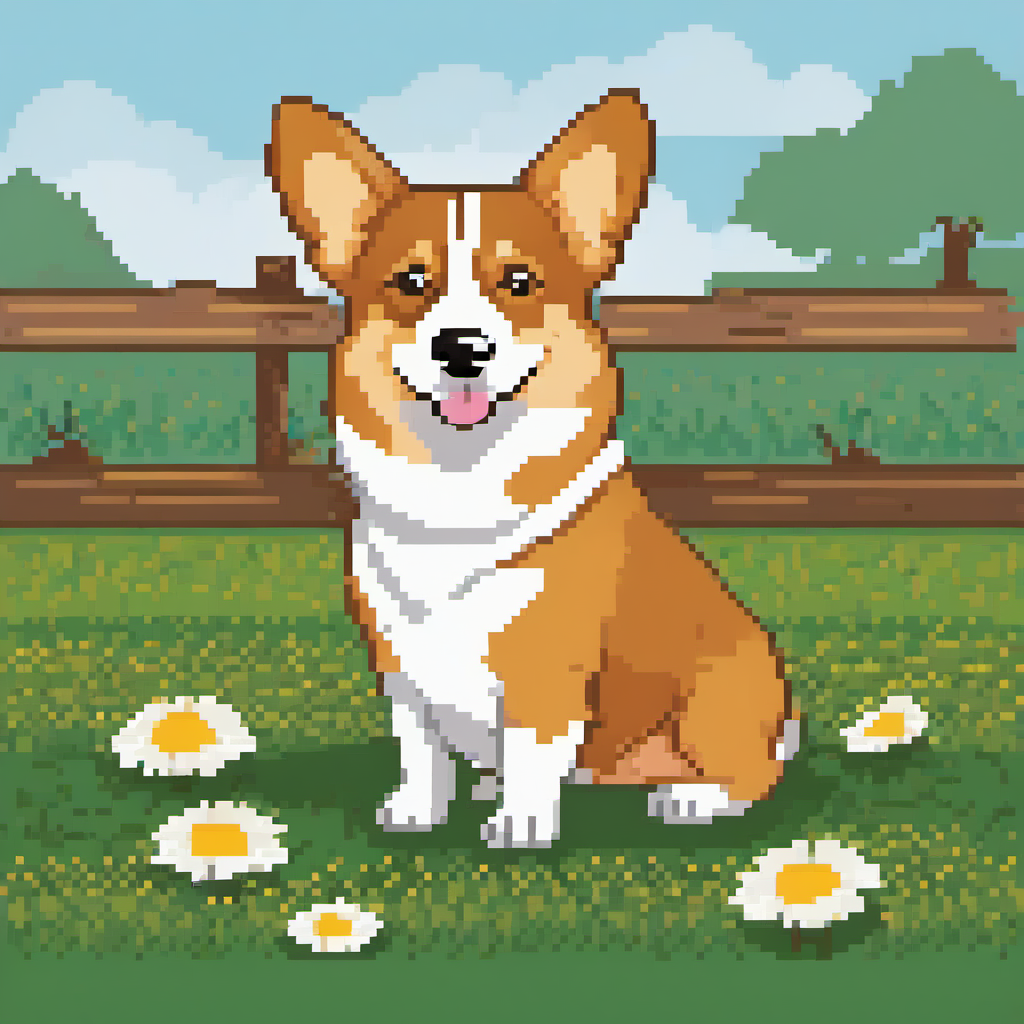}
\end{center}

\ussection{Model-file facts (shown in both versions)}
\usfield{Node \#4}{\textbf{sdxl.safetensors}\newline SDXL Base 1.0. All-in-one checkpoint holding the UNet, both text encoders (CLIP-L and OpenCLIP bigG) and the VAE. Trained at 1024x1024; a separate refiner model is published alongside it.}
\usfield{Node \#10}{\textbf{pixelart\_xl.safetensors}\newline Pixel-art style LoRA. Trained on SDXL Base 1.0; renders blocky, limited-palette images. Works best around strength 0.8-1.2 with 1024x1024 output, optionally downscaled afterwards to sharpen the grid.}

\ussection{Q1. Fault decision}
\usfield{Question}{Is anything wrong with this workflow?}
\begin{itemize}
\item[$\bullet$] Yes, something is wrong
\item[$\bullet$] No, it looks fine\quad\textbf{(Accepted)}
\item[$\bullet$] Can't determine
\end{itemize}
\noindent A Yes answer opens Q2--Q5. The other answers lead to the participant ratings.\par\smallskip

\ussection{Q2--Q3. Localization choices}
\usfield{Q2. Region}{Which part is at fault? Pick ONE option.}
\usfield{Q3. Node}{Which node exactly? Pick the node you would change.}
\noindent Each region below lists its Q2 option followed by the complete Q3 node menu. Select one node within the selected region.\par\smallskip
\begin{enumerate}
\item[\texttt{R1}] Checkpoint loading --- CheckpointLoaderSimple \#4
\begin{itemize}
\item \#4 CheckpointLoaderSimple
\item I don't know
\end{itemize}
\item[\texttt{R2}] LoRA loading --- LoraLoader \#10
\begin{itemize}
\item \#10 LoraLoader
\item I don't know
\end{itemize}
\item[\texttt{R3}] Prompt encoding --- CLIPTextEncode \#6, \#7
\begin{itemize}
\item \#6 CLIPTextEncode
\item \#7 CLIPTextEncode
\item I don't know
\end{itemize}
\item[\texttt{R4}] Sampling \& latent --- KSampler \#3, EmptyLatentImage \#5
\begin{itemize}
\item \#3 KSampler
\item \#5 EmptyLatentImage
\item I don't know
\end{itemize}
\item[\texttt{R5}] Decode \& save --- VAEDecode \#8, SaveImage \#9
\begin{itemize}
\item \#8 VAEDecode
\item \#9 SaveImage
\item I don't know
\end{itemize}
\end{enumerate}
\noindent A Yes answer opens these region and node choices and counts as a false alarm on this valid workflow.\par\smallskip

\ussection{Q4--Q5. Repair choices}
\usfield{Q4. Question}{What kind of change?}
\begin{enumerate}
\item[\texttt{A}] Swap the model/file it loads
\item[\texttt{B}] Change a setting/widget value
\item[\texttt{C}] Rewire connections
\item[\texttt{D}] Replace with a different node
\item[\texttt{E}] Other
\item[$\bullet$] I don't know
\end{enumerate}
\usfield{Q5. Branch behavior}{A Yes answer exposes Q4. Categories A--D and I don't know lead directly to the ratings; category E (Other) requires a free-text explanation. Q5 is omitted for this valid workflow. The correct response preserves the workflow unchanged.}

\ussection{Q6--Q8. Participant ratings}
\usfield{Q6. Confidence}{"How confident are you?" Options: 1, 2, 3, 4, 5 (1: not at all; 5: very).}
\usfield{Q7. Report helpfulness}{"How helpful was the checker report on this task?" Options: 1, 2, 3, 4, 5 (1: not helpful; 5: very helpful). Appears in the treatment version.}
\usfield{Q8. Mental effort}{"How much mental effort did this task take?" Options: 1, 2, 3, 4, 5, 6, 7, 8, 9 (1: very, very low; 9: very, very high).}
\noindent These ratings record participants' subjective judgments and remain separate from accuracy scoring. The rating fields appear for every Q1 answer.\par\smallskip
\usfield{Optional comment}{"Anything else about your answer? (optional)" Participants may add a free-text comment.}

\ussection{ComfyTy diagnosis (treatment only)}
\begin{lstlisting}[style=usreport]
OK  [100% typed (8/8)]
\end{lstlisting}

\ustask{T09\quad\texttt{g-blend-ok}}
\usfield{Task class}{Valid workflow}
\usfield{Source and construction}{Valid control for the resolution compatibility pattern in T02, validated by GPU execution. Two EmptyLatentImage nodes both use 512x512 resolution; their sampled outputs feed a LatentBlend node.}
\usfield{Upstream basis}{\href{https://github.com/Comfy-Org/ComfyUI/pull/1058}{ComfyUI pull request \#1058} introduces the two-input latent blending operation. This valid control uses two 512x512 inputs, so blending preserves their spatial dimensions.}
\usfield{Task description}{Generate a detailed photo of a golden retriever dog in a park (text-to-image).}
\usfield{Reference diagnosis}{The workflow is valid and its components are compatible.}

\ussection{Answer key}
\usfield{Q1. Decision}{No, it looks fine}
\usfield{Q2. Region}{The correct Q1 answer completes diagnosis.}
\usfield{Q3. Node}{The workflow's nodes are valid.}
\usfield{Q4--Q5. Repair}{Keep the valid workflow unchanged. Repair scoring applies to faulty workflows.}

\ussection{Workflow and observed behavior}
\begin{center}
\includegraphics[width=\linewidth,height=0.38\textheight,keepaspectratio]{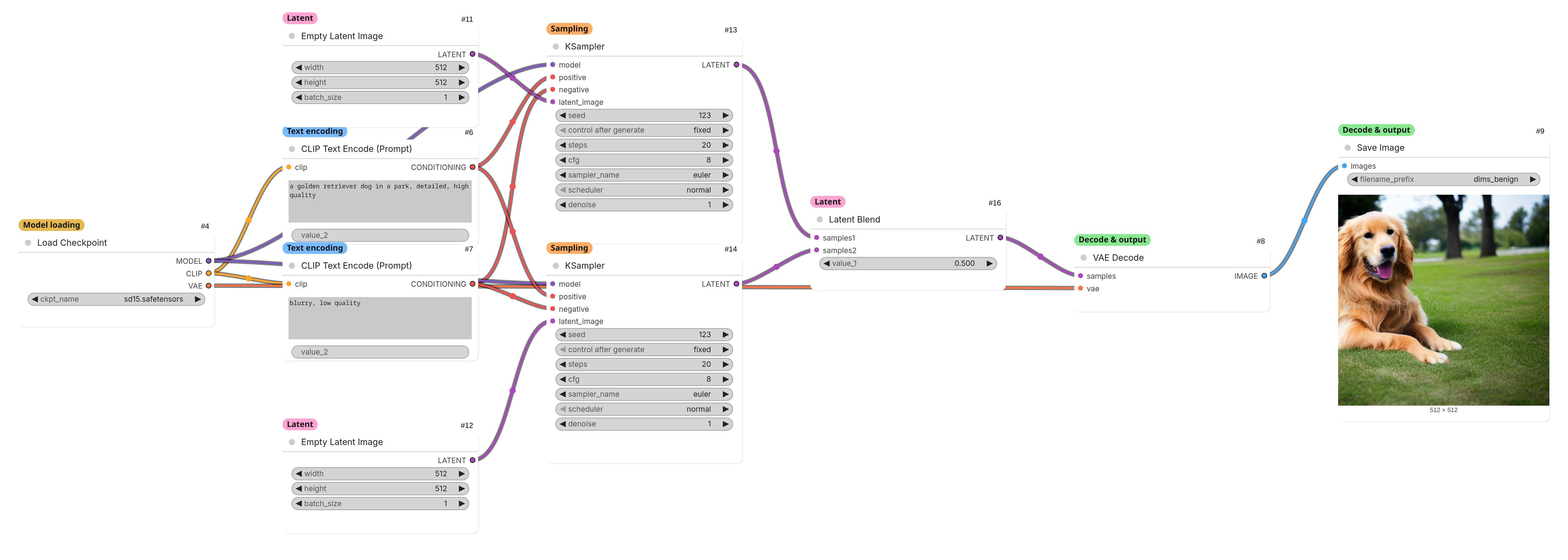}
\end{center}
\usfield{Graph}{10 nodes; 15 links. Node numbers connect the graph to the choices and diagnostic report.}
\usfield{Observed output}{The run completes and produces the following image. Participants inspect this output.}
\begin{center}
\includegraphics[width=0.25\linewidth,height=0.13\textheight,keepaspectratio]{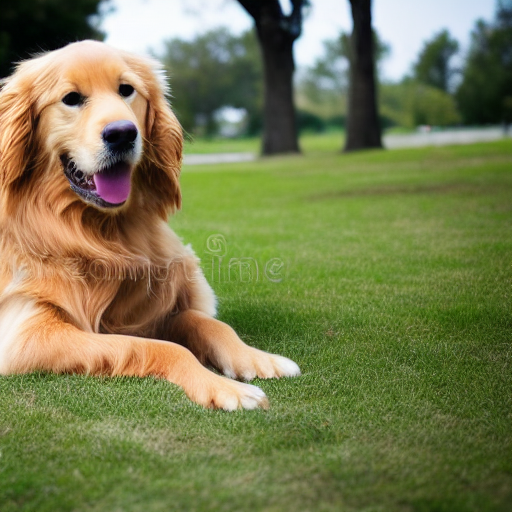}
\end{center}

\ussection{Model-file facts (shown in both versions)}
\usfield{Node \#4}{\textbf{sd15.safetensors}\newline Stable Diffusion 1.5. The classic text-to-image checkpoint; one file holding the UNet, the CLIP-L text encoder and the VAE. Trained at 512x512, and the base that most community 1.5 finetunes and LoRAs build on.}

\ussection{Q1. Fault decision}
\usfield{Question}{Is anything wrong with this workflow?}
\begin{itemize}
\item[$\bullet$] Yes, something is wrong
\item[$\bullet$] No, it looks fine\quad\textbf{(Accepted)}
\item[$\bullet$] Can't determine
\end{itemize}
\noindent A Yes answer opens Q2--Q5. The other answers lead to the participant ratings.\par\smallskip

\ussection{Q2--Q3. Localization choices}
\usfield{Q2. Region}{Which part is at fault? Pick ONE option.}
\usfield{Q3. Node}{Which node exactly? Pick the node you would change.}
\noindent Each region below lists its Q2 option followed by the complete Q3 node menu. Select one node within the selected region.\par\smallskip
\begin{enumerate}
\item[\texttt{R1}] Model loading --- CheckpointLoaderSimple \#4
\begin{itemize}
\item \#4 CheckpointLoaderSimple
\item I don't know
\end{itemize}
\item[\texttt{R2}] Latent sources \& blend --- EmptyLatentImage \#11, \#12, LatentBlend \#16
\begin{itemize}
\item \#11 EmptyLatentImage
\item \#12 EmptyLatentImage
\item \#16 LatentBlend
\item I don't know
\end{itemize}
\item[\texttt{R3}] The two samplers --- KSampler \#13, \#14
\begin{itemize}
\item \#13 KSampler
\item \#14 KSampler
\item I don't know
\end{itemize}
\item[\texttt{R4}] Prompt encoding --- CLIPTextEncode \#6, \#7
\begin{itemize}
\item \#6 CLIPTextEncode
\item \#7 CLIPTextEncode
\item I don't know
\end{itemize}
\item[\texttt{R5}] Decode \& save --- VAEDecode \#8, SaveImage \#9
\begin{itemize}
\item \#8 VAEDecode
\item \#9 SaveImage
\item I don't know
\end{itemize}
\end{enumerate}
\noindent A Yes answer opens these region and node choices and counts as a false alarm on this valid workflow.\par\smallskip

\ussection{Q4--Q5. Repair choices}
\usfield{Q4. Question}{What kind of change?}
\begin{enumerate}
\item[\texttt{A}] Swap the model/file it loads
\item[\texttt{B}] Change a setting/widget value
\item[\texttt{C}] Rewire connections
\item[\texttt{D}] Replace with a different node
\item[\texttt{E}] Other
\item[$\bullet$] I don't know
\end{enumerate}
\usfield{Q5. Branch behavior}{A Yes answer exposes Q4. Categories A--D and I don't know lead directly to the ratings; category E (Other) requires a free-text explanation. Q5 is omitted for this valid workflow. The correct response preserves the workflow unchanged.}

\ussection{Q6--Q8. Participant ratings}
\usfield{Q6. Confidence}{"How confident are you?" Options: 1, 2, 3, 4, 5 (1: not at all; 5: very).}
\usfield{Q7. Report helpfulness}{"How helpful was the checker report on this task?" Options: 1, 2, 3, 4, 5 (1: not helpful; 5: very helpful). Appears in the treatment version.}
\usfield{Q8. Mental effort}{"How much mental effort did this task take?" Options: 1, 2, 3, 4, 5, 6, 7, 8, 9 (1: very, very low; 9: very, very high).}
\noindent These ratings record participants' subjective judgments and remain separate from accuracy scoring. The rating fields appear for every Q1 answer.\par\smallskip
\usfield{Optional comment}{"Anything else about your answer? (optional)" Participants may add a free-text comment.}

\ussection{ComfyTy diagnosis (treatment only)}
\begin{lstlisting}[style=usreport]
OK  [100% typed (10/10)]
\end{lstlisting}

\ustask{T10\quad\texttt{g-modality-ok}}
\usfield{Task class}{Valid workflow}
\usfield{Source and construction}{Valid control for the modality compatibility pattern in T01, validated by GPU execution. An SD1.5 text-to-image pipeline pairs canny-edge conditioning with an SD1.5 canny ControlNet.}
\usfield{Upstream basis}{\href{https://github.com/comfyanonymous/ComfyUI_examples/tree/master/controlnet}{The official ComfyUI ControlNet examples} specify that the conditioning image format must match the ControlNet model. This valid control follows that requirement with canny maps and a canny ControlNet.}
\usfield{Task description}{Generate a dog photo whose composition follows the edge map extracted from a reference image (ControlNet).}
\usfield{Reference diagnosis}{The workflow is valid and its components are compatible.}

\ussection{Answer key}
\usfield{Q1. Decision}{No, it looks fine}
\usfield{Q2. Region}{The correct Q1 answer completes diagnosis.}
\usfield{Q3. Node}{The workflow's nodes are valid.}
\usfield{Q4--Q5. Repair}{Keep the valid workflow unchanged. Repair scoring applies to faulty workflows.}

\ussection{Workflow and observed behavior}
\begin{center}
\includegraphics[width=\linewidth,height=0.38\textheight,keepaspectratio]{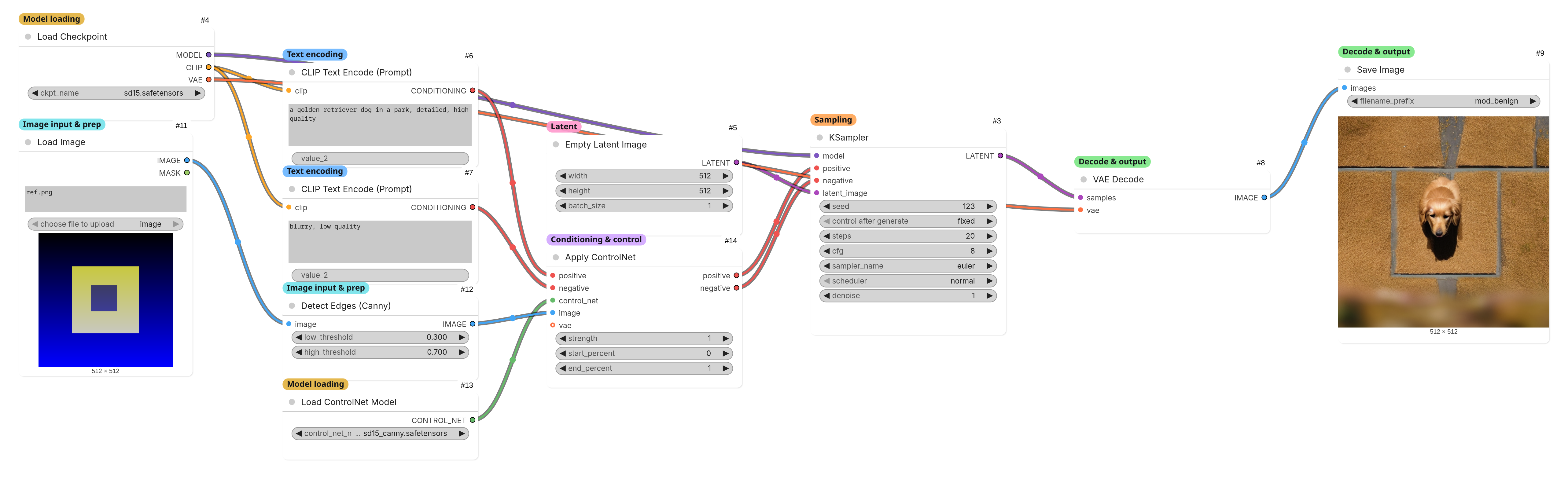}
\end{center}
\usfield{Graph}{11 nodes; 14 links. Node numbers connect the graph to the choices and diagnostic report.}
\usfield{Observed output}{The run completes and produces the following image. Participants inspect this output.}
\begin{center}
\includegraphics[width=0.25\linewidth,height=0.13\textheight,keepaspectratio]{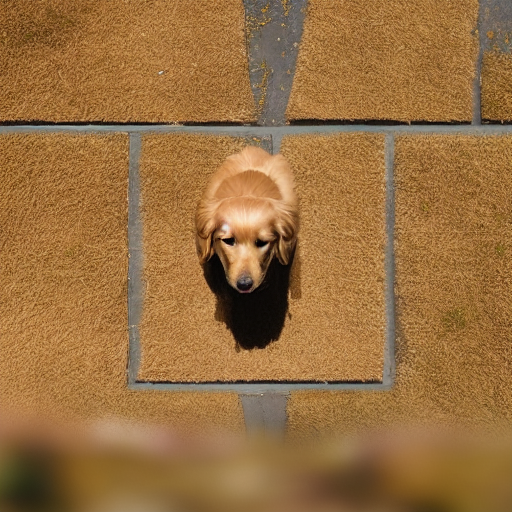}
\end{center}

\ussection{Model-file facts (shown in both versions)}
\usfield{Node \#4}{\textbf{sd15.safetensors}\newline Stable Diffusion 1.5. The classic text-to-image checkpoint; one file holding the UNet, the CLIP-L text encoder and the VAE. Trained at 512x512, and the base that most community 1.5 finetunes and LoRAs build on.}
\usfield{Node \#13}{\textbf{sd15\_canny.safetensors}\newline ControlNet 1.1, canny conditioning. Trained against Stable Diffusion 1.5 on roughly 3M edge-map and caption pairs. The control image it expects is a canny edge map.}

\ussection{Q1. Fault decision}
\usfield{Question}{Is anything wrong with this workflow?}
\begin{itemize}
\item[$\bullet$] Yes, something is wrong
\item[$\bullet$] No, it looks fine\quad\textbf{(Accepted)}
\item[$\bullet$] Can't determine
\end{itemize}
\noindent A Yes answer opens Q2--Q5. The other answers lead to the participant ratings.\par\smallskip

\ussection{Q2--Q3. Localization choices}
\usfield{Q2. Region}{Which part is at fault? Pick ONE option.}
\usfield{Q3. Node}{Which node exactly? Pick the node you would change.}
\noindent Each region below lists its Q2 option followed by the complete Q3 node menu. Select one node within the selected region.\par\smallskip
\begin{enumerate}
\item[\texttt{R1}] Model loading --- CheckpointLoaderSimple \#4
\begin{itemize}
\item \#4 CheckpointLoaderSimple
\item I don't know
\end{itemize}
\item[\texttt{R2}] ControlNet branch --- LoadImage \#11, Canny \#12, ControlNetLoader \#13, ControlNetApplyAdvanced \#14
\begin{itemize}
\item \#11 LoadImage
\item \#12 Canny
\item \#13 ControlNetLoader
\item \#14 ControlNetApplyAdvanced
\item I don't know
\end{itemize}
\item[\texttt{R3}] Prompt encoding --- CLIPTextEncode \#6, \#7
\begin{itemize}
\item \#6 CLIPTextEncode
\item \#7 CLIPTextEncode
\item I don't know
\end{itemize}
\item[\texttt{R4}] Sampling \& latent --- KSampler \#3, EmptyLatentImage \#5
\begin{itemize}
\item \#3 KSampler
\item \#5 EmptyLatentImage
\item I don't know
\end{itemize}
\item[\texttt{R5}] Decode \& save --- VAEDecode \#8, SaveImage \#9
\begin{itemize}
\item \#8 VAEDecode
\item \#9 SaveImage
\item I don't know
\end{itemize}
\end{enumerate}
\noindent A Yes answer opens these region and node choices and counts as a false alarm on this valid workflow.\par\smallskip

\ussection{Q4--Q5. Repair choices}
\usfield{Q4. Question}{What kind of change?}
\begin{enumerate}
\item[\texttt{A}] Swap the model/file it loads
\item[\texttt{B}] Change a setting/widget value
\item[\texttt{C}] Rewire connections
\item[\texttt{D}] Replace with a different node
\item[\texttt{E}] Other
\item[$\bullet$] I don't know
\end{enumerate}
\usfield{Q5. Branch behavior}{A Yes answer exposes Q4. Categories A--D and I don't know lead directly to the ratings; category E (Other) requires a free-text explanation. Q5 is omitted for this valid workflow. The correct response preserves the workflow unchanged.}

\ussection{Q6--Q8. Participant ratings}
\usfield{Q6. Confidence}{"How confident are you?" Options: 1, 2, 3, 4, 5 (1: not at all; 5: very).}
\usfield{Q7. Report helpfulness}{"How helpful was the checker report on this task?" Options: 1, 2, 3, 4, 5 (1: not helpful; 5: very helpful). Appears in the treatment version.}
\usfield{Q8. Mental effort}{"How much mental effort did this task take?" Options: 1, 2, 3, 4, 5, 6, 7, 8, 9 (1: very, very low; 9: very, very high).}
\noindent These ratings record participants' subjective judgments and remain separate from accuracy scoring. The rating fields appear for every Q1 answer.\par\smallskip
\usfield{Optional comment}{"Anything else about your answer? (optional)" Participants may add a free-text comment.}

\ussection{ComfyTy diagnosis (treatment only)}
\begin{lstlisting}[style=usreport]
OK  [100% typed (11/11)]
\end{lstlisting}

\ustask{T11\quad\texttt{c-5098-cnfam}}
\usfield{Task class}{Crash}
\usfield{Source}{Adapted from \href{https://github.com/Comfy-Org/ComfyUI/issues/5098}{ComfyUI issue \#5098}.}
\usfield{Task description}{Generate an image with canny-edge ControlNet guidance.}
\usfield{Reference diagnosis}{SD1.5 ControlNet (control\_v11p\_sd15\_canny) used with an SDXL checkpoint (juggernautXL); swap either the ControlNet to an SDXL one or the checkpoint to SD1.5 (maintainer-confirmed).}

\ussection{Answer key}
\usfield{Q1. Decision}{The task supplies the runtime failure; diagnosis starts at Q2.}
\usfield{Q2. Region}{R1: Checkpoint loading --- CheckpointLoaderSimple \#4; R2: ControlNet loading --- ControlNetLoader \#14}
\usfield{Q3. Node}{\#4 CheckpointLoaderSimple; \#14 ControlNetLoader}
\usfield{Q4. Category}{A: Swap the model/file it loads}
\usfield{Q5. Specific repair}{A1: change \#14 control\_v11p\_sd15\_canny\_fp16 to an SDXL canny ControlNet; A2: change \#4 juggernautXL to an SD1.5 checkpoint}

\ussection{Workflow and observed behavior}
\begin{center}
\includegraphics[width=\linewidth,height=0.38\textheight,keepaspectratio]{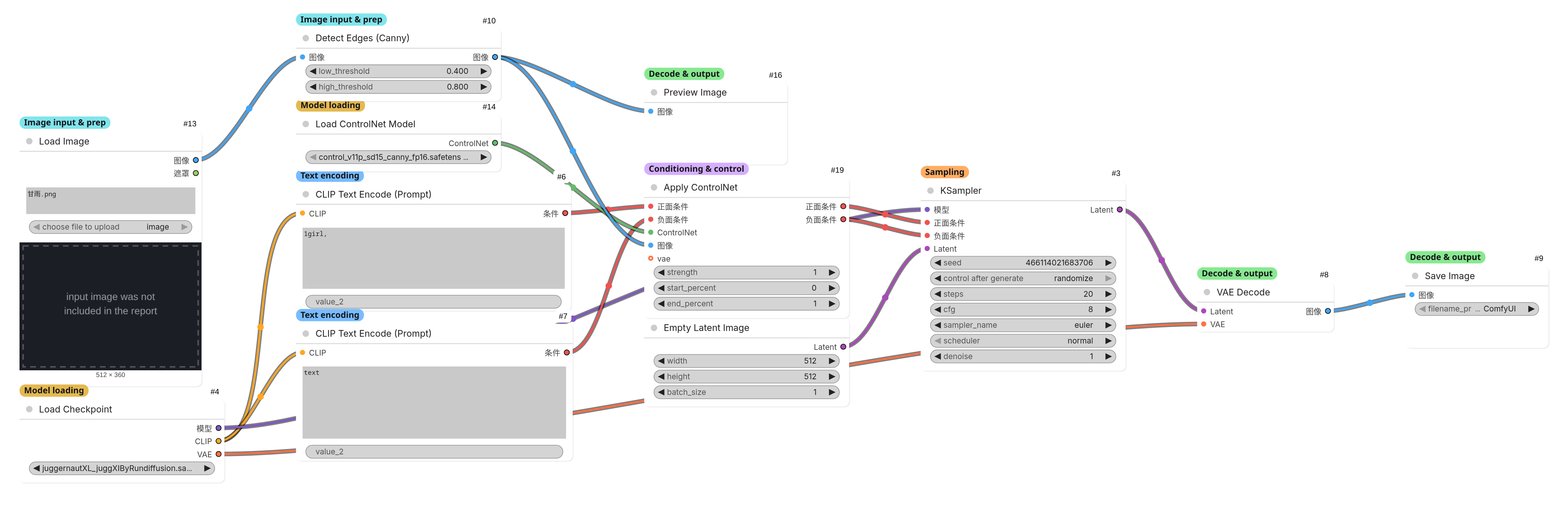}
\end{center}
\usfield{Graph}{12 nodes; 15 links. Node numbers connect the graph to the choices and diagnostic report.}
\usfield{Observed behavior}{ComfyUI reports the following runtime failure in both versions.}
\begin{lstlisting}[style=usreport]
# ComfyUI Error Report
## Error Details
- **Node ID:** 3
- **Node Type:** KSampler
- **Exception Type:** RuntimeError
- **Exception Message:** mat1 and mat2 shapes cannot be multiplied (154x2048 and 768x320)
\end{lstlisting}

\ussection{Model-file facts (shown in both versions)}
\usfield{Node \#4}{\textbf{juggernautXL\_juggXIByRundiffusion.safetensors}\newline Juggernaut XL v11, by RunDiffusion. Photorealistic community finetune of SDXL Base 1.0, tuned for portraits and product shots. Suggested settings: 832x1216 or 1024x1024, DPM++ 2M Karras, CFG 3-6.}
\usfield{Node \#14}{\textbf{control\_v11p\_sd15\_canny\_fp16.safetensors}\newline ControlNet 1.1 canny, fp16 build. The same weights as the full-precision release, stored at half precision to save VRAM. Base model: Stable Diffusion 1.5; the control image it expects is a canny edge map.}
\usfield{Input-image availability}{The study uses marked representative input images at nodes \#13 because the original inputs are unavailable in the source report.}

\ussection{Q1. Fault decision}
\noindent The error trace establishes the failure. Participants begin at Q2 and may select "Couldn't determine" to end the task.\par\smallskip

\ussection{Q2--Q3. Localization choices}
\usfield{Q2. Region}{1. Which part of the workflow caused the failure? Pick ONE option.}
\usfield{Q3. Node}{2. Your fix --- which node would you change?}
\noindent Each region below lists its Q2 option followed by the complete Q3 node menu. Select one node within the selected region.\par\smallskip
\begin{enumerate}
\item[\texttt{R1}] Checkpoint loading --- CheckpointLoaderSimple \#4\quad\textbf{(Accepted)}
\begin{itemize}
\item \#4 CheckpointLoaderSimple
\item I don't know
\end{itemize}
\item[\texttt{R2}] ControlNet loading --- ControlNetLoader \#14\quad\textbf{(Accepted)}
\begin{itemize}
\item \#14 ControlNetLoader
\item I don't know
\end{itemize}
\item[\texttt{R3}] Control image prep \& apply --- LoadImage \#13, Canny \#10, PreviewImage \#16, ControlNetApplyAdvanced \#19
\begin{itemize}
\item \#13 LoadImage
\item \#10 Canny
\item \#16 PreviewImage
\item \#19 ControlNetApplyAdvanced
\item I don't know
\end{itemize}
\item[\texttt{R4}] Prompt encoding --- CLIPTextEncode \#6, \#7
\begin{itemize}
\item \#6 CLIPTextEncode
\item \#7 CLIPTextEncode
\item I don't know
\end{itemize}
\item[\texttt{R5}] Sampling \& latent --- KSampler \#3, EmptyLatentImage \#5
\begin{itemize}
\item \#3 KSampler
\item \#5 EmptyLatentImage
\item I don't know
\end{itemize}
\item[\texttt{R6}] Decode \& save --- VAEDecode \#8, SaveImage \#9
\begin{itemize}
\item \#8 VAEDecode
\item \#9 SaveImage
\item I don't know
\end{itemize}
\end{enumerate}

\ussection{Q4--Q5. Repair choices}
\usfield{Q4. Question}{What kind of change?}
\begin{enumerate}
\item[\texttt{A}] Swap the model/file it loads\quad\textbf{(Accepted)}
\item[\texttt{B}] Change a setting/widget value
\item[\texttt{C}] Rewire connections
\item[\texttt{D}] Replace with a different node
\item[\texttt{E}] Other
\item[$\bullet$] I don't know
\end{enumerate}
\noindent Q5 displays the menu for the selected category A--D. Each menu also offers I don't know, which counts as an unsuccessful repair selection. Category E (Other) requires a free-text explanation. Selecting I don't know for Q4 proceeds directly to the participant ratings.\par\smallskip

\ussection{A. Swap the model/file it loads}
\usfield{Question}{Change the file to what?}
\begin{enumerate}
\item[\texttt{A1}] change \#14 control\_v11p\_sd15\_canny\_fp16 to an SDXL canny ControlNet\quad\textbf{(Accepted)}
\item[\texttt{A2}] change \#4 juggernautXL to an SD1.5 checkpoint\quad\textbf{(Accepted)}
\item[\texttt{A3}] change \#14 to an SD1.5 depth ControlNet
\item[\texttt{A4}] change \#14 to the full-precision build of the same ControlNet
\item[\texttt{A5}] change \#4 to a different SDXL checkpoint
\item[\texttt{A6}] change \#4 to the SDXL inpainting variant
\item[$\bullet$] I don't know
\end{enumerate}

\ussection{B. Change a setting/widget value}
\usfield{Question}{Which setting --- and to what?}
\begin{enumerate}
\item[\texttt{B1}] lower `strength` on ControlNetApplyAdvanced \#19 (now 1.0)
\item[\texttt{B2}] shorten `end\_percent` on \#19 (now 1.0)
\item[\texttt{B3}] tune the Canny \#10 thresholds (0.4 / 0.8)
\item[\texttt{B4}] raise EmptyLatentImage \#5 from 512 to 1024
\item[\texttt{B5}] increase `steps` on KSampler \#3 (now 20)
\item[$\bullet$] I don't know
\end{enumerate}

\ussection{C. Rewire connections}
\usfield{Question}{Which connection would you change?}
\begin{enumerate}
\item[\texttt{C1}] connect LoadImage \#13 straight into \#19 `image` (skip Canny \#10)
\item[\texttt{C2}] swap the positive/negative conditioning inputs of \#19
\item[\texttt{C3}] bypass the ControlNet branch entirely
\item[$\bullet$] I don't know
\end{enumerate}

\ussection{D. Replace with a different node}
\usfield{Question}{Replace it with what?}
\begin{enumerate}
\item[\texttt{D1}] replace Canny \#10 with a depth-map preprocessor
\item[\texttt{D2}] replace \#19 with the basic ControlNetApply node
\item[\texttt{D3}] replace ControlNetLoader \#14 with a T2I-Adapter loader
\item[$\bullet$] I don't know
\end{enumerate}

\ussection{Q6--Q8. Participant ratings}
\usfield{Q6. Confidence}{"How confident are you in this fix?" Options: 1, 2, 3, 4, 5 (1: not at all; 5: very).}
\usfield{Q7. Report helpfulness}{"How helpful was the checker report on this task?" Options: 1, 2, 3, 4, 5 (1: not helpful; 5: very helpful). Appears in the treatment version.}
\usfield{Q8. Mental effort}{"How much mental effort did this task take?" Options: 1, 2, 3, 4, 5, 6, 7, 8, 9 (1: very, very low; 9: very, very high).}
\noindent These ratings record participants' subjective judgments and remain separate from accuracy scoring. For crashes, the rating fields appear after region selection.\par\smallskip
\usfield{Optional comment}{"Anything else about your answer? (optional)" Participants may add a free-text comment.}

\ussection{ComfyTy diagnosis (treatment only)}
\begin{lstlisting}[style=usreport]
ILL-TYPED (brand)  [100% typed (12/12)]
    Model family/architecture conflict: SD15 vs SDXL.
    Involved (minimal core):
      - CheckpointLoaderSimple #4 (ckpt=SDXL via filename: juggernautXL_juggXIByRundiffusion.safetensors)
      - ControlNetLoader #14 (controlnet-arch=SD15 via filename: control_v11p_sd15_canny_fp16.safetensors)
      - edge CheckpointLoaderSimple #4 -> CLIPTextEncode #7.clip
      - edge CLIPTextEncode #7 -> ControlNetApplyAdvanced #19.negative
      - edge ControlNetLoader #14 -> ControlNetApplyAdvanced #19.control_net
    Possible fixes:
      - change the file at CheckpointLoaderSimple #4 (ckpt=SDXL via filename: juggernautXL_juggXIByRundiffusion.safetensors) to one matching the rest of the graph
      - change the file at ControlNetLoader #14 (controlnet-arch=SD15 via filename: control_v11p_sd15_canny_fp16.safetensors) to one matching the rest of the graph
\end{lstlisting}

\ustask{T12\quad\texttt{c-7335-ctrllora}}
\usfield{Task class}{Crash}
\usfield{Source}{Adapted from \href{https://github.com/Comfy-Org/ComfyUI/issues/7335}{ComfyUI issue \#7335}.}
\usfield{Task description}{Inpaint an image region with canny ControlNet guidance (SDXL).}
\usfield{Reference diagnosis}{A Control-LoRA (control-lora-canny) is applied while sampling with an SDXL *inpainting* model (9-channel UNet input); use a regular SDXL base for the ControlNet pass or a standard (non-Control-LoRA) ControlNet.}

\ussection{Answer key}
\usfield{Q1. Decision}{The task supplies the runtime failure; diagnosis starts at Q2.}
\usfield{Q2. Region}{R1: Checkpoint loading --- CheckpointLoaderSimple \#18; R2: ControlNet branch --- ControlNetLoader \#43, ControlNetApplyAdvanced \#44, CannyEdgePreprocessor \#45}
\usfield{Q3. Node}{\#18 CheckpointLoaderSimple; \#43 ControlNetLoader}
\usfield{Q4. Category}{A: Swap the model/file it loads}
\usfield{Q5. Specific repair}{A1: change \#43 control-lora-canny-rank256 to a standard (non-Control-LoRA) SDXL canny ControlNet; A2: change \#18 dreamshaperXL\_lightningInpaint to the non-inpaint (base) build}

\ussection{Workflow and observed behavior}
\begin{center}
\includegraphics[width=\linewidth,height=0.38\textheight,keepaspectratio]{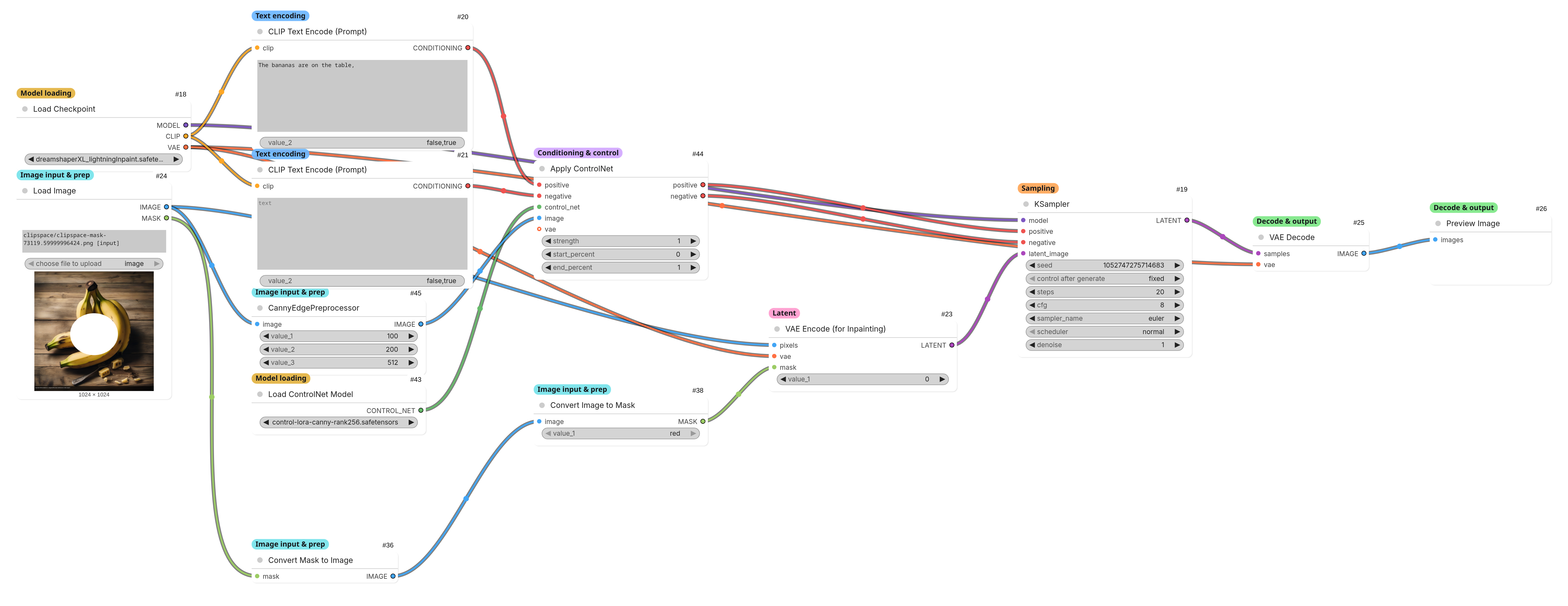}
\end{center}
\usfield{Graph}{13 nodes; 19 links. Node numbers connect the graph to the choices and diagnostic report.}
\usfield{Observed behavior}{ComfyUI reports the following runtime failure in both versions.}
\begin{lstlisting}[style=usreport]
!!! Exception during processing !!! shape '[320, 9, 3, 3]' is invalid for input of size 11520

# ComfyUI Error Report
## Error Details
- **Node ID:** 19
- **Node Type:** KSampler
- **Exception Type:** RuntimeError
- **Exception Message:** shape '[320, 9, 3, 3]' is invalid for input of size 11520
\end{lstlisting}

\ussection{Model-file facts (shown in both versions)}
\usfield{Node \#18}{\textbf{dreamshaperXL\_lightningInpaint.safetensors}\newline DreamShaper XL Lightning, inpainting edition. An SDXL finetune specialised for filling masked regions: its UNet takes nine input channels (four latent, one mask, four masked image) where a standard model takes four. Lightning schedule, 4-8 steps at CFG 1.5-2.}
\usfield{Node \#43}{\textbf{control-lora-canny-rank256.safetensors}\newline Control-LoRA, canny, rank 256, from Stability AI. A ControlNet distilled into LoRA form for SDXL: rather than running a parallel control network it patches the base model's own weights. Rank-128 and rank-256 variants are published.}

\ussection{Q1. Fault decision}
\noindent The error trace establishes the failure. Participants begin at Q2 and may select "Couldn't determine" to end the task.\par\smallskip

\ussection{Q2--Q3. Localization choices}
\usfield{Q2. Region}{1. Which part of the workflow caused the failure? Pick ONE option.}
\usfield{Q3. Node}{2. Your fix --- which node would you change?}
\noindent Each region below lists its Q2 option followed by the complete Q3 node menu. Select one node within the selected region.\par\smallskip
\begin{enumerate}
\item[\texttt{R1}] Checkpoint loading --- CheckpointLoaderSimple \#18\quad\textbf{(Accepted)}
\begin{itemize}
\item \#18 CheckpointLoaderSimple
\item I don't know
\end{itemize}
\item[\texttt{R2}] ControlNet branch --- ControlNetLoader \#43, ControlNetApplyAdvanced \#44, CannyEdgePreprocessor \#45\quad\textbf{(Accepted)}
\begin{itemize}
\item \#43 ControlNetLoader
\item \#44 ControlNetApplyAdvanced
\item \#45 CannyEdgePreprocessor
\item I don't know
\end{itemize}
\item[\texttt{R3}] Inpainting input \& mask --- LoadImage \#24, VAEEncodeForInpaint \#23, ImageToMask \#38, MaskToImage \#36
\begin{itemize}
\item \#24 LoadImage
\item \#23 VAEEncodeForInpaint
\item \#38 ImageToMask
\item \#36 MaskToImage
\item I don't know
\end{itemize}
\item[\texttt{R4}] Prompt encoding --- CLIPTextEncode \#20, \#21
\begin{itemize}
\item \#20 CLIPTextEncode
\item \#21 CLIPTextEncode
\item I don't know
\end{itemize}
\item[\texttt{R5}] Sampling \& decode --- KSampler \#19, VAEDecode \#25, PreviewImage \#26
\begin{itemize}
\item \#19 KSampler
\item \#25 VAEDecode
\item \#26 PreviewImage
\item I don't know
\end{itemize}
\end{enumerate}

\ussection{Q4--Q5. Repair choices}
\usfield{Q4. Question}{What kind of change?}
\begin{enumerate}
\item[\texttt{A}] Swap the model/file it loads\quad\textbf{(Accepted)}
\item[\texttt{B}] Change a setting/widget value
\item[\texttt{C}] Rewire connections
\item[\texttt{D}] Replace with a different node
\item[\texttt{E}] Other
\item[$\bullet$] I don't know
\end{enumerate}
\noindent Q5 displays the menu for the selected category A--D. Each menu also offers I don't know, which counts as an unsuccessful repair selection. Category E (Other) requires a free-text explanation. Selecting I don't know for Q4 proceeds directly to the participant ratings.\par\smallskip

\ussection{A. Swap the model/file it loads}
\usfield{Question}{Change the file to what?}
\begin{enumerate}
\item[\texttt{A1}] change \#43 control-lora-canny-rank256 to a standard (non-Control-LoRA) SDXL canny ControlNet\quad\textbf{(Accepted)}
\item[\texttt{A2}] change \#18 dreamshaperXL\_lightningInpaint to the non-inpaint (base) build\quad\textbf{(Accepted)}
\item[\texttt{A3}] change \#43 to an SD1.5 canny ControlNet
\item[\texttt{A4}] change \#43 to a higher-rank build of the Control-LoRA
\item[\texttt{A5}] change \#18 to an SD1.5 inpaint model
\item[$\bullet$] I don't know
\end{enumerate}

\ussection{B. Change a setting/widget value}
\usfield{Question}{Which setting --- and to what?}
\begin{enumerate}
\item[\texttt{B1}] lower `denoise` on KSampler \#19
\item[\texttt{B2}] change `grow\_mask\_by` on VAEEncodeForInpaint \#23
\item[\texttt{B3}] lower the ControlNet strength on \#44
\item[\texttt{B4}] increase the sampler `steps`
\item[$\bullet$] I don't know
\end{enumerate}

\ussection{C. Rewire connections}
\usfield{Question}{Which connection would you change?}
\begin{enumerate}
\item[\texttt{C1}] wire the mask from LoadImage \#24 directly (skip \#38/\#36)
\item[\texttt{C2}] wire \#24's image straight to the sampler branch
\item[\texttt{C3}] apply the ControlNet after the inpaint encode
\item[$\bullet$] I don't know
\end{enumerate}

\ussection{D. Replace with a different node}
\usfield{Question}{Replace it with what?}
\begin{enumerate}
\item[\texttt{D1}] replace ControlNetApplyAdvanced \#44 with basic ControlNetApply
\item[\texttt{D2}] replace VAEEncodeForInpaint \#23 with plain VAEEncode (drop the mask)
\item[\texttt{D3}] replace CannyEdgePreprocessor \#45 with a depth preprocessor
\item[$\bullet$] I don't know
\end{enumerate}

\ussection{Q6--Q8. Participant ratings}
\usfield{Q6. Confidence}{"How confident are you in this fix?" Options: 1, 2, 3, 4, 5 (1: not at all; 5: very).}
\usfield{Q7. Report helpfulness}{"How helpful was the checker report on this task?" Options: 1, 2, 3, 4, 5 (1: not helpful; 5: very helpful). Appears in the treatment version.}
\usfield{Q8. Mental effort}{"How much mental effort did this task take?" Options: 1, 2, 3, 4, 5, 6, 7, 8, 9 (1: very, very low; 9: very, very high).}
\noindent These ratings record participants' subjective judgments and remain separate from accuracy scoring. For crashes, the rating fields appear after region selection.\par\smallskip
\usfield{Optional comment}{"Anything else about your answer? (optional)" Participants may add a free-text comment.}

\ussection{ComfyTy diagnosis (treatment only)}
\begin{lstlisting}[style=usreport]
ILL-TYPED (brand)  [100% typed (13/13)]
    Model family/architecture conflict: SDXL_Inpaint vs ControlLora.
    Involved (minimal core):
      - KSampler #19
      - CheckpointLoaderSimple #18 (ckpt=SDXL_Inpaint via filename: dreamshaperXL_lightningInpaint.safetensors)
      - ControlNetApplyAdvanced #44
      - ControlNetLoader #43 (controlnet=Edge via filename: control-lora-canny-rank256.safetensors)
      - edge KSampler #19 -> VAEDecode #25.samples
      - edge CheckpointLoaderSimple #18 -> VAEDecode #25.vae
      - edge VAEDecode #25 -> PreviewImage #26.images
      - edge LoadImage #24 -> VAEEncodeForInpaint #23.pixels
      - edge CheckpointLoaderSimple #18 -> VAEEncodeForInpaint #23.vae
      - edge CheckpointLoaderSimple #18 -> KSampler #19.model
      - edge ControlNetApplyAdvanced #44 -> KSampler #19.positive
      - edge ControlNetApplyAdvanced #44 -> KSampler #19.negative
      - edge VAEEncodeForInpaint #23 -> KSampler #19.latent_image
      - edge LoadImage #24 -> CannyEdgePreprocessor #45.image
      - edge CLIPTextEncode #20 -> ControlNetApplyAdvanced #44.positive
      - edge CLIPTextEncode #21 -> ControlNetApplyAdvanced #44.negative
      - edge ControlNetLoader #43 -> ControlNetApplyAdvanced #44.control_net
      - edge CannyEdgePreprocessor #45 -> ControlNetApplyAdvanced #44.image
      - edge CheckpointLoaderSimple #18 -> CLIPTextEncode #20.clip
      - edge CheckpointLoaderSimple #18 -> CLIPTextEncode #21.clip
      - edge MaskToImage #36 -> ImageToMask #38.image
      - edge ImageToMask #38 -> VAEEncodeForInpaint #23.mask
      - edge LoadImage #24 -> MaskToImage #36.mask
    Possible fixes:
      - change the file at CheckpointLoaderSimple #18 (ckpt=SDXL_Inpaint via filename: dreamshaperXL_lightningInpaint.safetensors) to one matching the rest of the graph
      - change the file at ControlNetLoader #43 (controlnet=Edge via filename: control-lora-canny-rank256.safetensors) to one matching the rest of the graph
      - change the file at ControlNetLoader #43 (controlnet-arch=SDXL via filename: control-lora-canny-rank256.safetensors) to one matching the rest of the graph
      - change the file at ControlNetLoader #43 (controlnet-kind=ControlLora via filename: control-lora-canny-rank256.safetensors) to one matching the rest of the graph
\end{lstlisting}
\endgroup

%% file: sections/user-study-format.tex
\color{black}
\newcommand{\ustask}[1]{\par\addvspace{12pt}\noindent
  {\normalsize\bfseries #1}\par\nobreak\smallskip}
\newcommand{\ussection}[1]{\par\addvspace{6pt}\noindent
  {\bfseries #1}\par\nobreak\smallskip}
\newcommand{\usfield}[2]{\par\noindent\textbf{#1.} #2\par\smallskip}
\setlist[enumerate]{leftmargin=3em,labelsep=0.6em,
  topsep=3pt,itemsep=2pt,parsep=0pt,partopsep=0pt}
\setlist[itemize]{leftmargin=1.5em,
  topsep=2pt,itemsep=1pt,parsep=0pt,partopsep=0pt}
\lstdefinestyle{usreport}{language={},basicstyle=\footnotesize\ttfamily\color{black},
  columns=fullflexible,keepspaces=true,breaklines=true,breakatwhitespace=false,
  showstringspaces=false,numbers=none,frame=none,
  backgroundcolor=\color{white},aboveskip=3pt,belowskip=3pt,
  xleftmargin=0pt,xrightmargin=0pt,resetmargins=true}

%% file: study/data/user_study_results.tex
\begin{table}[!htbp]
\caption{\textbf{Per-trial user-study results.} Pnn/Tnn: user/task;
B/\underline{T}: baseline/treatment (without/with diagnostics).
Bits: decision, region, node, repair type, fix, overall (1/0: correct/incorrect).
*: timeout (zero credit); darker green: more correct bits;
summary pairs: baseline:treatment successes/trials.}
\label{tab:user-study}
\centering
\begingroup
\setlength{\tabcolsep}{0.85pt}
\renewcommand{\arraystretch}{1.15}
\definecolor{usgreen}{HTML}{238B45}
\newcommand{\uscell}[5]{\cellcolor{usgreen!#5}\if T#1\underline{\textbf{#1}#2#3}\else\textbf{#1}#2#3\fi\rlap{#4}}
\newcommand{\usvalue}[3]{\if T#3\underline{#1}\else#1\fi}
\scriptsize
\begin{tabular}{@{}l*{14}{c}@{}}
\toprule
& \multicolumn{6}{c}{Silent faults} & \multicolumn{4}{c}{Valid workflows}
& \multicolumn{2}{c}{Crashes} & \multicolumn{2}{c}{Summary} \\
\cmidrule(lr){2-7}\cmidrule(lr){8-11}\cmidrule(lr){12-13}\cmidrule(lr){14-15}
User & T01 & T02 & T03 & T04 & T05 & T06 & T07 & T08 & T09 & T10 & T11 & T12 & Solved & Repair \\
\midrule
P01 & \uscell{B}{000}{000}{}{0} & \uscell{B}{111}{100}{}{30} & \uscell{T}{111}{111}{}{45} & \uscell{T}{111}{100}{\textsuperscript{*}}{0} & \uscell{B}{000}{000}{}{0} & \uscell{T}{111}{111}{}{45} & \uscell{B}{111}{111}{}{45} & \uscell{T}{111}{111}{}{45} & \uscell{T}{111}{111}{}{45} & \uscell{B}{111}{111}{}{45} & \uscell{T}{111}{111}{}{45} & \uscell{B}{000}{000}{}{0} & \usvalue{3/6}{23}{B}:\usvalue{5/6}{38}{T} & \usvalue{0/4}{0}{B}:\usvalue{3/4}{34}{T} \\
P02 & \uscell{B}{000}{000}{}{0} & \uscell{T}{111}{000}{}{23} & \uscell{B}{100}{110}{\textsuperscript{*}}{0} & \uscell{T}{111}{111}{}{45} & \uscell{B}{111}{111}{}{45} & \uscell{T}{111}{111}{}{45} & \uscell{B}{111}{111}{}{45} & \uscell{B}{111}{111}{}{45} & \uscell{T}{111}{111}{}{45} & \uscell{T}{111}{111}{}{45} & \uscell{T}{111}{111}{}{45} & \uscell{B}{111}{111}{}{45} & \usvalue{4/6}{30}{B}:\usvalue{6/6}{45}{T} & \usvalue{2/4}{23}{B}:\usvalue{3/4}{34}{T} \\
P03 & \uscell{T}{111}{111}{}{45} & \uscell{T}{111}{111}{}{45} & \uscell{B}{100}{000}{}{8} & \uscell{T}{111}{000}{}{23} & \uscell{B}{100}{000}{}{8} & \uscell{B}{111}{111}{}{45} & \uscell{T}{111}{111}{}{45} & \uscell{T}{111}{111}{}{45} & \uscell{B}{111}{111}{}{45} & \uscell{B}{000}{000}{}{0} & \uscell{B}{000}{000}{}{0} & \uscell{T}{111}{100}{}{30} & \usvalue{2/6}{15}{B}:\usvalue{6/6}{45}{T} & \usvalue{1/4}{11}{B}:\usvalue{2/4}{23}{T} \\
P04 & \uscell{B}{000}{000}{}{0} & \uscell{T}{111}{111}{}{45} & \uscell{B}{100}{000}{}{8} & \uscell{T}{111}{100}{}{30} & \uscell{B}{111}{111}{}{45} & \uscell{T}{111}{111}{}{45} & \uscell{B}{111}{111}{}{45} & \uscell{T}{111}{111}{}{45} & \uscell{B}{111}{111}{}{45} & \uscell{T}{111}{111}{}{45} & \uscell{T}{111}{111}{}{45} & \uscell{B}{100}{000}{}{8} & \usvalue{3/6}{23}{B}:\usvalue{6/6}{45}{T} & \usvalue{1/4}{11}{B}:\usvalue{3/4}{34}{T} \\
P05 & \uscell{T}{111}{111}{}{45} & \uscell{B}{111}{000}{}{23} & \uscell{T}{111}{111}{}{45} & \uscell{B}{111}{000}{\textsuperscript{*}}{0} & \uscell{T}{111}{111}{}{45} & \uscell{B}{111}{111}{}{45} & \uscell{B}{111}{111}{}{45} & \uscell{T}{111}{111}{}{45} & \uscell{B}{111}{111}{}{45} & \uscell{T}{111}{111}{}{45} & \uscell{T}{111}{111}{}{45} & \uscell{B}{100}{000}{}{8} & \usvalue{4/6}{30}{B}:\usvalue{6/6}{45}{T} & \usvalue{1/4}{11}{B}:\usvalue{4/4}{45}{T} \\
P06 & \uscell{T}{100}{100}{}{15} & \uscell{B}{110}{000}{}{15} & \uscell{B}{000}{000}{}{0} & \uscell{T}{100}{000}{}{8} & \uscell{B}{000}{000}{}{0} & \uscell{T}{111}{100}{}{30} & \uscell{T}{111}{111}{}{45} & \uscell{B}{111}{111}{}{45} & \uscell{T}{111}{111}{}{45} & \uscell{B}{111}{111}{}{45} & \uscell{B}{100}{000}{}{8} & \uscell{T}{111}{100}{}{30} & \usvalue{2/6}{15}{B}:\usvalue{4/6}{30}{T} & \usvalue{0/4}{0}{B}:\usvalue{0/4}{0}{T} \\
P07 & \uscell{B}{100}{000}{}{8} & \uscell{T}{111}{111}{}{45} & \uscell{T}{100}{000}{}{8} & \uscell{B}{100}{100}{}{15} & \uscell{T}{111}{111}{}{45} & \uscell{B}{111}{111}{}{45} & \uscell{T}{111}{111}{}{45} & \uscell{T}{111}{111}{}{45} & \uscell{B}{111}{111}{}{45} & \uscell{B}{111}{111}{}{45} & \uscell{T}{111}{111}{}{45} & \uscell{B}{111}{111}{}{45} & \usvalue{4/6}{30}{B}:\usvalue{5/6}{38}{T} & \usvalue{2/4}{23}{B}:\usvalue{3/4}{34}{T} \\
P08 & \uscell{T}{111}{111}{}{45} & \uscell{B}{111}{111}{}{45} & \uscell{T}{111}{000}{\textsuperscript{*}}{0} & \uscell{B}{100}{100}{}{15} & \uscell{T}{111}{111}{}{45} & \uscell{B}{111}{111}{}{45} & \uscell{B}{111}{111}{}{45} & \uscell{B}{111}{111}{}{45} & \uscell{T}{111}{111}{}{45} & \uscell{T}{111}{111}{}{45} & \uscell{B}{111}{111}{}{45} & \uscell{T}{111}{111}{}{45} & \usvalue{5/6}{38}{B}:\usvalue{5/6}{38}{T} & \usvalue{3/4}{34}{B}:\usvalue{3/4}{34}{T} \\
P09 & \uscell{B}{110}{000}{}{15} & \uscell{B}{111}{111}{}{45} & \uscell{T}{111}{000}{}{23} & \uscell{B}{100}{000}{}{8} & \uscell{T}{111}{111}{}{45} & \uscell{T}{111}{111}{}{45} & \uscell{T}{111}{111}{}{45} & \uscell{B}{111}{111}{}{45} & \uscell{B}{111}{111}{}{45} & \uscell{T}{111}{111}{}{45} & \uscell{B}{100}{100}{}{15} & \uscell{T}{111}{111}{}{45} & \usvalue{3/6}{23}{B}:\usvalue{6/6}{45}{T} & \usvalue{1/4}{11}{B}:\usvalue{3/4}{34}{T} \\
P10 & \uscell{T}{111}{100}{}{30} & \uscell{T}{111}{000}{}{23} & \uscell{B}{100}{000}{}{8} & \uscell{B}{100}{100}{}{15} & \uscell{T}{111}{100}{}{30} & \uscell{B}{111}{111}{}{45} & \uscell{T}{111}{111}{}{45} & \uscell{B}{111}{111}{}{45} & \uscell{T}{111}{111}{}{45} & \uscell{B}{111}{111}{}{45} & \uscell{B}{100}{000}{}{8} & \uscell{T}{111}{111}{}{45} & \usvalue{3/6}{23}{B}:\usvalue{6/6}{45}{T} & \usvalue{1/4}{11}{B}:\usvalue{1/4}{11}{T} \\
\midrule
\multicolumn{5}{@{}l}{Overall summary}
& \multicolumn{4}{c}{Without diagnostics}
& \multicolumn{4}{c}{With diagnostics}
& \multicolumn{2}{c}{Change (pp)} \\
\midrule
\multicolumn{5}{@{}l}{Correct diagnosis (all workflows)}
& \multicolumn{4}{c}{\usvalue{33/60 (55.0\%)}{25}{B}}
& \multicolumn{4}{c}{\usvalue{55/60 (91.7\%)}{41}{T}}
& \multicolumn{2}{c}{+36.7} \\
\multicolumn{5}{@{}l}{Fault localization (faulty workflows)}
& \multicolumn{4}{c}{\usvalue{14/40 (35.0\%)}{16}{B}}
& \multicolumn{4}{c}{\usvalue{35/40 (87.5\%)}{39}{T}}
& \multicolumn{2}{c}{+52.5} \\
\multicolumn{5}{@{}l}{Valid repair selection (faulty workflows)}
& \multicolumn{4}{c}{\usvalue{12/40 (30.0\%)}{14}{B}}
& \multicolumn{4}{c}{\usvalue{25/40 (62.5\%)}{28}{T}}
& \multicolumn{2}{c}{+32.5} \\
\bottomrule
\end{tabular}
\endgroup
\end{table}

%% file: study/data/user_study_details.tex
This section reports the submitted choices and subjective ratings for the 120 scored trials. B and T denote baseline and treatment. The results cover the 12 study tasks and exclude practice trials.
\begingroup
\small
\color{black}
\setlength{\tabcolsep}{4pt}\setlength{\LTleft}{\fill}\setlength{\LTright}{\fill}
\renewcommand{\arraystretch}{1.15}
\par\medskip\noindent\textbf{Answer choices.} Table~\ref{tab:user-choice-details} shows the answers for each participant and task, using the anonymous participant IDs in Table~\ref{tab:user-study}. R1, R2, \ldots\ identify regions; \#N identifies a node; A--E identify repair categories; and A1, A2, \ldots\ identify concrete repairs, using the labels in Appendix~\ref{app:user-study:tasks}. Bold labels mark accepted individual answers. A correct individual choice can accompany errors at other stages. DK means `I don't know'; NE marks a field omitted by question branching; UA marks an unanswered field in an open branch. An em dash marks a question absent from the task: Q1 for crashes and Q5 for valid workflows. Timeout and give-up submissions retain their partial answers; Status records how each trial ends.\par
\begin{longtable}{@{}lllllllll@{}}
\caption{Submitted answer choices.}\label{tab:user-choice-details}\\
\toprule
Task & ID & B/T & Q1: fault & Q2: region & Q3: node & Q4: category & Q5: repair & Status \\
\midrule
\endfirsthead
\multicolumn{9}{l}{\tablename~\thetable\ (continued)}\\
\toprule
Task & ID & B/T & Q1: fault & Q2: region & Q3: node & Q4: category & Q5: repair & Status \\
\midrule
\endhead
\midrule
\multicolumn{9}{r}{Continued on next page}\\
\endfoot
\bottomrule
\endlastfoot
T01 & P01 & B & No & NE & NE & NE & NE & Submit \\
T01 & P02 & B & No & NE & NE & NE & NE & Submit \\
T01 & P03 & T & \textbf{Yes} & \textbf{R2} & \textbf{\#13} & \textbf{A} & \textbf{A1} & Submit \\
T01 & P04 & B & No & NE & NE & NE & NE & Submit \\
T01 & P05 & T & \textbf{Yes} & \textbf{R2} & \textbf{\#13} & \textbf{A} & \textbf{A1} & Submit \\
T01 & P06 & T & \textbf{Yes} & R1 & DK & \textbf{A} & A2 & Submit \\
T01 & P07 & B & \textbf{Yes} & R4 & \#3 & B & B4 & Submit \\
T01 & P08 & T & \textbf{Yes} & \textbf{R2} & \textbf{\#13} & \textbf{A} & \textbf{A1} & Submit \\
T01 & P09 & B & \textbf{Yes} & \textbf{R2} & \#14 & C & DK & Submit \\
T01 & P10 & T & \textbf{Yes} & \textbf{R2} & \textbf{\#13} & \textbf{A} & A5 & Submit \\
\midrule
T02 & P01 & B & \textbf{Yes} & \textbf{R2} & \textbf{\#12} & \textbf{B} & B2 & Submit \\
T02 & P02 & T & \textbf{Yes} & \textbf{R2} & \textbf{\#16} & C & C2 & Submit \\
T02 & P03 & T & \textbf{Yes} & \textbf{R2} & \textbf{\#12} & \textbf{B} & \textbf{B1} & Submit \\
T02 & P04 & T & \textbf{Yes} & \textbf{R2} & \textbf{\#12} & \textbf{B} & \textbf{B1} & Submit \\
T02 & P05 & B & \textbf{Yes} & \textbf{R2} & \textbf{\#16} & C & C2 & Submit \\
T02 & P06 & B & \textbf{Yes} & \textbf{R2} & DK & DK & NE & Submit \\
T02 & P07 & T & \textbf{Yes} & \textbf{R2} & \textbf{\#12} & \textbf{B} & \textbf{B1} & Submit \\
T02 & P08 & B & \textbf{Yes} & \textbf{R2} & \textbf{\#12} & \textbf{B} & \textbf{B1} & Submit \\
T02 & P09 & B & \textbf{Yes} & \textbf{R2} & \textbf{\#12} & \textbf{B} & \textbf{B1} & Submit \\
T02 & P10 & T & \textbf{Yes} & \textbf{R2} & \textbf{\#16} & D & D2 & Submit \\
\midrule
T03 & P01 & T & \textbf{Yes} & \textbf{R1} & \textbf{\#12} & \textbf{B} & \textbf{B1} & Submit \\
T03 & P02 & B & \textbf{Yes} & R5 & \#10 & \textbf{B} & \textbf{B1} & Timeout \\
T03 & P03 & B & \textbf{Yes} & R5 & \#73 & C & C2 & Submit \\
T03 & P04 & B & \textbf{Yes} & R5 & \#10 & A & A1 & Submit \\
T03 & P05 & T & \textbf{Yes} & \textbf{R1} & \textbf{\#12} & \textbf{B} & \textbf{B1} & Submit \\
T03 & P06 & B & Can't determine & NE & NE & NE & NE & Submit \\
T03 & P07 & T & \textbf{Yes} & R5 & \#10 & A & DK & Submit \\
T03 & P08 & T & \textbf{Yes} & \textbf{R1} & \textbf{\#12} & A & A3 & Timeout \\
T03 & P09 & T & \textbf{Yes} & \textbf{R1} & \textbf{\#12} & A & A4 & Submit \\
T03 & P10 & B & \textbf{Yes} & R5 & \#10 & A & A2 & Submit \\
\midrule
T04 & P01 & T & \textbf{Yes} & \textbf{R1} & \textbf{\#4} & \textbf{A} & A2 & Timeout \\
T04 & P02 & T & \textbf{Yes} & \textbf{R1} & \textbf{\#4} & \textbf{D} & \textbf{D1} & Submit \\
T04 & P03 & T & \textbf{Yes} & \textbf{R1} & \textbf{\#4} & C & C3 & Submit \\
T04 & P04 & T & \textbf{Yes} & \textbf{R1} & \textbf{\#4} & \textbf{A} & A2 & Submit \\
T04 & P05 & B & \textbf{Yes} & \textbf{R1} & \textbf{\#4} & B & UA & Timeout \\
T04 & P06 & T & \textbf{Yes} & R5 & \#3 & C & C3 & Submit \\
T04 & P07 & B & \textbf{Yes} & R4 & \#29 & \textbf{A} & A3 & Submit \\
T04 & P08 & B & \textbf{Yes} & R3 & \#56 & \textbf{A} & A5 & Submit \\
T04 & P09 & B & \textbf{Yes} & R2 & DK & DK & NE & Submit \\
T04 & P10 & B & \textbf{Yes} & R6 & \#43 & \textbf{A} & A2 & Submit \\
\midrule
T05 & P01 & B & No & NE & NE & NE & NE & Submit \\
T05 & P02 & B & \textbf{Yes} & \textbf{R2} & \textbf{\#10} & \textbf{A} & \textbf{A1} & Submit \\
T05 & P03 & B & \textbf{Yes} & R4 & \#3 & B & B3 & Submit \\
T05 & P04 & B & \textbf{Yes} & \textbf{R2} & \textbf{\#10} & \textbf{A} & \textbf{A1} & Submit \\
T05 & P05 & T & \textbf{Yes} & \textbf{R2} & \textbf{\#10} & \textbf{A} & \textbf{A1} & Submit \\
T05 & P06 & B & No & NE & NE & NE & NE & Submit \\
T05 & P07 & T & \textbf{Yes} & \textbf{R2} & \textbf{\#10} & \textbf{A} & \textbf{A1} & Submit \\
T05 & P08 & T & \textbf{Yes} & \textbf{R2} & \textbf{\#10} & \textbf{A} & \textbf{A1} & Submit \\
T05 & P09 & T & \textbf{Yes} & \textbf{R2} & \textbf{\#10} & \textbf{A} & \textbf{A1} & Submit \\
T05 & P10 & T & \textbf{Yes} & \textbf{R1} & \textbf{\#4} & \textbf{A} & DK & Submit \\
\midrule
T06 & P01 & T & \textbf{Yes} & \textbf{R2} & \textbf{\#20} & \textbf{A} & \textbf{A1} & Submit \\
T06 & P02 & T & \textbf{Yes} & \textbf{R2} & \textbf{\#20} & \textbf{A} & \textbf{A1} & Submit \\
T06 & P03 & B & \textbf{Yes} & \textbf{R2} & \textbf{\#20} & \textbf{A} & \textbf{A1} & Submit \\
T06 & P04 & T & \textbf{Yes} & \textbf{R2} & \textbf{\#20} & \textbf{A} & \textbf{A1} & Submit \\
T06 & P05 & B & \textbf{Yes} & \textbf{R2} & \textbf{\#20} & \textbf{C} & \textbf{C1} & Submit \\
T06 & P06 & T & \textbf{Yes} & \textbf{R2} & \textbf{\#20} & \textbf{A} & A2 & Submit \\
T06 & P07 & B & \textbf{Yes} & \textbf{R2} & \textbf{\#20} & \textbf{A} & \textbf{A1} & Submit \\
T06 & P08 & B & \textbf{Yes} & \textbf{R5} & \textbf{\#8} & \textbf{C} & \textbf{C1} & Submit \\
T06 & P09 & T & \textbf{Yes} & \textbf{R2} & \textbf{\#20} & \textbf{A} & \textbf{A1} & Submit \\
T06 & P10 & B & \textbf{Yes} & \textbf{R2} & \textbf{\#20} & \textbf{A} & \textbf{A1} & Submit \\
\midrule
T07 & P01 & B & \textbf{No} & NE & NE & NE & --- & Submit \\
T07 & P02 & B & \textbf{No} & NE & NE & NE & --- & Submit \\
T07 & P03 & T & \textbf{No} & NE & NE & NE & --- & Submit \\
T07 & P04 & B & \textbf{No} & NE & NE & NE & --- & Submit \\
T07 & P05 & B & \textbf{No} & NE & NE & NE & --- & Submit \\
T07 & P06 & T & \textbf{No} & NE & NE & NE & --- & Submit \\
T07 & P07 & T & \textbf{No} & NE & NE & NE & --- & Submit \\
T07 & P08 & B & \textbf{No} & NE & NE & NE & --- & Submit \\
T07 & P09 & T & \textbf{No} & NE & NE & NE & --- & Submit \\
T07 & P10 & T & \textbf{No} & NE & NE & NE & --- & Submit \\
\midrule
T08 & P01 & T & \textbf{No} & NE & NE & NE & --- & Submit \\
T08 & P02 & B & \textbf{No} & NE & NE & NE & --- & Submit \\
T08 & P03 & T & \textbf{No} & NE & NE & NE & --- & Submit \\
T08 & P04 & T & \textbf{No} & NE & NE & NE & --- & Submit \\
T08 & P05 & T & \textbf{No} & NE & NE & NE & --- & Submit \\
T08 & P06 & B & \textbf{No} & NE & NE & NE & --- & Submit \\
T08 & P07 & T & \textbf{No} & NE & NE & NE & --- & Submit \\
T08 & P08 & B & \textbf{No} & NE & NE & NE & --- & Submit \\
T08 & P09 & B & \textbf{No} & NE & NE & NE & --- & Submit \\
T08 & P10 & B & \textbf{No} & NE & NE & NE & --- & Submit \\
\midrule
T09 & P01 & T & \textbf{No} & NE & NE & NE & --- & Submit \\
T09 & P02 & T & \textbf{No} & NE & NE & NE & --- & Submit \\
T09 & P03 & B & \textbf{No} & NE & NE & NE & --- & Submit \\
T09 & P04 & B & \textbf{No} & NE & NE & NE & --- & Submit \\
T09 & P05 & B & \textbf{No} & NE & NE & NE & --- & Submit \\
T09 & P06 & T & \textbf{No} & NE & NE & NE & --- & Submit \\
T09 & P07 & B & \textbf{No} & NE & NE & NE & --- & Submit \\
T09 & P08 & T & \textbf{No} & NE & NE & NE & --- & Submit \\
T09 & P09 & B & \textbf{No} & NE & NE & NE & --- & Submit \\
T09 & P10 & T & \textbf{No} & NE & NE & NE & --- & Submit \\
\midrule
T10 & P01 & B & \textbf{No} & NE & NE & NE & --- & Submit \\
T10 & P02 & T & \textbf{No} & NE & NE & NE & --- & Submit \\
T10 & P03 & B & Yes & R2 & \#12 & B & --- & Submit \\
T10 & P04 & T & \textbf{No} & NE & NE & NE & --- & Submit \\
T10 & P05 & T & \textbf{No} & NE & NE & NE & --- & Submit \\
T10 & P06 & B & \textbf{No} & NE & NE & NE & --- & Submit \\
T10 & P07 & B & \textbf{No} & NE & NE & NE & --- & Submit \\
T10 & P08 & T & \textbf{No} & NE & NE & NE & --- & Submit \\
T10 & P09 & T & \textbf{No} & NE & NE & NE & --- & Submit \\
T10 & P10 & B & \textbf{No} & NE & NE & NE & --- & Submit \\
\midrule
T11 & P01 & T & --- & \textbf{R2} & \textbf{\#14} & \textbf{A} & \textbf{A1} & Submit \\
T11 & P02 & T & --- & \textbf{R2} & \textbf{\#14} & \textbf{A} & \textbf{A1} & Submit \\
T11 & P03 & B & --- & UA & NE & NE & NE & Give up \\
T11 & P04 & T & --- & \textbf{R1} & \textbf{\#4} & \textbf{A} & \textbf{A2} & Submit \\
T11 & P05 & T & --- & \textbf{R1} & \textbf{\#4} & \textbf{A} & \textbf{A2} & Submit \\
T11 & P06 & B & --- & R3 & \#10 & B & B3 & Submit \\
T11 & P07 & T & --- & \textbf{R2} & \textbf{\#14} & \textbf{A} & \textbf{A1} & Submit \\
T11 & P08 & B & --- & \textbf{R2} & \textbf{\#14} & \textbf{A} & \textbf{A1} & Submit \\
T11 & P09 & B & --- & R3 & \#13 & \textbf{A} & DK & Submit \\
T11 & P10 & B & --- & R5 & \#3 & C & C2 & Submit \\
\midrule
T12 & P01 & B & --- & R4 & UA & E & Free text & Give up \\
T12 & P02 & B & --- & \textbf{R2} & \textbf{\#43} & \textbf{A} & \textbf{A1} & Submit \\
T12 & P03 & T & --- & \textbf{R2} & \textbf{\#43} & \textbf{A} & A3 & Submit \\
T12 & P04 & B & --- & R5 & \#19 & B & B2 & Submit \\
T12 & P05 & B & --- & R3 & \#24 & B & DK & Submit \\
T12 & P06 & T & --- & \textbf{R2} & \textbf{\#43} & \textbf{A} & A5 & Submit \\
T12 & P07 & B & --- & \textbf{R2} & \textbf{\#43} & \textbf{A} & \textbf{A1} & Submit \\
T12 & P08 & T & --- & \textbf{R2} & \textbf{\#43} & \textbf{A} & \textbf{A1} & Submit \\
T12 & P09 & T & --- & \textbf{R2} & \textbf{\#43} & \textbf{A} & \textbf{A1} & Submit \\
T12 & P10 & T & --- & \textbf{R2} & \textbf{\#43} & \textbf{A} & \textbf{A1} & Submit \\
\end{longtable}
\par\medskip\noindent\textbf{Participant ratings.} Table~\ref{tab:user-task-ratings} shows each participant's ratings for all 12 tasks. Confidence (Conf.) ranges from 1 (not at all confident) to 5 (very confident); helpfulness ranges from 1 (not helpful) to 5 (very helpful); mental effort ranges from 1 (very, very low) to 9 (very, very high). Helpfulness appears only in treatment. For crash tasks, confidence specifically concerns the selected fix; Appendix~\ref{app:user-study:tasks} gives the exact prompts. Available ratings from partial submissions are retained. Underlined task entries indicate treatment; plain entries indicate baseline. An asterisk marks a missing rating, and an em dash marks helpfulness in baseline, where the question is absent. The rightmost columns give each participant's median by version; the bottom rows give task medians and, at the right, medians over all available trial ratings. Medians exclude missing ratings.\par
\begingroup
\setlength{\tabcolsep}{3pt}
\begin{longtable}{@{}ll*{14}{r}@{}}
\caption{Participant and task ratings.}\label{tab:user-task-ratings}\\
\toprule
ID & Rating & T01 & T02 & T03 & T04 & T05 & T06 & T07 & T08 & T09 & T10 & T11 & T12 & Med. B & Med. T \\
\midrule
\endfirsthead
\multicolumn{16}{l}{\tablename~\thetable\ (continued)}\\
\toprule
ID & Rating & T01 & T02 & T03 & T04 & T05 & T06 & T07 & T08 & T09 & T10 & T11 & T12 & Med. B & Med. T \\
\midrule
\endhead
\midrule
\multicolumn{16}{r}{Continued on next page}\\
\endfoot
\bottomrule
\endlastfoot
P01 & Conf. & 2 & 4 & \underline{4} & \underline{4} & 5 & \underline{5} & 5 & \underline{5} & \underline{5} & 3 & \underline{5} & 2 & 3.5 & 5 \\
P01 & Help & --- & --- & \underline{3} & \underline{5} & --- & \underline{5} & --- & \underline{5} & \underline{5} & --- & \underline{5} & --- & --- & 5 \\
P01 & Effort & 5 & 7 & \underline{6} & \underline{7} & 2 & \underline{3} & 1 & \underline{1} & \underline{1} & 4 & \underline{6} & 7 & 4.5 & 4.5 \\
\addlinespace
P02 & Conf. & 4 & \underline{4} & * & \underline{2} & 2 & \underline{4} & 5 & 4 & \underline{5} & \underline{5} & \underline{4} & 4 & 4 & 4 \\
P02 & Help & --- & \underline{5} & --- & \underline{3} & --- & \underline{4} & --- & --- & \underline{5} & \underline{5} & \underline{5} & --- & --- & 5 \\
P02 & Effort & 6 & \underline{4} & * & \underline{8} & 4 & \underline{4} & 1 & 3 & \underline{2} & \underline{2} & \underline{3} & 3 & 3 & 3.5 \\
\addlinespace
P03 & Conf. & \underline{2} & \underline{5} & 1 & \underline{4} & 3 & 5 & \underline{4} & \underline{4} & 2 & 3 & * & \underline{3} & 3 & 4 \\
P03 & Help & \underline{4} & \underline{5} & --- & \underline{5} & --- & --- & \underline{4} & \underline{3} & --- & --- & --- & \underline{5} & --- & 4.5 \\
P03 & Effort & \underline{7} & \underline{3} & 8 & \underline{5} & 6 & 2 & \underline{2} & \underline{1} & 4 & 8 & * & \underline{8} & 6 & 4 \\
\addlinespace
P04 & Conf. & 3 & \underline{5} & 3 & \underline{5} & 4 & \underline{5} & 4 & \underline{5} & 3 & \underline{5} & \underline{4} & 4 & 3.5 & 5 \\
P04 & Help & --- & \underline{5} & --- & \underline{5} & --- & \underline{5} & --- & \underline{5} & --- & \underline{5} & \underline{5} & --- & --- & 5 \\
P04 & Effort & 5 & \underline{3} & 7 & \underline{3} & 6 & \underline{4} & 3 & \underline{2} & 6 & \underline{2} & \underline{4} & 5 & 5.5 & 3 \\
\addlinespace
P05 & Conf. & \underline{3} & 4 & \underline{2} & * & \underline{4} & 4 & 4 & \underline{4} & 3 & \underline{4} & \underline{3} & 1 & 4 & 3.5 \\
P05 & Help & \underline{5} & --- & \underline{3} & --- & \underline{5} & --- & --- & \underline{4} & --- & \underline{4} & \underline{4} & --- & --- & 4 \\
P05 & Effort & \underline{6} & 5 & \underline{7} & * & \underline{3} & 5 & 3 & \underline{2} & 5 & \underline{2} & \underline{7} & 9 & 5 & 4.5 \\
\addlinespace
P06 & Conf. & \underline{3} & 1 & 5 & \underline{2} & 3 & \underline{5} & \underline{5} & 3 & \underline{5} & 5 & 3 & \underline{4} & 3 & 4.5 \\
P06 & Help & \underline{5} & --- & --- & \underline{5} & --- & \underline{5} & \underline{5} & --- & \underline{5} & --- & --- & \underline{5} & --- & 5 \\
P06 & Effort & \underline{4} & 7 & 8 & \underline{5} & 7 & \underline{3} & \underline{3} & 2 & \underline{1} & 1 & 3 & \underline{3} & 5 & 3 \\
\addlinespace
P07 & Conf. & 3 & \underline{5} & \underline{2} & 2 & \underline{5} & 4 & \underline{4} & \underline{4} & 4 & 3 & \underline{4} & 2 & 3 & 4 \\
P07 & Help & --- & \underline{5} & \underline{5} & --- & \underline{5} & --- & \underline{4} & \underline{4} & --- & --- & \underline{5} & --- & --- & 5 \\
P07 & Effort & 6 & \underline{6} & \underline{7} & 8 & \underline{3} & 5 & \underline{5} & \underline{3} & 3 & 3 & \underline{6} & 8 & 5.5 & 5.5 \\
\addlinespace
P08 & Conf. & \underline{4} & 4 & \underline{*} & 5 & \underline{5} & 5 & 4 & 3 & \underline{5} & \underline{4} & 4 & \underline{2} & 4 & 4 \\
P08 & Help & \underline{4} & --- & \underline{*} & --- & \underline{5} & --- & --- & --- & \underline{3} & \underline{4} & --- & \underline{4} & --- & 4 \\
P08 & Effort & \underline{4} & 6 & \underline{*} & 2 & \underline{1} & 1 & 6 & 3 & \underline{3} & \underline{2} & 6 & \underline{5} & 4.5 & 3 \\
\addlinespace
P09 & Conf. & 1 & 5 & \underline{3} & 1 & \underline{5} & \underline{5} & \underline{4} & 3 & 3 & \underline{4} & 2 & \underline{3} & 2.5 & 4 \\
P09 & Help & --- & --- & \underline{5} & --- & \underline{5} & \underline{5} & \underline{5} & --- & --- & \underline{5} & --- & \underline{5} & --- & 5 \\
P09 & Effort & 7 & 3 & \underline{6} & 7 & \underline{2} & \underline{2} & \underline{3} & 5 & 5 & \underline{5} & 8 & \underline{7} & 6 & 4 \\
\addlinespace
P10 & Conf. & \underline{4} & \underline{3} & 2 & 2 & \underline{2} & 5 & \underline{4} & 4 & \underline{3} & 4 & 4 & \underline{2} & 4 & 3 \\
P10 & Help & \underline{4} & \underline{2} & --- & --- & \underline{2} & --- & \underline{4} & --- & \underline{3} & --- & --- & \underline{3} & --- & 3 \\
P10 & Effort & \underline{6} & \underline{6} & 7 & 4 & \underline{4} & 3 & \underline{3} & 3 & \underline{6} & 2 & 3 & \underline{5} & 3 & 5.5 \\
\midrule
Median B & Conf. & 3 & 4 & 2.5 & 2 & 3 & 5 & 4 & 3 & 3 & 3 & 3.5 & 2 & 3 &  \\
Median B & Effort & 6 & 6 & 7.5 & 5.5 & 6 & 3 & 3 & 3 & 5 & 3 & 4.5 & 7 & 5 &  \\
Median T & Conf. & 3 & 5 & 2.5 & 4 & 5 & 5 & 4 & 4 & 5 & 4 & 4 & 3 &  & 4 \\
Median T & Help & 4 & 5 & 4 & 5 & 5 & 5 & 4 & 4 & 5 & 5 & 5 & 5 &  & 5 \\
Median T & Effort & 6 & 4 & 6.5 & 5 & 3 & 3 & 3 & 2 & 2 & 2 & 6 & 5 &  & 4 \\
\end{longtable}
\endgroup
\noindent Table~\ref{tab:user-post-ratings} reports the closing questionnaire once per participant. Adoption ranges from 1 (definitely not) to 5 (definitely yes), and trust from 1 (not at all) to 5 (fully). Overall workload items range from 1 (very low) to 7 (very high), except perceived performance, whose prompt specifies 1 (perfect) and 7 (failure). The workload ratings cover the whole study.\par
\noindent Mental and Time denote mental demand and time pressure; Performance denotes perceived performance.\par
\begin{longtable}{@{}l*{7}{r}@{}}
\caption{Post-study ratings.}\label{tab:user-post-ratings}\\
\toprule
ID & Adoption & Trust & Mental & Time & Performance & Effort & Frustration \\
\midrule
\endfirsthead
\multicolumn{8}{l}{\tablename~\thetable\ (continued)}\\
\toprule
ID & Adoption & Trust & Mental & Time & Performance & Effort & Frustration \\
\midrule
\endhead
\midrule
\multicolumn{8}{r}{Continued on next page}\\
\endfoot
\bottomrule
\endlastfoot
P01 & 5 & 4 & 6 & 5 & 5 & 6 & 4 \\
P02 & 5 & 4 & 4 & 4 & 3 & 6 & 3 \\
P03 & 5 & 4 & 6 & 6 & 7 & 6 & 7 \\
P04 & 5 & 5 & 4 & 2 & 4 & 6 & 2 \\
P05 & 5 & 4 & 5 & 5 & 4 & 6 & 2 \\
P06 & 5 & 5 & 3 & 2 & 4 & 5 & 1 \\
P07 & 5 & 4 & 5 & 5 & 5 & 6 & 6 \\
P08 & 5 & 3 & 5 & 4 & 4 & 5 & 5 \\
P09 & 5 & 5 & 5 & 4 & 4 & 6 & 5 \\
P10 & 4 & 4 & 5 & 6 & 4 & 4 & 5 \\
\midrule
Median & 5 & 4 & 5 & 4.5 & 4 & 6 & 4.5 \\
\end{longtable}
\noindent For future similar tasks, 10 participants prefer the report, 0 prefer working without it, and 0 express no preference. These responses describe willingness to use the checker and perceived trust.\par
\endgroup

%% file: sections/appendix-threats.tex
\section{Extended Threats to Validity}
\label{sec:appendix-threats}

Section~\ref{sec:discussion:validity} summarizes the main threats to the
validity of our results. This appendix discusses them in full, across the
failure study (RQ1), the evaluation (RQ2--RQ4), and the user study (RQ5). It
groups the threats by type and states, for each, what the results do and do not
show.

\parh{Construct validity.} Construct validity asks whether our measures
capture the concepts they are meant to. Two of our measures depend on labels,
and for each we chose the most reliable labeling approach available at its
scale.
\begin{itemize}
\item \emph{Failure-study labels (RQ1).} The census labels all 10{,}134 reports,
  far more than manual labeling can cover. We therefore label them with an LLM
  classifier, following prior use of LLMs to annotate software artifacts and
  issue reports~\cite{ahmed2025,colavito2024}. To guard against the bias of a
  single model, four further models, three of them from other
  developers~\cite{verga2024}, relabeled a random subsample. They agree with the
  classifier at $\kappa=0.71$--$0.81$, and an adjudicated two-rater consensus
  lies within one percentage point of its prevalence. These checks establish
  reliability across models. In addition, the census counts only failures that
  users report, and users may not notice silent failures. Any undercount would
  thus make our silent fraction conservative.
\item \emph{Labels and ground truth (RQ2).} To avoid circularity, the labels
  come from public documentation and from cited conditions in ComfyUI's source,
  never from \tool or its schemas. Documentation can be wrong, so we score
  detection only on faults that either execution confirms or confirmed in our
  manual audit. A candidate fault that raises an error confirms itself. A
  candidate fault that completes needs a judgment of its output, which we make
  in two conservative steps. We use an image oracle to judge the outputs of the
  candidate faults, and confirm the outputs by visual that the image oracle
  flags as potentially faulty. We then review the remaining outputs. Candidate
  faults that neither step confirms, such as an ignored LoRA, do not count in
  our measurements in RQ2.
\end{itemize}

\parh{Internal validity.} Internal validity asks whether the evidence itself,
and not the assumptions and settings of our experiments, supports our
conclusions. For each conclusion, we chose assumptions that can be audited and
settings that do not favor \tool.
\begin{itemize}
\item \emph{Faithfulness proofs.} As is standard for proofs about systems built
  on large libraries, the proofs (Appendix~\ref{app:faithful}) rely on trusted
  lemmas about PyTorch operators and ComfyUI's model code. We keep these lemmas
  auditable: each states one narrow behavior and is justified by PyTorch
  semantics or by cited source lines of ComfyUI~0.27.0. The proofs also assume
  that the model database identifies files correctly. A later ComfyUI release
  requires re-checking only the schemas whose cited code changes.
\item \emph{Case selection (RQ4).} RQ4 does not test \tool on all reported
  failures, but on the 108 of 753 whose graphs we can rebuild and whose failures
  execution reproduces. None of the curation steps refers to \tool or its
  verdicts. They ask whether a report states its components, whether the graph
  can be built and run, and whether the failure recurs. The steps deliberately
  \emph{keep} failures that lie outside the type system and graphs with nodes
  that have no schema.
\item \emph{Timing.} We compare checking time with the execution time of the
  same graphs, run with each workflow's own settings
  (Appendix~\ref{app:eval:injection}). Each run loads its models, which an
  interactive session may keep in memory. Even the fastest execution, about two
  seconds, exceeds the median checking time by more than two orders of
  magnitude, so the comparison does not depend on model loading.
\end{itemize}

\parh{External validity.} Several results may not generalize beyond the
evaluated settings.
\begin{itemize}
\item \emph{Injected faults.} The hosts are official workflows that run with
  public models within our GPU budget, and graphs derived from one host are
  correlated. Community workflows with custom nodes may expose other faults,
  and the quantization and modality families have few confirmed faults.
\item \emph{Missing schemas.} The experiment removes classes uniformly at
  random. Real schema libraries omit classes selectively, which may affect
  detection differently. The experiment also does not isolate individual
  solver components.
\item \emph{Scale.} Larger graphs and constraint systems require further
  evaluation.
\item \emph{Venues and workflows.} Reports from the ComfyUI core repository may
  contain a different mix of failures from those encountered across community
  workflows. False-alarm rates on arbitrary community workflows remain
  unmeasured. The compatibility dimensions arise from the underlying models, but
  their prevalence and \tool's coverage may differ in other workflow systems.
\item \emph{User study.} The study covers 10 participants on 12 read-only
  tasks. Usability while users edit and execute their own workflows, and the
  effectiveness of automatically suggested repairs, remain to be evaluated.
\end{itemize}